\documentclass[a4paper,11pt,oneside]{amsart} % defaults: (US) letterpaper, 10pt

\usepackage{geometry}
	\usepackage{fullpage}

\usepackage{mathtools} % extension of amsmath that also fixes two bugs
\numberwithin{equation}{section} % amsmath option

\usepackage{amssymb}
\usepackage{bbm} % for \mathbbm{1}

\usepackage{bm} % cf https://tex.stackexchange.com/a/10643
	\newcommand{\vect}[1]{\bm{{#1}}}

\usepackage[utf8]{inputenc}
\usepackage[T1]{fontenc}
\usepackage{slantsc} % italics small caps
\usepackage{lmodern}
\AtBeginDocument{% cf https://tex.stackexchange.com/a/284343
	\DeclareFontShape{T1}{lmr}{m}{scit}{<->ssub*lmr/m/scsl}{}%
}

\usepackage[british]{babel} 

\usepackage[dvipsnames]{xcolor} % for colours used for hyperref below
\usepackage{graphicx}
\usepackage{tikz}
	\usetikzlibrary{calc} % coordinate calculation with TikZ
	\usetikzlibrary{matrix,arrows,decorations.pathmorphing} % commutative diagrams in tikz

\usepackage[justification=centering]{subcaption} % omitting parentheses is done by loading the package with the option [labelformat=simple]

\usepackage[
	ocgcolorlinks, % cf tex.stackexchange.com/questions/4425/
	linkcolor=BlueViolet, citecolor=OliveGreen, urlcolor=RawSienna, filecolor=Sepia,
	hyperfootnotes=false,
	]{hyperref} % after most other packages; ocgcolorlinks cf tex.stackexchange.com/questions/4425/

\usepackage{multirow} % \multirow, \multicol
\usepackage{booktabs} % nicer tables

\usepackage{enumitem}

\usepackage{amsthm} % after hyperref since that alters the \newtheorem 

\newtheorem{theorem}{Theorem}[section]
\newtheorem{proposition}[theorem]{Proposition}
\newtheorem{lemma}[theorem]{Lemma}
\newtheorem*{corollary}{Corollary}

\theoremstyle{definition}

\newtheorem*{definition}{Definition}
\newtheorem*{remark}{Remark}
\newtheorem*{remarks}{Remarks}

\usepackage[alphabetic]{amsrefs}
\let\originalleft\left
\let\originalright\right
\renewcommand{\left}{\mathopen{}\mathclose\bgroup\originalleft}
\renewcommand{\right}{\aftergroup\egroup\originalright}

\DeclareMathOperator{\E}{e}
\newcommand{\I}{\mathrm{i}}

\makeatletter
	\newcommand\niton{\mathrel{\m@th\mathpalette\canc@l\owns}}
	\newcommand\canc@l[2]{{\ooalign{$\hfil#1/\mkern1mu\hfil$\crcr$#1#2$}}}
\makeatother

\DeclarePairedDelimiter{\bra}{\langle}{\rvert}
\DeclarePairedDelimiter{\ket}{\lvert}{\rangle}
\DeclarePairedDelimiterX{\braket}[2]{\langle}{\rangle}{#1\vert#2}

\DeclarePairedDelimiter{\cbra}{\langle\!\langle}{\rvert}
\DeclarePairedDelimiter{\cket}{\lvert}{\rangle\!\rangle}
\DeclarePairedDelimiterX{\cbraket}[2]{\langle\!\langle}{\rangle}{#1\vert#2}
\DeclarePairedDelimiterX{\bracket}[2]{\langle}{\rangle\!\rangle}{#1\vert#2}
\DeclarePairedDelimiterX{\cbracket}[2]{\langle\!\langle}{\rangle\!\rangle}{#1\vert#2}

\newcommand{\To}{\cdots\mspace{-1mu}}

\DeclareMathOperator{\diag}{diag}

\DeclareMathOperator{\id}{\mathbbm{1}} %{\mathbf{1}}
\DeclareMathOperator{\tr}{tr}
\DeclareMathOperator{\qdet}{qdet}

\DeclareMathOperator{\ev}{ev_\omega}

\newcommand{\checkedMma}{}

\begin{document}

\title[Spin-Ruijsenaars, \textsf{q}-deformed Haldane--Shastry and Macdonald polynomials]{Spin-Ruijsenaars, \textsf{q}-deformed Haldane--Shastry \\ and Macdonald polynomials}

\author{Jules Lamers}
\address{School of Mathematics and Statistics, University of Melbourne, Vic 3010, Australia} 
\email{jules.l@unimelb.edu.au}

\author{Vincent Pasquier}
\address{Institut de Physique Théorique, Université Paris Saclay, CEA, CNRS, F-91191 Gif-sur-Yvette, France}
\email{vincent.pasquier@ipht.fr}

\author{Didina Serban}
\address{Institut de Physique Théorique, Université Paris Saclay, CEA, CNRS, F-91191 Gif-sur-Yvette, France}
\email{didina.serban@ipht.fr}

\date{\today}
	
\begin{abstract}
We study the \textsf{q}-analogue of the Haldane--Shastry model, a partially isotropic (\textsc{xxz}-like) long-range spin chain that by construction enjoys quantum-affine (really: quantum-loop) symmetries at finite system size.

We derive the pairwise form of the Hamiltonian, found by one of us building on work of D.~Uglov, via `freezing' from the affine Hecke algebra. To this end we first obtain explicit expressions for the spin-Macdonald operators of the (trigonometric) spin-Ruijsenaars model. Through freezing these give rise to the higher Hamiltonians of the spin chain, including another Hamiltonian of the opposite `chirality'. The sum of the two chiral Hamiltonians has a real spectrum also when $|\mathsf{q}|=1$, so in particular when \textsf{q} is a root of unity. 

For generic $\mathsf{q}$ the eigenspaces are known to be labelled by `motifs'. We clarify the relation between these patterns and the corresponding degeneracies (multiplicities) in the crystal limit $\textsf{q}\to \infty$. For each motif we obtain an explicit expression for the exact eigenvector, valid for generic~\textsf{q}, that has (`pseudo' or `\textit{l}-') highest weight in the sense that, in terms of the operators from the monodromy matrix, it is an eigenvector of \textit{A} and~\textit{D} and annihilated by \textit{C}. It has a simple component featuring the `symmetric square' of the \textsf{q}-$\mspace{-1mu}$Vandermonde polynomial times a Macdonald polynomial\,---\,or more precisely its quantum spherical zonal special case. All other components of the eigenvector are obtained from this through the action of the Hecke algebra, followed by `evaluation' of the variables to roots of unity. We prove that our vectors have highest weight upon evaluation. Our description of the exact spectrum is complete.

The entire model, including the quantum-loop action, can be reformulated in terms of polynomials. Our main tools are the \textit{Y}$\!$-operators from the affine Hecke algebra. From a more mathematical perspective the key step in our diagonalisation is as follows. We show that on a subspace of suitable polynomials the first~$M$ `classical' (i.e.\ no difference part) \textit{Y}$\!$-operators in $N$ variables reduce, upon evaluation as above, to \textit{Y}$\!$-operators in $M$ variables with parameters at the quantum zonal spherical point.
\end{abstract}

\maketitle	

\DeclareRobustCommand{\SkipTocEntry}[5]{} % cf https://www.ams.org/faq?faq_id=238

\setcounter{tocdepth}{3}    
\tableofcontents

\section{Overview and main results} \label{s:intro}
\noindent 
The Haldane--Shastry spin chain \cite{Hal_88,Sha_88} and its partially isotropic (\textsc{xxz}-like) counterpart~\cite{BG+_93,Ugl_95u,Lam_18} are quantum-integrable spin chains with long-range interactions that are such that the model's spectrum admits an exact description in closed form\,---\,there are no Bethe-type equations that remain to be solved. We begin with a guided tour to introduce these models and state our results.

In \textsection\ref{s:intro_HS} we recall the salient features of the ordinary Haldane--Shastry spin chain. The reader who is familiar with the isotropic case may wish to glance at the three definitions in \textsection\ref{s:intro_HS} before skipping to the partially anisotropic generalisation in \textsection\ref{s:intro_qHS}, where we introduce the model and give an overview of its remarkable properties, many of which are new results. In \textsection\ref{s:intro_spin-RM} we preview the plan of our derivations in the main text. These proofs exploit a connection with a more general model, the spin-version of the (quantum trigonometric) Ruijsenaars model, whose spin-Macdonald operators give rise to the spin chain by `freezing'. 

In this tour we follow~\cite{BG+_93,Ugl_95u,Lam_18} and denote the deformation (anisotropy) parameter by~$\mathsf{q}$ as usual for quantum groups. We use $\mathsf{p}$ for the second parameter of Macdonald polynomials. From \textsection\ref{s:setup} onwards we'll switch to the notation $t^{1/2} = \mathsf{q}$ and $q = \mathsf{p}$, standard in the world of Macdonald polynomials and affine Hecke algebras~\cite{Mac_95,Mac_98,Che_05}.

\subsection{Recap of the isotropic case} \label{s:intro_HS} 
In short, the ordinary (isotropic) Haldane--Shastry spin chain~\cite{Hal_88,Sha_88} is a physically motivated quantum spin chain\,---\,e.g.\ serving as a toy model for the fractional quantum Hall effect\,---\,with many remarkable properties: it
\begin{enumerate}
	\item[i.] (\emph{abelian symmetries}) belongs to a family of commuting operators~\cite{Ino_90,HH+_92,BG+_93,TH_95}, each of which
	\item[ii.] (\emph{nonabelian symmetries}) commutes with an action of the Yangian \cite{HH+_92,BG+_93};
	\item[iii.] (\emph{explicit eigenvectors}) has eigenvectors that are determined by a symmetric polynomial~\cite{BG+_93}, which for Yangian highest-weight \cite{BPS_95a} eigenvectors is known explicitly and involves a Jack polynomial~\cite{Hal_91a}.
\end{enumerate}
Each is a hallmark of quantum integrability: (i)~a tower of higher Hamiltonians, (ii)~an underlying quantum-algebraic structure, and (iii)~exact solvability. Let us review these three properties of the Haldane--Shastry spin chain, introducing some useful notation along the way.

\subsubsection{Abelian symmetries} In this work we focus on rank one; we will address higher rank, cf.~\cite{Lam_18}, elsewhere. Consider a chain with $N$ spin-$1/2$ sites: the spin-chain Hilbert space is $\mathcal{H} \coloneqq (\mathbb{C}^2)^{\otimes N}$ where $\mathbb{C}^2 = \mathbb{C} \, \ket{\uparrow} \oplus \mathbb{C} \, \ket{\downarrow}$. Write $P_{ij}$ for the permutation of the $i$th and $j$th factors of $\mathcal{H}$, so $P_{ij} = (1+\vec{\sigma}_i \cdot \vec{\sigma}_j)/2$
\checkedMma
with $\vec{\sigma} = (\sigma^x,\sigma^y,\sigma^z)$ the Pauli matrices. The Hamiltonian is
\begin{equation} \label{eq:HS_pre}
	H^\textsc{hs} = \sum_{i<j}^N \frac{1-P_{ij}}{4\,\sin^2[\pi\,(i-j)/N]} \, .
	\checkedMma
\end{equation}
This operator is positive:
\checkedMma
$(-)H^\textsc{hs}$ models an (anti)ferromagnet. Following Uglov we introduce

\begin{definition}[\cite{Ugl_95u}]
Let $\omega \coloneqq \E^{2\pi\I/N} \in \mathbb{C}^\times \coloneqq \mathbb{C}\setminus\{0\}$ be the primitive $N$th root of unity. Define the \emph{evaluation}
\begin{equation} \label{eq:ev}
	\ev \colon z_j \longmapsto \omega^j = \E^{2\pi\I j/N}
\end{equation}
of $z_1,\To,z_N$ at the corresponding $N$th roots of unity. \emph{On shell}, i.e.\ after evaluation, we think of $z_j$ as (the multiplicative notation for) the position of site~$j$ of the chain, viewed as embedded in the unit circle $S^1 \subseteq \mathbb{C}$. We will refer to the $z_j$ as \emph{coordinates}.
\end{definition}
 
With this notation \eqref{eq:HS_pre} can be rewritten as
\begin{equation} \label{eq:HS}
	H^\textsc{hs} = \ev \widetilde{H}^\textsc{hs} \, , \qquad 
	\widetilde{H}^\textsc{hs} = \sum_{i<j}^N V^\textsc{hs}(z_i,z_j) \, (1-P_{ij}) \, , \qquad V^\textsc{hs}(z_i,z_j) = \frac{-z_i\,z_j}{(z_i-z_j)^2} \, .
	\checkedMma
\end{equation}
The pair potential has a neat geometric interpretation: $\ev V^\textsc{hs}(z_i,z_j) = 1/d^2$,
\checkedMma
where $d = 2 \, \left|\sin(\pi(i-j)/N)\right|$ is the chord distance between sites $i$ and~$j$, cf.~Figure~\ref{fg:pot}.

The Hamiltonian \eqref{eq:HS} is a member of a hierarchy of higher Hamiltonians that pairwise commute~\cite{BG+_93} and, in principle, can be constructed explicitly and systematically~\cite{TH_95}. The first few of these abelian symmetries, apart from \eqref{eq:HS} and the translation operator
\begin{equation} \label{eq:HS_translation}
	G^\textsc{hs} \coloneqq P_{N,N-1} \cdots P_{12} \, ,
\end{equation} 
are given in \cite{Ino_90,HH+_92,TH_95}. 

The spectrum of the Haldane--Shastry spin chain is particularly simple. The joint eigenspaces of the abelian symmetries are labelled by simple combinatorial patterns.
\begin{definition} [\cite{HH+_92}]
A \emph{motif} (though `$N$-site $\mathfrak{sl}_2$ motif' would be more precise) is a sequence in $\{1,\To,N-1\}$ increasing with steps of at least two. As in \cite{Ugl_95u} we denote the set of all motifs by
\begin{equation} \label{eq:motif}
	\mathcal{M}_N \coloneqq \Bigl\{\, \mu \subset \{1,\To,N-1\} \Bigm| \mu_{m+1} > \mu_m + 1 \,\Bigr\} \, .
\end{equation}
Let us define the \emph{length}~$\ell(\mu)$ of $\mu$ to be the number of parts~$\mu_m$. We further write $|\mu| \coloneqq \sum_m \mu_m$.
\end{definition} 

Denote the empty motif by~$0$. For example,
\checkedMma 
$\mathcal{M}_2 = \{0,(1)\}$, $\mathcal{M}_3 = \{0,(1),(2)\}$ and $\mathcal{M}_4 = \{0,(1),(2),(3),(1,3)\}$. Motifs are stable under increase of the system size, $\mathcal{M}_{N-1} \subset \mathcal{M}_N$. Conditioning on whether $N-1 \in \mu$ yields a recursion $\mathcal{M}_N \cong \mathcal{M}_{N-1} \, \amalg \, \mathcal{M}_{N-2}$ (disjoint union), so the number of motifs forms a Fibonacci sequence with offset one: $\#\mathcal{M}_N = \text{Fib}_{N+1}$. 
\checkedMma

As for any homogeneous (translationally invariant) spin chain the momentum~$p^\textsc{hs}$ is defined such that \eqref{eq:HS_translation} has eigenvalue $\E^{\I \, p^\textsc{hs}}$.
For the eigenspace labelled by $\mu \in \mathcal{M}_N$ it is given by
\begin{align} \label{eq:HS_mtm}
	p^\textsc{hs}(\mu) & = \sum_{m=1}^M p^\textsc{hs}_m \ \ \mathrm{mod} \, 2\pi \, , 
	&& \ \, p^\textsc{hs}_m = \frac{2\pi}{N} \, \mu_m \, .
\intertext{The energy is (strictly) additive too, with a quadratic dispersion relation:}
	\label{eq:HS_dispersion}
	E^\textsc{hs}(\mu) & = \sum_{m=1}^M \varepsilon^\textsc{hs}(\mu_m) \, , \quad
	&& \varepsilon^\textsc{hs}(\mu_m) = \frac{1}{2} \, \mu_m \,(N-\mu_m) = \frac{N^2}{8\,\pi^2} \, p^\textsc{hs}_m \, \bigl( 2\pi - p^\textsc{hs}_m \bigr) \, .
\end{align}
Thanks to additivity all energies are half integral, i.e.\ lie in $\tfrac12 \mathbb{Z}_{\geq 0}$. The $\mu_m$ can be seen as the `Bethe quantum numbers', or, up to a factor, quasimomenta~$p^\textsc{hs}_m$. Indeed, $\mu_m$ parametrises the contribution of the $m$th magnon not just to the momentum~\eqref{eq:HS_mtm} but, by \eqref{eq:HS_dispersion}, also to the energy. We stress that, in view of the definition~\eqref{eq:motif} of motifs these energies are \emph{strictly} additive: the quasimomenta~$p^\textsc{hs}_m$ are all real\,---\,there are only `1-strings'\,---\,and there is no interaction (bound-state) energy. The physical picture is that of a gas of anyons: free quasiparticles that interact through their fractional (exclusion) statistics only~\cite{Hal_91a,Hal_91b}. See also \cite{Hal_94} and \cite{Pol_99}.

The spectrum is highly degenerate~\cite{Hal_88}. In part this is because the eigenvalues~\eqref{eq:HS_dispersion} may be (`accidentally'~\cite{FG_15}) degenerate, e.g.\ for $\mu=(1,3)$ and $\mu=(5,7)$ at $N=8$. Another reason is the presence of a large nonabelian symmetry algebra. 

\subsubsection{Nonabelian symmetries} The Hamiltonian \eqref{eq:HS} is clearly isotropic, i.e.\ invariant under $\mathfrak{U}^\textsc{hs} \coloneqq U \mathfrak{sl}_2$, the universal enveloping algebra of $\mathfrak{sl}_2 = (\mathfrak{su}_2)_\mathbb{C}$. The latter acts as usual: if $\sigma^\pm \coloneqq (\sigma^x \pm \I \, \sigma^y)/2$ then
\begin{equation} \label{eq:sl_2}
	S^\pm \coloneqq \sum_{i=1}^N \sigma^\pm_i \, , \quad S^z \coloneqq \frac{1}{2} \sum_{i=1}^N \sigma^z_i \, , \qquad [S^z,S^\pm] = \pm S^\pm \, , \quad [S^+,S^-] = 2 \, S^z \, .
\checkedMma 
\end{equation}
For the Haldane--Shastry spin chain this symmetry is enhanced to the \emph{Yangian} $\widehat{\mathfrak{U}}^\textsc{hs} \coloneqq Y\mspace{-2mu}(\mathfrak{sl}_2)$, with additional generators~\cite{HH+_92} (note that $\ev \I \, (z_i + z_j)/(z_i - z_j) = \cot\bigl(\pi(i-j)/N\bigr)$)
\checkedMma
\begin{equation} \label{eq:Yangian_gens}
	\begin{aligned}
		Q^\pm & = \ev \widetilde{Q}^\pm \, , \qquad && \widetilde{Q}^\pm = \mp\frac{\I}{2} \sum_{i<j}^N \frac{z_i + z_j}{z_i - z_j} \, (\sigma^\pm_i \, \sigma^z_j - \sigma^z_i \, \sigma^\pm_j ) \, , \\
		Q^z & = \ev \widetilde{Q}^z \, , \qquad && \, \widetilde{Q}^z = \hphantom{\pm}\frac{\I}{2} \sum_{i<j}^N \frac{z_i + z_j}{z_i - z_j} \, (\sigma^+_i \, \sigma^-_j - \sigma^-_i \, \sigma^+_j ) \, .
	\end{aligned}
	\checkedMma
\end{equation}
Even off shell these live in the adjoint representation of $\mathfrak{sl}_2$,
\begin{subequations} \label{eq:Yangian_rels}
	\begin{gather}
	[S^z , \widetilde{Q}^\pm] = \pm \widetilde{Q}^\pm \, , \qquad [S^\pm, \widetilde{Q}^\mp] = \pm 2 \, \widetilde{Q}^z \, , \qquad [S^\pm , \widetilde{Q}^z] = \mp \widetilde{Q}^\pm \, ,
\checkedMma
\intertext{and obey the Serre relation}
	[\widetilde{Q}^z,[\widetilde{Q}^+,\widetilde{Q}^-]] = {-}(S^+ \,\widetilde{Q}^- - \widetilde{Q}^+ \, S^-)\,S^z \, .
	\end{gather}
\end{subequations}
On shell, \eqref{eq:Yangian_gens} moreover commute with the abelian symmetries, including $G^\textsc{hs}$ and $H^\textsc{hs}$. 

The upshot is that the Hilbert space decomposes as
\begin{equation} \label{eq:HS_motif_decomp}
	\mathcal{H} = \bigoplus_{\mu \,\in\, \mathcal{M}_N} \!\!\! \mathcal{H}^{\mu,\textsc{hs}} \, .
\end{equation}
Each $\mathcal{H}^{\mu,\textsc{hs}}$ is a joint eigenspace for the abelian symmetries, with energy and momentum~\eqref{eq:HS_dispersion}, as well as an irreducible Yangian module with known Drinfeld polynomial~\cite{BG+_93}. It contains a unique (up to rescaling) vector $\ket{\mu}^{\mspace{-1mu}\textsc{hs}} \in \mathcal{H}^{\mu,\textsc{hs}}$ with \emph{Yangian highest weight}, i.e.\ $S^+\,\ket{\mu}^{\mspace{-1mu}\textsc{hs}} = Q^+ \ket{\mu}^{\mspace{-1mu}\textsc{hs}} = 0$. Conversely, all of $\mathcal{H}^{\mu,\textsc{hs}}$ is generated by the $\widehat{\mathfrak{U}}^\textsc{hs}$-action on $\ket{\mu}^{\mspace{-1mu}\textsc{hs}}$. This structure of $\mathcal{H}$ is illustrated in Figure~\ref{fg:N=6}. The vector $\ket{\mu}^{\mspace{-1mu}\textsc{hs}}$ can be written down in closed form, as follows.

\begin{figure}[h]
	\centering
	\begin{tikzpicture}[scale=0.5]
		\draw[->] (-1,-3.1) -- (-1,3.3) node[above] {$\,S^z$};
		\foreach \y in {-3,...,3} \draw (-1,\y) -- (-1.1,\y);
		\node at (-2,3.3) [above] {$M$};
		\foreach \M in {0,...,6} \node at (-2,3-\M) {$\M$};
		\foreach \y in {-3,...,3} \fill[black] (0,\y) circle (.2);
		\foreach \x in {1,2,4,7,9} \foreach \y in {-2,...,2} \fill[black] (\x,\y) circle (.2);
		\foreach \x in {3,5,8,10,11,13,14,16,18} \foreach \y in {-1,...,1} \fill[black] (\x,\y) circle (.2);
		\foreach \x in {6,12,15,17,19} \fill[black] (\x,0) circle (.2);
		\draw (0,3) -- (0,-3);
		\foreach \x in {1,2,4,7,9} \draw (\x,2) -- (\x,-2);
		\foreach \x in {3,5,8,10,11,13,14,16,18} \draw (\x,1) -- (\x,-1);
		\foreach \x in {2,4,7} {
			\foreach \y in {0,...,2} \draw[dotted] (\x,\y) -- (\x+1,\y-1);
			\foreach \y in {-2,...,0} \draw[dotted] (\x,\y) -- (\x+1,\y+1);
		};
		\foreach \x in {5,11,14,16} {
			\draw[dotted] (\x,1) -- (\x+1,0);
			\draw[dotted] (\x,-1) -- (\x+1,0);
		};
		\node at (0,3) [above] {$\scriptstyle0$};
		\node at (1,2) [above] {$\scriptstyle(1)$};
		\node at (2,2) [above] {$\scriptstyle(2)$};
		\node at (4,2) [above] {$\scriptstyle(3)$};
		\node at (7,2) [above] {$\scriptstyle(4)$};
		\node at (9,2) [above] {$\scriptstyle(5)$};
		\node at (10,1) [above] {$\scriptstyle(1,3)$};
		\node at (11,1) [above] {$\scriptstyle\ \ \ (1,4)$};
		\node at (13,1) [above] {$\scriptstyle(1,5)$};
		\node at (14,1) [above] {$\scriptstyle\ \ \ (2,4)$};
		\node at (16,1) [above] {$\scriptstyle(2,5)$};
		\node at (18,1) [above] {$\scriptstyle(3,5)$};
		\node at (19,0) [above] {$\scriptstyle \ \ \ \ \ (1,3,5)$};
		\draw[lightgray] (1,-4.45) -- (1,-4.85) -- (9,-4.85) -- (9,-4.45);
		\draw[lightgray] (2.5,-4.45) -- (2.5,-4.7) -- (7.5,-4.7) -- (7.5,-4.45);
		\draw[lightgray] (10,-4.45) -- (10,-4.85) -- (18,-4.85) -- (18,-4.45);
		\draw[lightgray] (11.5,-4.45) -- (11.5,-4.7) -- (16.5,-4.7) -- (16.5,-4.45);
		\node at (-1,-3.8) {$p^\textsc{hs}$};
		\node at (0,-3.8) {$\scriptstyle 0$};
		\node at (1,-3.8) {$\scriptstyle \tfrac{\pi}{3}$};
		\node at (2.5,-3.8) {$\scriptstyle \tfrac{2\pi}{3}$};
		\node at (5,-3.8) {$\scriptstyle \pi$};
		\node at (7.5,-3.8) {$\scriptstyle \tfrac{4\pi}{3}$};
		\node at (9,-3.8) {$\scriptstyle \tfrac{5\pi}{3}$};
		\node at (10,-3.8) {$\scriptstyle \tfrac{4\pi}{3}$};
		\node at (11.5,-3.8) {$\scriptstyle \tfrac{5\pi}{3}$};
		\node at (13,-3.8) {$\scriptstyle 0$};
		\node at (14.5,-3.8) {$\scriptstyle 0$};
		\node at (16.5,-3.8) {$\scriptstyle \tfrac{\pi}{3}$};
		\node at (18,-3.8) {$\scriptstyle \tfrac{2\pi}{3}$};
		\node at (19,-3.8) {$\scriptstyle \pi$};
	\end{tikzpicture}
	\caption{Schematic picture of the structure of the Hilbert space $\mathcal{H}$ for $N=6$. Each 
	% of the $2^6 = 7 + 5{\times}5 + 9{\times}3 +5{\times}1$ dots represents 
	dot represents an eigenvector of the abelian symmetries. The vertical axis records its $S^z$, equal to $3-M$ for $\mathcal{H}_M$. The $\widehat{\mathfrak{U}}^\textsc{hs}$ highest-weight $H^\textsc{hs}$-eigenvectors $\ket{\mu}^{\mspace{-1mu}\textsc{hs}}$ are labelled by their motif. Vertical lines connect vectors in an $\mathfrak{U}^\textsc{hs}$-irrep, which are combined by dotted lines into $\widehat{\mathfrak{U}}^\textsc{hs}$-irreps $\mathcal{H}^{\mu,\textsc{hs}}$. The value of the momentum~$p^\textsc{hs}$ is indicated below each irrep, where we have also linked parity-conjugate pairs with opposite momentum and mirror-image motifs. (In the \textsf{q}-deformed case the picture is the same: just drop the superscripts `\textsc{hs}'.)}
	\label{fg:N=6}
\end{figure}
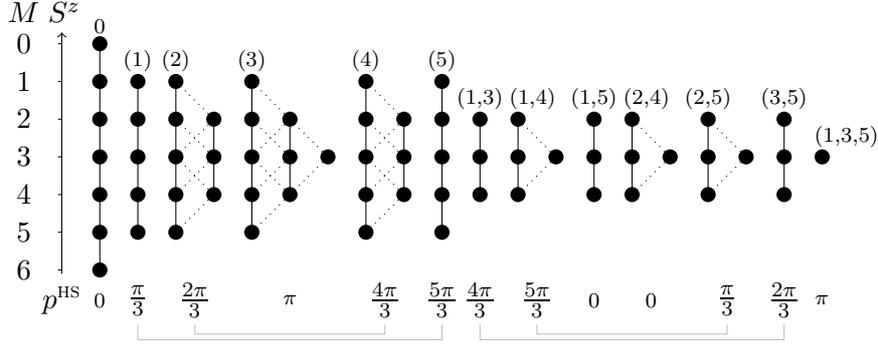

\subsubsection{Explicit eigenvectors}
The abelian symmetries preserve the decomposition
\begin{equation} \label{eq:weight_decomp}
	\mathcal{H} = \bigoplus_{M=0}^N \! \mathcal{H}_M \, , \qquad \mathcal{H}_M \coloneqq \ker\bigl[S^z - \bigl(\tfrac{1}{2} N-M\bigr)\bigr] \, .
\end{equation}
Any vector in the \emph{$M$-particle sector} (weight space)~$\mathcal{H}_M$ can be written via the coordinate basis:
\begin{equation} \label{eq:coord_basis}
	\!\! \sum_{i_1<\cdots<i_M}^N \!\!\!\!\! \Psi(i_1,\To,i_M) \, \cket{i_1,\cdots,i_M} \, , \qquad 
	\cket{i_1,\To,i_M} \coloneqq \sigma^-_{i_1} \cdots \sigma^-_{i_M} \, \ket{\uparrow\cdots\uparrow} \, .
\end{equation}
Thus, $\cket{\varnothing} = \ket{\uparrow\cdots\uparrow} \in \mathcal{H}_0$ is the pseudovacuum, while $\cket{i}\in \mathcal{H}_1$ has a~$\downarrow$ at site~$i$, and so on. Property~(iii) of the Haldane--Shastry spin chain comprises two statements. Firstly, every $H^\textsc{hs}$-eigenvector in $\mathcal{H}_M$ is completely determined by some symmetric polynomial $\widetilde{\Psi}^\textsc{hs}(z_1,\To,z_M)$ via \cite{BG+_93}
\begin{equation} \label{eq:HS_vec_ev}
	\Psi^\textsc{hs}(i_1,\To,i_M) = \cbraket{i_1,\To,i_M}{\Psi}^\textsc{hs} = \ev \widetilde{\Psi}^\textsc{hs}(z_{i_1},\To,z_{i_M}) \, .
\end{equation}
Secondly, for the $\widehat{\mathfrak{U}}^\textsc{hs}$ highest-weight eigenvectors these polynomials take an elegant form.

\begin{definition} Recall that a \emph{partition} $\nu = (\nu_1 \geq \nu_2 \geq \dots \geq 0)$ is a weakly decreasing sequence of integers, with \emph{length} $\ell(\nu)$ the number of nonzero parts. There is a length-preserving bijection from $\mathcal{M}_N$ to the set of partitions with $\nu_1 \leq N-2\,\ell(\nu)+1$ (see also Figure~\ref{fg:motifs_vs_partitions}): at length~$M$ set
\checkedMma
\begin{subequations} \label{eq:motifs_vs_partitions}
	\begin{gather}
	\nu_m \coloneqq \mu_{M-m+1} - 2\,(M-m) \, , \qquad 1\leq m\leq M \, .
\checkedMma
\intertext{If $\delta_M \coloneqq (M-1,M-2,\cdots)$ denotes the staircase partition of length~$M-1$, and $\mu^+$ the partition obtained from $\mu\in\mathcal{M}_N$ by reversal, this relation takes the succinct form}
	\nu + 2 \, \delta_M = \mu^+ \, , 
\end{gather}
\end{subequations}
where addition and scalar multiplication are pointwise. 
\end{definition}

With this notation the (unnormalised) wave function \eqref{eq:HS_vec_ev} of $\ket{\mu}^\textsc{hs}$ is determined by the polynomial \cite{Hal_91a,BG+_93}
\begin{equation} \label{eq:HS_polynomial}
	\widetilde{\Psi}^\textsc{hs}_\nu(z_1,\To,z_M) = \prod_{m<n}^M \! (z_m-z_n)^2 \ P_\nu^{(1/2)}(z_1,\To,z_M) \, .
\end{equation}
Here $P_\nu^{(\alpha)}$ is a Jack polynomial~\cite{Jac_70} with parameter~$\alpha = k^{-1}$ related to the coupling $k\,(k-1)$ of the (trigonometric quantum) Calogero--Sutherland model~\cite{Sut_71,Sut_72}. These symmetric polynomials are studied extensively in the literature, see e.g.~\cite{Sta_89,Mac_95}. They play an important role in~\cite{Mat_92}, and appear for the fractional quantum Hall effect~\cite{KP_07,BH_08}. If $\alpha = 1/2$, as in \eqref{eq:HS_polynomial}, one gets \emph{zonal spherical polynomials}, see e.g.\ \textsection{}VII.6 in \cite{Mac_95}. (For comparison: $\alpha=1$ gives Schur and $\alpha=2$ zonal polynomials; cf.~Figure~\ref{fg:Macdonalds} on p.\,\pageref{fg:Macdonalds}.)

\begin{remarks} 
\textbf{i.}~Note that $\ell(\nu)=M$ means that $\nu_M \geq 1$ and $\nu_{M+1}=0$, so $\nu_m = \bar{\nu}_m + 1$ for some partition $\bar{\nu}$ with $\ell(\bar{\nu})\leq M$ (see again Figure~\ref{fg:motifs_vs_partitions}). Jack polynomials have the property
\begin{equation} \label{eq:Jack_shift_property}
	P_\nu^{(\alpha)}(z_1,\To,z_M) = z_1 \cdots z_M \, P_{\mspace{1mu}\bar{\nu}}^{(\alpha)}(z_1,\To,z_M) \, , \qquad \nu_m = \bar{\nu}_m + 1 \, .
\end{equation}
In the literature on the Haldane--Shastry model this relation is often used to extract an explicit centre-of-mass factor $z_1 \cdots z_M$ and end up with a polynomial associated to~$\bar{\nu}$ as on the right-hand side of \eqref{eq:Jack_shift_property}. This factor (or, equivalently, the condition $\ell(\nu)=M$) ensures that the resulting eigenvector has Yangian highest-weight on shell~\cite{BPS_95a}. 

\textbf{ii.}~The relation \eqref{eq:motifs_vs_partitions} has the following origin.
Write $\vect{z}^\nu = z_1^{\nu_1} \cdots z_N^{\nu_N}$, appending zeros to $\nu$ if necessary to get a (weak) partition with $N$ entries. With respect to the dominance ordering, see \eqref{eq:dominance} in \textsection\ref{s:Macdonald}, the highest term in \eqref{eq:HS_polynomial} receives contributions from
\begin{equation*}
	\prod_{m<n}^M \! (z_m-z_n) = \vect{z}^{\delta_M} + \text{lower} \, , \qquad 
	P_\nu^{(\alpha)}(z_1,\To,z_M) = \vect{z}^\nu + \text{lower} \, .
\end{equation*}
Therefore
\begin{equation} \label{eq:largest_monomial}
	\widetilde{\Psi}^\textsc{hs}_\nu(z_1,\To,z_M) = \vect{z}^{\mu^+} + \text{lower monomials} \, , 
\end{equation}
explaining \eqref{eq:motifs_vs_partitions}. Next, the degree of \eqref{eq:HS_polynomial} in any variable is $\deg_{z_1} \! \widetilde{\Psi}^\textsc{hs}_\nu = \nu_1 + 2\,(M-1)  = \mu_M$. As $\ev z_1^N = 1$ it suffices to consider partitions $\nu$ such that $\deg_{z_1} \! \widetilde{\Psi}^\textsc{hs}_\nu \leq N-1$. This reproduces the condition $\nu_1 \leq N-2\,M +1$, i.e.\ $\mu_M < N$, from the line preceding \eqref{eq:motifs_vs_partitions}. 

\textbf{iii.}~Since $P_\nu^{(\alpha)}$ is a homogeneous polynomial of total degree $|\nu|$, the polynomial \eqref{eq:HS_polynomial} is homogeneous of total degree $|\mu|$. This readily yields \eqref{eq:HS_mtm}. The proof of \eqref{eq:HS_dispersion} is more intricate.
\end{remarks}

\begin{figure}[h]
	\centering
	\begin{tikzpicture}[scale=0.3]
		\foreach \x in {1,...,15} \draw (\x*1.25,5) circle (.5); 
		\foreach \x in {2,6,8,11} \draw[fill=black] (\x*1.25,5) circle (.5); 
		\node at (2*1.25,5.5) [above] {$\mu_1$};
		\node at (6*1.25,5.5) [above] {$\mu_2$};
		\node at (8*1.25,5.5) [above] {$\vphantom{\mu_1}\cdots$};
		\node at (11*1.25,5.5) [above] {$\mu_M$};
		\draw[dotted] (2*1.25,4.5) -- (2*1.25,2.5-.33);
		\draw[dotted] (6*1.25,4.5) -- (6*1.25,2.5-.66);
		\draw[dotted] (8*1.25,4.5) -- (8*1.25,.5+.66);
		\draw[dotted] (11*1.25,4.5) -- (11*1.25,.5);
		\foreach \x in {1,...,15} \draw (\x*1.25,-.5) circle (.5); 
		\foreach \x in {1,3,5,7} \draw[fill=black] (\x*1.25,-.5) circle (.5);
		\node at (1*1.25,-1) [below] {$1$};
		\node at (3*1.25,-1) [below] {$3$};
		\node at (5*1.25,-1) [below] {$\vphantom{M}\cdots$\ \ \ };
		\node at (7*1.25,-1) [below] {\ \ $2M{-}1$}; 
		\node at (15*1.25,-1) [below] {\ \ $N{-}1$}; 
		\draw[dotted] (1*1.25,-.5) -- (1*1.25,3-.33);
		\draw[dotted] (3*1.25,-.5) -- (3*1.25,3-.66);
		\draw[dotted] (5*1.25,-.5) -- (5*1.25,1+.66);
		\draw[dotted] (7*1.25,-.5) -- (7*1.25,1);
		\draw[<->] (1*1.25,2.75-.33) -- (2*1.25,2.75-.33);
		\draw[<->] (3*1.25,2.75-.66) -- (6*1.25,2.75-.66);
		\draw[<->] (5*1.25,.75+.66) -- (8*1.25,.75+.66);
		\draw[<->] (7*1.25,.75) -- (11*1.25,.75);
		\node at (1.5*1.25,2.75-.33) [shift={(-.12cm,.3cm)}] {$\bar{\nu}_M$};
		\node at (4.5*1.25,2.75-.66) [yshift=.25cm] {$\bar{\nu}_{M-1}$};
		\node at (6.5*1.25,.75+.66-.05) [inner sep=1.5pt, fill=white] {$\To$};
		\node at (9*1.25,.75) [shift={(.2cm,.25cm)}] {$\bar{\nu}_1$};
	\end{tikzpicture}
	\caption{The correspondence \eqref{eq:motifs_vs_partitions} between a motif $\mu \in \mathcal{M}_N$ of length $M = \ell(\mu) \geq 1$ and a partition with $\nu_1 \leq N-2\,M+1$ and $\ell(\nu) = M$, given by $\nu_m = \bar{\nu}_m +1$, $1\leq m\leq M$. Note that $\bar{\nu}$ characterises the extent by which $\mu$ differs from the left-most filled motif of length~$M$.}
	\label{fg:motifs_vs_partitions}
\end{figure}
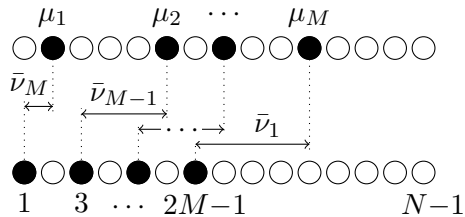

\subsubsection{Spin-Calogero--Sutherland and freezing} \label{s:intro_spin-CS}
The key insight of \cite{BG+_93} is that the many special properties of the Haldane--Shastry spin chain naturally arise from a connection with a \emph{dynamical} model. Consider the spin-version of the Calogero--Sutherland model, with $N$~spin-1/2 particles moving on a circle while interacting in pairs, and Hamiltonian~\cite{HH_92,MP_93,HW_93} 
\begin{equation} \label{eq:spin-CS}
	\begin{aligned} 
	\widetilde{H}^{\text{nr}} & = \frac{1}{2} \, \sum_{i=1}^N \, \bigl(z_i \, \partial_{z_i}\bigr)^{\!2} + \sum_{i<j}^N \frac{-z_i\,z_j}{(z_i-z_j)^2} \ k \, (k-P_{ij}) \\
	& = {-}\frac{1}{2}\sum_{i=1}^N \partial_{x_i}^2 + \sum_{i<j}^N \frac{k\,(k- P_{ij})}{4\, \sin^2[(x_i-x_j)/2]} \, , \qquad z_j = \E^{\I\, x_j} \, .
	\end{aligned}
\end{equation}
Here $k$ is the reduced coupling parameter. In the second line we switched to additive notation. 
This model already
\begin{enumerate}
	\item[i.] (\emph{abelian symmetries}) belongs to a family of commuting operators~\cite{BG+_93,Che_94b,Res_17},
% viz (2.4) and, more specifically, (2.30) and Corollary 2.8 in \cite{Che_94b}; Appendix A in \cite{Res_17}
	each of which
	\item[ii.] (\emph{nonabelian symmetries}) commutes with an action of the Yangian \cite{BG+_93}, cf.~\cite{Dri_86};
	\item[iii.] (\emph{explicit eigenvectors}) has eigenvectors that are determined by a suitably symmetric polynomial, which for Yangian highest-weight eigenvectors is known explicitly in terms of a Jack polynomial \cite{TU_97,Ugl_98}.
\end{enumerate}

As foreseen in \cite{Sha_88} the spin chain emerges through \emph{freezing}~\cite{Pol_93,SS_93,BG+_93,TH_95}: when $k \to \infty$ the kinetic energy becomes negligible compared to the potential energy and the particles `freeze' at their equally spaced classical equilibrium positions $\ev z_j$ to yield \eqref{eq:HS}. By carefully evaluating this limit one shows that properties (i) and~(ii) are inherited by the spin chain~\cite{BG+_93,TH_95}. The derivation of property~(iii) for the spin chain is also based on freezing, though the argument is more subtle; we do \emph{not} know how to get \eqref{eq:HS_polynomial} directly from the spin-Calogero--Sutherland eigenvectors of \cite{TU_97,Ugl_98} via freezing.

\subsection{\textsf{q}-deformed Haldane--Shastry} \label{s:intro_qHS} 
Our goal is to study the partially isotropic (\textsc{xxz}-like) counterpart of the Haldane--Shastry spin chain, building on \cite{BG+_93,TH_95,Ugl_95u,Lam_18}. Let $\mathsf{q} \in \mathbb{C}^\times$ denote the anisotropy parameter. In a nutshell, the isotropy of the Haldane--Shastry spin chain can be \textsf{q}-deformed in such a way that the result
\begin{enumerate}
	\item[i.] (\emph{abelian symmetries}) belongs to a family of commuting operators~\cite{BG+_93,Ugl_95u}, each of which
	\item[ii.] (\emph{nonabelian symmetries}) commutes with an action of the quantum-affine (more pecisely: quantum-loop) algebra \cite{BG+_93};
	\item[iii.] (\emph{explicit eigenvectors}) has eigenvectors that are determined by a symmetric polynomial, which for (`pseudo') highest-weight eigenvectors are known explicitly and involves a Macdonald polynomial.
\end{enumerate}
The state of the art can be summarised as follows. A Hamiltonian that \textsf{q}-deforms \eqref{eq:HS} was found by Uglov~\cite{Ugl_95u} and simplified significantly by one of us \cite{Lam_18}. We shall give a direct derivation of the latter and of the appropriate translation operator, which was proposed in \cite{Lam_18}. These abelian symmetries are reviewed in \textsection\ref{s:intro_abelian}, where we will moreover present two new \textsf{q}-deformations of \eqref{eq:HS}. One of them has real spectrum also for the regime $|\mathsf{q}|=1$ that should be most interesting physically. We give the higher abelian symmetries in \textsection\ref{s:intro_freezing}.

By construction~\cite{BG+_93} the \textsf{q}-deformation is such that the nonabelian symmetries are deformed to 
\checkedMma
the quantum-affine level. We preview these symmetries in \textsection\ref{s:intro_nonabelian_1}. Unfortunately we have not been able to formulate their action, which was also studied in \cite{Ugl_95u}, as concretely as that of the Yangian generators~\eqref{eq:Yangian_gens}. We describe this action a little later, in \textsection\ref{s:intro_nonabelian_2}, where we also explain the appropriate notion of highest weight, which we call pseudo highest weight.

The structure of the Hilbert space parallels the isotropic case. In particular we find that all eigenvectors are still determined by some polynomial, which we are able to give explicitly for the eigenvectors with (pseudo) highest weight. We present these eigenvectors in \textsection\ref{s:intro_explicit_evrs}.
	
These remarkable properties once again stem from a dynamical model that reduces to the  \textsf{q}-deformed Haldane--Shastry spin chain via freezing. We return to this in \textsection\ref{s:intro_spin-RM}.

\subsubsection{Abelian symmetries} \label{s:intro_abelian} 
A Hamiltonian for the \textsf{q}-deformed Haldane--Shastry spin chain was found in \cite{Ugl_95u}. Like \eqref{eq:HS} it admits an expression in a long-range pairwise form~\cite{Lam_18}:\,%
\footnote{\ \label{fn:CPT} We set the coupling constant from \cite{Lam_18} to $J=[N]/N$ as in \cite{Ugl_95u}. Note that in \cite{Lam_18} all spectral parameters were inverted, cf.~\eqref{eq:R_graphical} and \eqref{eq:Sij_left_ex}, in order to stay close to the expressions of \cite{Ugl_95u}. Equivalently, $H$ from \cite{Lam_18} is related to \eqref{eq:ham_left} by inverting~$\mathsf{q}$ and flipping all spins $\ket{\uparrow} \leftrightarrow \ket{\downarrow}$. Indeed, in \cite{Lam_18} it was observed that $H^\textsc{l}$ is (`\textsc{cpt}') invariant under simultaneous reversal of spins~$\ket{\downarrow} \leftrightarrow \ket{\uparrow}$, the order of the coordinates $z_j \mapsto z_{N-j+1}$ (on shell equivalent to $z_j \mapsto 1/z_j$ as all coordinates occur in ratios), and inverting $\mathsf{q}$. This symmetry is easy to understand: the \textit{R}-matrix~\eqref{eq:R_mat} is invariant under inverting its arguments, as well as $\mathsf{q}$, together with conjugation by ($P$ or) $\sigma^x \otimes \sigma^x$. This extends to the $S_{[i,j]}^{\mspace{1mu}\textsc{l}}$, and thus the Hamiltonian, where conjugation by the antidiagonal matrix $(\sigma^x)^{\otimes N}$ implements global spin reversal. (In particular, complex conjugation of $H^\textsc{l}$ is equivalent to global spin reversal if $\mathsf{q}\in S^1$.)}%
\begin{equation} \label{eq:ham_left}
	\normalfont % avoid italics text
	H^\textsc{l} = \ev \widetilde{H}^\textsc{l} \, , \qquad \widetilde{H}^\textsc{l} = \frac{[N]}{N} \sum_{i<j}^N V\mspace{-1mu}(z_i,z_j) \, S_{[i,j]}^{\mspace{1mu}\textsc{l}} \, .
\checkedMma
\end{equation}
(The superscript `L' will make sense soon.) The prefactor involves the \textsf{q}-analogue of $N \in \mathbb{N}$,
\begin{equation} \label{eq:N_q}
	[N] \coloneqq \frac{\mathsf{q}^N - \mathsf{q}^{-N}}{\mathsf{q}-\mathsf{q}^{-1}} = \mathsf{q}^{N-1} + \mathsf{q}^{N-3} + \cdots + \mathsf{q}^{3-N} + \mathsf{q}^{1-N} \, .
\checkedMma
\end{equation} 
Next, the potential in \eqref{eq:ham_left} reads
\begin{equation} \label{eq:pot}
	V\mspace{-1mu}(z_i,z_j) = \frac{z_i \, z_j}{(\mathsf{q}\,z_i-\mathsf{q}^{-1}z_j)(\mathsf{q}\,z_j-\mathsf{q}^{-1}z_i)} \, ,
\checkedMma
\end{equation}
where the sign is chosen such that $\ev V\mspace{-1mu}(z_i,z_j) > 0$ for $\mathsf{q} \in \mathbb{R}$. A geometric way to think about this potential is shown in Figure~\ref{fg:pot}.

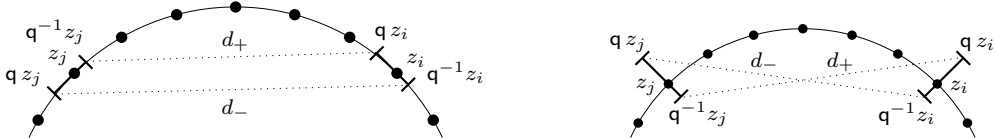
\begin{figure}[h]
	\centering
	\subcaptionbox*{}{%
		\begin{tikzpicture}[scale=0.3,font=\scriptsize]
		\draw[very thin] (25:10) arc (25:155:10);
		\foreach \th in {0,...,8} \fill[black] (30+15*\th:10) circle (.25);
		\draw[thick,|-|] (30+15*1:10) +(45-90:.75) -- +(45+90:1.3);
		\draw[thick,|-|] (30+15*7:10) +(135-90:.75) -- +(135+90:1.3);
		\draw[dotted] ($(45:10) +(45-90:.75)$) -- ($(135:10) + (135+90:1.3)$);
		\draw[dotted] ($(45:10) + (45+90:1.3)$) -- ($(135:10) + (135-90:.75)$);
		\node at (44:11) {$z_i$}; 
		\node at ($(45:11) + (45+90:1.3)$) [yshift=1.25] {$\mathsf{q}\,z_i$};
		\node at ($(45:11) + (45-90:1)$) [xshift=10] {$\mathsf{q}^{-1}z_i$};
		\node at (135:11) {$z_j$};
		\node at ($(135:11) + (135+90:1.3)$) [xshift=-5] {$\mathsf{q}\,z_j$};
		\node at ($(135:11) + (135-90:1)$) [shift={(-.25,.125)}] {$\mathsf{q}^{-1}z_j$};
		\node at (0,5.5) {$d_-$};
		\node at (0,8.5) {$d_+$};
		\end{tikzpicture}
	}\hfil % note: use \hfil and NOT \hfill
	\subcaptionbox*{}{%
		\begin{tikzpicture}[scale=0.25,font=\scriptsize]
		\draw[very thin] (25:10) arc (25:155:10);
		\foreach \th in {0,...,8} \fill[black] (30+15*\th:10) circle (.25);
		\draw[thick,|-|] (30+15*1:10) + (45:-1) -- +(45:2);
		\draw[thick,|-|] (30+15*7:10) + (135:-1) -- +(135:2);
		\draw[dotted] ($(45:10) + (45:-1)$) -- ($(135:10) + (135:2)$);
		\draw[dotted] ($(45:10) + (45:2)$) -- ($(135:10) + (135:-1)$);
		\node at (40:10.75) {$z_i$}; 
		\node at (45:13) {$\mathsf{q}\,z_i$};
		\node at (45:7.75) {$\mathsf{q}^{-1}z_i$};
		\node at (140:10.75) {$z_j$};
		\node at (135:13) {$\mathsf{q}\,z_j$};
		\node at (135:7.75) {$\mathsf{q}^{-1}z_j$};
		\node at (-2,8.25) {$d_-$}; 
		\node at (2,8.25) {$d_+$};
		\end{tikzpicture}
	}
	\caption{The potential \eqref{eq:pot} is a point splitting of the inverse square in \eqref{eq:HS}. Consider a little `dipole' at each site, with length set by~$\mathsf{q}-\mathsf{q}^{-1}$. Then $\ev V\mspace{-1mu}(z_i,z_j) = 1/d_+ \, d_-$, where $d_\pm$ are illustrated for $\mathsf{q}\in \I\,\mathbb{R}_{>1}$ (left) and $\mathsf{q}\in\mathbb{R}_{>1}$ (right). As $\mathsf{q}\to1$ both $d_\pm \to d$ reduce to the chord distance.} \label{fg:pot}
\end{figure}

Finally, the operators $S_{[i,j]}^{\mspace{1mu}\textsc{l}}$ in \eqref{eq:ham_left} deform the long-range exchange interactions of \eqref{eq:HS}. The deformation is accomplished via the spin-1/2 \textsc{xxz} (six-vertex) \textit{R}-matrix
\begin{equation} \label{eq:R_mat}
	\check{R}(u) \coloneqq 
	\begin{pmatrix} 
	\, 1 & \color{gray!80}{0} & \color{gray!80}{0} & \color{gray!80}{0} \, \\
	\, \color{gray!80}{0} & u\,g(u) & f(u) & \color{gray!80}{0} \, \\
	\, \color{gray!80}{0} & f(u) & g(u) & \color{gray!80}{0} \, \\
	\, \color{gray!80}{0} & \color{gray!80}{0} & \color{gray!80}{0} & 1 \, \\
	\end{pmatrix}
	\, , \qquad
	f(u) \coloneqq \frac{u-1}{\mathsf{q}\,u-\mathsf{q}^{-1}} \, , \quad
	g(u) \coloneqq \frac{\mathsf{q}-\mathsf{q}^{-1}}{\mathsf{q}\,u-\mathsf{q}^{-1}} \, . % = 1 - \mathsf{q} \, f(u) \, .
\checkedMma
\end{equation}
Here the $4\times4$ matrix is with respect to the standard basis $\ket{\uparrow\uparrow},\ket{\uparrow\downarrow},\ket{\downarrow\uparrow},\ket{\downarrow\downarrow}$ of $\mathbb{C}^2 \otimes \mathbb{C}^2$. 
The functions $f$ and $g$ can be recognised as the ratios of the six-vertex model's local weights.\,%
\footnote{\ Namely, $f(u) = b/a$ and $u\,g(u) = c_+/a$, $g(u) = c_-/a$ in terms of the vertex weights of the asymmetric six-vertex model. Note that one usually thinks of $R(u) = P \, \check{R}(u)$ as encoding these vertex weights, so $b \leftrightarrow c_\pm$ are swapped. The asymmetry $c_+ \neq c_-$ stems from the connection to the \textsf{q}-deformed algebras in \textsection\ref{s:Hecke_TL_Uqsl}, see \eqref{eq:baxterisation}.} 
The properties of \eqref{eq:R_mat}, notably including the Yang--Baxter equation, are reviewed in \textsection\ref{s:quantum-affine}. Note that $\check{R}(u) \to P$ as $\mathsf{q}\to1$.

Note that the isotropic interactions in \eqref{eq:HS} can be decomposed into nearest-neighbour steps consisting of transport to the left, interaction, and transport back:
\begin{equation*}
	1-P_{ij} = P_{j-1,j} \cdots P_{i+1,i+2} \, (1-P_{i,i+1}) \, P_{i+1,i+2} \cdots P_{j-1,j} \, .
\end{equation*}
The appropriate \textsf{q}-deformation has the same structure, cf.~\cite{HS_96}. It is perhaps most clearly defined using graphical notation:
\begin{equation} \label{eq:Sij_left}
	S_{[i,j]}^{\mspace{1mu}\textsc{l}} \coloneqq \!
	\tikz[baseline={([yshift=-.5*11pt*0.4]current bounding box.center)},	xscale=0.6,yscale=0.4,font=\footnotesize]{
		\draw[->] (10.5,0) node[below]{$z_N$} -- (10.5,8)  node[above]{$z_N\vphantom{z_j}$};
		\draw[->] (9,0) node[below]{$z_{j+1}$} -- (9,8)  node[above]{$z_{j+1}$};
		\draw[rounded corners=3.5pt,->] (8,0) node[below]{$z_j$} -- (8,.5) -- (7,1.5) node[inner sep=1.5pt,fill=white]{$z_j$} -- (6,2.5) node[inner sep=1.5pt,fill=white]{$z_j$} -- (5,3.5) node[shift={(.05,-.1)}, inner sep=1.5pt,fill=white]{$z_j$} -- (5,4) -- (5,4.5) node[shift={(.05,.1)}, inner sep=1.5pt,fill=white]{$z_j$} -- (6,5.5) node[inner sep=1.5pt,fill=white]{$z_j$} -- (7,6.5) node[inner sep=1.5pt,fill=white]{$z_j$} -- (8,7.5) -- (8,8) node[above]{$z_j$};
		\draw[rounded corners=3.5pt,->] (7,0) node[below]{$z_{j-1}$} -- (7,.5) -- (8,1.5) -- (8,4) node[inner sep=1.5pt,fill=white]{$z_{j-1}$} -- (8,6.5) -- (7,7.5) -- (7,8) node[above]{$z_{j-1}$};
		\draw[rounded corners=3.5pt,->] (6,0) node[below]{$\cdots$} -- (6,1.5) -- (7,2.5) -- (7,4) node[inner sep=1.5pt,fill=white]{$\cdots$} -- (7,5.5) -- (6,6.5) -- (6,8) node[above]{$\cdots{\vphantom{z_j}}$};
		\draw[rounded corners=3.5pt,->] (5,0) node[below]{$z_{i+1}$} -- (5,2.5) -- (6,3.5) -- (6,4) node[inner sep=1.5pt,fill=white]{$z_{i+1}$} -- (6,4.5) -- (5,5.5) -- (5,8) node[above]{$z_{i+1}\vphantom{z_j}$};
		\draw[->] (4,0) node[below]{$z_i$} -- (4,8) node[above]{$z_i\vphantom{z_j}$};
		\draw[->] (3,0) node[below]{$z_{i-1}$} -- (3,8)  node[above]{$z_{i-1}\vphantom{z_j}$};
		\draw[->] (1.5,0) node[below]{$z_1$} -- (1.5,8)  node[above]{$z_1\vphantom{z_j}$};
		\draw[style={decorate, decoration={zigzag,amplitude=.5mm,segment length=2mm}}] (4,4) -- (5,4);
		\foreach \x in {-1,...,1} \draw (2.25+.2*\x,4) node{$\cdot\mathstrut$};
		\foreach \x in {-1,...,1} \draw (9.75+.2*\x,4) node{$\cdot\mathstrut$};	
	} .
\end{equation}
The little arrows at the top indicate that the diagrams are read from bottom to top (time goes up). The coordinates, here in the role of inhomogeneity parameters, follow the lines as indicated. The nearest-neighbour transport is accounted for by the \textit{R}-matrix,
\begin{equation} \label{eq:R_graphical}
	\tikz[baseline={([yshift=-.5*11pt*0.4]current bounding box.center)},	xscale=0.5,yscale=0.4,font=\footnotesize]{
		\draw[rounded corners=3.5pt,->] (1,0) node[below]{$v$} -- (1,.5) -- (0,1.5) -- (0,2) node[above]{$v$};
		\draw[rounded corners=3.5pt,->] (0,0) node[below]{$u$} -- (0,.5) -- (1,1.5) -- (1,2) node[above]{$u$};
	} 
	\coloneqq \check{R}(u/v) \, ,
\end{equation}
while the nearest-neighbour exchange is deformed to the Temperley--Lieb generator\,%
\footnote{\ Unlike the usual graphical notation for $e^\text{sp}_i$ this does not naturally represent the Temperley--Lieb relations (\textsection\ref{s:Hecke_TL_Uqsl}), but it correctly accounts for the flow of (spectral) parameters along the lines.}
\begin{equation} \label{eq:R'(1)}
	\tikz[baseline={([yshift=-.5*11pt*0.4]current bounding box.center)},xscale=0.5,yscale=0.4,font=\footnotesize]{
		\draw[->] (1,0) node[below]{$u$} -- (1,2) node[above]{$u$};
		\draw[->] (2,0) node[below]{$v$} -- (2,2) node[above]{$v$};
		\draw[style={decorate, decoration={zigzag,amplitude=.5mm,segment length=2mm}}] (1,1) -- (2,1);
	} 
	\coloneqq e^\text{sp} = -(\mathsf{q}\,{-}\,\mathsf{q}^{-1}) \, \check{R}'(1) = 
	\begin{pmatrix} 
	\, 0 & \color{gray!80}{0} & \color{gray!80}{0} & \color{gray!80}{0} \, \\
	\, \color{gray!80}{0} & \mathsf{q}^{-1} & -1 & \color{gray!80}{0} \, \\
	\, \color{gray!80}{0} & -1 & \mathsf{q} & \color{gray!80}{0} \, \\
	\, \color{gray!80}{0} & \color{gray!80}{0} & \color{gray!80}{0} & 0 \, \\
	\end{pmatrix}
	\, .
\checkedMma
\end{equation}
This \textsf{q}-antisymmetriser (up to normalisation) is the local Hamiltonian of the quantum-$\mathfrak{sl}_2$ invariant Heisenberg spin chain~\cite{PS_90}, see~\textsection\ref{s:spin_chains}. It reduces to $e^\text{sp} \to 1-P$ when $\mathsf{q}\to1$.

An example of the long-range spin interactions~\eqref{eq:Sij_left} is
\begin{equation} \label{eq:Sij_left_ex}
	\begin{aligned}
	S_{[1,5]}^\textsc{l} = \ & \check{R}_{45}(z_5/z_4) \, \check{R}_{34}(z_5/z_3) \, \check{R}_{23}(z_5/z_2) \\
	& \quad \times -(\mathsf{q}-\mathsf{q}^{-1})\,\check{R}_{12}'(1) \\
	& \qquad \times \check{R}_{23}(z_2/z_5) \, \check{R}_{34}(z_3/z_5) \, \check{R}_{45}(z_4/z_5) \, .
	\end{aligned}
\checkedMma
\end{equation}
We stress that in the graphical notation the parameters follow the lines, but (unlike if one would draw $R=P\,\check{R}$ or $\check{R}\,P$) the vector spaces do \emph{not}, cf.~the subscripts in \eqref{eq:Sij_left_ex}. 
The notation `$[i,j]$' as an interval in \eqref{eq:Sij_left}, which is borrowed from \cite{HS_96}, reflects the fact that the intermediate spins are affected by the transport via the \textit{R}-matrix: the model involves multi-spin interactions when $\mathsf{q}\neq \pm1$. As a result the direct computation of the action of $H^\textsc{l}$ on any vector is quite complicated even for a single excited spin.

\begin{remarks} 
\textbf{i.}~If $\mathsf{q}\in \mathbb{R}^\times$ the hermiticity of \eqref{eq:R'(1)} is inherited by $H^\textsc{l}$~\cite{Lam_18}. \textbf{ii.}~The structure of $H^\textsc{l}$, with its multi-spin interactions, might be somewhat involved, yet is precisely such that the key properties of \eqref{eq:HS} generalise to the \textsf{q}-case. We will \emph{derive} the formula for $H^\textsc{l}$ in \textsection\ref{s:abelian_freezing}, see \textsection\ref{s:intro_freezing} for a sketch. \textbf{iii.}~$H^\textsc{l}$ has a \emph{stochastic} version too: see \textsection\ref{s:app_stochastic_twist}. \textbf{iv.}~The Hamiltonian depends mildly on the sign of $\mathsf{q}$: $H^\textsc{l}|_{\mathsf{q}\mapsto -\mathsf{q}}$ differs from $(-1)^N \, H^\textsc{l}$ by a conjugation. We will prove this in \textsection\ref{s:app_stochastic_twist}, see \eqref{eq:ham_left_q_to_-q}.
\end{remarks}

As in the isotropic case, the energy spectrum of the Hamiltonian~\eqref{eq:ham_left} can be given explicitly. They are still labelled by motifs~\eqref{eq:motif} and remain (strictly) additive.
\begin{theorem}[cf.~\cite{Ugl_95u}] \label{thm:intro_energy_left}
The spectrum of~$\normalfont H^\textsc{l}$ is given by\,%
\footnote{\ The dispersion in \cite{Ugl_95u,Lam_18} differs from~$\varepsilon^\textsc{l}$ by $\mathsf{q} \mapsto \mathsf{q}^{-1}$, cf.~`\textsc{cpt}' invariance from Footnote~\ref{fn:CPT} on p.\,\pageref{fn:CPT}.}
\begin{equation} \label{eq:energy_left}
	\normalfont E^{\mspace{1mu}\textsc{l}}(\mu) = \sum_{m=1}^M \varepsilon^\textsc{l}(\mu_m) \, , \qquad
	\varepsilon^\textsc{l}(\mu_m) = \frac{1}{\mathsf{q}-\mathsf{q}^{-1}} \, \biggl( \frac{\mathsf{q}^N}{\mathsf{q}^{\mu_m}} \,[\mu_m] - \frac{\mu_m}{N} \, [N] \biggr) \, ,
\end{equation}
for $\mu \in \mathcal{M}_N$. More precisely,

{\normalfont\textbf{i.} (\cite{Ugl_95u})} All eigenvalues of the Hamiltonian~$\normalfont H^\textsc{l}$ are of the form \eqref{eq:energy_left} for appropriate $\mu$.

{\normalfont\textbf{ii.}} For each $\mu \in \mathcal{M}_N$ the eigenvalue \eqref{eq:energy_left} does indeed occur in the spectrum of $\normalfont H^\textsc{l}$.

{\normalfont\textbf{iii.}} For generic values of $\mathsf{q} \in\mathbb{C}^\times$, i.e.\ $\mathsf{q}^N \neq 1$, the assignment~\eqref{eq:energy_left} is injective on $\mathcal{M}_N$.

{\normalfont\textbf{iv.} (Completeness)} All eigenvalues of $\normalfont H^\textsc{l}$ are given by \eqref{eq:energy_left} with $\mu \in \mathcal{M}_N$.
\end{theorem}

\noindent Part~(i) is due to Uglov~\cite{Ugl_95u}; we will review his proof in \textsection\ref{s:abelian_freezing}. Concerning part~(ii), we will give an explicit eigenvector with that energy in \textsection\ref{s:intro_explicit_evrs}\,---\,see also Remark~(vi) therein; these eigenvectors will be derived in \textsection\ref{s:explicit_evrs}. Part~(iv) will be shown at the end of \textsection\ref{s:intro_nonabelian_1} based on a result that will be established in \textsection\ref{s:hw}. Here we can already give the simple
\begin{proof}[Proof of Theorem~\ref{thm:intro_energy_left}~(iii)]
As a Laurent polynomial in $\mathsf{q}$, \eqref{eq:energy_left} is given by~\cite{Lam_18}
\begin{equation*}
	\varepsilon^\textsc{l}(\mu_m) = \frac{1}{N} \sum_{n=1}^{N-1} \min\Bigl( \mu_m \, (N-n), (N-\mu_m) \, n \Bigr) \, \mathsf{q}^{N-2n} \, .
\end{equation*}
The coefficient in this expression, viewed as a function of $n$ on the interval $[0,N]$, is piecewise linear, from the origin to a maximum at $n=\mu_m$ and back down to zero at $n=N$. Therefore the coefficients of $E^\textsc{l}(\mu)$ as a Laurent polynomial in $\mathsf{q}$ look like a piecewise linear function with kinks at the parts of $\mu$. This allows us to reconstruct $\mu \in \mathcal{M}_N$ uniquely from $E^\textsc{l}(\mu)$.
\end{proof}

The \textsf{q}-deformation~\eqref{eq:Sij_left} breaks left-right symmetry; the model described by \eqref{eq:ham_left} is \emph{chiral}. One of our new results is a Hamiltonian with the opposite chirality, which also \textsf{q}-deforms \eqref{eq:HS} and is very similar to~\eqref{eq:ham_left}:
\begin{subequations} \label{eq:ham_right}
	\begin{gather}
	H^\textsc{r} = \ev \widetilde{H}^\textsc{r} \, , \qquad \widetilde{H}^\textsc{r} = \frac{[N]}{N} \sum_{i<j}^N V\mspace{-1mu}(z_i,z_j) \, S_{[i,j]}^{\mspace{1mu}\textsc{r}} \, ,
\checkedMma
\intertext{now featuring long-range spin interactions where the \textsf{q}-antisymmetrisation takes place on the right,}
	S_{[i,j]}^{\mspace{1mu}\textsc{r}} \coloneqq \!
	\tikz[baseline={([yshift=-.5*11pt*0.4]current bounding box.center)},	xscale=0.6,yscale=0.4,font=\footnotesize]{
		\draw[->] (10.5,0) node[below]{$z_N$} -- (10.5,8)  node[above]{$z_N\vphantom{z_j}$};
		\draw[->] (9,0) node[below]{$z_{j+1}$} -- (9,8)  node[above]{$z_{j+1}$};
		\draw[->] (8,0) node[below]{$z_j$} -- (8,8) node[above]{$z_j$};
		\draw[rounded corners=3.5pt,->] (7,0) node[below]{$z_{j-1}$} -- (7,2.5) -- (6,3.5) -- (6,4) node[inner sep=1.5pt,fill=white]{$z_{j-1}$} -- (6,4.5) -- (7,5.5) -- (7,8) node[above]{$z_{j-1}$};
		\draw[rounded corners=3.5pt,->] (6,0) node[below]{$\cdots$} -- (6,1.5) -- (5,2.5) -- (5,4) node[inner sep=1.5pt,fill=white]{$\cdots$} -- (5,5.5) -- (6,6.5) -- (6,8) node[above]{$\cdots{\vphantom{z_j}}$};		
		\draw[rounded corners=3.5pt,->] (5,0) node[below]{$z_{i+1}$} -- (5,.5) -- (4,1.5) -- (4,4) node[inner sep=1.5pt,fill=white]{$z_{i+1}$} -- (4,6.5) -- (5,7.5) -- (5,8) node[above]{$z_{i+1}\vphantom{z_j}$};
		\draw[rounded corners=3.5pt,->] (4,0) node[below]{$z_i$} -- (4,.5) -- (5,1.5) node[inner sep=1.5pt,fill=white]{$z_i$} -- (6,2.5) node[inner sep=1.5pt,fill=white]{$z_i$} -- (7,3.5) node[shift={(.05,-.1)}, inner sep=1.5pt,fill=white]{$z_i$} -- (7,4) -- (7,4.5) node[shift={(.05,.1)}, inner sep=1.5pt,fill=white]{$z_i$} -- (6,5.5) node[inner sep=1.5pt,fill=white]{$z_i$} -- (5,6.5) node[inner sep=1.5pt,fill=white]{$z_i$} -- (4,7.5) -- (4,8) node[above]{$z_i \vphantom{z_j}$};
		\draw[->] (3,0) node[below]{$z_{i-1}$} -- (3,8)  node[above]{$z_{i-1}\vphantom{z_j}$};
		\draw[->] (1.5,0) node[below]{$z_1$} -- (1.5,8)  node[above]{$z_1\vphantom{z_j}$};
		\draw[style={decorate, decoration={zigzag,amplitude=.5mm,segment length=2mm}}] (7,4) -- (8,4);
		\foreach \x in {-1,...,1} \draw (2.25+.2*\x,4) node{$\cdot\mathstrut$};
		\foreach \x in {-1,...,1} \draw (9.75+.2*\x,4) node{$\cdot\mathstrut$};	
	} .
	\end{gather}
\end{subequations}

\begin{theorem} \label{thm:ham_energy_right}
{\normalfont\textbf{i.}}~The abelian symmetries of the \textsf{q}-deformed Haldane--Shastry spin chain include the operator $\normalfont H^\textsc{r}$. In particular, the two chiral Hamiltonians $\normalfont H^\textsc{l},H^\textsc{r}$ commute.

{\normalfont\textbf{ii.}} The eigenvalues of $\normalfont H^\textsc{r}$ are as in Theorem~\ref{thm:intro_energy_left} with \eqref{eq:energy_left} modified by inverting \textsf{q} or, equivalently, by reflecting the motif:
\begin{equation} \label{eq:energy_right}
	\normalfont E^{\mspace{1mu}\textsc{r}}(\mu) = \sum_{m=1}^M \varepsilon^\textsc{r}(\mu_m) \, ,  
	\qquad \varepsilon^\textsc{r}(\mu_m) = \varepsilon^\textsc{l}(\mu_m) \big|_{\mathsf{q}\mapsto \mathsf{q}^{-1}} = \varepsilon^\textsc{l}(N-\mu_m) \, .
\end{equation}
\end{theorem}

\noindent In \textsection\ref{s:abelian_freezing} we will show that the first part is true by construction. The second part follows from Theorem~\ref{thm:intro_energy_left} except that the first part therein has to be adjusted, which will be done in \textsection\ref{s:abelian_freezing}. The results in Theorem~\ref{thm:ham_energy_right} are new.

The \textsf{q}-deformed chiral energies~\eqref{eq:energy_left}--\eqref{eq:energy_right} are real when $\mathsf{q}\in\mathbb{R}^\times$, in which case $H^\textsc{l}$ is hermitian~\cite{Lam_18}. Recall that for the Heisenberg \textsc{xxz} spin chain the regime $|\mathsf{q}|=1$, where its parameter $\Delta = (\mathsf{q}+\mathsf{q}^{-1})/2$ obeys $|\Delta|\leq 1$, is physically most interesting. With this in mind, Theorem~\ref{thm:ham_energy_right} implies the following important

\begin{corollary} \label{cor:ham_full_energy}
{\normalfont\textbf{i.}}~The abelian symmetries include the \emph{full} Hamiltonian
\begin{equation} \label{eq:ham_full}
	\normalfont H^\text{full} \coloneqq \frac{1}{2} \, (H^\textsc{l} + H^\textsc{r}) = \ev \widetilde{H}^\text{full} \, , \qquad \widetilde{H}^\text{full} =\frac{[N]}{2N} \, \sum_{i< j}^N V\mspace{-1mu}(z_i,z_j) \, \bigl(S_{[i,j]}^{\mspace{1mu}\textsc{l}} + S_{[i,j]}^{\mspace{1mu}\textsc{r}}\bigr) \, .
\checkedMma
\end{equation}

{\normalfont\textbf{ii.}}~The eigenvalue of $\normalfont H^\text{full}$ for $\mu \in \mathcal{M}_N$ involves a beautiful \textsf{q}-deformed dispersion relation:
\begin{equation} \label{eq:dispersion_full} 
	\normalfont 
	E^\text{full}(\mu) = \sum_{m=1}^M \varepsilon^\text{full}(\mu_m) \, , \qquad
	\varepsilon^\text{full}(\mu_m) = \frac{\varepsilon^\textsc{l}(\mu_m) + \varepsilon^\textsc{r}(\mu_m)}{2} = \frac{1}{2} \, [\mu_m] \, [N-\mu_m] \, .
	\checkedMma
\end{equation}
In particular, its spectrum is manifestly real also for $\mathsf{q}\in S^1 \subset \mathbb{C}^\times$. 
\end{corollary}

The elegant expression \eqref{eq:dispersion_full} causes $H^\text{full}$ to have some extra degeneracies with respect to the chiral Hamiltonians. Indeed, for $H^\textsc{l}$ (or $H^\textsc{r}$) the set of distinct energies is typically equinumerous with the set of motifs: the chiral Hamiltonians `only' have representation-theoretic degeneracies to reflect the nonabelian symmetry, which we will explore in \textsection\ref{s:intro_nonabelian_1}. The Hamiltonian $H^\text{full}$ has some additional degeneracies, as mirror-image motifs yield equal $E^\text{full}$, reflecting parity invariance. At special values of \textsf{q} the dispersion relation simplifies, making additional `accidental' degeneracies possible. Examples include the isotropic limit $\mathsf{q}\to1$, with many~\cite{FG_15} accidental degeneracies for $H^\textsc{hs}$, and the crystal limit $\mathsf{q}\to\infty$, which we will treat in \textsection\ref{s:intro_crystal_limit}.
\checkedMma 
The root-of-unity case, for which very large degeneracies occur too, will be investigated elsewhere.

Before we introduce another abelian symmetry it is instructive to pause for a moment and investigate the boundary conditions. The periodicity of the isotropic Hamiltonian~\eqref{eq:HS}, which is invariant under conjugation by the cyclic translation operator~\eqref{eq:HS_translation}, is affected by the \textsf{q}-deformation. Consider $H^\textsc{l}$ for definiteness. The potential~\eqref{eq:pot} remains periodic as it depends on the ratio $z_i/z_j$, i.e.\ on the distance $i-j$ in additive language. On the other hand, the long-range interactions~\eqref{eq:Sij_left} are not periodic: compare the highly non-local multispin operator~$S_{[1,N]}^{\mspace{1mu}\textsc{l}}$ with any genuine nearest-neighbour interaction~$S_{[i,i+1]}^{\mspace{1mu}\textsc{l}} = e^\text{sp}_i$. From this perspective the model lives on a strip rather than a circle: no particle ever really wraps around the back of the chain. This periodicity breaking is required by the coproduct of the nonabelian symmetries, cf.~\cite{HS_96}. As $\mathsf{q}\to1$ the `wall' between sites $N$ and $1$ becomes transparent. For $\mathsf{q}\to \infty$ we instead get an open chain, as we will show in \textsection\ref{s:intro_crystal_limit}. In general the spin chain can be viewed as having some sort of \emph{twisted} (quasiperiodic) boundary conditions:

\begin{lemma}[braid limit] \label{prop:ham_braid_limit}
Let $\normalfont T^\text{sp}_i = \mathsf{q} - e^\text{sp}_i$ be the \textsf{q}-deformed permutation (Hecke generator), with inverse $\normalfont T^{\text{sp}\,-1}_i = \mathsf{q}^{-1} - e^\text{sp}_i$. 
The Hamiltonian~\eqref{eq:ham_left} formally contains the twisted Heisenberg \textsc{xxz} spin chain:
\begin{equation} \label{eq:ham_formal_limit}
	\normalfont
	{-\varpi} \ \mathrm{ev}_\varpi \mspace{1mu} \widetilde{H}^\textsc{l} \ \to \ \sum_{i=1}^{N-1} e^\text{sp}_i \ + \, T^\text{sp}_{N-1} \cdots T^\text{sp}_2 \, e^\text{sp}_1 \, T^{\text{sp}\,-1}_2 \cdots T^{\text{sp}\,-1}_{N-1} \, , \qquad\quad \varpi\to\infty \, ,
\end{equation}
where we replaced \eqref{eq:ev} by the (hyperbolic) evaluation $\mathrm{ev}_\varpi \colon z_j \longmapsto \varpi^j$ with $\varpi \in\mathbb{R}^\times$.
\end{lemma}
\noindent The last term in \eqref{eq:ham_formal_limit}, describing the twisted boundary conditions, is known as a `braid translation'~\cite{MS_93} (with trivial `blob' generator~\cite{MS_94}).
\begin{proof}[Proof of Proposition~\ref{prop:ham_braid_limit}]
We need to evaluate the limit $\varpi\to\infty$ for the components of $\widetilde{H}^\textsc{l}$ from~\eqref{eq:ham_left}. Up to a simple rescaling the potential~\eqref{eq:pot} boils down to the usual, \textsf{q}-independent \emph{nearest-neighbour} pair potential:
\begin{equation*}
	{-\varpi} \ \mathrm{ev}_\varpi V\mspace{-1mu}(z_i,z_j) \to \delta_{|i-j| \, \mathrm{mod} \, N , \,1} \, , \qquad i \neq j \, , \qquad\quad \varpi\to\infty \, .
	\checkedMma
\end{equation*}
In the bulk only the nearest-neighbour interactions $S_{[i,i+1]}^{\mspace{1mu}\textsc{l}} = e^\text{sp}_i$ from~\eqref{eq:R'(1)} survive. The only other term in \eqref{eq:ham_left} that survives this limit is $S_{[1,N]}^{\mspace{1mu}\textsc{l}}$. Now from the viewpoint of the \textit{R}-matrix~\eqref{eq:R_mat}, $\varpi \to\infty$ ($\varpi \to0$) is the braid limit, yielding $T^\text{sp}_i$ (resp.\ $T^{\text{sp}\,-1}_i$), 
\checkedMma
yielding the final term in \eqref{eq:ham_formal_limit}.
\checkedMma
\end{proof}

Despite these somewhat subtle boundary conditions, the \textsf{q}-deformed Haldane--Shastry spin chain is formally periodic.

\begin{proposition}[cf.~\cite{Lam_18}] \label{prop:intro_q-translation}
{\normalfont\textbf{i.}}~The \textsf{q}-deformed Haldane--Shastry spin chain is \emph{\textsf{q}-homoge\-neous}: its abelian symmetries include the (left) \emph{\textsf{q}-translation operator} 
\begin{equation} \label{eq:q-translation}
	G \coloneqq \ev\widetilde{G} \, , \qquad \widetilde{G} \coloneqq \check{R}_{N-1,N}(z_1/z_N) \cdots \check{R}_{12}(z_1/z_2) \, = \!
	\tikz[baseline={([yshift=-.5*11pt*0.4]current bounding box.center)},	xscale=0.5,yscale=0.4,font=\footnotesize]{
		\foreach \x in {1,...,4} \draw[rounded corners=3.5pt] (\x,0) -- ++ (0,\x-.5) -- ++ (-1,1) -- (\x-1,5);
		\node at (0,5) [above]{$z_2$};
		\node at (1,0) [below]{$z_2$};
		\node at (2.5,0) [below]{$\cdots$};
		\node at (1.5,5) [above]{$\cdots$};
		\node at (3,5) [above]{$z_N$};
		\node at (4,0) [below]{$z_N$};
		\draw[rounded corners=3.5pt] (0,0) node[below]{$z_1$} -- (0,.5) -- (1,1.5) node[inner sep=1.5pt,fill=white]{$z_1$} -- (2,2.5) node[inner sep=1.5pt,fill=white]{$z_1$} -- (3,3.5) node[inner sep=1.5pt,fill=white]{$z_1$} -- (4,4.5) -- (4,5) node[above]{$z_1$};
		\foreach \x in {0,...,4} \draw[->] (\x,4.95) -- (\x,5);	
	} .
\checkedMma
\end{equation}

{\normalfont\textbf{ii.}}~For each $\mu \in \mathcal{M}_N$ as in Theorems \ref{thm:intro_energy_left} and~\ref{thm:ham_energy_right} the eigenvalue of $G$ is $\E^{\I p}$ where the \emph{\textsf{q}-momentum} is
\begin{equation} \label{eq:mtm}
	p(\mu) = \sum_{m=1}^M p_m \ \ \mathrm{mod} \, 2\pi \, , \qquad p_m = \frac{2\pi}{N} \, \mu_m \, .
\end{equation}
\end{proposition}
\noindent This was conjectured in \cite{Lam_18}. The first part of the proposition will be established in \textsection\ref{s:abelian_freezing}. As for the second part, we note that for general $\mathsf{q} \in \mathbb{C}^\times$ the multi-spin interactions make it rather hard to verify \eqref{eq:mtm} by direct computation on the explicit expression for the eigenvectors that we will give in \textsection\ref{s:intro_explicit_evrs} even for $M=1$. Nevertheless, \eqref{eq:mtm} has a simple proof.

\begin{proof}[Proof of Proposition~\ref{prop:intro_q-translation}~(ii)]
The Yang--Baxter equation for the \textit{R}-matrix implies $G^N=1$,
\checkedMma 
so $G$ has eigenvalues of the form $\E^{\I\, p}$, with $p \in (2\pi/N) \, \mathbb{Z}_N$ for $\mathbb{Z}_N \coloneqq \mathbb{Z}/N\,\mathbb{Z}$, quantised as usual for particles on a circle. Since the value of $p$ is discrete, it cannot depend on $\mathsf{q} \in \mathbb{C}^\times$. (Instead, the dependence on $\mathsf{q}$ is hidden in the meaning of $p$, as eigenvalue of $-\I \log G$.) As the entries of $G$ are continuous in the deformation parameter its eigenvalues can be calculated at any $\mathsf{q}$. The isotropic point $\mathsf{q}=1$ suffices, where $G \to G^\textsc{hs}$ so that \eqref{eq:mtm} follows from \eqref{eq:HS_mtm}. (One can also use the crystal limit $\mathsf{q}\to\infty$, see \eqref{eq:crystal_mtm} in \textsection\ref{s:intro_crystal_limit} below.)
\end{proof} 

In summary we have encountered two chiral Hamiltonians, \eqref{eq:ham_left} and \eqref{eq:ham_right}, which combine to give the full Hamiltonian~\eqref{eq:ham_full}. We have also met the \textsf{q}-translation operator~\eqref{eq:q-translation}. These operators commute with each other. Their eigenvalues are known explicitly. The different \textsf{q}-deformed dispersion relations are plotted in Figure~\ref{fg:dispersions}. Each of them reduces to the isotropic dispersion relation~\eqref{eq:HS_dispersion} as $\mathsf{q}\to1$. The full Hamiltonian has a real spectrum even if $|\mathsf{q}|=1$, which is expected to be most relevant physically. We will get back to the remaining abelian symmetries in \textsection\ref{s:intro_freezing}; see Table~\ref{tb:abelian_symmetries} on p.\,\pageref{tb:abelian_symmetries} for an overview.

\begin{figure}[h]
	\centering
	\begin{tikzpicture}[yscale=.2]
	\pgfmathsetmacro{\q}{1.2}
	\pgfmathsetmacro{\N}{10}
	\pgfmathdeclarefunction{qnumber}{1}{\pgfmathparse{(\q^#1 - 1/\q^#1)/(\q - 1/\q)}}
	\pgfmathdeclarefunction{nviap}{1}{\pgfmathparse{\N*#1/(2*pi)}}
	% check \begin{axis} etc from \usepackage{pgfplots}
		\draw[->] (0,0) -- (0,20) node[above]{$\varepsilon$};
			\draw (0,0) -- (-.2*.3,0) node[left] {$0$};
			\draw[smooth,samples=2,domain=-.2*.3:0] plot(\x,{qnumber(\N/2)^2/2}) node[left] {$\frac{1}{2} \mspace{1mu} [N/2]^2$};
		\draw[->] (0,0) -- (2*pi+.3,0) node[right] {$p$};
			\draw (0,0) -- (0,-.3) node[below] {$0$};
			\draw (pi,0) -- (pi,-.3) node[below] {$\vphantom{1}\pi$};
			\draw (2*pi,0) -- (2*pi,-.3) node[below] {$2\pi$};
		\draw[smooth,samples=100,domain=0:2*pi,dotted] 
			plot(\x,{1/2*nviap(\x)*(\N-nviap(\x))*qnumber(\N/2)^2/(\N/2)^2});
		\draw[smooth,samples=100,domain=0:2*pi,dashed] plot(\x,{
			-1/(\q-1/\q) * ( nviap(\x)/\N * qnumber(\N) - pow(\q,\N-nviap(\x)) * qnumber(nviap(\x)) )
			});
		\node at (1.8*pi/3,18.5) {$\varepsilon^\textsc{l}$};
		\draw[smooth,samples=100,domain=0:2*pi,dashdotted] plot(\x,{
			1/(\q-1/\q) * ( nviap(\x)/\N * qnumber(\N) - pow(\q,nviap(\x)-\N) * qnumber(nviap(\x)) )
			});
		\node at (4.3*pi/3,18.5) {$\varepsilon^\textsc{r}$};
		\draw[smooth,samples=100,domain=0:2*pi] plot(\x,{
			1/2 * qnumber(nviap(\x)) * qnumber(\N-nviap(\x))
			});
		\node[right] at (7.35,19) {\textsf{q}-deformed dispersions};
			\draw[dashed] (7.5,16) -- (8.25,16) node[right] {\eqref{eq:energy_left} left};
			\draw (7.5,13) -- (8.25,13) node[right] {\eqref{eq:dispersion_full} full};
			\draw[dashdotted] (7.5,10) -- (8.25,10) node[right] {\eqref{eq:energy_right} right};
		\node[right] at (7.35,7) {cf.~isotropic case};
			\draw[dotted] (7.5,4) -- (8.25,4) node[right] 
			{$p\,(2\pi\mspace{1mu}{-}\mspace{1mu}p) \mspace{1mu}{\times}\mspace{1mu} \tfrac{[N/2]^2}{2 \mspace{1mu}\pi^2 \vphantom{k^k_a}}$};
	\end{tikzpicture}
	\caption{The one-particle energies as functions of the \textsf{q}-momentum $p$. For this plot we have taken $\mathsf{q} = 1.2$ and $N=10$. Though the quadratic dispersion~\eqref{eq:HS_dispersion} of isotropic Haldane--Shastry is a function of the ordinary ($\mathsf{q} = 1$) momentum we included a parabola for comparison.}
	\label{fg:dispersions}
\end{figure}
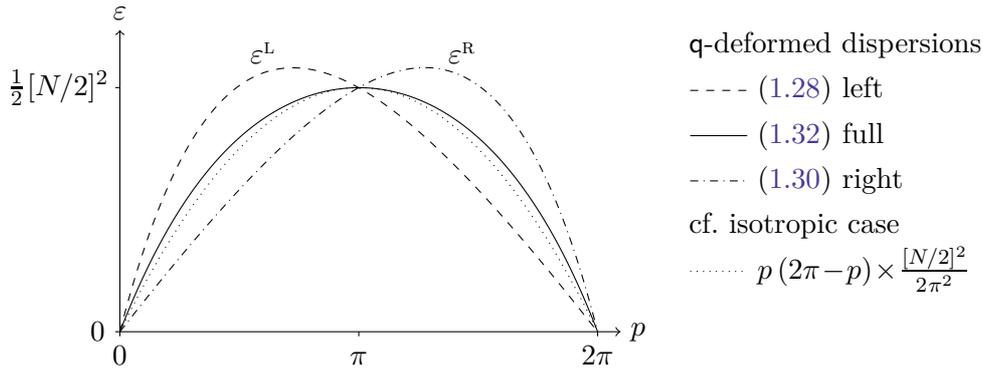

\subsubsection{Nonabelian symmetries: preview} \label{s:intro_nonabelian_1}
To understand the structure of the joint eigenspaces of the abelian symmetries we turn to the \emph{non}abelian symmetries. By construction~\cite{BG+_93} the \textsf{q}-deformation is such that the Yangian is deformed to 
\checkedMma
the quantum-affine (or more precisely: quantum-loop) algebra
\begin{equation*}
	\widehat{\mathfrak{U}} \, \coloneqq \, U_\mathsf{q}'(\widehat{\vphantom{t}\smash{\mathfrak{gl}_2}})^{\vphantom{\prime}}_{c=0} \, = \, U_\mathsf{q}(L\mspace{1mu}\mathfrak{gl}_2) \, .
\end{equation*}
In this section we examine the concrete consequence: the presence of this large symmetry algebra is directly visible in the degeneracies in the spectrum\,---\,which are much higher than those of the Heisenberg \textsc{xxz} or even \textsc{xxx} spin chain.

Just like the Yangian incorporates $\mathfrak{sl}_2$ as a subalgebra, $\widehat{\mathfrak{U}}$ contains the \textsf{q}-deformation of $\mathfrak{sl}_2$,
\begin{equation*}
	\mathfrak{U} \coloneqq U_{\mathsf{q}}(\mathfrak{sl}_2) \, .
\end{equation*}
For more about this subalgebra, including its action on the Hilbert space $\mathcal{H}$, see \textsection\ref{s:Hecke_TL_Uqsl}. Here we stress that this action commutes with the abelian symmetries. The representation theory of $\mathfrak{U}$ for generic $\mathsf{q}$ parallels that of $\mathfrak{sl}_2$; in particular the \textsf{q}-deformed Haldane--Shastry model has at least the same degeneracies as any isotropic model, despite its partial isotropy. This is true for $\widetilde{G},\widetilde{H}^\textsc{l}, \dots$ off shell (at arbitrary $z_i$)\,---\,and in the braid limit~\eqref{eq:ham_formal_limit}: the twisted Heisenberg \textsc{xxz} spin chain is well known to be $\mathfrak{U}$-invariant.

On shell (upon evaluation) the abelian symmetries $G, H^\textsc{l}, \dots$ furthermore have their nonabelian symmetries enhanced to invariance under the much larger algebra~$\widehat{\mathfrak{U}}$. We will recall the relevant facts about $\widehat{\mathfrak{U}}$ in \textsection\ref{s:quantum-affine}, and summarise its connection to the Heisenberg \textsc{xxz} spin chain in \textsection\ref{s:spin_chains}. For the \textsf{q}-deformed Haldane--Shastry spin chain we need a more intricate representation of the affine generators to ensure that they commute with the abelian symmetries. This action will be described in \textsection\ref{s:intro_nonabelian_2}.

Let us preview the structure of the Hilbert space~$\mathcal{H}$ for generic $\mathsf{q}$, which is like the structure at $\mathsf{q}=1$. For each motif $\mu \in \mathcal{M}_N$ there is a unique (up to normalisation) vector $\ket{\mu} \in \mathcal{H}_M$ at $M=\ell(\mu)$ with eigenvalues as in \textsection\ref{s:intro_abelian}. We will explicitly construct $\ket{\mu}$ in \textsection\ref{s:intro_explicit_evrs}.
\begin{definition}
Let 
\begin{equation} \label{eq:Hmu}
	\mathcal{H}^\mu \coloneqq \widehat{\mathfrak{U}} \, \cdot \, \mathbb{C}\,\ket{\mu}
\end{equation}
be the subspace of $\mathcal{H}$ generated by the action of the nonabelian symmetries on the vector $\ket{\mu}$.
\end{definition}
\noindent Since the $\widehat{\mathfrak{U}}$-action commutes with the abelian symmetries, $\mathcal{H}^\mu$ is a joint eigenspace of the latter, with eigenvalues \eqref{eq:energy_left}--\eqref{eq:energy_right} and \eqref{eq:mtm}. By Theorem~\ref{thm:intro_energy_left}~(iii) these eigenvalues are pairwise distinct for generic $\mathsf{q}$, so the subspaces~$\mathcal{H}^\mu$ only intersect at the origin. The Hilbert space decomposes as a direct sum
\begin{equation} \label{eq:motif_decomp}
	\mathcal{H} = \bigoplus_{\mu \,\in\, \mathcal{M}_N} \!\!\! \mathcal{H}^\mu \, .
\end{equation}
We will momentarily verify that we did not miss any eigenspace, as asserted in part~(iv) of Theorem~\ref{thm:intro_energy_left}. For generic $\mathsf{q}$ each $\mathcal{H}^\mu$ has the structure of an irreducible $\widehat{\mathfrak{U}}$-module, characterised as follows.

Recall that there is a bijection between finite-dimensional $\widehat{\mathfrak{U}}$-irreps (up to equivalence) and Drinfeld polynomials (normalised so that $P(0)=1$)~\cite{CP_91}. One can think of the Drinfeld polynomial as an affine analogue of a highest weight. For us it is given by
\begin{proposition}[\cite{Ugl_95u}] \label{prop:Drinfeld}
	The Drinfeld polynomial for $\mathcal{H}^\mu$, $\mu \in\mathcal{M}_N$, is
	\begin{equation} \label{eq:Drinfeld_polyn}
		P^\mu(u) \coloneqq \frac{\prod_{i=1}^N (1- \mathsf{q}^{N-2\, i+1} \, u)}{\displaystyle \prod_{n\in \mu} (1- \mathsf{q}^{N-2\, n-1} \, u) (1- \mathsf{q}^{N-2\, n+1} \, u)} \, ,
		\checkedMma
	\end{equation}
	where $\mu$ specifies which (consecutive pairs of) factors to omit.
\end{proposition}
\noindent In \cite{BG+_93} it was argued (for $\mathsf{q}=1$) that any Drinfeld polynomial that can occur in the model must be of this form for some $\mu$. Uglov~\cite{Ugl_95u} gave \eqref{eq:Drinfeld_polyn} without direct derivation. We prove Proposition~\ref{prop:Drinfeld} in \textsection\ref{s:hw}.

The zeros of $P^\mu$ form \emph{(\textsf{q}-)strings}, i.e.\ sets of the form $\{v,\mathsf{q}^2\mspace{1mu}v,\mathsf{q}^4\mspace{1mu}v,\dots\}$~\cite{CP_91}. For $N=4$, for example, $P^0(u) = (1-\mathsf{q}^3 \mspace{1mu} u) (1-\mathsf{q} \mspace{1mu} u) (1-\mathsf{q}^{-1} \mspace{1mu} u) (1-\mathsf{q}^{-3} \mspace{1mu} u)$ gives a string of length four, both $P^{(1)}(u) = (1-\mathsf{q}^{-1} \mspace{1mu} u) (1-\mathsf{q}^{-3} \mspace{1mu} u)$ and $P^{(3)}(u) = (1-\mathsf{q}^3 \mspace{1mu} u) (1-\mathsf{q} \mspace{1mu} u)$ a string of length two, $P^{(2)}(u) = (1-\mathsf{q}^3 \mspace{1mu} u) (1-\mathsf{q}^{-3} \mspace{1mu} u)$ two strings of length one, and $P^{(1,3)}(u) = 1$ an empty string (of length~zero). The Drinfeld polynomial in particular describes the structure of $\mathcal{H}^\mu$ as a module for $\mathfrak{U} \subset \widehat{\mathfrak{U}}$: each string of length~$j$ among the zeros of $P^\mu$ corresponds to a factor of dimension~$j+1$ (spin $j/2$) \cite{CP_91}, cf.~the graphical rule from~\cite{Lam_18}. Proposition~\ref{prop:Drinfeld} thus leads to

\begin{corollary}
For $\mu \in \mathcal{M}_N$ set $M = \ell(\mu)$. For generic $\mathsf{q}$ the eigenspace $\mathcal{H}^\mu$ has dimension
\begin{equation} \label{eq:dim}
	\dim\mathcal{H}^\mu = \prod_{m=0}^M (\mu_{m+1} - \mu_m -1) =
	\begin{cases*}
		N+1 & if $\mu = 0 \, ,$ \\
		\mu_1 \, (N-\mu_M) \displaystyle\prod_{\smash{m=1}}^{\smash{M-1}} \! (\mu_{m+1} - \mu_m - 1) \quad & if $M \geq 1 \, ,$
	\end{cases*}
\checkedMma
\end{equation}
where on the left one has to interpret $\mu_0 \coloneqq -1$ and $\mu_{M+1} \coloneqq N+1$. 
\end{corollary}

\noindent For example, at $N=4$ we have $\dim\mathcal{H}^0 =5$, $\dim \mathcal{H}^{(1)} = \dim \mathcal{H}^{(3)} =3$, and $\dim\mathcal{H}^{(1,3)}=1$. The affine symmetries appear for $\mathcal{H}^{(2)}$. As a $\mathfrak{U}$-module it is the tensor product of two 2-dimen\-sional factors, and decomposes into irreducible pieces of dimensions 3 and~1. The affine generators connect them to turn $\mathcal{H}^{(2)}$ into a $\widehat{\mathfrak{U}}$-irrep. Figure~\ref{fg:N=6} illustrates this structure for $N=6$.

Once we construct the eigenvector $\ket{\mu}$ (\textsection\ref{s:intro_explicit_evrs}) and establish that $\widehat{\mathfrak{U}}$ does indeed act by symmetries (\textsection\ref{s:intro_nonabelian_2})\,---\,so that the definition \eqref{eq:Hmu} of $\mathcal{H}^\mu$ makes sense\,---\,the dimension formula~\eqref{eq:dim} implies that we didn't miss anything in the decomposition~\eqref{eq:motif_decomp}:
\begin{proof}[Sketch of the proof of Theorem~\ref{thm:intro_energy_left}~(iii)]
To count dimensions we add \eqref{eq:dim} for all $\mu \in \mathcal{M}_N$. To see that the result is $2^N = \dim \mathcal{H}$ observe that motifs allow us to organise the coordinate basis in terms of the pattern $\cdots\! \downarrow\uparrow\!\cdots$. In short, the parts of $\mu \in \mathcal{M}_N$ record the positions of the $\downarrow$ in the pattern, and \eqref{eq:dim} is the number of coordinate basis vectors prescribed in this way. (This combinatorial interpretation is no coincidence, as we'll see in \textsection\ref{s:intro_crystal_limit}.)
\end{proof}

Next we turn to the eigenvector $\ket{\mu}$.

\subsubsection{Explicit eigenvectors} \label{s:intro_explicit_evrs} 
As in the isotropic case any eigenvector in the $M$-particle sector~\eqref{eq:weight_decomp} is completely determined by some symmetric polynomial in $M$ variables. For $\mathsf{q}=1$ the relation was given in~\eqref{eq:HS_vec_ev}, which can be rewritten as follows. Let us identify the permutation $(i,i+1)$ with its action $s_i : z_i \leftrightarrow z_{i+1}$ on coordinates. Using cycle notation $(j,j-1,\To,i) = (j-1,j) \cdots (i,i+1)$ we write 
\begin{equation} \label{eq:perm_grassmannian}
	\{i_1,\To,i_M\} \coloneqq (i_1,i_1-1,\To,1) \, \cdots \, (i_M,i_M -1,\To,M) 
\end{equation}
for the shortest permutation that sends $m \longmapsto i_m$ for all $1\leq m\leq M$. Then \eqref{eq:HS_vec_ev} states
\begin{equation*}
 	\Psi^\textsc{hs}(i_1,\To,i_M) = \ev \Bigl(s_{\{i_1,\To,i_M\}} \widetilde{\Psi}^\textsc{hs}(z_1,\To,z_M) \Bigr) \, .
 \checkedMma
\end{equation*}
In the \textsf{q}-case the different components of the wave function are likewise related through the action of \textsf{q}-deformed permutations before evaluation. Here we suffice with the minimum needed to state this result; see \textsection\ref{s:Hecke_AHA} for more about the Hecke algebra.

The \textsf{q}-analogue of the coordinate permutation~$s_i$ is the Hecke generator. In terms of the functions $f,g$ from \eqref{eq:R_mat} it reads
\begin{equation} \label{eq:intro_Hecke_pol}
	T_i^\text{pol} = f(z_i/z_{i+1})^{-1} \, \bigl(s_i - g(z_i/z_{i+1})\bigr) \, .
\checkedMma
\end{equation}
The \textsf{q}-deformed cyclic permutation  and the \textsf{q}-analogue of \eqref{eq:perm_grassmannian} are constructed from it as
\begin{subequations} \label{eq:Hecke_graphical}
\begin{gather}
	T_{(j,j-1,\To,i)} = T_{j-1} \cdots T_i \, , \qquad T_{\{i_1,\To,i_M\}} = T_{(i_1,\To,1)} \cdots T_{(i_M,\To,M)} \, ,
\shortintertext{or, in terms of (braid) diagrams,}
	\begin{aligned}
	T_i & =
	\tikz[baseline={([yshift=-.5*11pt*0.4*.9+6pt]current bounding box.center)},	xscale=0.5*.9,yscale=0.4*.9,font=\scriptsize, cross line/.style={-,preaction={draw=white,-,line width=6pt}}]{
		\foreach \x in {-2,-.5,1.5,3} \draw[->] (\x,0) -- (\x,2); \node at (-2,0) [below] {$1$}; \node at (3,0) [below] {$N$};
		\foreach \xx in {-1,...,1} \draw (-1.25+.2*\xx,1) node{$\cdot\mathstrut$} (2.25+.2*\xx,1) node{$\cdot\mathstrut$};
		\draw[rounded corners=3.5pt,->] (1,0) node[below]{$i{+}1$} -- (1,.5) -- (0,1.5) -- (0,2);
		\draw[cross line,rounded corners=3.5pt,->] (0,0) node[below]{$i$} -- (0,.5) -- (1,1.5) -- (1,2);
	} \, , \\
	T^{-1}_i & =
	\tikz[baseline={([yshift=-.5*11pt*0.4*.9+6pt]current bounding box.center)},	xscale=0.5*.9,yscale=0.4*.9,font=\scriptsize, cross line/.style={-,preaction={draw=white,-,line width=6pt}}]{
		\foreach \x in {-2,-.5,1.5,3} \draw[->] (\x,0) -- (\x,2); \node at (-2,0) [below] {$1$}; \node at (3,0) [below] {$N$};
		\foreach \xx in {-1,...,1} \draw (-1.25+.2*\xx,1) node{$\cdot\mathstrut$} (2.25+.2*\xx,1) node{$\cdot\mathstrut$};
		\draw[rounded corners=3.5pt,->] (0,0) node[below]{$i$} -- (0,.5) -- (1,1.5) -- (1,2);
		\draw[cross line,rounded corners=3.5pt,->] (1,0) node[below]{$i{+}1$} -- (1,.5) -- (0,1.5) -- (0,2);
	} \, ,
	\end{aligned}
	\qquad\qquad T_{\{i_1,\To,i_M\}} = 
	\tikz[baseline={([yshift=-.5*11pt*0.4*.9]current bounding box.center)},	xscale=0.5*.9,yscale=0.4*.9,font=\scriptsize, cross line/.style={-,preaction={draw=white,-,line width=6pt}}]{
		\foreach \x in {7,8.5} \draw[->] (\x,0) -- (\x,5); \node at (8.5,0) [below] {$N$};
		\foreach \xx in {-1,...,1} \node at (7.75+.2*\xx,2.5) {$\cdot\mathstrut$};
		\draw[rounded corners=3.5pt,->] (6,0) -- (6,3.5) -- (5,4.5) -- (5,5);
		\draw[rounded corners=3.5pt,->] (5,0) -- (5,2.5) -- (3,4.5) -- (3,5);
		\draw[rounded corners=3.5pt,->] (4,0) -- (4,1.5) -- (2,3.5) -- (2,5);
		\draw[rounded corners=3.5pt,->] (3,0) -- (3,.5) -- (0,3.5) -- (0,5);
		\draw[cross line,rounded corners=3.5pt,->] (2,0) node[below]{$M$} -- (2,.5) -- (6,4.5) -- (6,5) node[above]{$i_M$};
		\draw[cross line,rounded corners=3.5pt,->] (1,0) node[below]{$\cdots$} -- (1,1.5) -- (4,4.5) -- (4,5) node[above]{$\cdots$};
		\draw[cross line,rounded corners=3.5pt,->] (0,0) node[below]{$1$} -- (0,2.5) -- (1,3.5) -- (1,5) node[above]{$i_1$};
	} \, .
\end{gather}
\end{subequations}
(Here we only included the inverse generator for completeness. Note that these diagrams differ from \eqref{eq:Sij_left}, as the lines do not carry parameters and there are under- and overcrossings.)

\begin{theorem} \label{thm:phys_vector_spinchain}
Any eigenvector $\ket{\Psi} \in \mathcal{H}_M$ in the $M$-particle sector of the \textsf{q}-deformed Haldane--Shastry spin chain is determined by some symmetric polynomial $\widetilde{\Psi}(z_1,\To,z_M)$ via
\begin{equation} \label{eq:intro_wavefn}
	\normalfont
	\Psi(i_1,\To,i_M) = \cbraket{i_1,\To,i_M}{\Psi} = \ev \Bigl(T^\text{pol}_{\{i_1,\To,i_M\}} \widetilde{\Psi}(z_1,\To,z_M) \Bigr) \, .
\checkedMma % physical wave function
\end{equation}
\end{theorem}
\noindent In \textsection\ref{s:explicit_evrs} we will prove this powerful new result for the explicit eigenvectors that we will give just below. Thanks to the nonabelian symmetries (see \textsection\ref{s:intro_nonabelian_1}, especially the dimension counting) this extends to the full Hilbert space $\mathcal{H}$.

The wave function $\Psi$ depends on the positions (sites) $i_m \in \mathbb{Z}_N$ of the magnons on the spin chain, while the polynomial~$\widetilde{\Psi}$ depends on the $M$~`\textsf{q}-magnon coordinates'~$z_m$, which we think of as being transported by the Hecke action to the same location~$z_{i_m} \in S^1\subset \mathbb{C}$ upon evaluation. Whereas $\widetilde{\Psi}$ is a symmetric polynomial in $M$~variables, its image under $T_{\{i_1,\To,i_M\}}^{\text{pol}}$ is not symmetric and depends on $i_M > M$ variables. In particular, the Hecke operators in \eqref{eq:intro_wavefn} act nontrivially even though $\widetilde{\Psi}$ is symmetric.

Now we return to the decomposition \eqref{eq:motif_decomp} of the Hilbert space into joint eigenspaces of the abelian symmetries. For every motif we have obtained the wave function that \textsf{q}-deforms \eqref{eq:HS_polynomial}:

\begin{theorem} \label{thm:nice_polynomial} 
For every $\mu \in \mathcal{M}_N$ the eigenspace $\mathcal{H}^\mu$ contains a unique (up to rescaling) vector $\ket{\mu} \in \mathcal{H}^\mu \cap \mathcal{H}_M$ at $M=\ell(\mu)$. If $\nu$ is the partition associated to $\mu$ as in \eqref{eq:motifs_vs_partitions} its wave function is determined through \eqref{eq:intro_wavefn} by the `symmetric square' of the \textsf{q}-$\mspace{-2mu}$Vandermonde polynomial times a Macdonald polynomial:
\begin{equation} \label{eq:intro_polynomial}
	\widetilde{\Psi}_{\nu}(z_1,\To,z_M) \coloneqq \Biggl( \, \prod_{m<n}^M (\mathsf{q} \, z_m - \mathsf{q}^{-1} z_n) \, (\mathsf{q}^{-1} z_m - \mathsf{q}\,z_n) \! \Biggr) P_{\nu}^\star (z_1,\To,z_M) \, ,
\checkedMma
\end{equation}
where $P_\nu^\star$ denotes the special case $\mathsf{p}^\star = \mathsf{q}^\star = \mathsf{q}^2$ of a Macdonald polynomial.
\end{theorem}
\noindent This key result will be established in \textsection\ref{s:explicit_evrs}.
 
Macdonald polynomials $P_{\nu}$ are a family of homogeneous symmetric polynomials depending on two parameters (\textsection\ref{s:Macdonald}; Figure~\ref{fg:Macdonalds} on p.\,\pageref{fg:Macdonalds}). In the notation $q \equiv \mathsf{p}$, $t\equiv \mathsf{q}^2$ of Macdonald~\cite{Mac_95} we are dealing with the special case $q^\star = t^{\star \,\alpha}$ for (Jack) parameter $\alpha =1/2$. That is, $P_\nu^\star$ is (essentially) a \emph{quantum spherical zonal polynomial}, related to harmonic analysis on a quantum homogeneous space~\cite{Nou_96}.\,%
\footnote{\ \label{fn:casimirs} The pairwise form of e.g.\ \eqref{eq:ham_left} looks like an affine ($\vect{z}$-dependent) version of the trace formula for the first quantum Casimir operator~$C_1$ of $U_{\mathsf{q}}(\mathfrak{gl}_N)$ from \cite{RTF_89}; cf.~\mbox{(5.12)} in \cite{Nou_96}. Noumi showed \cite{Nou_96} (`case Sp' of Theorem~5.2) that the radial component of $C_1$ is the quantum zonal spherical operator, i.e.\ $D_1^\star$ in our notation. This must be closely related to the appearance of $P^\star_\nu$ in \eqref{eq:intro_polynomial}.}
In \textsection\ref{s:intro_examples} we give several concrete examples to get some feeling for these $P_\nu^\star$. 

\begin{remarks} \label{rmk:direct_verification}
\textbf{i}.~There is no freedom to pick $\alpha$: the value $\alpha = 1/2$ rolls out of the proof in \textsection\ref{s:explicit_evrs}. 

\textbf{ii}.~The dependence of $P_\nu^\star$ on $\mathsf{q}^2$ reflects a symmetry of $H^\textsc{l}$ under $\mathsf{q} \mapsto -\mathsf{q}$, see \eqref{eq:ham_left_q_to_-q} in \textsection\ref{s:app_stochastic_twist}. 

\textbf{iii.}~At the end of \textsection\ref{s:intro_crystal_limit} we will sketch a proof showing that the eigenvectors given by Theorems \ref{thm:phys_vector_spinchain}--\ref{thm:nice_polynomial} are nonzero despite the evaluation\,---\,which, after all, kills some polynomials.

\textbf{iv.}~In \textsection\ref{s:intro_nonabelian_2} we will see that Theorem~\ref{thm:nice_polynomial} describes (all) highest-weight eigenvectors. The relation \eqref{eq:Jack_shift_property} extends to Macdonald polynomials (\textsection{VI.4} in \cite{Mac_95}),
\begin{equation} \label{eq:Macdonald_shift_property}
	P_\nu(z_1,\To,z_M) = z_1 \cdots z_M \, P_{\mspace{1mu}\bar{\nu}}(z_1,\To,z_M) \, , \qquad \nu = \bar{\nu} + (1^M) \, ,
\end{equation} 
so that \eqref{eq:intro_polynomial} is divisible by $z_1 \cdots z_M$. This is a highest-weight condition, see Theorem~\ref{eq:thm_hw_condition}. 

\textbf{v.} The property \eqref{eq:largest_monomial} carries over to the polynomial \eqref{eq:intro_polynomial} since
\begin{equation*}
	\prod_{m<n}^M (\mathsf{q}^{\pm1} \, z_m-\mathsf{q}^{\mp1} z_n) = \mathsf{q}^{\pm M(M-1)/2} \, \vect{z}^{\delta_M} + \text{lower} \, ,  \qquad 
	P_\nu(z_1,\To,z_M) = \vect{z}^\nu + \text{lower} \, .
\end{equation*}

\textbf{vi.}~Our proof of Theorem~\ref{thm:nice_polynomial} (in \textsection\ref{s:explicit_evrs}) is separate from the derivation of the corresponding eigenvalues from Theorems \ref{thm:intro_energy_left} and~\ref{thm:ham_energy_right}~(ii). At present we do not properly understand the direct connection between the two derivations; we'll return to this in the future. For general $\mathsf{q} \in \mathbb{C}^\times$ the multi-spin interactions make it rather hard to verify by direct computation that
\begin{equation} \label{eq:intro_eigenvalue_eqns}
	H^\textsc{l} \, \ket{\mu} = E^\textsc{l}(\mu) \, \ket{\mu} \, , \qquad
	H^\textsc{r} \, \ket{\mu} = E^\textsc{r}(\mu) \, \ket{\mu} \, .
\end{equation}
At the isotropic point $\mathsf{q}=1$ these eigenvalue equations can be confirmed analytically~\cite{Hal_91b}. We have further checked \eqref{eq:intro_eigenvalue_eqns} for all $\mu \in \mathcal{M}_N$, $N\leq 10$ numerically with (pseudo)random values of $\mathsf{q} \in \mathbb{C}^\times$. Yet another verification comes from $\mathsf{q} \to \infty$, see \textsection\ref{s:intro_crystal_limit}.

\textbf{vii.}~As a corollary of our description of the eigenvectors we get the following `on-shell identities' for quantum spherical zonal polynomials. Denote the \emph{left} eigenvector with the same eigenvalues~\eqref{eq:energy_left}, \eqref{eq:energy_right} and \eqref{eq:mtm} by $\bra{\mu}$. If $\mathsf{q}\in\mathbb{R}^\times$ it is the complex transpose of $\ket{\mu}$ since $H^\textsc{l}$ is hermitian in that case~\cite{Lam_18}. Write $\,\cdot\,^*$ for complex conjugation. For generic $\mathsf{q}\in\mathbb{C}^\times$ the wave function of $\bra{\mu}$ is determined by \eqref{eq:intro_polynomial} as
\begin{equation} \label{eq:wavefn_hecke_left}
	\bracket{\mu}{i_1,\To,i_M} = \mathrm{ev}_\omega^* \Bigl(T^\text{pol}_{\{i_1,\To,i_M\}} \widetilde{\Psi}_\nu(z_1,\To,z_M) \Bigr) \, , \qquad \mathrm{ev}_\omega^* = \mathrm{ev}_{\omega^*} = \mathrm{ev}_{\omega^{-1}} \, .
\checkedMma
\end{equation}
For two motifs $\mu^{(1)}, \mu^{(2)} \in\mathcal{M}_N$ of the same length~$\ell(\mu^{(1)}) = \ell(\mu^{(2)}) = M$ let $\nu^{(1)},\nu^{(2)}$ be the associated partitions. The corollary following \eqref{eq:motif_decomp} implies that the polynomials \eqref{eq:intro_polynomial} obey the \emph{on-shell identities}
\begin{equation*}
	\normalfont
	\sum_{i_1<\cdots < i_M}^N \!\!\!\!\!\! \mathrm{ev}_\omega^*  \Bigl(T^\text{pol}_{\{i_1,\To,i_M\}} \widetilde{\Psi}_{\nu^{(1)}}(z_1,\To,z_M) \Bigr) \, \ev \Bigl( T^\text{pol}_{\{i_1,\To,i_M\}} \widetilde{\Psi}_{\nu^{(2)}}(z_1,\To,z_M) \Bigr) \propto \delta_{\mu^{(1)}\!,\, \mu^{(2)}} \, .
\checkedMma
\end{equation*}
The evaluation is crucial; for example, for $\mu^{(1)}=(1)$ and $\mu^{(2)}=(2)$ it is the evaluation of
\begin{equation*}
 	\sum_{i=1}^N \, \Bigl(T^\text{pol}_{(i,\To,1)} \, z_1 \Bigr)^{\!*} \, \Bigl(T^\text{pol}_{(i,\To,1)} \, z_1^2 \Bigr) = \mathsf{q}^{-2\,(N-1)} \, \sum_{i=1}^N z_i \, , \qquad z_i^* \coloneqq z_i^{-1} \, ,
 \checkedMma
\end{equation*}
which only vanishes upon evaluation. We do not have a direct proof of these on-shell identities, nor a general expression for the norms of our eigenvectors.
\end{remarks}

\subsubsection{Nonabelian symmetries: details} \label{s:intro_nonabelian_2}
The eigenvector determined by \eqref{eq:intro_wavefn}--\eqref{eq:intro_polynomial} has highest weight in the appropriate sense and gives rise to all other eigenvectors in $\mathcal{H}^\mu$ through the action of the nonabelian symmetries. These are constructed as follows.

Let $s_{ij} \coloneqq s_{(ij)}$ swap $z_i \leftrightarrow z_j$, so $s_{i,i+1} =s_i$. Following \cite{BG+_93} we set
\begin{equation} \label{eq:intro_xij}
	x_{ij} \coloneqq 
	\begin{cases}
		f(z_i/z_j)^{-1} - f(z_i/z_j)^{-1} \, g(z_i/z_j) \, s_{ij} \, , \qquad & i<j \, , \\ 
		f(z_i/z_j)^{-1} + f(z_j/z_i)^{-1} \, g(z_j/z_i) \, s_{ji} \, , \qquad & i>j \, .
	\end{cases}
\end{equation}
Then $x_{i,i+1} = T^\text{pol}_i \, s_i$, we have $x_{ji} = x_{ij}^{-1}$, and the $x_{ij}$ obey the Yang--Baxter equation. The `$\!$\textit{Y}-operators'
\begin{equation} \label{eq:intro_Yi_circ}
	Y_i^\circ \coloneqq x_{i,i+1} \, x_{i,i+2} \cdots x_{iN} \, x_{i1} \cdots x_{i,i-2} \, x_{i,i-1}
\end{equation}
mutually commute. We use the superscript `${}^\circ$' to indicate that these are the `classical' (i.e.\ no difference part) version of \textsf{q}-deformed Dunkl operators $Y_i$ from the affine Hecke algebra, which will be reviewed in \textsection\ref{s:Hecke_AHA}.

The monodromy matrix of the \textsf{q}-deformed Haldane--Shastry spin chain is obtained from that of the Heisenberg \textsc{xxz} spin chain (\textsection\ref{s:spin_chains}) by `quantising' its inhomogeneity parameters to (the inverse of) the \textit{Y}-operators~\eqref{eq:intro_Yi_circ}.

\begin{theorem}[\cite{BG+_93}] \label{thm:monodromy_frozen}
Consider an auxiliary space $V_a \cong \mathbb{C}^2$. Using $R(u) = P\,\check{R}(u)$ define the monodromy matrix on $V_a \otimes \mathcal{H}$ as
\begin{equation} \label{eq:intro_L_spin_chain}
	L_a(u) = \ev \widetilde{L}^\circ_a(u) \, , \qquad \widetilde{L}^\circ_a(u) = R_{aN}(u\,Y_N^\circ) \cdots R_{a1}(u\,Y_1^\circ) \, .
\end{equation}
Written as a $2\times 2$ matrix on $V_a$ its entries, acting on $\mathcal{H}$, preserve the property from Theorem~\ref{thm:phys_vector_spinchain} (possibly changing the value of $M$) and commute with the abelian symmetries. In addition, \eqref{eq:intro_L_spin_chain} obeys the `$\mspace{-1mu}$\textit{RLL}-relations', thus turning $\mathcal{H}$ into a representation of $\widehat{\mathfrak{U}}$.
\end{theorem}
\noindent We recall the proof from \cite{BG+_93} and give the corresponding Chevalley generators in \textsection\ref{s:nonabelian_freezing}. 

There are a couple of ways to understand \eqref{eq:intro_L_spin_chain}. We take the following viewpoint. Theorem~\ref{thm:phys_vector_spinchain} guarantees that each vector is determined by a polynomial as in \eqref{eq:intro_wavefn}. Thus one can work with \eqref{eq:intro_L_spin_chain} by letting the $Y_i^\circ$ act on the polynomials prior to evaluation. That is, we really define \eqref{eq:intro_L_spin_chain} to mean (cf.~\eqref{eq:ev_operator} in \textsection\ref{s:abelian_freezing})
\begin{equation} \label{eq:intro_L_spin_chain_def}
	L_a(u) \, \ket{\Psi} = \ev \Biggl( \widetilde{L}^\circ_a(u) \!\! \sum_{i_1 < \cdots < i_M}^N \!\!\!\!\! T_{\{i_1,\To,i_M\}}^\text{pol} \widetilde{\Psi}(\vect{z}) \, \cket{i_1,\To,i_M} \Biggr) \, .
\end{equation}

\begin{remark} 
The approach taken in \cite{BG+_93} is essentially as follows. One can check that any vector of the form described in Theorem~\ref{thm:phys_vector_spinchain} has the property that \emph{prior to evaluation} the coordinate transposition $s_i: z_i \leftrightarrow z_{i+1}$, acts in the same way as $\check{R}_{i,i+1}(z_i/z_{i+1})$. (This is no coincidence: see \textsection\ref{s:intro_abelian_tilde}.) This property, which is well-known in the context of the quantum Knizhnik--Zamolodchikov (\textsf{q}KZ) equation~\cite{FR_92,Che_92b}, enables one to translate the \textit{Y}$\!$-operators in \eqref{eq:intro_L_spin_chain} into an action on spins. Indeed, decompose \eqref{eq:intro_Yi_circ} into simple transpositions~$s_i$, and move these all the way to right to exchange them one by one for \textit{R}-matrices. In \cite{BG+_93} this procedure was called the \emph{projection onto the physical space}. This shows that it is possible to turn \eqref{eq:intro_L_spin_chain} into an operator that acts on spins only, up to some rational factor in the $z_j$ as for other operators encountered so far. The resulting monodromy matrix will still obey the \textit{RLL}-relations and commute with the abelian symmetries (\textsection\ref{s:nonabelian_tilde}). This procedure is feasible in the isotropic case (where $\check{R}(u) \to P$); this is how the Yangian generators~\eqref{eq:Yangian_gens} of \cite{HH+_92} can be obtained from the monodromy matrix of \cite{BG+_93}. In the \textsf{q}-deformed setting, however, it is much more cumbersome, cf.~\cite{Ugl_95u}. In \textsection\ref{s:proofs} we will use various tricks to do this efficiently, yielding the expressions for the abelian symmetries (\textsection\ref{s:abelian_tilde})\,---\,or to avoid it altogether, as in our proof of the highest-weight property (\textsection\ref{s:hw}). Unfortunately we have not yet been able to get a `projected' (spin-only) form of the nonabelian symmetries.
\end{remark}

As usual the monodromy matrix~\eqref{eq:intro_L_spin_chain}, viewed as a matrix in auxiliary space, contains four `quantum' operators
\begin{equation} \label{eq:intro_ABCD}
	L_a(u) = 
	\begin{pmatrix}
	A(u) & B(u) \\ C(u) & D(u)
	\end{pmatrix}_{\!a}
\end{equation}
that generate the $\widehat{\mathfrak{U}}$-action on $\mathcal{H}$. Unlike for the Heisenberg \textsc{xxz} spin chain these \emph{commute} with the Hamiltonian, so the commutation relations between the four quantum operators is not important for us. Since $L_a(u)$ commutes with the abelian symmetries, it follows that the joint eigenspace $\mathcal{H}^\mu$ is a $\widehat{\mathfrak{U}}$-module as announced in \textsection\ref{s:intro_nonabelian_1}. It has highest weight in the following sense, which \textsf{q}-deforms the notion of Yangian highest weight.

\begin{definition} \label{p:intro_hw_def}
A $\widehat{\mathfrak{U}}$-module has \emph{pseudo highest weight} \cite[\textsection12.2]{CP_94} if it contains a vector $\ket{\mu}$ that is an eigenvector for $A(u)$ and $D(u)$, and annihilated by $C(u)$, for all $u$. (In \cite{Nak_01} this is called \emph{\textit{l}-highest weight}.)
\end{definition}

\noindent The Drinfeld polynomial characterises the eigenvalues of a pseudo highest-weight vector for the diagonal entries in \eqref{eq:intro_ABCD} up to a common normalisation~\cite{CP_91,JK+_95b}:\,%
\footnote{\ This differs from \cite{CP_91} by inverting $\mathsf{q}$ on the right-hand side. In \cite{JK+_95b} that is accounted for via $\alpha^\mu(u^{-1})/\delta^\mu(u^{-1}) = \mathsf{q}^{\deg P^\mu} P^\mu(\mathsf{q}^{-2}\,u)/P^\mu(u)$. That convention would require inverting $\mathsf{q}$ in \eqref{eq:Drinfeld_polyn}.}
\begin{equation} \label{eq:intro_ABCD_on_shell}
	\begin{aligned}
	& \begin{aligned} 
	A(u) \, \ket{\mu} & = \alpha^\mu(u) \, \ket{\mu} \, , \\
	D(u) \, \ket{\mu} & = \delta^\mu(u) \, \ket{\mu} \, ,
	\end{aligned}
	\qquad \quad
	\frac{\alpha^\mu(u)}{\delta^\mu(u)} = \mathsf{q}^{-\!\deg P^\mu} \, \frac{P^\mu(\mathsf{q}^{2}\,u)}{P^\mu(u)} \, , \\
	& \mspace{1mu} C(u) \, \ket{\mu} = 0 \, .
	\end{aligned} 
\checkedMma
\end{equation} 
Thus $B(u)$ acts as a lowering operator, generating from $\ket{\mu}$ all other vectors in $\mathcal{H}^\mu$. 

The derivation of our eigenvectors in \textsection\ref{s:explicit_evrs} in fact shows that $\widetilde{\Psi}_\nu$ from \eqref{eq:intro_polynomial} yields an eigenvector for \emph{any} partition~$\nu$ with $\ell(\nu) \leq M$. In \textsection\ref{s:hw} we will follow \cite{BPS_95a} to prove
\begin{theorem} \label{eq:thm_hw_condition}
The condition $\ell(\nu)=M$ for the partition in \eqref{eq:intro_polynomial} ensures that the eigenvector from Theorem~\ref{thm:nice_polynomial} has pseudo highest weight with Drinfeld polynomial~\eqref{eq:Drinfeld_polyn}.
\end{theorem}
\noindent The partitions with $\ell(\nu)<M$ account for all (non-affine) $\mathfrak{U}$-descendants of $\ket{\mu}$. We do not yet have direct expressions for its affine descendants.

\subsubsection{Crystal limit} \label{s:intro_crystal_limit}
Let us explore the freedom of having a new parameter to play with and consider the crystal limit $\mathsf{q}\to\infty$ \cite{Kas_90}; see also \cite{Jim_92}.\,%
\footnote{\ One can also let $\mathsf{q}\to 0$. By `\textsc{cpt}' from Footnote~\ref{fn:CPT} on p.\,\pageref{fn:CPT} the resulting limits differ by inverting the argument (for the \textit{R}-matrix) and spin reversal $\ket{\uparrow}\leftrightarrow\ket{\downarrow}$. The corresponding crystal Hamiltonians thus are as in \eqref{eq:crystal_ham} but now count domain walls of the form $\cdots\! \uparrow\downarrow\!\cdots$. The energy thus stays the same when domains are extended to the \emph{right}, and all eigenvectors are effectively left-right flipped. The crystal dispersions in~\eqref{eq:crystal_dispersion} are swapped; cf.~the relations in \eqref{eq:energy_right}. One is led to the same notion of motifs.} 
In this extreme case the spin chain simplifies drastically and we obtain a simple combinatorial model. Since the \textit{R}-matrix and potential depend continuously on the deformation parameter this provides a useful toy model to understand the structure of the Haldane--Shastry model, in the \textsf{q}-deformed case as well as at the isotropic point.

As $\mathsf{q}\to\infty$ the potential~\eqref{eq:pot} diverges, with constant leading term: $([4]/[2]) \,V\mspace{-1mu}(z_i,z_j) \to 1$.
\checkedMma % NB. rescaling by $[3]$ works too
For the spin interactions we need to determine the limits of the Temperley--Lieb generator and \textit{R}-matrix. Rescaling the former to be a projector we get
\begin{equation*}
	\begin{aligned}
		\frac{e^\text{sp}}{[2]} & \to \chi \coloneqq \ket{\downarrow\uparrow} \bra{\downarrow\uparrow} = \diag(0,0,1,0) \, , \\
		\check{R}(u) & \to u^{-\chi} = \diag\bigl(1,1,u^{-1},1\bigr)\, ,
	\end{aligned}
	\qquad \mathsf{q} \to \infty \, ,
	\checkedMma
\end{equation*}
One can check that $[2]^{-1} \, S_{[i,j]}^{\mspace{1mu}\textsc{l}} \to \chi_{i,i+1}$ (and $[2]^{-1} \, S_{[i,j]}^{\mspace{1mu}\textsc{r}} \to \chi_{j-1,j}$) 
\checkedMma
concentrate at the left (right) to become diagonal nearest-neighbour operators independent of $j$ ($i$, respectively). The Hamiltonians~\eqref{eq:ham_left} and~\eqref{eq:ham_right} thus become simple diagonal matrices:

\begin{lemma}
	In the crystal limit the Hamiltonians count domain walls $\cdots\! \downarrow\uparrow\!\cdots$, weighted by their distance from the last (first) site for $\normalfont\bar{H}^\textsc{l}$ ($\normalfont\bar{H}^\textsc{r}$):
	\begin{equation} \label{eq:crystal_ham}
		\normalfont
		\begin{aligned}
			\frac{[4]\mspace{1mu}N}{[2]^2\mspace{1mu}[N]} \, H^\textsc{l} & \to \bar{H}^\textsc{l} \coloneqq \sum_{i=1}^{N-1} (N-i) \,
			\chi_{i,i+1} \, , \\
			\frac{[4]\mspace{1mu}N}{[2]^2\mspace{1mu}[N]} \, H^\textsc{r} & \to \bar{H}^\textsc{r} \coloneqq \sum_{j=2}^N \, (j-1) \, \chi_{j-1,j} \, , \\
		\end{aligned}
		\qquad\qquad \mathsf{q} \to \infty \, .
		\checkedMma
	\end{equation}
\end{lemma}

It is instructive to work out the representation theory, cf.~\textsection0 of~\cite{Ugl_95u}. Recall from~\eqref{eq:coord_basis} our notation $\cket{\,\cdot\,}$ for the coordinate basis. At $M=0$ there are no domain walls, so $\cket{\varnothing} = \ket{\uparrow\cdots \uparrow}$ has eigenvalue zero. We can flip spins without affecting the energy as long as we do not create any domain wall~$\downarrow\uparrow$. The joint kernel of the crystal Hamiltonians thus consists of the $N+1$ vectors 
\begin{equation*}
	\ket{ \uparrow\cdots\uparrow\uparrow\uparrow\uparrow} \, , \qquad 
	\ket{ \uparrow\cdots\uparrow\uparrow\uparrow\downarrow} \, , \qquad 
	\ket{ \uparrow\cdots\uparrow\uparrow\downarrow\downarrow} \, ,
	\qquad \dots \, , \qquad 
	\ket{ \downarrow\cdots\downarrow\downarrow\downarrow\downarrow} \, .
\end{equation*}
For $M=1$ the coordinate basis $\cket{n} = \sigma_n^- \, \cket{\varnothing}$ gives $N-1$ new eigenvectors (excluding $\cket{N}  \in \ker \bar{H}^\textsc{l} = \ker \bar{H}^\textsc{r}$), with linear energy
\begin{equation} \label{eq:crystal_dispersion}
	\bar{H}^\textsc{l} \, \cket{n} = (N-n) \, \cket{n} \, , \qquad \bar{H}^\textsc{r} \, \cket{n} = n \, \cket{n} \, , \qquad n<N \, .
\end{equation}
This is consistent with the crystal limit of \eqref{eq:energy_left}--\eqref{eq:energy_right}. Again the energy is unchanged if we flip spins avoiding $\downarrow\uparrow$, giving the $n\,(N-n)$ vectors
\begin{equation} \label{eq:crystal_magnon}
	\begin{aligned}
		& \ket{ \uparrow\cdots\uparrow\uparrow\underset{n}{\downarrow}\uparrow\uparrow\cdots \uparrow\uparrow } \, , \quad \
		& & \ket{ \uparrow\cdots\uparrow\uparrow\underset{n}{\downarrow}\uparrow\uparrow\cdots \uparrow\downarrow }\, , \quad \
		& & \dots \, , \quad \
		& & \ket{ \uparrow\cdots\uparrow\uparrow\underset{n}{\downarrow}\uparrow\downarrow\cdots \downarrow\downarrow } \, , \\
		& \ket{ \uparrow\cdots\uparrow\downarrow\underset{n}{\downarrow}\uparrow\uparrow\cdots \uparrow\uparrow } \, , 
		& & \ket{ \uparrow\cdots\uparrow\downarrow\underset{n}{\downarrow}\uparrow\uparrow\cdots \uparrow\downarrow }\, , 
		& & \dots \, ,
		& & \ket{ \uparrow\cdots\uparrow\downarrow\underset{n}{\downarrow}\uparrow\downarrow\cdots \downarrow\downarrow } \, , \\
		& \qquad\qquad\!\!\vdots & & \qquad\qquad\!\!\vdots & & \ \ddots & & \qquad\qquad\!\!\vdots \\
		& \ket{ \downarrow\cdots\downarrow\downarrow\underset{n}{\downarrow}\uparrow\uparrow\cdots \uparrow\uparrow } \, , 
		& & \ket{ \downarrow\cdots\downarrow\downarrow\underset{n}{\downarrow}\uparrow\uparrow\cdots \uparrow\downarrow }\, , 
		& & \dots \, ,
		& & \ket{ \downarrow\cdots\downarrow\downarrow\underset{n}{\downarrow}\uparrow\downarrow\cdots\downarrow\downarrow } \, .
	\end{aligned}
\end{equation}
At $M=2$ the new (highest-weight) vectors have two $\downarrow\uparrow$s, so they are of the form $\cket{\mu_1,\mu_2}$ where the two excited spins are at least one apart and $\mu_2 <N$. Each of these has descendants with the same energy, obtained by extending the domains of $\downarrow$s\,---\,or starting a new domain all the way at the right\,---\,to the left without merging any domains like in \eqref{eq:crystal_magnon}. Continuing in this way naturally leads to the notion of motifs given in \eqref{eq:motif}.

In terms of the spectrum of the crystal Hamiltonians the motif $\mu\in \mathcal{M}_N$ corresponds to the (highest-weight) vector $\cket{\mu} = \prod_{i\in\mu} \sigma^-_i \cket{\varnothing} \in \mathcal{H}_M$ with $M = \ell(\mu)$. The motif condition ensures that the pattern $\downarrow\uparrow$ occurs exactly $M$ times; in other words, that it has $M$ domains (of size one) with spin~$\downarrow$. In the crystal limit $\cket{\mu}$ thus has `crystal energy' eigenvalues
\begin{equation} \label{eq:crystal_energies}
	\bar{E}^{\mspace{1mu}\textsc{l}}(\mu) = M\,N - |\mu| \, , \qquad \bar{E}^{\mspace{1mu}\textsc{r}}(\mu) = |\mu| \, ,
\end{equation}
where we recall the notation $|\mu| \coloneqq \sum_m \mu_m$. 
As before these eigenvalues stay the same if we extend the domains to the left without merging. Let $\bar{\mathcal{H}}^\mu$ denote the linear span of all vectors obtained from $\cket{\mu}$ in this way. Its dimension clearly is as in \eqref{eq:dim}. Any coordinate basis vector lies in $\bar{\mathcal{H}}^\mu$ with $\mu \in \mathcal{M}_N$ recording the locations of its $\downarrow\uparrow$s. This yields the crystal analogue of the decomposition~\eqref{eq:motif_decomp}:
	
\begin{proposition}[\cite{Ugl_95u}] In the crystal limit the decomposition of the Hilbert space into joint eigenspaces for $\normalfont \bar{H}^\textsc{l}$ and $\normalfont \bar{H}^\textsc{r}$ from \eqref{eq:crystal_ham} is labelled by motifs as
\begin{equation} \label{eq:crystal_motif_decomp}
	\normalfont
	\mathcal{H} = \bigoplus_{\mu \,\in\, \mathcal{M}_N} \!\!\! \bar{\mathcal{H}}^\mu \, , \qquad \mathsf{q}\to\infty \, .
\end{equation}
The pseudo highest-weight vector in $\bar{\mathcal{H}}^\mu$ is $\cket{\mu}$, with eigenvalues \eqref{eq:crystal_energies}, and multiplicity \eqref{eq:dim}.
\end{proposition}

The \textsf{q}-translation operator~\eqref{eq:q-translation} also becomes a simple domain-wall counting operator in the crystal limit:
	
\begin{lemma} \label{lem:crystal_q-translation}
In the crystal limit the \textsf{q}-translation operator reduces to
\begin{equation} \label{eq:crystal_q-translation}
	G \to \bar{G} \coloneqq \ev \; (z_1/z_2)^{-\chi_{12}} \cdots (z_1/z_N)^{-\chi_{N-1,N}} = \exp\Biggl(\frac{2\pi \mspace{1mu}\I}{N} \sum_{k=1}^{N-1} k \, \chi_{k,k+1} \Biggr) \, , \qquad \mathsf{q} \to \infty \, .
\checkedMma
\end{equation}
\end{lemma}
\noindent We can directly read off the `crystal momentum' of $\cket{\mu}$ and its descendants:
\begin{equation} \label{eq:crystal_mtm}
	\bar{p}(\mu) \coloneqq \frac{2\pi}{N} \, |\mu| \mod 2\pi \, .
\checkedMma
\end{equation}
This is one way of computing the eigenvalue~\eqref{eq:mtm}.
	
To conclude the discussion of the crystal limit let us verify that $\cket{\mu}$ is indeed obtained from our pseudo highest-weight eigenvectors as $\mathsf{q}\to\infty$. The polynomial $\widetilde{\Psi}_\nu$ from~\eqref{eq:intro_polynomial} becomes a Schur polynomial:
\begin{equation} \label{eq:intro_pol_crystal}
	\mathsf{q}^{-M(M-1)} \, \widetilde{\Psi}_\nu \to (-1)^{M(M-1)/2} \, (z_1\cdots z_M)^{M-1} \, s_\nu \, , \qquad \mathsf{q} \to \infty \, .
\checkedMma
\end{equation}
(The definition of $s_\nu$ is recalled in~\eqref{eq:Schur} below.) Indeed, $P^\star_\nu = s_\nu + \mathcal{O}(\mathsf{q}^{-1})$; cf.~the examples in Tables~\ref{tb:Macdonald_N=8}--\ref{tb:Macdonald_N=10} in \textsection\ref{s:intro_examples}. The Hecke generators $T_i^\text{pol} = \mathsf{q} \, \bigl(\bar{\pi}_i + \mathcal{O}(\mathsf{q}^{-1})\bigr)$ give rise to (idempotent) \emph{0-Hecke} generators $\bar{\pi}_i = (z_i - z_{i+1})^{-1} \, (z_i \, s_i - z_{i+1})$.
\checkedMma % \cite{Nor_79}
Thus the leading components of the resulting eigenvectors are those that maximise the length $\ell(\{i_1,\To,i_M\}) = \sum_m i_m - M(M+1)/2$ while surviving the action of 0-Hecke. It can be shown that a single component dominates the crystal limit, reproducing the eigenvectors described above up to a phase on shell:
	
\begin{proposition} \label{prop:crystal_eigenvector}
The crystal limit of the eigenvectors from Theorem~\ref{thm:nice_polynomial} is given by
\begin{equation} \label{eq:intro_eigenvector_crystal}
	\normalfont
	\mathsf{q}^{-|\mu|-M(M-3)/2} \!\! \sum_{i_1 < \cdots < i_M}^N \!\!\!\!\! T_{\{i_1,\To,i_M\}}^\text{pol} \widetilde{\Psi}_\nu(\vect{z}) \, \cket{i_1,\To,i_M} \to (-1)^{|\mu|-M} \, \vect{z}^\lambda \, \cket{\mu} \, , \qquad \mathsf{q}\to\infty \, ,
\checkedMma
\end{equation}
where $\lambda$ is the partition conjugate to $\mu^+$,
\begin{equation} \label{eq:motifs_vs_partitions_lambda}
	\lambda \coloneqq (M,\To,\underset{\mu_1\vphantom{a^k_k}}{M}, M{-}1, \To, \underset{\mu_2\vphantom{a^k_k}}{M{-}1}, \ \To \ , \underset{\mathclap{\mu_M}\vphantom{a^k_k}}{1},0,\To,0) \, , \qquad \lambda' = \mu^+ \, .
\end{equation} 
\end{proposition}
\noindent The latter is the partition associated to $\mu \in \mathcal{M}_N$ in \cite{Ugl_95u} following~\cite{JK+_95b}; see also~\textsection\ref{s:abelian_freezing}. The monomial $\vect{z}^\lambda$ in \eqref{eq:intro_eigenvector_crystal} certainly survives evaluation:
\begin{equation*}
	\ev \vect{z}^\lambda = \ev \prod_{m=1}^M \prod_{i=1}^{\mu_m} z_i = \prod_{m=1}^M \! \omega^{\mspace{1mu}\mu_m(\mu_m+1)/2} = \omega^{\sum_m \mu_m (\mu_m+1)/2} \, ,
\checkedMma
\end{equation*}
reproducing \mbox{(0.0.58)} of \cite{Ugl_95u}, derived in \textsection4.4 therein. This moreover means that our eigenvectors survive evaluation away from the crystal limit too:
\begin{corollary}[cf.~\cite{Ugl_95u}]
Each of the eigenvectors from Theorem~\ref{thm:nice_polynomial} survives the evaluation in Theorem~\ref{thm:phys_vector_spinchain} for generic values of $\mathsf{q} \in \mathbb{C}^\times$.
\end{corollary}
\noindent We leave the proofs and detailed description of the crystal limit, including the representation theory, for a separate publication.

\subsubsection{Examples} \label{s:intro_examples} Let us give some concrete examples of our pseudo highest-weight eigenvectors.

We start from the ferromagnetic ground state. Here $M=0$ with the empty motif, $\mu = 0$, and vanishing energy and \textsf{q}-momentum. The partition is the same, with $P_0^\star = 1$ trivial. This eigenspace (multiplet) has dimension $N+1$. Next, for $M=1$ the motif and partition coincide as well, $\mu = \nu = (n)$ for $1\leq n < N$, with $P_{(n)}^\star(z_1) = z_1^n$. In fact, as for any translationally invariant model, the  one-particle sector~$\mathcal{H}_1$ is already diagonalised by \textsf{q}-homogeneity. 
\begin{lemma} \label{lem:q-magnon}
By construction the `\textsf{q}-magnons'
\begin{equation} \label{eq:q-magnon}
	\sum_{j=1}^N \omega^{n\,j} \, G^{1-j} \, \cket{1} = \ev \sum_{j=1}^N z_j^n \, \widetilde{G}^{1-j} \, \cket{1} \, , \qquad 0\leq n < N \, ,
\checkedMma
\end{equation}
have $G$-eigenvalue $\omega^n$, i.e.\ \textsf{q}-momentum $p=2\pi n/N$.
\checkedMma  
\end{lemma} 

The proof is a direct verification. On shell \eqref{eq:q-magnon} matches our general eigenvectors: by a somewhat tedious calculation one can verify that
\begin{equation*} 
	\ev \, \cbra{1} \sum_{j=1}^N z_j^n \, \widetilde{G}^{1-j} \cket{1} \, = \, t^{(N-2\mspace{1mu}n+1)/2} \, \frac{N}{[N]} \, \ev z_1^n \, .
\checkedMma
\end{equation*}  
This yields $N$ linearly independent vectors 
\checkedMma 
(orthogonal for $\textsf{q} \in \mathbb{R}$)
\checkedMma
that span $\mathcal{H}_1$. One can check that $n=0$ gives $N/[N]$ times the $\mathfrak{U}$-descendant of the $M=0$ eigenvector. 
\checkedMma
For each $1\leq n<N$ the associated partition has $\ell(\nu)=M$, giving a highest-weight vector. According to \eqref{eq:dim} the corresponding eigenspace $\mathcal{H}^{(n)}$ has dimension $n \, (N-n)$, cf.~\eqref{eq:crystal_magnon}. Viewed as a module for $\mathfrak{U}$ its tensor-product decomposition is
\begin{equation*}
	n \otimes (N-n) = \!\! \bigoplus_{k=1}^{\min(n,N-n)} \!\!\!\!\!\!\! (N-2\,k+1) = (N-1) \oplus (N-3) \oplus \cdots \oplus (|N-2\,n|+1)\, ,
\checkedMma
\end{equation*}
which is irreducible as a $\widehat{\mathfrak{U}}$-representation. (One could call these `affine magnons'; the $\mathfrak{U}$-irrep of dimension $N-1$, which occurs for any $n$, then is the ordinary magnon.) The \textsf{q}-momentum is $p=2\pi\,n/N$. For even $N$ there is one magnon with $n=N/2$, so $p = \pi$. All other magnons come in parity-conjugate pairs with mirror-image motifs $\mu = (n), (N-n)$, opposite momentum (mod $2\pi$) and energy differing by $\mathsf{q}\mapsto \mathsf{q}^{-1}$, cf.~\eqref{eq:energy_right}.

If $N$ is even another particularly simple case occurs at the other side of the spectrum, $M=N/2$ (`half filling'). Here the motif $\mu^\textsc{af} = (1,3,\To,2M-1)$ corresponds to partition $\nu^\textsc{af} = (1^M)$. By \eqref{eq:Macdonald_shift_property} we have $P_{(1^M)}^\star(z_1,\To,z_M) = z_1 \cdots z_M$. The simple component \eqref{eq:intro_polynomial} has an additional squared \textsf{q}-$\mspace{-1mu}$Vandermonde factor: this is the \textsf{q}-deformed Jastrow wave function in multiplicative notation. This is the antiferromagnetic ground state of $-H^\textsc{l}$. The eigenspace is one dimensional. The \textsf{q}-momentum is $p=N\,\pi/2 \,\mathrm{mod}\, 2\pi$,
\checkedMma 
so $p=0$ for $N=4n$ and $p=\pi$ if $N=4n+2$:
\checkedMma
these are the only two values invariant under parity reversal $p\mapsto -p \; \mathrm{mod}\, 2\pi$. 

The lowest excitations around the antiferromagnetic singlet occur when $N$ is odd, with $(N+1)/2$ motifs of length $M=(N-1)/2$ (as close as possible to half filling) differing from $(1,3,\To,N-2)$ in that the last~$s$ parts are increased by one, $0\leq s\leq M$, where $s=0$ means that nothing is changed. The corresponding partition is $\nu = (2^s,1^{M-s})$, with associated polynomial $P_{(2^s,1^{M-s})}^\star(z_1,\To,z_M) = e_{(M,s)}(z_1,\To,z_M)$ a product of elementary symmetric polynomials, see just below. Each such eigenspace has dimension two. Its physical interpretation is a (\textsf{q}-)spinon, a quasiparticle with spin~1/2, cf.~\cite{Hal_91a}. 

Here and below we use the following bases for $\mathbb{C}[z_1,\To,z_M]^{\mathfrak{S}_M}$, cf.~\textsection{I} of \cite{Mac_95}. Each of these bases is labelled by partitions $\lambda$ with $\ell(\lambda)\leq M$, viewed as having $M$ parts by appending zeros if necessary. The elementary symmetric polynomials are
\begin{equation} \label{eq:elementary}
	e_r(z_1,\To,z_M) = \! \sum_{m_1<\cdots<m_r} \!\!\!\!\!\! z_{m_1} \cdots z_{m_r} \, , \qquad e_\lambda(z_1,\To,z_M) = \prod_{r\in\lambda} e_r(z_1,\To,z_M) \, .
\end{equation}
The monomial symmetric polynomials are `minimal symmetrisations' of $\vect{z}^\lambda$:
\begin{equation} \label{eq:monomial}
	m_\lambda(z_1,\To,z_M) = \! \sum_{\alpha \in \mathfrak{S}_M \lambda} \!\!\! \vect{z}^\alpha \, ,
\end{equation}
summed over all distinct rearrangements of $\lambda$. For example, $m_{(1^r)}(z_1,\To,z_M) = e_r(z_1,\To,z_M)$. The most efficient basis for our examples is given by Schur polynomials
\begin{equation} \label{eq:Schur}
	s_\lambda(z_1,\To,z_M) = \prod_{m<n}^M \! (z_m - z_n)^{-1} \det_{1\leq m,n\leq M} z_m^{\lambda_n + M - n} \, ,
\end{equation}
where the determinant is totally antisymmetric whence divisible by the Vandermonde factor. Each of the preceding symmetric polynomials obeys \eqref{eq:Macdonald_shift_property}, and they are all limiting cases of Macdonald polynomials, cf.~Figure~\ref{fg:Macdonalds} on p.\,\pageref{fg:Macdonalds}.

\begin{table}[ht]
	\begin{tabular}{clll} \toprule
		$M$ & motif $\mu$ & partition $\nu$ & $P_\nu^\star(z_1,\To,z_M)$ \\ \midrule
		$0$ & $0$ & $0$ & $1$ \\ \midrule
		$1$ & $(n)$ & $(n)$ & $s_{(n)} = e_1^n$ \\ \midrule
		$2$ & $(n,n+2)$ & $(n,n)$ & $s_{(n,n)} = e_2^n$ \\
		& $(n,n+3)$ & $(n+1,n)$ & $s_{(n+1,n)} = e_2^n \, e_1$ \\
		& $(n,n+4)$ & $(n+2,n)$ & $s_{(n+2,n)} + \frac{1}{[3]} \, s_{(n+1,n+1)}$ \\[.8ex]
		& $(n,n+5)$ & $(n+3,n)$ & $s_{(n+3,n)} + \frac{[2]}{[4]} \, s_{(n+2,n+1)}$ 	\\[.8ex]
		& $(n,n+6)$ & $(n+4,n)$ & $s_{(n+4,n)} + \frac{[3]}{[5]} \, s_{(n+3,n+1)} + \frac{1}{[5]} \, s_{(n+2,n+2)}$ \\ \midrule
		$3$ & $(n,n+2,n+4)$ & $(n,n,n)$ & $s_{(n,n,n)} = e_3^n$ \\
		& $(n,n+2,n+5)$ & $(n+1,n,n)$ & $s_{(n+1,n,n)} = e_3^n \, e_1$ \\
		& $(n,n+3,n+5)$ & $(n+1,n+1,n)$ & $s_{(n+2,n,n)} = e_3^n \, e_2$ \\
		& $(n,n+2,n+6)$ & $(n+2,n,n)$ & $s_{(n+2,n,n)} + \frac{1}{[3]} \, s_{(n+1,n+1,n)}$ \\
		& $(n,n+3,n+6)$ & $(n+2,n+1,n)$ & $s_{(n+2,n+1,n)} + \frac{[4]}{[5]\mspace{1mu}[2]} \, s_{(n+1,n+1,n+1)}$ \\[.8ex]
		& $(n,n+4,n+6)$ & $(n+2,n+2,n)$ & $s_{(n+2,n+2,n)} + \frac{1}{[3]} \, s_{(n+2,n+1,n+1)}$ \\ \midrule
		$4$	& $(n,n+2,n+4,n+6)\!$ & $(n,n,n,n)$ & $s_{(n,n,n,n)} = e_4^n$ \\
		\bottomrule
		\hline
	\end{tabular}
	\caption{The quantum zonal spherical polynomials 
	(with parameters $\mathsf{p}^\star = \mathsf{q}^\star = \mathsf{q}^2$, i.e.\ $q^\star = t^{\star\,1/2} = t$ in the notation of \cite{Mac_95}) 
	needed to construct all highest-weight vectors for $N \leq 8$, given in terms of Schurs.}
	\label{tb:Macdonald_N=8}
\end{table}

So far the examples of $P_\nu^\star$ were independent of $\mathsf{q}$. (Each of these simple instances can be recognised as $e_{\nu'}(z_1,\To,z_M)$ as well as $m_\nu(z_1,\To,z_M)$ and $s_\nu(z_1,\To,z_M)$.) This covers all polynomials needed for $N \leq 5$. The first \textsf{q}-dependent quantum zonal polynomial appears when $N=6$, for $\mu = (1,5)$ so $\nu = (3,1)$:
\begin{equation*}
	\begin{aligned}
	P_{(3,1)}^\star(z_1,z_2) & = e_2(z_1,z_2) \, \biggl( e_1(z_1,z_2)^2 - \frac{[4]}{[3]\mspace{1mu}[2]} \, e_2(z_1,z_2) \biggr) \\ 
	& = m_{(3,1)}(z_1,z_2) + \frac{[2]^2}{[3]} \, m_{(2,2)}(z_1,z_2) \\ 
	& = s_{(3,1)}(z_1,z_2) + \frac{1}{[3]} \, s_{(2,2)}(z_1,z_2) \, .
	\end{aligned}
\checkedMma
\end{equation*}
The reason is that this is the first time that there is a partition of length~$M$ that is smaller than $\nu$ in the dominance ordering~\eqref{eq:dominance} from \textsection\ref{s:Macdonald}: $(3,1) > (2,2)$. 
The coefficients in the expansion over Schur polynomials are known as Kostka--Macdonald coefficients.
More generally for $M=2$ we can use \eqref{eq:Macdonald_shift_property} to write $P_{(\nu_1,\nu_2)}^\star = (z_1 \, z_2)^{\nu_2} \, P_{(\nu_1 - \nu_2,0)}^\star$, where
\begin{equation*}
	\begin{aligned}
	P_{(n,0)}^\star(z_1,z_2) = \sum_{i=0}^{\lfloor n/2 \rfloor} \frac{[n-2\,i+1]}{[n+1]} \, s_{(n-i,i)}(z_1,z_2) \, ,
	\end{aligned}
\checkedMma
\end{equation*}
with $\lfloor n/2 \rfloor$ the integer part of $n/2$.
Tables~\ref{tb:Macdonald_N=8}--\ref{tb:Macdonald_N=10} contain all polynomials required to construct the complete spectrum for $N\leq 10$. Note the stability property $P_{\bar\nu}(z_1,\To,z_{M-1},0) = P_{\bar\nu}(z_1,\To,z_{M-1})$ for Macdonald polynomials with $\ell(\bar\nu)<M$ as well as the symmetry between polynomials with mirror-image motifs.

\begin{table}[ht]
	\begin{tabular}{cll} \toprule
		$M$ & `reduced' partition $\bar{\nu}$ & $P_{\bar{\nu}}^\star(z_1,\To,z_M)$ \\ \midrule
		$2$ & $(5,0)$ & $s_{(5,0)} + \frac{[4]}{[6]} \, s_{(4,1)} + \frac{[2]}{[6]} \, s_{(3,2)}$ \\ \midrule
		$3$ & $(3,0,0)$ & $s_{(3,0,0)} + \frac{[2]}{[4]} \, s_{(2,1,0)}$ \\
		& $(3,1,0)$ & $s_{(3,1,0)} + \frac{1}{[3]} \, s_{(2,2,0)}  + \frac{[5]\mspace{1mu}[2]}{[6]\mspace{1mu}[3]} \, s_{(2,1,1)}$ \\[.8ex]
		& $(3,2,0)$ & $s_{(3,2,0)} + \frac{1}{[3]} \, s_{(3,1,1)}  + \frac{[5]\mspace{1mu}[2]}{[6]\mspace{1mu}[3]} \, s_{(2,2,1)}$ \\[.8ex]
		& $(3,3,0)$ & $s_{(3,3,0)} + \frac{[2]}{[4]} \, s_{(3,2,1)}$ \\ \midrule
		$4$	& $(1,0,0,0)$ & $s_{(1,0,0,0)} = e_1$ \\
		& $(1,1,0,0)$ & $s_{(1,1,0,0)} = e_2$ \\
		& $(1,1,1,0)$ & $s_{(1,1,1,0)} = e_3$ \\
		\bottomrule
		\hline
	\end{tabular}
	\caption{Continuation of Table~\ref{tb:Macdonald_N=8} required for $N = 9$. We use \eqref{eq:Macdonald_shift_property} to set $n = 0$. The particularly simple polynomials correspond to \textsf{q}-spinons.}
	\label{tb:Macdonald_N=9}
\end{table}

\begin{table}[ht]
	\begin{tabular}{cll} \toprule
		$M$ & $\bar{\nu}$ & $P_{\bar{\nu}}^\star(z_1,\To,z_M)$ \\ \midrule
		$2$ & $(6,0)$ & $s_{(6,0)} + \frac{[5]}{[7]} \, s_{(5,1)} + \frac{[3]}{[7]} \, s_{(4,2)} + \frac{1}{[7]} \, s_{(3,3)}$ \\ \midrule
		$3$ & $(4,0,0)$ & $s_{(4,0,0)} + \frac{[3]}{[5]} \, s_{(3,1,0)} + \frac{1}{[5]} \, s_{(2,2,0)}$ \\[.8ex]
		& $(4,1,0)$ & $s_{(4,1,0)} + \frac{[2]}{[4]} \, s_{(3,2,0)} + \frac{[6]\mspace{1mu}[3]}{[7]\mspace{1mu}[4]} \, s_{(3,1,1)} + \frac{[6]}{[7]\mspace{1mu}[4]} \, s_{(2,2,1)}$ \\[.8ex]
		& $(4,2,0)$ & $s_{(4,2,0)} + \frac{1}{[3]} \, s_{(3,3,0)} + \frac{1}{[3]} \,s_{(4,1,1)} + \frac{[6]\mspace{1mu}[2]^3}{[7]\mspace{1mu}[3]^2} \, s_{(3,2,1)} + \frac{[5]}{[7]\mspace{1mu}[3]} \, s_{(2,2,2)}$ \\[.8ex]
		& $(4,3,0)$ & $s_{(4,3,0)} + \frac{[2]}{[4]} \, s_{(4,2,1)} + \frac{[6]\mspace{1mu}[3]}{[7]\mspace{1mu}[4]} \, s_{(3,3,1)} + \frac{[6]}{[7]\mspace{1mu}[4]} \, s_{(3,2,2)}$ \\[.8ex]
		& $(4,4,0)$ & $s_{(4,4,0)} + \frac{[3]}{[5]} \, s_{(4,3,1)} + \frac{1}{[5]} \, s_{(4,2,2)}$ \\ \midrule
		$4$ & $(2,0,0,0)$ & $s_{(2,0,0,0)} + \frac{1}{[3]} \, s_{(1,1,0,0)}$ 	\\[.8ex]
		& $(2,1,0,0)$ & $s_{(2,1,0,0)} + \frac{[4]}{[5]\mspace{1mu}[2]} \, s_{(1,1,1,0)}$ \\[.8ex] 
		& $(2,1,1,0)$ & $s_{(2,1,1,0)} + \frac{[6]}{[7]\mspace{1mu}[2]} \, s_{(1,1,1,1)}$ \\[.8ex]
		& $(2,2,0,0)$ & $s_{(2,2,0,0)} + \frac{1}{[3]} \, s_{(2,1,1,0)} + \frac{1}{[5]} \, s_{(1,1,1,1)}$	\\[.8ex] 
		& $(2,2,1,0)$ & $s_{(2,2,1,0)} + \frac{[4]}{[5]\mspace{1mu}[2]} \, s_{(2,1,1,1)}$ \\[.8ex]
		& $(2,2,2,0)$ & $s_{(2,2,2,0)} + \frac{1}{[3]} \, s_{(2,2,1,1)}$ \\ \midrule
		\bottomrule
		\hline
	\end{tabular}
	\caption{Continuation of Table~\ref{tb:Macdonald_N=9} needed for $N = 10$. We omit $P_{(0^5)} = 1$.}
	\label{tb:Macdonald_N=10}
\end{table} 

\subsection{Plan of proofs: spin-Ruijsenaars and freezing} \label{s:intro_spin-RM} The origin of the \textsf{q}-deformed Haldane--Shastry spin chain with all its remarkable properties is once again an integrable quantum many-body system~\cite{Pol_93,BG+_93,TH_95,Ugl_95u}: the spin-version of the trigonometric Ruijsenaars model~\cite{Rui_87}. The Ruijsenaars model is the \textsf{q}-deformation of the Calogero--Sutherland model from \textsection\ref{s:intro_spin-CS}, parametrised by $\mathsf{q}$ and $\mathsf{p} = \mathsf{q}^{2\hbar/k}$ with $k$ as in \eqref{eq:spin-CS}. The spin-Ruijsenaars model was studied in \cite{BG+_93,Che_94a,JK+_95a,JK+_95b} and more explicitly in \cite{Kon_96}. Like in the isotropic case, this model already
\begin{enumerate}
	\item[i.] (\emph{abelian symmetries}) belongs to a family of commuting operators~\cite{BG+_93,Che_94a}, each of which
	\item[ii.] (\emph{nonabelian symmetries}) commutes with an action of the quantum-loop algebra \cite{BG+_93}, cf.~\cite{CP_96};
	\item[iii.] (\emph{explicit eigenvectors}) has eigenvectors that are determined by a suitably symmetric polynomial, which for pseudo highest-weight eigenvectors can be described explicitly in terms of a Macdonald polynomial.
\end{enumerate}

\subsubsection{Physical (\textsf{q}-bosonic) space} \label{s:intro_phys_space}
Consider $N$ relativistic spin-1/2 particles of equal mass moving on a circle. The particles are `\textsf{q}-bosons' in that they are invariant under simultaneous \textsf{q}-exchange of spins and coordinates. More precisely, 

\begin{definition}[\cite{FR_92,BG+_93}]
We call an element of $(\mathbb{C}^2[z])^{\otimes N} \cong (\mathbb{C}^2)^{\otimes N} \otimes \mathbb{C}[z_1,\To,z_N]$ on which $\check{R}_{i,i+1}(z_i/z_{i+1}) = s_i$ a \emph{physical vector}. The subspace consisting of all such vectors is the \emph{physical space}, denoted by $\widetilde{\mathcal{H}}$. The \emph{simple component} of a physical vector $\ket{\widetilde{\Psi}}$ in the $M$-particle sector is defined to be $\cbraket{1\To M}{\widetilde{\Psi}}$, the component with all $\downarrow$s on the left.
\end{definition}
\noindent Note that the property defining physical vectors already appeared in the remark of~\textsection\ref{s:intro_nonabelian_2}. For more about the physical space, including a characterisation via the Hecke algebra, see \textsection\ref{s:physical_space}.

Any vector in the $M$-particle sector of $\widetilde{\mathcal{H}}$ is determined by a single polynomial. Several explicit descriptions are available in the literature, cf.~e.g.~\cite{DZ_05a,DZ_05b,KP_07,dGP_10}. We will use the following characterisation in terms of the coordinate basis, which explains the significance of the simple component.

\begin{proposition}[cf.~\cite{RSZ_07}] \label{prop:intro_physical_vectors}
Let $\mathbb{C}[z_1,\To,z_N]^{\mathfrak{S}_M \times \mathfrak{S}_{N-M}}$ denote the ring of polynomials that are symmetric in $z_1,\To,z_M$ and in $z_{M+1},\To,z_N$ separately. Any physical vector in the $M$-particle sector is determined by such a polynomial: with the notation \eqref{eq:Hecke_graphical},
\begin{equation} \label{eq:phys_vector_intro}
	\normalfont
	\sum_{i_1 < \cdots < i_M}^N \!\!\!\!\! T_{\{i_1,\To,i_M\}}^\text{pol} \widetilde{\Psi}(\vect{z}) \, \cket{i_1,\To,i_M} \, , \qquad \widetilde{\Psi}(\vect{z}) \in \mathbb{C}[z_1,\To,z_N]^{\mathfrak{S}_M \times \mathfrak{S}_{N-M}} \, .
\checkedMma
\end{equation}
\end{proposition}
\noindent The recursion leading to \eqref{eq:phys_vector} was already given in \cite{RSZ_07}, see the unnumbered equation after \mbox{(14)} therein. We will give a proof of Proposition~\ref{prop:intro_physical_vectors} in \textsection\ref{s:physical_space}.

Since physical vectors are determined by their simple component, any operator that preserves the physical space can be reconstructed from its action on the simple component:

\begin{corollary}
Let $\widetilde{O}$ be an operator on $\widetilde{\mathcal{H}}$. Let $0\leq M\leq N$ and assume that $\widetilde{O}$ maps the $M$-particle sector into some $M'$-particle sector. Then the restriction of $\widetilde{O}$ to the $M$-particle sector is completely determined by the assignment $\cbraket{1,\To, M}{\widetilde{\Psi}} \mapsto \cbra{1,\To,M'} \, \widetilde{O}\,\ket{\widetilde{\Psi}}$.
\end{corollary}
\noindent This trick, which is described in more detail in \textsection\ref{s:physical_space} and exploited in \textsection\ref{s:nonabelian_tilde}, provides an efficient tool for computations. To the best of our knowledge it is new.

\subsubsection{Abelian symmetries (spin-Macdonald operators)} \label{s:intro_abelian_tilde} 
Let $\mathsf{p} \in \mathbb{C}^\times$ set the speed of light~$c$ via $\mathsf{p} = \E^{\I \, \hbar/m\mspace{1mu}c}$, with $m$ the rest mass of the particles. The spin-Ruijsenaars model is quantum integrable, with a hierarchy of commuting Hamiltonians. In \textsection\ref{s:abelian_tilde} we obtain explicit expressions for these spin-Macdonald operators governing the dynamics.\,%
\footnote{\ Another type of matrix-valued Macdonald operators were constructed in~\cite{EV_00}.} 
Our expressions are as follows. 

Consider the $j$th momentum (translation) operator in multiplicative notation (\textsection\ref{s:Hecke_AHA}),
\begin{equation} \label{eq:p_hat_graphical}
	\hat{\mathsf{p}}_j \colon z_j \mapsto \mathsf{p} \, z_j \, , \qquad 
	\hat{\mathsf{p}}_j = 
	\tikz[baseline={([yshift=-.5*11pt*0.4]current bounding box.center)},	xscale=0.6,yscale=0.4,font=\footnotesize]{
		\draw[->] (0,0) node[below] {$z_1$} -- (0,2) node[above] {$\vphantom{z_j}z_1$};
		\draw[->] (1.5,0) node[below] {$z_{j-1}$} -- (1.5,2) node[above] {$z_{j-1}$};
		\draw[very thick] (2.5,0) node[below] {$\mathsf{p}\mspace{1mu}z_j$} -- (2.5,1) node[yshift=-1pt] {$\mathllap{\mathsf{p}\,}\tikz[baseline={([yshift=-.5*12pt*.35]current bounding box.center)},scale=.35]{\fill[black] (0,0) circle (.25)}$};
		\draw[->] (2.5,1) -- (2.5,2) node[above]{$z_j$};
		\draw[->] (3.5,0) node[below] {$z_{j+1}$} -- (3.5,2) node[above] {$z_{j+1}$};
		\draw[->] (5,0) node[below] {$z_N$} -- (5,2) node[above] {$\vphantom{z_j}z_N$};
		\foreach \x in {-1,...,1} \draw (.75+.2*\x,1) node{$\cdot\mathstrut$};	
		\foreach \x in {-1,...,1} \draw (4.25+.2*\x,1) node{$\cdot\mathstrut$};	
	} \, ,
\end{equation}
with \textsf{p}-deformed canonical commutation relations $\hat{\mathsf{p}}_j \, z_i = \mathsf{p}^{\delta_{ij}} \, z_i \, \hat{\mathsf{p}}_j$. 

\begin{definition} 
Using the graphical notation~\eqref{eq:R_graphical}, \eqref{eq:p_hat_graphical} the (first) \emph{spin-Macdonald operator} is
\begin{equation} \label{eq:intro_spin-D_1}
	\widetilde{D}_1 = \sum_{j=1}^N A_j(\vect{z}) \times \ \, \tikz[baseline={([yshift=-.5*11pt*0.4]current bounding box.center)},	xscale=0.6,yscale=0.4,font=\footnotesize]{
		\draw[->] (10.5,0) node[below]{$z_N$} -- (10.5,9) node[above]{$z_N\vphantom{z_j}$};
		\draw[->] (9,0) node[below]{$z_{j+1}$} -- (9,9) node[above]{$z_{j+1}$};
		\draw[very thick,rounded corners=3.5pt] (8,0) node[below]{$\mathsf{p}\mspace{1mu}z_j$} -- (8,.5) -- (7,1.5) node[inner sep=1.5pt,fill=white]{$\mathsf{p}\mspace{1mu}z_j$} -- (6,2.5) node[inner sep=1.5pt,fill=white]{$\mathsf{p}\mspace{1mu}z_j$} -- (5,3.5) node[shift={(.05,-.1)}, inner sep=1.5pt,fill=white]{$\mathsf{p}\mspace{1mu}z_j$} -- (5,4.5) node[yshift=-1pt]{$\mathllap{\mathsf{p}\,}\tikz[baseline={([yshift=-.5*12pt*.35]current bounding box.center)},scale=.35]{\fill[black] (0,0) circle (.25)}$};
		\draw[rounded corners=3.5pt,->] (5,4.5) -- (5,5.5) node[shift={(0,.05)}, inner sep=1.7pt,fill=white]{$z_j$} -- (6,6.5) node[inner sep=1.5pt,fill=white]{$z_j$} -- (7,7.5) node[inner sep=1.5pt,fill=white]{$z_j$} -- (8,8.5) -- (8,9) node[above]{$z_j$};
		\draw[rounded corners=3.5pt,->] (7,0) node[below]{$z_{j-1}$} -- (7,.5) -- (8,1.5) -- (8,4.5) node[inner sep=1.5pt,fill=white]{$z_{j-1}$} -- (8,7.5) -- (7,8.5) -- (7,9) node[above]{$z_{j-1}$};
		\draw[rounded corners=3.5pt,->] (6,0) node[below]{$\cdots$} -- (6,1.5) -- (7,2.5) -- (7,4.5) node[inner sep=1.5pt,fill=white]{$\cdots$} -- (7,6.5) -- (6,7.5) -- (6,9) node[above]{$\cdots{\vphantom{z_j}}$};
		\draw[rounded corners=3.5pt,->] (5,0) node[below]{$z_1$} -- (5,2.5) -- (6,3.5) -- (6,4.5) node[inner sep=1.5pt,fill=white]{$z_1$} -- (6,5.5) -- (5,6.5) -- (5,9) node[above]{$z_1\vphantom{z_j}$};
		\foreach \x in {-1,...,1} \draw (9.75+.2*\x,4.5) node{$\cdot\mathstrut$};		
	} , \qquad
	\begin{aligned}
	A_j(\vect{z}) \coloneqq {} & \prod_{\bar\jmath \mspace{1mu} (\neq j)}^N \!\! f(z_j/z_{\bar\jmath})^{-1} \, , \\
	= {} & \prod_{\bar\jmath \mspace{1mu} (\neq j)}^N \!\! \frac{\mathsf{q}\,z_j - \mathsf{q}^{-1}z_{\bar\jmath}}{z_j-z_{\bar\jmath}} \,.
	\end{aligned}
\end{equation}
\end{definition}

This operator was also found by Cherednik~\cite{Che_94a}.\,%
\footnote{\ Equation \mbox{(4.15)} in \cite{Che_94a} can be recognised as the last form in \eqref{eq:intro_spin-D_1_N=3}. See also Footnote~2 in \cite{Che_94a}.}
For $N=3$ it becomes 
\begin{equation} \label{eq:intro_spin-D_1_N=3}
	\begin{aligned}
	\widetilde{D}_1 = {} & A_1(z_1,z_2,z_3) \, \hat{\mathsf{p}}_1 + A_2(z_1,z_2,z_3) \, \check{R}_{12}(z_2/z_1) \, \hat{\mathsf{p}}_2 \, \check{R}_{12}(z_1/z_2) \\
	& + A_3(z_1,z_2,z_3) \, \check{R}_{23}(z_3/z_2) \, \check{R}_{12}(z_3/z_1) \, \hat{\mathsf{p}}_3 \, \check{R}_{12}(z_1/z_3) \, \check{R}_{23}(z_2/z_3) \\
	= {} & A_1(z_1,z_2,z_3) \, \hat{\mathsf{p}}_1 + A_2(z_1,z_2,z_3) \, \check{R}_{12}(z_2/z_1) \, \check{R}_{12}(\mathsf{p}^{-1} z_1/z_2) \, \hat{\mathsf{p}}_2 \\
	& + A_3(z_1,z_2,z_3) \, \check{R}_{23}(z_3/z_2) \, \check{R}_{12}(z_3/z_1) \, \check{R}_{12}(\mathsf{p}^{-1}\,z_1/z_3) \, \check{R}_{23}(\mathsf{p}^{-1}\,z_2/z_3) \, \hat{\mathsf{p}}_3 \, .
	\end{aligned}
\checkedMma
\end{equation}
The difference with the spinless case (\textsection\ref{s:Macdonald}) is that the $\hat{\mathsf{p}}_j$ are `dressed' by \textit{R}-matrices. 
In the nonrelativistic limit $c\to\infty$, taken by setting $\mathsf{p} = \mathsf{q}^{2\hbar/k}$ and Taylor expanding at $\mathsf{q} = 1$, \eqref{eq:intro_spin-D_1} reduces to the effective Hamiltonian $\widetilde{H}^{\text{eff},\mspace{1mu}\text{nr}}$ of the spin-Calogero--Sutherland model, related to~\eqref{eq:spin-CS} by a `gauge transformation'. This limit is reviewed in \textsection\ref{s:app_nonrlt_limit}.

The higher spin-Macdonald operators $\widetilde{D}_r$, $1\leq r \leq N$, involve more and more `layers' of \textit{R}-matrices, see \eqref{eq:spin_D_r}. For example,
\begin{equation} \label{eq:intro_spin-D_2}
	\widetilde{D}_2 = \sum_{j<j'}^N A_{j j'}(\vect{z}) \times \ \,
	\tikz[baseline={([yshift=-.5*11pt*0.4*.9]current bounding box.center)},	xscale=0.6*.9,yscale=0.4*.9,font=\footnotesize]{
		\draw[->] (13.5,0) node[below]{$z_N$} -- (13.5,13) node[above]{$z_N\vphantom{z_j}$};
		\draw[->] (12,0) -- (12,13);
		\draw[very thick,rounded corners=3.5pt] (11,0) node[below]{$\mathsf{p}\mspace{1mu}z_{j'}$} -- (11,.5) -- (6,5.5) -- (6,6.5) node[yshift=-1pt]{$\mathllap{\mathsf{p}\,}\tikz[baseline={([yshift=-.5*12pt*.35]current bounding box.center)},scale=.35]{\fill[black] (0,0) circle (.25)}$};
		\draw[rounded corners=3.5pt,->] (6,6.5) -- (6,7.5) -- (11,12.5) -- (11,13) node[above]{$z_{j'}$};
		\draw[rounded corners=3.5pt,->] (10,0) -- (10,.5) -- (11,1.5) -- (11,11.5) -- (10,12.5) -- (10,13);
		\node at (9.5,0) [below] {$\cdots$}; 
		\node at (9.5,13) [above] {$\cdots{\vphantom{z_j}}$};
		\draw[rounded corners=3.5pt,->] (9,0) -- (9,1.5) -- (10,2.5) -- (10,10.5) -- (9,11.5) -- (9,13);
		\draw[very thick,rounded corners=3.5pt] (8,0) node[below]{$\mathsf{p}\mspace{1mu}z_j$} -- (8,1.5) -- (5,4.5) -- (5,6.5) node[yshift=-1pt]{$\mathllap{\mathsf{p}\,}\tikz[baseline={([yshift=-.5*12pt*.35]current bounding box.center)},scale=.35]{\fill[black] (0,0) circle (.25)}$};
		\draw[rounded corners=3.5pt,->] (5,6.5) -- (5,8.5) -- (8,11.5) -- (8,13) node[above]{$z_j$};
		\draw[rounded corners=3.5pt,->] (7,0) -- (7,1.5) -- (9,3.5)	-- (9,9.5) -- (7,11.5) -- (7,13);
		\node at (6.5,0) [below] {$\cdots$}; 
		\node at (6.5,13) [above] {$\cdots{\vphantom{z_j}}$};
		\draw[rounded corners=3.5pt,->] (6,0) -- (6,2.5) -- (8,4.5) -- (8,8.5) -- (6,10.5) -- (6,13);
		\draw[rounded corners=3.5pt,->] (5,0) node[below]{$z_1$} -- (5,3.5) -- (7,5.5)	-- (7,7.5) -- (5,9.5) -- (5,13) node[above]{$z_1\vphantom{z_j}$};
		\foreach \x in {-1,...,1} \draw (12.75+.2*\x,6.5) node{$\cdot\mathstrut$};	
	}
	, \quad 
	\begin{aligned}
	& \\ 
	A_{j j'}(\vect{z}) \coloneqq {} & \! \prod_{\bar\jmath (\neq j,j')}^N \!\!\!\! f(z_j/z_{\bar\jmath})^{-1} f(z_{j'},z_{\bar\jmath})^{-1} \\
	= {} & f(z_j/z_{j'}) \, f(z_{j'}/z_j) \\
		& \quad \times \, A_j(\vect{z}) \, A_{j'}(\vect{z}) \, .
	\end{aligned}
\checkedMma
\end{equation}
When $N=3$ this gives
\begin{equation*}
	\begin{aligned}
	\widetilde{D}_2 = {} & A_{12}(z_1,z_2,z_3) \, \hat{\mathsf{p}}_1 \, \hat{\mathsf{p}}_2 + A_{13}(z_1,z_2,z_3) \, \check{R}_{23}(z_3/z_2) \, \hat{\mathsf{p}}_1 \, \hat{\mathsf{p}}_3 \, \check{R}_{23}(z_2/z_3) \\
	& + A_{23}(z_1,z_2,z_3) \, \check{R}_{12}(z_2/z_1) \, \check{R}_{23}(z_3/z_1) \, \hat{\mathsf{p}}_2 \, \hat{\mathsf{p}}_3 \, \check{R}_{23}(z_1/z_3) \, \check{R}_{12}(z_1/z_2) \, .
	\end{aligned}
\checkedMma
\end{equation*}

Beyond the `equator' $r =\lfloor N/2\rfloor$ the expressions become simpler again. In particular, the (multiplicative) translation operator is the same as in the spinless case, 
\begin{equation} \label{eq:intro_spin-D_N}
	\widetilde{D}_N = \hat{\mathsf{p}}_1 \cdots \hat{\mathsf{p}}_N \, ,
\checkedMma
\end{equation} 
and the counterpart of \eqref{eq:intro_spin-D_1} with opposite chirality is
\begin{equation} \label{eq:intro_spin-D_-1}
	\widetilde{D}_{-1} \coloneqq \widetilde{D}_N^{-1} \, \widetilde{D}_{N-1} = \sum_{i=1}^N A_{-i}(\vect{z}) \times \tikz[baseline={([yshift=-.5*11pt*0.4*.9]current bounding box.center)},	xscale=-0.6*.9,yscale=0.4*.9,font=\footnotesize]{
		\draw[->] (10.5,0) node[below]{$\vphantom{\mathsf{p}^{-1}}z_1$} -- (10.5,9) node[above]{$z_1$};
		\draw[->] (9,0)	-- (9,9);
		\draw[very thick,rounded corners=3.5pt] (8,0) node[below]{$\mathsf{p}^{-1} z_i$} -- (8,.5) -- (5,3.5)	-- (5,4.5) node[yshift=-1pt]{$\tikz[baseline={([yshift=-.5*12pt*.35]current bounding box.center)},scale=.35]{\fill[black] (0,0) circle (.25)} \mathrlap{\,\mathsf{p}^{-1}}$};
		\draw[rounded corners=3.5pt,->] (5,4.5) -- (5,5.5) -- (8,8.5) -- (8,9) node[above]{$z_i$};
		\draw[rounded corners=3.5pt,->] (7,0) -- (7,.5) -- (8,1.5) -- (8,7.5) -- (7,8.5) -- (7,9);
		\node at (6.5,0) [below] {$\vphantom{\mathsf{p}^{-1}}\cdots$}; \node at (6.5,9) [above] {$\cdots$};
		\draw[rounded corners=3.5pt,->] (6,0) -- (6,1.5) -- (7,2.5)	-- (7,6.5) -- (6,7.5) -- (6,9);
		\draw[rounded corners=3.5pt,->] (5,0) node[below]{$\vphantom{\mathsf{p}^{-1}}z_N$} -- (5,2.5) -- (6,3.5)	-- (6,5.5) -- (5,6.5) -- (5,9) node[above]{$z_N$};
		\foreach \x in {-1,...,1} \draw (9.75+.2*\x,4.5) node{$\cdot\mathstrut$};		
	} 
	\ \ \ , \qquad A_{-i}(\vect{z}) \coloneqq \prod_{\bar\imath \mspace{1mu} (\neq i)}^N \!\! f(z_{\bar\imath},z_i)^{-1} \, .
\checkedMma
\end{equation}

The key property of these operators is
\begin{theorem}[cf.~\cite{BG+_93,Che_94a}] \label{thm:intro_spin_D_r}
The spin-Macdonald operators \eqref{eq:intro_spin-D_1}, \eqref{eq:intro_spin-D_2}--\eqref{eq:intro_spin-D_-1} are part of a commuting family of operators on the physical space~$\widetilde{\mathcal{H}}$.
\end{theorem}
\noindent The existence of this family of commuting operators on $\widetilde{\mathcal{H}}$ was shown in \cite{BG+_93}. Their expressions were found by \cite{Che_94a}, albeit in a less explicit form. We will prove Theorem in~\ref{thm:intro_spin_D_r} in \textsection\ref{s:abelian_tilde}, see Theorem~\ref{thm:spin_D_r} therein.

\begin{remarks} \label{p:Ruijsenaars}
\textbf{i.}~See \eqref{eq:spin_D_r} for a general expression for $\widetilde{D}_r$. \textbf{ii.}~The eigenvalues of the $\widetilde{D}_r$ are, by construction (\textsection\ref{s:abelian_tilde}), as in the spinless case (\textsection\ref{s:Macdonald}). 
\textbf{iii.}~The `full', or physical, spin-Ruijsenaars Hamiltonian is $(\widetilde{D}_1 + \widetilde{D}_{-1})/2$, while the physical momentum operator is $(\widetilde{D}_1 - \widetilde{D}_{-1})/2$. By a conjugation (`gauge transformation') one can pass to the spin-generalisation of Ruijsenaars's manifestly Hermitian form~\cite{Rui_87}, see \textsection\ref{s:Ruijsenaars}.
\end{remarks}

\subsubsection{Nonabelian symmetries} \label{s:intro_nonab_dyn}
One reason for going through the spin-Ruijsenaars model is that the latter already enjoys quantum-loop symmetry. Recalling~\eqref{eq:intro_xij} let
\begin{equation} \label{eq:intro_Y-ops}
	Y_i \coloneqq x_{i,i+1} \, x_{i,i+2} \cdots x_{iN} \; \hat{\mathsf{p}}_i \; x_{i1} \cdots x_{i,i-2} \, x_{i,i-1}
\end{equation}
be the \textsf{q}-deformed Dunkl operators of the affine Hecke algebra, see \textsection\ref{s:Hecke_AHA}.

\begin{theorem}[\cite{BG+_93}, cf.~\cite{CP_96}] \label{thm:monodromy_tilde}
The physical space $\widetilde{\mathcal{H}}$ carries an action of $\widehat{\mathfrak{U}}$, given by the monodromy matrix
\begin{equation} \label{eq:intro_L_tilde}
	\widetilde{L}_a(u) \coloneqq R_{aN}(u\,Y_N) \cdots  R_{a1}(u\,Y_1)
\end{equation}
on $V_a \otimes \widetilde{\mathcal{H}}$. This operator commutes with the spin-Macdonald operators.
\end{theorem}

\noindent In \textsection\ref{s:nonabelian_tilde} we recall the proof from \cite{BG+_93}, which uses all relations of the affine Hecke algebra, and give the action in terms of Chevalley generators, which is due to~\cite{CP_96}.

The Corollary of Proposition~\ref{prop:intro_physical_vectors} provides a convenient way for working with the $\widehat{\mathfrak{U}}$-action on $\widetilde{\mathcal{H}}$, as we will show in \textsection\ref{s:physical_space}. For $\mathfrak{U} \subset \widehat{\mathfrak{U}}$ the result is quite simple, see \eqref{eq:Uqsl2_pol} and \eqref{eq:Uqsl2_pol_alt}: up to a prefactor these are partial (Hecke) \textsf{q}-symmetrisers that ensure the resulting polynomials have the correct symmetry. Likewise, the affine generators are essentially partial $\mathsf{q}^{-1}$-symmetrisers, besides a simple factor depending on the parameter~$\mathsf{p}$ from the $Y_i$, see \eqref{eq:affinisation_pol} and \eqref{eq:affinisation_pol_alt}.

\subsubsection{Abelian spin-chain symmetries from freezing} \label{s:intro_freezing} The \textsf{q}-deformed Haldane--Shastry spin chain arises from the spin-Ruijsenaars model by freezing. This is the topic of \textsection\ref{s:freezing}. 
As we mentioned in \textsection\ref{s:intro_spin-CS}, the idea of freezing is due to Polychronakos~\cite{Pol_93} and was further developed in \cite{BG+_93,TH_95}. In the \textsf{q}-case it is due to \cite{BG+_93,Ugl_95u}. In the limit\,%
\footnote{\ The limit $\mathsf{p}\to 1$ (at fixed $\mathsf{q}$) should not be confused with the (Jack) limit $\mathsf{p} = \mathsf{q}^{2\alpha} \to 1$ at fixed $\alpha$. Physically, $\mathsf{p} = \mathsf{q}^{2\hbar/k} = \E^{\I \, \hbar/m\mspace{1mu}c}$ so $\mathsf{p}\to 1$ corresponds to $\alpha = \hbar/k \to 0$. This can be interpreted either as a \emph{classical} limit ($\hbar \to 0$) or as letting $k \to \infty$, cf.~\textsection\ref{s:intro_spin-CS}. (The \emph{semiclassical} limit is the next order in $\hbar$.) Instead, $\mathsf{p} = \mathsf{q}^{2\alpha} \to 1$ can be interpreted as the \emph{non-relativistic} limit $c\to\infty$, see \textsection\ref{s:app_nonrlt_limit}, or the \emph{isotropic} limit from the spin-chain perspective.}
$\mathsf{p}\to 1$ the kinetic energy is negligible compared to the potential energy and the particles slow down to come to a halt at their (equispaced) classical equilibrium positions~$\ev z_j$, which are the same as for Calogero--Sutherland (\textsection\ref{s:intro_HS}), to give rise to the spin chain. 

More precisely, as in the spinless case, the $\widetilde{D}_r$ become trivial at $\mathsf{p}=1$. The spin-chain Hamiltonians thus arise as the `semiclassical' limit of the spin-Macdonald operators, by linearising at $\mathsf{p} = 1$. We denote this operation by $\partial/\partial\mathsf{p}\big|_{\mathsf{p}=1}$. Let $\begin{bmatrix} N \\ r \end{bmatrix} = \displaystyle \frac{[N]\,[N-1]\cdots [N-r+1]}{[r]\,[r-1]\cdots [2]}$ be the \textsf{q}-binomial coefficients.

\begin{theorem}[cf.~\cite{Ugl_95u}] \label{thm:freezing_abelian}
{\normalfont\textbf{i.}}~The spin-chain Hamiltonians, for $1\leq r \leq N-1$, arise from the spin-Macdonald operators as
\begin{equation} \label{eq:H_r}
	H_r = \ev \widetilde{H}_r \, , \qquad \widetilde{H}_r = \frac{1}{\mathsf{q}-\mathsf{q}^{-1}} \, \frac{\partial}{\partial\mspace{1mu}\mathsf{p}}\bigg|_{\mathsf{p}=1} \! \biggl( \widetilde{D}_r - \frac{r}{N} \mspace{1mu} \begin{bmatrix} N \\ r \end{bmatrix} \mspace{1mu} \widetilde{D}_N \biggr) \, .
\checkedMma
\end{equation}
In particular we recover $\normalfont H^\textsc{l} = H_1$ and $\normalfont H^\textsc{r} = H_{N-1}$.

{\normalfont\textbf{ii.}}~The eigenvalue of \eqref{eq:H_r} on the joint eigenspace $\mathcal{H}^\mu$, labelled by the motif $\mu \in \mathcal{M}_N$, is
\begin{equation} \label{eq:energy_r}
	\begin{aligned}
		E_r(\mu) = \sum_{m=1}^M & \varepsilon_r(\mu_m) \, , \\ 
		& \varepsilon_r(\mu_m) = \frac{1}{\mathsf{q}-\mathsf{q}^{-1}} \,\Biggl(\,\sum_{s=1}^r \, (-1)^{s-1} \, \begin{bmatrix} N \\ r - s \end{bmatrix} \, \mathsf{q}^{s \,(N-\mu_m)} \, \frac{[s\,\mu_m]}{[s]} \ \, - \, \frac{r}{N} \begin{bmatrix} N \\ r \end{bmatrix} \, \mu_m \Biggr) \, .
	\end{aligned}
	\checkedMma
\end{equation}
In particular we retrieve the dispersions $\normalfont \varepsilon^\textsc{l}(\mu_m) = \varepsilon_1(\mu_m)$ and $\normalfont \varepsilon^\textsc{r}(\mu_m) = \varepsilon_{N-1}(\mu_m)$.
\checkedMma
\end{theorem}
\noindent For $r=1$ this result is due to Uglov~\cite{Ugl_95u}. We prove Theorem~\ref{thm:freezing_abelian} in \textsection\ref{s:abelian_freezing}.

We discuss the two parts of Theorem~\ref{thm:freezing_abelian} in turn. The subtraction involving $\widetilde{D}_N$ in \eqref{eq:H_r} is to get rid of the differential operators $z_j \, \partial_{z_j}$ coming from the linearisation of the $\hat{\mathsf{p}}_j$ in $\widetilde{D}_r$.
\checkedMma
Concretely this subtraction amounts to moving the $\hat{\mathsf{p}}$s in $\widetilde{D}_r$ to the right as in~\eqref{eq:intro_spin-D_1_N=3} and then discarding them. Let us illustrate this with a

\begin{proof}[Example of \eqref{eq:H_r} for $r=1$ (sketch)]
\checkedMma
Note that $\ev A_j(\vect{z}) = [N]/N$. We compute
\checkedMma 
\begin{equation*}
	\frac{\partial}{\partial \mathsf{p}}\bigg|_{\mathsf{p}=1} \,
		\tikz[baseline={([yshift=-.5*11pt*0.4*.9]current bounding box.center)},	xscale=0.6*.9,yscale=0.4*.9,font=\footnotesize]{
		\draw[->] (9.5,0) node[below]{$z_N$} -- (9.5,9) node[above]{$z_N\vphantom{z_j}$};
		\draw[very thick,rounded corners=3.5pt] (8,0) node[below]{$\mathsf{p}\mspace{1mu}z_j$} -- (8,.5) -- (5,3.5) -- (5,4.5) node[yshift=-1pt]{$\mathllap{\mathsf{p}\,}\tikz[baseline={([yshift=-.5*12pt*.35]current bounding box.center)},scale=.35]{\fill[black] (0,0) circle (.25)}$};
		\draw[rounded corners=3.5pt,->] (5,4.5) -- (5,5.5) -- (8,8.5) -- (8,9) node[above]{$z_j$};
		\draw[rounded corners=3.5pt,->] (7,0) node[below]{$z_{j-1}$} -- (7,.5) -- (8,1.5) -- (8,4.5) -- (8,7.5) -- (7,8.5) -- (7,9) node[above]{$z_{j-1}$};
		\draw[rounded corners=3.5pt,->] (6,0) node[below]{$\cdots$} -- (6,1.5) -- (7,2.5) -- (7,4.5) -- (7,6.5) -- (6,7.5) -- (6,9) node[above]{$\cdots{\vphantom{z_j}}$};
		\draw[rounded corners=3.5pt,->] (5,0) node[below]{$z_1$} -- (5,2.5) -- (6,3.5) -- (6,4.5) -- (6,5.5) -- (5,6.5) -- (5,9) node[above]{$z_1\vphantom{z_j}$};
		\foreach \x in {-1,...,1} \draw (8.75+.2*\x,4.5) node{$\cdot\mathstrut$};		
	} 
	\!\! = z_j \, \partial_{z_j} + \sum_{i=1}^{j-1} \!\!
	\tikz[baseline={([yshift=-.5*11pt*0.4*.9]current bounding box.center)},	xscale=0.6*.9,yscale=0.4*.9,font=\footnotesize]{
		\draw[->] (9.5,0) node[below]{$z_N$} -- (9.5,9) node[above]{$z_N$};
		\draw[rounded corners=3.5pt,->] (8,0) node[below]{$\mspace{1mu}z_j$} -- (8,.5) -- (5,3.5) -- (5,5.5) -- (8,8.5) -- (8,9) node[above]{$z_j$};
		\draw[rounded corners=3.5pt,->] (7,0) node[below]{$\cdots{\vphantom{z_j}}$} -- (7,.5) -- (8,1.5) -- (8,7.5) -- (7,8.5) -- (7,9) node[above]{$\cdots{\vphantom{z_j}}$};
		\draw[rounded corners=3.5pt,->] (6,0) node[below]{$z_i$} -- (6,1.5) -- (7,2.5) -- (7,6.5) -- (6,7.5) -- (6,9) node[above]{$z_i{\vphantom{z_j}}$};
		\draw[rounded corners=3.5pt,->] (5,0) node[below]{$\cdots{\vphantom{z_j}}$} -- (5,2.5) -- (6,3.5) -- (6,5.5) -- (5,6.5) -- (5,9) node[above]{$\cdots{\vphantom{z_j}}$};
		\foreach \x in {-1,...,1} \draw (8.75+.2*\x,4.5) node{$\cdot\mathstrut$};
		\draw (6.5,2) [thick] ellipse (.25*.66 and .25);
	} 
	\!, \qquad 
	\tikz[baseline={([yshift=-.5*11pt*0.4*.9]current bounding box.center)},	xscale=0.6*.9,yscale=0.4*.9,font=\footnotesize]{
		\draw[rounded corners=3.5pt,->] (1,0) node[below]{$v$} -- (1,.5) -- (0,1.5) -- (0,2) node[above]{$v$};
		\draw[rounded corners=3.5pt,->] (0,0) node[below]{$u$} -- (0,.5) -- (1,1.5) -- (1,2) node[above]{$u$};
		\draw (.5,1) [thick] ellipse (.25*.66 and .25);
	}
	\! \coloneqq \frac{\partial}{\partial \mathsf{p}}\bigg|_{\mathsf{p}=1} \!\! \check{R}(u/\mathsf{p}\,v) \, .
\end{equation*}
Discard the derivative and observe that
\begin{equation*} \label{p:V_computation}
	\tikz[baseline={([yshift=-.5*11pt*0.4*.9]current bounding box.center)},	xscale=0.6*.9,yscale=0.4*.9,font=\footnotesize]{
		\draw[->] (9.5,0) node[below]{$z_N$} -- (9.5,9) node[above]{$z_N$};
		\draw[rounded corners=3.5pt,->] (8,0) node[below]{$\mspace{1mu}z_j$} -- (8,.5) -- (5,3.5) -- (5,5.5) -- (8,8.5) -- (8,9) node[above]{$z_j$};
		\draw[rounded corners=3.5pt,->] (7,0) node[below]{$\cdots{\vphantom{z_j}}$} -- (7,.5) -- (8,1.5) -- (8,7.5) -- (7,8.5) -- (7,9) node[above]{$\cdots{\vphantom{z_j}}$};
		\draw[rounded corners=3.5pt,->] (6,0) node[below]{$z_i$} -- (6,1.5) -- (7,2.5) -- (7,6.5) -- (6,7.5) -- (6,9) node[above]{$z_i{\vphantom{z_j}}$};
		\draw[rounded corners=3.5pt,->] (5,0) node[below]{$\cdots{\vphantom{z_j}}$} -- (5,2.5) -- (6,3.5) -- (6,5.5) -- (5,6.5) -- (5,9) node[above]{$\cdots{\vphantom{z_j}}$};
		\foreach \x in {-1,...,1} \draw (8.75+.2*\x,4.5) node{$\cdot\mathstrut$};
		\draw (6.5,2) [thick] ellipse (.25*.66 and .25);
	} 
	\! = \!
	% \sum_{i=1}^{j-1} \!
	\tikz[baseline={([yshift=-.5*11pt*0.4*.9]current bounding box.center)},	xscale=0.6*.9,yscale=0.4*.9,font=\footnotesize]{
		\draw[->] (9.5,0) node[below]{$z_N$} -- (9.5,6) node[above]{$z_N$};
		\draw[rounded corners=3.5pt,->] (8,0) node[below]{$\mspace{1mu}z_j$} -- (8,.5) -- (6,2.5) -- (6,3.5) -- (8,5.5) -- (8,6) node[above]{$z_j$};
		\draw[rounded corners=3.5pt,->] (7,0) node[below]{$\cdots{\vphantom{z_j}}$} -- (7,.5) -- (8,1.5) -- (8,4.5) -- (7,5.5) -- (7,6) node[above]{$\cdots{\vphantom{z_j}}$};
		\draw[rounded corners=3.5pt,->] (6,0) node[below]{$z_i$} -- (6,1.5) -- (7,2.5) -- (7,3.5) -- (6,4.5) -- (6,6) node[above]{$z_i{\vphantom{z_j}}$};
		\draw[rounded corners=3.5pt,->] (5,0) node[below]{$\cdots{\vphantom{z_j}}$} -- (5,6) node[above]{$\cdots{\vphantom{z_j}}$};
		\foreach \x in {-1,...,1} \draw (8.75+.2*\x,3) node{$\cdot\mathstrut$};
		\draw (6.5,2) [thick] ellipse (.25*.66 and .25);
	} 
	\! = \, \bigl(\mathsf{q}-\mathsf{q}^{-1}\bigr) 
	% \sum_{i=1}^{j-1} 
	\, V(z_i,z_j) \times\!
	\tikz[baseline={([yshift=-.5*11pt*0.4*.9]current bounding box.center)},	xscale=0.6*.9,yscale=0.4*.9,font=\footnotesize]{
		\draw[->] (9.5,0) node[below]{$z_N$} -- (9.5,6) node[above]{$z_N$};
		\draw[rounded corners=3.5pt,->] (8,0) node[below]{$z_j$} -- (8,1.5) -- (7,2.5) -- (7,3.5) -- (8,4.5) -- (8,6) node[above]{$z_j$};
		\draw[rounded corners=3.5pt,->] (7,0) node[below]{$\cdots{\vphantom{z_j}}$} -- (7,1.5) -- (8,2.5) -- (8,3.5) -- (7,4.5) -- (7,6) node[above]{$\cdots{\vphantom{z_j}}$};
		\draw[rounded corners=3.5pt,->] (6,0) node[below]{$z_i$} -- (6,6) node[above]{$z_i{\vphantom{z_j}}$};
		\draw[rounded corners=3.5pt,->] (5,0) node[below]{$\cdots{\vphantom{z_j}}$} -- (5,6) node[above]{$\cdots{\vphantom{z_j}}$};
		\foreach \x in {-1,...,1} \draw (8.75+.2*\x,3) node{$\cdot\mathstrut$};
		\draw[style={decorate, decoration={zigzag,amplitude=.5mm,segment length=2mm}}] (6,3) -- (7,3);
	} .
\end{equation*}
We thus get $H^\textsc{l}$ from \eqref{eq:ham_left} at $r=1$, where we remove the prefactor $\mathsf{q}-\mathsf{q}^{-1}$ to ensure that the limit $\mathsf{q}\to 1$ is nontrivial. A detailed proof will be given in \textsection\ref{s:abelian_freezing}.
\end{proof}

Likewise, \eqref{eq:intro_spin-D_-1} yields $H_{N-1} = H^\textsc{r}$, as we will also show in \textsection\ref{s:abelian_freezing}.
\checkedMma
Explicit expressions for the higher spin-chain Hamiltonians are similarly computed. For instance, \eqref{eq:intro_spin-D_2} gives rise to
\begin{equation} \label{eq:H_2}
	\begin{aligned}
	\widetilde{H}_2 = \
	\vphantom{\sum^N} \smash{ 
	\frac{1}{\mathsf{q}-\mathsf{q}^{-1}} \, \sum_{j<j'}^N A_{j j'}(\vect{z}) \ \frac{\partial}{\partial \mathsf{p}}\bigg|_{\mathsf{p}=1} 
	} & \check{R}_{j-1,j}(z_j/z_{j-1}) \cdots \check{R}_{12}(z_j/z_1) \\
	& \times \check{R}_{j'-1,j'}(z_{j'}/z_{j'-1}) \cdots \check{R}_{j+1,j+2}(z_{j'}/z_{j+1}) \\
	& \quad \times \check{R}_{j,j+1}(z_{j'}/z_{j-1}) \cdots \check{R}_{23}(z_{j'}/z_1) \\
	& \quad \times \check{R}_{23}(z_1/\mathsf{p}\,z_{j'}) \cdots \check{R}_{j,j+1}(z_{j-1}/\mathsf{p}\,z_{j'}) \\
	& \times \check{R}_{j+1,j+2}(z_{j+1}/\mathsf{p}\,z_{j'}) \cdots \check{R}_{j'-1,j'}(z_{j'-1}/\mathsf{p}\,z_{j'}) \! \\
	\times \ & \check{R}_{12}(z_1/\mathsf{p}\,z_j) \cdots \check{R}_{j-1,j}(z_{j-1}/\mathsf{p}\,z_j) \, .
	\end{aligned}
\checkedMma
\end{equation}
Here the linearisation can be explicitly evaluated as for $r=1$. Notice that $H_N=0$.
\checkedMma
 
As for the second part of Theorem~\ref{thm:freezing_abelian} we conclude with some
\begin{remarks} \textbf{i.}~Note the symmetries $\varepsilon_{N-r}(\mu_m) = \varepsilon_r(\mu_m)|_{\mathsf{q}\mapsto\mathsf{q}^{-1}} = \varepsilon_r(N-\mu_m)$.
\checkedMma

\textbf{ii.}~The isotropic limit of these eigenvalues is conveniently computed from those of $H^\text{full}_r \coloneqq (H_r + H_{N-r})/2$, which are determined from the dispersion
\begin{equation} \label{eq:energy_r_full}
	\begin{aligned}
	\varepsilon^\text{full}_r(\mu_m) = {} & \frac{1}{2} \, \sum_{s=1}^r (-1)^{s-1} \, \begin{bmatrix} N \\ r - s \end{bmatrix} \, [s\,(N-\mu_m)] \, \frac{[s\,\mu_m]}{[s]} \\
	\to {} & \frac{1}{2} \, \sum_{s=1}^r (-1)^{s-1} \, \binom{N}{r - s} \, s \ (N-\mu_m) \, \mu_m = \binom{N-2}{r-1} \, \varepsilon^\textsc{hs}(\mu_m) \, , \qquad \mathsf{q} \to 1 \, .
	\end{aligned}
	\checkedMma
\end{equation}
As the very weak dependence on $r$ in the result signals, the higher spin-chain Hamiltonians all become dependent in the isotropic limit. It should be possible to extract the explicit expressions for the first few higher Hamiltonians of the ordinary Haldane--Shastry chain~\cite{Ino_90,HH+_92,TH_95} from the above by carefully taking the isotropic limit. 

\textbf{iii.}~Observe that the \textsf{q}-deformed spin-chain Hamiltonians are obtained by linearising at $\mathsf{p}= 1$ and give the ordinary Haldane--Shastry spin chain by setting $\mathsf{q}=1$. Instead, the quantum-affine symmetries~\eqref{eq:intro_L_spin_chain} are obtained from \eqref{eq:intro_L_tilde} for the spin-Ruijsenaars model by putting $\mathsf{p}=1$ (\textsection\ref{s:nonabelian_freezing}) but, as usual, have to be linearised in at $\mathsf{q}=1$ to get the (double) Yangian symmetry of the Haldane--Shastry model. Both specialisations involve linearising once.
\end{remarks}

Table~\ref{tb:abelian_symmetries} gives an overview of the abelian symmetries.

\begin{table}[h]
	\begin{tabular}{cccc|cc} \toprule
		\multicolumn{2}{c}{spin-Ruijsenaars model (\textsection\ref{s:spin_Ruijsenaars})} & 
		\begin{tabular}{c} classical \\ ($\mathsf{p}=1$) \end{tabular} & 
		$\!\!\!\!$\begin{tabular}{c} semiclass. \\ $\bigl(\frac{\partial}{\partial\mspace{1mu}\mathsf{p}}\big|_{\mathsf{p}=1}\bigr)$ \end{tabular} & 
		\multicolumn{2}{c}{ \begin{tabular}{c} \textsf{q}-deformed \\ Haldane--Shastry \end{tabular} (\textsection\ref{s:freezing})} \\ \midrule
		$\widetilde{D}_0 = D_0 =1$ & & $1$ & $\!\!\!0$ & $0$ & \\
		%%%
		$\widetilde{D}_1$ & \eqref{eq:intro_spin-D_1}, \eqref{eq:spin-D_1}$\!\!\!$ & $[N]$ & $\!\!\!\!$\eqref{eq:H_tilde_expansion} & $H^\textsc{l}=H_1$ & \eqref{eq:ham_left}, \eqref{eq:ham_left_from_freezing} \\
		%%%
		$\widetilde{D}_2$ & \eqref{eq:intro_spin-D_2}$\!\!\!$ & 
		% $[N]\,[N-1]/[2]$ 
		& & $H_2$ & \eqref{eq:H_2} \\[0ex]
		%%%
		$\smash{\vdots}$ & & & & $\smash{\vdots}$ & \\
		%%%
		$\widetilde{D}_r$ & \eqref{eq:spin_D_r}, \eqref{eq:spin_D_symmetry}$\!\!\!$ & \eqref{eq:Macdonald_gen_fn_class_pt} & & $H_r$ & \eqref{eq:H_r}, \textsection\ref{s:abelian_freezing} \\
		%%%
		$\smash{\vdots}$ & & & & $\smash{\vdots}$ & \\
		%%%
		$\widetilde{D}_{N-1} = \widetilde{D}_N \, \widetilde{D}_{-1}$ & \eqref{eq:intro_spin-D_-1}, \eqref{eq:spin_D_-1}$\!\!\!$ & $[N]$ & $\!\!\!\!$\eqref{eq:ham_right_semiclassical} & $H^\textsc{r} = H_{N-1}$ & \eqref{eq:ham_right}, \eqref{eq:ham_right_from_freezing} \\
		%%%
		$\widetilde{D}_N = D_N$ & \eqref{eq:intro_spin-D_N}, \eqref{eq:D_N}$\!\!\!$ & $1$ & $\!\!\!\!$\eqref{eq:total_degree_semiclassical} & $G$ & \eqref{eq:q-translation}, \textsection\ref{s:abelian_freezing} \\ 
		\bottomrule
		\hline
	\end{tabular}
	\caption{Summary of the abelian symmetries: the spin-Macdonald operators, their leading terms near $\mathsf{p}= 1$, and the derived symmetries of the spin chain. (Recall that in \textsection\ref{s:setup}--\ref{s:proofs} we write $q=\mathsf{p}$.)}
	\label{tb:abelian_symmetries}
\end{table}

\subsubsection{Explicit spin-chain eigenvectors from freezing} \label{s:intro_explicit_evrs_freezing} To find the eigenvectors of the spin chain we exploit the algebraic structure available prior to evaluation.\,%
\footnote{\ Note that we do \emph{not} derive the explicit spin-chain eigenvectors by freezing those of the spin-Ruijsenaars model. A reason is that the procedure of freezing is highly surjective; many vectors simplify significantly or are killed in the process. We briefly comment on the exact eigenvectors of the spin-Ruijsenaars model in \textsection\ref{s:explicit_evrs_tilde}.}
This is the topic of \textsection\ref{s:explicit_evrs}, where we derive the pseudo highest-weight eigenvectors that we presented in \textsection\ref{s:intro_explicit_evrs}. In a nutshell we proceed as follows.

As usual we work per $M$-particle sector. The Corollary of Proposition~\ref{prop:intro_physical_vectors} allows us to pass to the world of polynomials by focussing on the simple component~$\cbraket{1,\To,M}{\widetilde{\Psi}} = \widetilde{\Psi}(\vect{z})$, symmetric in $z_1,\To,z_M$ and in $z_{M+1},\To,z_N$. Evaluation helps selecting a suitable subspace of polynomials: it does not just tell us to restrict to degree at most $N-1$ in each variable, but allows us to consider polynomials that depend only on the first $M$ variables. In \textsection\ref{s:explicit_evrs} we will show that this is how we get from Proposition~\ref{prop:intro_physical_vectors} to Theorem~\ref{thm:phys_vector_spinchain}.

The simple component may thus be taken to be a symmetric polynomial in $z_1,\To,z_M$. Like in \cite{BG+_93} we determine it by passing through the \emph{non-symmetric} theory: the spin-chain Hamiltonians can be diagonalised along with the \textit{Y}$\mspace{-2mu}$-operators with `classical' parameters
\begin{equation*}
	\mathsf{p}^\circ = 1 \, , \qquad \mathsf{q}^\circ = \mathsf{q} \, ,
\end{equation*}
Indeed, before evaluation the spin-chain Hamiltonians~\eqref{eq:H_r} commute with $Y_i^\circ = Y_i \, |_{\mathsf{p}=1}$ from \eqref{eq:intro_Yi_circ}. We may therefore look for simultaneous eigenfunctions of these \emph{classical} \textit{Y}$\mspace{-2mu}$-operators. The restriction to polynomials in $z_1,\To,z_M$ suggests focussing on $Y_m^\circ$ for $1 \leq m \leq M$. These operators still depend on all $N$ variables; although they don't preserve the subspace of polynomials in $z_1,\To,z_M$ in general, they do so \emph{on shell}, i.e.\ upon evaluation:

\begin{theorem} \label{thm:intro_Ym_to_Y'm} 
The classical ($\mathsf{p}=1$) representation of the affine Hecke algebra on $\mathbb{C}[z_1,\To,z_N]$ contains a finite-dimensional subspace (though not submodule) that is \emph{on shell} preserved by $Y_1^\circ,\To,Y_M^\circ$. On this subspace the latter are \emph{on shell} conjugate to (a multiple of) \textit{Y}$\mspace{-2mu}$-operators that act on polynomials in only $M$ variables and have their parameters shifted to 
\begin{equation*}
	\mathsf{p}' = \mathsf{q}^{\mspace{1mu}\prime\,2} = \mathsf{q}^{-2} \, .
\end{equation*}
\end{theorem}
\noindent This is the main technical result of \textsection\ref{s:explicit_evrs}. For $Y_1^\circ$ a part of our derivation is closely related to a result of \cite{NS_17} and proof of \cite{Cha_19}, see our Lemma~\ref{lem:Y1_circ} and Proposition~\ref{prop:Ym_circ} plus ensuing discussion.

The parameters of the \textit{Y}$\mspace{-2mu}$-operators are shifted further when we pass from the joint eigenfunctions of the \textit{Y}$\mspace{-2mu}$-operators (nonsymmetric Macdonald polynomials) back to symmetric Macdonald polynomials. At the end of the day we obtain the wave functions~\eqref{eq:intro_wavefn} from Theorem~\ref{thm:nice_polynomial} involving Macdonald polynomials with parameters at the quantum zonal spherical point
\begin{equation*}
	\mathsf{p}^\star = \mathsf{q}^\star = \mathsf{q}^2 \, .
\end{equation*} 

The precise steps are summarised at the end of \textsection\ref{s:explicit_evrs}. These results suggest that it should be possible to relate the (polynomial) action of the spin chain on the $M$-particle sector to that of the quantum zonal spherical case of Macdonald operators, $D^\star_r$. We have not yet managed to find such a relation, which would also allow for a direct way of computing the energy eigenvalues from \textsection\ref{s:intro_abelian}. At the moment our derivation is computational; it would be desirable to understand it from a more structural (representation-theoretic, or perhaps geometric) point of view.

We finally prove that the condition $\ell(\mu) = M$ in Theorem~\ref{thm:nice_polynomial} from \textsection\ref{s:intro_explicit_evrs} is a pseudo highest-weight condition in the sense of \textsection\ref{s:intro_nonabelian_2}, and compute the Drinfeld polynomial~\eqref{eq:Drinfeld_polyn} from \textsection\ref{s:intro_nonabelian_1} with the help of the trick described at the end of \textsection\ref{s:intro_phys_space}.

\subsection{Outline} The main text is organised as follows. In \textsection\ref{s:setup} we review the algebraic preliminaries. The \emph{polynomial} representation of the affine Hecke algebra and its relation to Macdonald polynomials and the Ruijsenaars model are discussed in \textsection\ref{s:polynomial} in a way that will readily extend to the spin-Ruijsenaars setting. The \emph{spin} representation of the (finite) Hecke algebra, the quantum groups $\mathfrak{U}$ and $\widehat{\mathfrak{U}}$, and their relation to Heisenberg-type spin chains is the topic of \textsection\ref{s:spin}. 

The core of this work is \textsection\ref{s:proofs}, where we prove the results described above. Following \cite{BG+_93,TH_95,Ugl_95u} we derive the Hamiltonian of the \textsf{q}-deformed Haldane--Shastry spin chain in pairwise form~\cite{Lam_18} from the trigonometric spin-Ruijsenaars model (\textsection\ref{s:spin_Ruijsenaars}) by freezing (\textsection\ref{s:freezing}). In \textsection\ref{s:explicit_evrs} we construct the exact spin-chain eigenvectors and prove their on-shell pseudo highest-weight property.

There are three appendices. \textsection\ref{s:app_glossary} contains a glossary of our notation. In \textsection\ref{s:app_nonrlt_limit} we evaluate the istropic/nonrelativistic limit to facilitate comparison with the literature on the Haldane--Shastry model. Finally, in \textsection\ref{s:app_stochastic_twist} we discuss the stochastic version of the \textsf{q}-deformed Haldane--Shastry model in \textsection\ref{s:app_presentations} we derive the Chevalley generators of the Drinfeld--Jimbo presentation of $\widehat{\mathfrak{U}}$ from the monodromy matrix of the Faddeev--Reshetikhin--Takhtajan presentation.

\subsection{Acknowledgements} JL acknowledges the Knut and Alice Wallenberg Foundation (\textsc{kaw}) and the Australian Research Council Centre of Excellence for Mathematical and Statistical Frontiers (\textsc{acems}) for subsequent support during the course of this work. 

We thank the program \textit{Correlation functions in solvable models} (\textsc{Nordita}, 2018) and the meeting \textit{New trends in integrable systems} (Osaka City University, 2019), where part of this work was done. JL thanks the IPhT Saclay for hospitality at several stages of this work, and VP and DS the University of Melbourne for hospitality.

We thank P.~Di~Francesco, P.~Fendley, J.~de Gier, T.~Görbe, M.~Hallnäs, J.~Jacobsen, R.~Kedem, I.~Kostov, E.~Langmann, P.~McNamara, A.~Ram, D.~Volin, S.~O.~Warnaar, R.~Weston and M.~Wheeler for useful discussions. JL is especially grateful to O.~Chalykh, J.~Stokman and P.~Zinn-Justin for numerous discussions. We furthermore thank the anonymous referees for their careful reading of the manuscript and their detailed feedback, which enabled us to significantly improve the structure and exposition of this work. Finally we thank R.~Klabbers for proof reading \textsection\ref{s:intro}.

\section{Algebraic setup} \label{s:setup}
\noindent
In this section we recall various notions from the \textsf{q}-world to fix our notation and conventions and pave the way for the algebraic framework that we will use in \textsection\ref{s:proofs}. One might wish to skip this section; we will refer to the relevant parts when we need them in \textsection\ref{s:proofs}.

From now on we follow~\cite{Mac_95,Mac_98} and work with parameters $t^{1/2} = \mathsf{q}$ and $q =\mathsf{p}$. (The latter is denoted by $\rho$ in \cite{BG+_93} and $p$ in \cite{JK+_95a,JK+_95b,Ugl_95u}.) We will keep using the terminology `\textsf{q}-deformed'. One can either fix $t^{1/4} \in \mathbb{C} \setminus \{-1,0,1\}$\,---\,with exponent~$1/4$ in view of e.g.~\eqref{eq:T_sp_eigenspaces}\,---\,or work over the ring $\mathbb{C}(\!(t^{1/4})\!)$ of formal Laurent polynomials in $t^{1/4}$; to keep the notation light we use the former point of view. We work with the symmetric definition of the \textsf{q}-analogues of integers, factorials and binomial coefficients (Gaussian polynomials),
\begin{equation} \label{eq:q-numbers_etc}
	[n] \coloneqq [n]_{t^{1/2}} = \frac{t^{n/2} - t^{-n/2}}{t^{1/2} - t^{-1/2}} \, , \qquad [n]! \coloneqq [n] \, [n-1] \cdots [2] \, , \qquad \begin{bmatrix} n \\ k \end{bmatrix} \coloneqq \frac{[n]!}{[k]!\,[n-k]!} \, .
\end{equation}
We'll often factor out fractional powers of $t$ but all normalisations remain as in \textsection\ref{s:intro}.

\subsection{Polynomial side} \label{s:polynomial} Consider the algebra $\mathbb{C}[\vect{z}] \coloneqq \mathbb{C}[z_1,\To,z_N] \cong \mathbb{C}[z]^{\otimes N}$ of polynomials in $N$ variables. This space naturally is a module of the symmetric group~$\mathfrak{S}_N$ by permuting variables, generated by simple transpositions~$s_1,\To,s_{N-1}$ acting as $s_i \, z_i = z_{i+1} \, s_i$, so $w\in\mathfrak{S}_N$ acts by $(w \, F)(\vect{z}) = F(\vect{z}_w)$ where $(\vect{z}_w)_i \coloneqq z_{w\,i}$. As the notation suggests the latter is a right action on $\vect{z}$, yielding a left action on $F \in \mathbb{C}[\vect{z}]$. We use the cycle notation for permutations. We write $\mathbb{C}[\vect{z}]^{\mathfrak{S}_N}$ for the ring of symmetric polynomials in $N$ variables.

\subsubsection{Hecke algebras} \label{s:Hecke_AHA} 
The following \textsf{q}-deformation of (the group algebra $\mathbb{C}[\mathfrak{S}_N]$ of) the symmetric group plays a central role in this work.
\begin{definition}
The \emph{(Iwahori--)Hecke algebra} $\mathfrak{H}_N \coloneqq \mathfrak{H}_N\bigl(t^{1/2}\bigr)$ of type $A_{N-1}$ is the unital associative algebra with generators $T_1,\To,T_{N-1}$ obeying
\begin{equation} \label{eq:Hecke}
	\begin{gathered}
	\text{braid relations:} \qquad \hfill T_i \, T_{i+1} \, T_i = T_{i+1} \, T_i \, T_{i+1} \, , \qquad T_i \, T_j = T_j \, T_i \quad \text{if } |i-j|>1 \, , \hfill \\
	\text{Hecke condition:} \qquad \hfill \bigl(T_i-t^{1/2}\bigr)\bigl(T_i+t^{-1/2}\bigr) = 0 \, . \hfill
	\end{gathered}
\checkedMma
\end{equation}
\end{definition}

The Hecke condition means that $T_i$ is invertible, with $t^{1/2}-t^{-1/2}$ measuring the extent by which $T_i$ fails to be an involution:
\begin{equation} \label{eq:Hecke_inverse}
	T_i^{-1} = T_i -  (t^{1/2}-t^{-1/2}) \, .
\checkedMma
\end{equation}
The Hecke algebra has dimension $\dim \mathfrak{H}_N =N!$ for generic $t^{1/2} \in\mathbb{C}^\times$, with a basis $\{T_w\}_{w\in \mathfrak{S}_N}$ indexed by the symmetric group, $T_w = T_{i_1} \cdots T_{i_r}$ for any reduced decomposition $w = s_{i_1} \cdots s_{i_r}$;  e.g.\ $T_e = 1$, $T_{s_i} = T_i$ and \eqref{eq:Hecke_graphical}.

The Hecke condition fixes the possible eigenvalues of any representation of $T_i$ to $t^{1/2}$ and $-t^{-1/2}$. Although \eqref{eq:Hecke} is invariant under replacing $t^{1/2} \rightsquigarrow {-t^{-1/2}}$, this symmetry might be broken when picking a representation, cf.\ the dimensions in \eqref{eq:T_sp_eigenspaces} (Appendix~\ref{s:app_stochastic_twist}). We will work with representations where eigenvectors with eigenvalue $t^{1/2}$ ($-t^{-1/2}$) become (anti)symmetric at $t = 1$, see \eqref{eq:T_pol_eigenspaces} and \eqref{eq:T_sp_eigenspaces}. 

The Hecke algebra has two well-known representations: one on polynomials, and one on spins (\textsection\ref{s:Hecke_TL_Uqsl}). On $\mathbb{C}[\vect{z}]$ the action of $\mathfrak{S}_N$ is deformed to \eqref{eq:Hecke_graphical}, i.e.\ to the Demazure--Lusztig operator
\begin{equation} \label{eq:Hecke_pol_2a}
	T_i^\text{pol} \coloneqq {-}t^{-1/2} \, (t \, z_i- z_{i+1}) \, \partial_i + t^{1/2} \, ,
\checkedMma
\end{equation}
where the (Newton) divided difference is defined as
\begin{equation} \label{eq:div_diff}
	\partial_i \coloneqq (z_i-z_{i+1})^{-1} \, (1-s_i) \, .
\checkedMma
\end{equation}
Since $1-s_i$ antisymmetrises, $\partial_i$ preserves polynomials despite its denominator, so \eqref{eq:Hecke_pol_2a} does indeed act on $\mathbb{C}[\vect{z}]$. The divided differences obey the braid relations and $\partial_i^2 = 0$, yielding a representation of the \emph{nil-Hecke algebra}. In terms of the rational functions
\begin{equation} \label{eq:a,b}
	\begin{aligned}
	a(u) & \coloneqq t^{-1/2} \, \frac{t\,u-1}{u-1} \, , \qquad\quad & a_{ij} & \coloneqq a(z_i/z_j) \, , \\
	b(u) & \coloneqq -t^{-1/2} \, \frac{t-1}{u-1} \, , \qquad\quad & b_{ij} & \coloneqq b(z_i/z_j) \, ,
	\end{aligned}
\checkedMma
\end{equation}
we have
\begin{equation} \label{eq:Hecke_pol_2b}
	T^\text{pol}_i = a_{i,i+1} \, s_i + b_{i,i+1} \, , \qquad T^{\text{pol}\,-1}_i = a_{i,i+1} \, s_i - b_{i+1,i} \, .
\checkedMma
\end{equation}

For $N=2$ the decomposition into $t^{1/2}$- and $-t^{-1/2}$-eigenspaces of \eqref{eq:Hecke_pol_2a} is
\begin{equation} \label{eq:T_pol_eigenspaces}
	\mathbb{C}[z_1,z_2] \cong \, \mathbb{C}[z_1,z_2\,]^{\mathfrak{S}_2} \, \oplus \, (t\,z_1- z_2) \, \mathbb{C}[z_1,z_2\,]^{\mathfrak{S}_2} \, ,
\checkedMma
\end{equation}
In general the $\mathfrak{H}_N$-irreps in $\mathbb{C}[\vect{z}]$ are classified by Young diagrams, just as for $\mathfrak{S}_N$. We will be interested in the totally \textsf{q}-(anti)symmetric cases. Denote the \textsf{q}-$\mspace{-1mu}$Vandermonde polynomial by
\begin{equation} \label{eq:q-Vand}
	\Delta_t(z_1,\To,z_N) \coloneqq t^{-N\,(N-1)/4} \, \prod_{i<j}^N (t\,z_i-\,z_j) \, ,
\checkedMma
\end{equation}
and write $\ell(w)$ for the length of $w \in \mathfrak{S}_N$. The total \textsf{q}-(anti)symmetrisers are~\cite{Jim_86a}
\begin{equation} \label{eq:Hecke_projectors}
	\Pi_\pm \coloneqq \frac{t^{\mp N\,(N-1)/4}}{[N]!} \! \sum_{w \in \mathfrak{S}_N} \!\! (\pm t^{\pm 1/2})^{\ell(w)} \, T_w \, , 
\checkedMma
\end{equation}
The exponent in the prefactor is $N(N-1)/2 = \ell(w_0)$, with $w_0 \coloneqq (1\cdots N)\cdots (123)(12)$ the longest permutation in $\mathfrak{S}_N$, reversing the order of the coordinates ($z_i \leftrightarrow z_{N-i+1}$ for all $i$). In the polynomial case an efficient implementation uses the associated divided difference $\partial_{w_0} = (\partial_1 \cdots \partial_{N-1}) \cdots (\partial_1 \partial_2) \, \partial_1$, see Theorem~3.1 in \cite{DKL+_95}:
\begin{equation} \label{eq:Hecke_projectors_pol}
	\begin{aligned}
	\Pi_+^\text{pol} & = \frac{1}{[N]!} \, \partial_{w_0} \bigl( \Delta_{1/t}(\vect{z}) \, \cdot \, \bigr) \, , 
	\qquad \qquad
	& \Pi_+^\text{pol} \, \mathbb{C}[\vect{z}] & = \mathbb{C}[\vect{z}]^{\mathfrak{S}_N} \, , \\ 
	\Pi_-^\text{pol} & = \frac{1}{[N]!} \, \Delta_t(\vect{z}) \, \partial_{w_0} \, ,
	& \Pi_-^\text{pol} \, \mathbb{C}[\vect{z}] & = \Delta_t(\vect{z})\,\mathbb{C}[\vect{z}]^{\mathfrak{S}_N} \, .
	\end{aligned}
\checkedMma
\end{equation}
Note that $\mathfrak{H}_N$-symmetric polynomials are $\mathfrak{S}_N$-symmetric, yet $\mathfrak{H}_N$-skew (totally antisymmetric) polynomials are not $\mathfrak{S}_N$-skew.

\begin{definition}
The (extended) \emph{affine Hecke algebra}, or \textsc{aha}, $\widehat{\mathfrak{H}}_N \coloneqq \widehat{\mathfrak{H}}_N\bigl(t^{1/2}\bigr)$ of type $\widehat{A}_{N-1}$ (or more precisely $\mathfrak{gl}_N$) \cite{Lus_83,Lus_89} is a unital associative algebra that extends the (`finite') Hecke algebra~$\mathfrak{H}_N$ by $\mathbb{C}[\vect{Y}] = \mathbb{C}[Y_1,\To,Y_N]$\,---\,this notation means that the additional (Jucys--Murphy) generators $Y_i$ commute\,---\,with cross relations
\begin{equation} \label{eq:AHA}
	T_i^{-1} \, Y_i \, T_i^{-1} = Y_{i+1} \, , \qquad\quad T_i \, Y_j = Y_j \, T_i \quad \text{if } j\neq i,i+1 \, .
\checkedMma
\end{equation}
\end{definition}
\noindent Observe that this `chiral' setting may be extended to the `full' \textsc{aha} by including the inverses of $Y_i$. These will play a role in \textsection\ref{s:Macdonald} and \ref{s:abelian_tilde}.

The \emph{basic} representation of $\widehat{\mathfrak{H}}_N$ is an extension,  depending on a parameter~$q$, of the polynomial representation~\eqref{eq:Hecke_pol_2a} of $\mathfrak{H}_N$. To keep the notation light we'll think of $q\in \mathbb{C}^\times$ as fixed. Since we will only work with the polynomial representation of the \textsc{aha} we omit the superscript `pol' for the following operators. Define the \textit{q}-dilatation, or (multiplicative) difference, operator~$\hat{q}_i$ on $\mathbb{C}[\vect{z}]$ by 
\begin{subequations} \label{eq:q_hat}
\begin{gather}
	(\hat{q}_i \, F)(\vect{z}) \coloneqq F(z_1,\To,z_{i-1},q\,z_i,z_{i+1},\To,z_N) \, . 
\checkedMma
\shortintertext{It formally shifts the position of the $i$th coordinate, and can be expressed as} 
	\hat{q}_i = \sum_{n \geq 0} \frac{1}{n!} \, (q-1)^n \, z_i^n \, \partial_{z_i}^n = q^{z_i \, \partial_{z_i}} \, , 
\end{gather}
\end{subequations}
Here $z_i \, \partial_{z_i}$ counts the degree in $z_i$, and is the $i$th (continuum) momentum operator~$-\I \, \partial_{x_i}$ ($\hbar \equiv 1$) in multiplicative notation, cf.~\textsection\ref{s:app_CalSut_limit}. (The partial derivatives $\partial_{z_i}$, etc., should not be confused with divided differences $\partial_i$.)

There are two ways to express the affine generators, found independently in~\cite{Che_92b,BG+_93}. One features the twisted cyclic shift operator $\pi$ acting on $\mathbb{C}[\vect{z}]$ by
\begin{equation} \label{eq:pi}
	(\pi \, F)(\vect{z}) \coloneqq F(q\,z_N,z_1,\To,z_{N-1}) \, ,
\checkedMma
\end{equation}
In this notation the \textsf{q}-deformed (difference) Dunkl operators are~\cite[\textsection{A}]{Che_92b}
\begin{equation} \label{eq:Yi}
	Y_i \coloneqq T_i^\text{pol} \cdots T_{N-1}^\text{pol} \, \pi \, T_1^{\text{pol}\,-1} \cdots T_{i-1}^{\text{pol}\,-1} \, .
\checkedMma
\end{equation}
In view of \eqref{eq:AHA} the first affine generator $Y_1 = T_1^\text{pol} \cdots T_{N-1}^\text{pol} \pi$ determines $Y_2,\cdots,Y_N$. 

In terms of the braid diagrams \eqref{eq:Hecke_graphical} supplemented with the graphical notation~\eqref{eq:p_hat_graphical} the expression~\eqref{eq:Yi} may be depicted as in the first diagram in
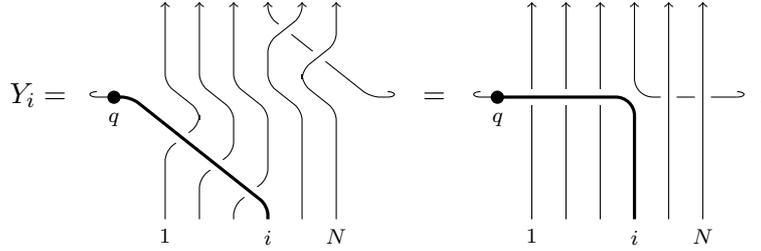
\begin{equation*}
	Y_i = \ \,
	\tikz[baseline={([yshift=-.5*11pt*0.4*.9+6pt+4pt]current bounding box.center)},xscale=0.5*.9,yscale=0.4*.9,font=\scriptsize, cross line/.style={-,preaction={draw=white,-,line width=6pt}}]{
	\draw[rounded corners=3.5pt,->] (0,0) node[below]{$1$} -- (0,2.5) -- (1,3.5) -- (1,4) -- (0,5) -- (0,8);
		\draw[rounded corners=3.5pt,->] (1,0) -- (1,1.5)-- (2,2.5) -- (2,4) -- (1,5) -- (1,8);
		\draw[rounded corners=3.5pt,->] (2,0) -- (2,.5) -- (3,1.5) -- (3,4) -- (2,5) -- (2,8);
		\draw[cross line,rounded corners=3.5pt,very thick] (3,0) node[below]{$i$} -- (3,.5) -- (0,3.5) -- (-1,4.5) -- (-1.5,4.5) node{$\tikz[baseline={([yshift=-.5*12pt*.35]current bounding box.center)},scale=.35]{\fill[black] (0,0) circle (.25)}$};
		\draw (-1.5,4.5) node[below,yshift=-2pt]{$q$} -- (-2,4.5) arc(-90:-270:.2 and .1);
		\draw (6.5,4.5) arc(-90:90:.2 and .1);
		\draw[rounded corners=3.5pt,->] (6.5,4.5) -- (6,4.5) -- (3,7.5) -- (3,8);
		\draw[cross line,rounded corners=3.5pt,->] (5,0) node[below]{$N$} -- (5,4) -- (4,5) -- (4,5.5) -- (5,6.5) -- (5,8);
		\draw[cross line,rounded corners=3.5pt,->] (4,0) -- (4,4) -- (3,5) -- (3,6.5) -- (4,7.5) -- (4,8);
	} \ \ = \ \
	\tikz[baseline={([yshift=-.5*11pt*0.4*.9+6pt+4pt]current bounding box.center)},xscale=0.5*.9,yscale=0.4*.9,font=\scriptsize, cross line/.style={-,preaction={draw=white,-,line width=6pt}}]{
		\draw[->] (0,0) node[below]{$1$} -- (0,8);
		\draw[->] (1,0) -- (1,8);
		\draw[->] (2,0) -- (2,8);
		\draw[cross line,rounded corners=7pt,very thick] (3,0) node[below]{$i$} |- (-1,4.5) node{$\tikz[baseline={([yshift=-.5*12pt*.35]current bounding box.center)},scale=.35]{\fill[black] (0,0) circle (.25)}$};
		\draw (-1,4.5) node[below,yshift=-2pt]{$q$} -- (-1.5,4.5) arc(-90:-270:.2 and .1);
		\draw (6,4.5) arc(-90:90:.2 and .1);
		\draw[rounded corners=7pt,->] (6,4.5) -| (3,8);
		\draw[cross line,->] (5,0) node[below]{$N$} -- (5,8);
		\draw[cross line,->] (4,0) -- (4,8);
	} \ \, .
\end{equation*}
This is compatible with the relations satisfied by the $Y\mspace{-1mu}$s. The second diagram corresponds to the following way of rewriting these difference operators, cf.~\eqref{eq:intro_Yi_circ}.

For calculations it's convenient (\textsection\ref{s:Macdonald}) to use a manifestly triangular form of $Y_i$. It is obtained from \eqref{eq:Yi} by distributing the simple transpositions in $\pi= s_{N-1} \cdots s_1 \, \hat{q}_1$
\checkedMma
over the Hecke generators. Write $s_{ij}$ for the transposition $z_i \leftrightarrow z_j$ (so $s_{i,i+1} = s_i$) and set
\begin{gather} \label{eq:x_ij}
	x_{ij} \coloneqq \begin{cases}
	\displaystyle \hphantom{-}t^{-1/2}\,\frac{(t-1) \, z_j}{z_i - z_j} \, (1-s_{ij}) + t^{1/2\hphantom{-}} = a_{ij} + b_{ij} \, s_{ij} \, , & i<j \, , \\ 
	\displaystyle -t^{-1/2}\,\frac{(t^{\vphantom{1}} -1) \, z_i}{z_j - z_i} \, (1-s_{ji}) + t^{-1/2} = a_{ij} - b_{ji} \, s_{ji} \, , \qquad & i>j \, .
	\end{cases}
\checkedMma
\end{gather}
These are defined so that $x_{i,i+1} = T_i^\text{pol} \mspace{1mu} s_i$,
\checkedMma
cf.~\eqref{eq:Hecke_pol_2b}, and $x_{ij} \, x_{ji} = 1$.
\checkedMma 
These operators obey the Yang--Baxter equation $x_{ij} \, x_{ik} \, x_{jk} = x_{jk} \, x_{ik} \, x_{ij}$, while $x_{ij}$ and $x_{kl}$ commute if $\{i,j\} \cap \{k,l\} = \varnothing$. In terms of this notation~\cite[\textsection4]{BG+_93}, cf.~\cite{Pas_96},
\begin{equation} \label{eq:Yi_via_x}
	Y_i = x_{i,i+1} \, x_{i,i+2} \cdots x_{iN} \, \hat{q}_i \, x_{i1} \cdots x_{i,i-2} \, x_{i,i-1} \, .
\checkedMma
\end{equation}

As an aside note that multiplication by $z_i^{-1}$ also obeys the relations~\eqref{eq:AHA}, though it does not preserve the space of polynomials. One can avoid the passage to Laurent polynomials by considering operators $Z_i$ that act on $\mathbb{C}[\vect{z}]$ by multiplying by $z_i$, at the price that the relations \eqref{eq:AHA} are inverted to
\begin{equation} \label{eq:AHA'}
	T_i \, Z_i \, T_i = Z_{i+1} \, , \qquad\quad T_i \, Z_j = Z_j \, T_i \quad \text{if } j\neq i,i+1 \, .
\checkedMma
\end{equation}
The $Z_i$ can be combined with \eqref{eq:Yi} into a polynomial representation of the \emph{double} affine Hecke algebra (\textsc{daha})~\cite{Che_92a}, \cite[\textsection1.4.3]{Che_05}. This unital associative algebra extends the \textsc{aha} by $\mathbb{C}[\vect{Z}]$, where the (mutually commuting) affine generators~$Z_i$ obey the cross relations \eqref{eq:AHA'} along with~\cite{Che_92a}
\begin{equation*} % \label{eq:DAHA}
	\begin{gathered}
	Y_i \, Z_1 \cdots Z_N = q \, Z_1 \cdots Z_N \, Y_i \, , \qquad\quad Z_i \, Y_1 \cdots Y_N = q^{-1} \, Y_1 \cdots Y_N \, Z_i \, , \\
	Y_2^{-1} \, Z_1 \, Y_2 \, Z_1^{-1} = T_1^2 \, .
	\end{gathered}
\checkedMma
\end{equation*}
In particular~$q$ is a parameter of the \textsc{daha} itself, just as $t$ already is for the Hecke algebra, whereas for the \textsc{aha} the parameter~$q$ is associated to the representation~\eqref{eq:Yi}. The \textsc{daha} has a graphical representation in terms of ribbon diagrams~\cite{BW+_13}.

\subsubsection{Macdonald theory} \label{s:Macdonald}
Bernstein \cite{Lus_83,Lus_89} noticed that the centre of the \textsc{aha} consists of symmetric polynomials in the $Y_i$:
\begin{equation} \label{eq:centre} 
	Z\bigl(\widehat{\mathfrak{H}}_N\bigr) = \mathbb{C}[\vect{Y}]^{\mathfrak{S}_N} \, .
\end{equation}
This is also known as the \emph{spherical \textsc{aha}}. As generators of \eqref{eq:centre} we choose elementary symmetric polynomials in the $Y_i$, which are packaged together in the generating function 
\begin{subequations} \label{eq:Delta}
\begin{gather}
	\Delta(u) \coloneqq \prod_{i=1}^N (1+ u\,Y_i) = \sum_{r=0}^N u^r \, e_r(\vect{Y}) \, , \qquad e_r(\vect{Y}) = \!\!\! \sum_{i_1 < \cdots < i_r}^N \!\!\!\!\! Y_{i_1} \cdots Y_{i_r} \, .
\checkedMma
\shortintertext{(The notation $\Delta(u)$ should not be confused with the \textsf{q}-$\mspace{-1mu}$Vandermonde~\eqref{eq:q-Vand}.) So this operator commutes with all generators of the \textsc{aha}, and of course}
	\bigl[ \Delta(u) , \Delta(v) \bigr] = 0 \, .
\end{gather}
\end{subequations}
From the viewpoint of integrability the latter says that $\Delta(u)$ is a good candidate for a generating function of commuting charges for an integrable model: see  \textsection\ref{s:Ruijsenaars}.

Consider the subspace of (completely) symmetric polynomials,
\begin{equation} \label{eq:phys_space_scalar}
	\mathbb{C}[\vect{z}]^{\mathfrak{S}_N} = \bigcap_{i=1}^{N-1} \! \ker(s_i - 1) = \bigcap_{i=1}^{N-1} \! \ker\bigl(T_i^\text{pol} - t^{1/2} \bigr) \, ,
\end{equation}
where the second equality uses~\eqref{eq:Hecke_pol_2a}, cf.~\eqref{eq:T_pol_eigenspaces}. The description in terms of Hecke generators makes clear that the generating function~\eqref{eq:Delta} preserves \eqref{eq:phys_space_scalar}. 

\begin{proposition}[\cite{Che_92a}] \label{prop:D_r}
Write
\begin{equation} \label{eq:e_to_D}
	D_r = e_r(\vect{Y}) \qquad \text{on} \quad \mathbb{C}[\vect{z}]^{\mathfrak{S}_N} \ , \qquad 0 \leq r \leq N \, .
\checkedMma
\end{equation}
Then $D_r$ are Macdonald operators~\cite{Mac_95,Mac_98,Mac_03},
\begin{equation} \label{eq:D_r}
	D_r = \!\!\! \sum_{J \colon \# J = r} \!\!\!\!\! A_J(\vect{z}) \, \hat{q}^{\phantom{x}}_J \, , \qquad 
	A_J(\vect{z}) \coloneqq \prod_{j \in J \niton \bar{\jmath}} \!\!\! a_{j\bar{\jmath}} \, , \quad 
	\hat{q}^{\phantom{x}}_J \coloneqq \prod_{j\in J} \! \hat{q}_j \, ,
\checkedMma
\end{equation}
where the sum ranges over all $r$-element subsets $J \subseteq \{ 1,\To,N\}$. 
\end{proposition}
\noindent For example,
\begin{equation} \label{eq:D_1}
	D_1 \coloneqq \sum_{j=1}^N A_j(\vect{z}) \, \hat{q}_j \, , 
	\qquad A_j(\vect{z}) = \prod_{\bar{\jmath}\mspace{1mu}(\neq j)}^N \!\! a_{j\mspace{1mu}\bar{\jmath}} = \, t^{-(N-1)/2} \! \prod_{\bar{\jmath}\mspace{1mu}(\neq j)}^N \!\! \frac{t\,z_j - z_{\bar{\jmath}}}{z_j - z_{\bar{\jmath}}} \, .
\checkedMma
\end{equation}
In particular we get the multiplicative translation operator, cf.~\eqref{eq:q_hat}, which counts the total degree:
\begin{equation} \label{eq:D_N}
	D_N = Y_1 \cdots Y_N = \pi^N = \hat{q}_1 \cdots \hat{q}_N \, .
\checkedMma
\end{equation}
The following proof of these well-known facts will be useful in \textsection\ref{s:abelian_tilde}. After we obtained this proof ourselves we came across it in Appendix~B of \cite{JK+_95b}.

\begin{proof}[Proof of Proposition~\ref{prop:D_r} (\/\cite{JK+_95b})]
We start with $r=1$. Let us consider the contribution due to $Y_i$ written as in~\eqref{eq:Yi_via_x}. On $\mathbb{C}[\vect{z}]^{\mathfrak{S}_N}$ we can replace $x_{ji} = t^{-1/2}$
\checkedMma
to the right of $\hat{q}_i$. Since the individual $Y_i$ do not preserve $\mathbb{C}[\vect{z}]^{\mathfrak{S}_N}$ the $x_{ij}$ to the left have to be commuted through $\hat{q}_i$ before we can replace $x_{ij} = t^{1/2}$.
\checkedMma 
The result is a linear combination of terms with $\hat{q}_j$ for $j \geq i$. It follows that $D_1$ can be written in the (`normal') form $\sum_j A_j(\vect{z}) \, \hat{q}_j$ for some rational function~$A_j(\vect{z})$ that we have to find. 

Note that $A_j(\vect{z})$ receives contributions from the $Y_i$ with $i\leq j$. One of the coefficients is therefore easy to determine: for $j=1$ we only need to consider
\begin{equation} \label{eq:Y1_contribution}
	\begin{aligned}
	Y_1 & = x_{12} \, x_{13} \cdots x_{1N} \, \hat{q}_1 \\
	& = 
	(a_{12} + b_{12} \, s_{12}) \, (a_{13} + b_{13} \, s_{13}) \cdots (a_{1N} + b_{1N} \, s_{1N}) \, \hat{q}_1 \\
	& = a_{12} \cdots a_{1N} \, \hat{q}_1 \ + \ \text{contributions to all other } A_j \,\hat{q}_j \, (j>1) \, .
	\end{aligned}
\checkedMma
\end{equation}
Thus $A_1(\vect{z}) = a_{12} \cdots a_{1N}$. To find the other coefficients we exploiting the fact that $\Delta(u)$ preserves \eqref{eq:phys_space_scalar}. Since permutations act trivially on symmetric polynomials we have\,%
\footnote{\ Here we could equally well conjugate by $s_{1j}$.
\checkedMma
However, $s_{(j\cdots 21)}$ is the shortest permutation such that $1\mapsto j$, which will be the prudent choice in \textsection\ref{s:abelian_tilde}. (A similar remark applies to higher~$r$.)}
\begin{equation*}
	e_1(\vect{Y}) \, = \, \underbrace{s_{j-1} \cdots s_1}_{= \, s_{(j\cdots 21)}} \, e_1(\vect{Y}) \, \underbrace{s_1 \cdots s_{j-1}}_{= \, s_{(12\cdots j)}} \qquad \text{on} \quad \mathbb{C}[\vect{z}]^{\mathfrak{S}_N} \, .
\checkedMma
\end{equation*}
But on the right-hand side the term with $i=1$ is the only to contribute to $A_j(\vect{z})$. By \eqref{eq:Y1_contribution} we have $s_{(j\cdots 21)} \, Y_1 \, s_{(12\cdots j)} = s_{(j\cdots 21)} \, \bigl(A_1(\vect{z}) \,\hat{q}_1 + \dots) \, s_{(12\cdots j)} = A_j(\vect{z}) \, \hat{q}_j + \dots$
\checkedMma
for $A_j(\vect{z})$ as in \eqref{eq:D_1} and with final ellipsis denoting terms with $\hat{q}_i$ for $i\neq j$. This proves \eqref{eq:D_1}.

To see how to adapt the argument to the general case we turn to $r=2$. Like before on symmetric polynomials the result can be written in normal form $\sum_{j<j'} A_{j j'}(\vect{z}) \, \hat{q}_j \, \hat{q}_{j'}$, where $A_{j j'}(\vect{z})$ receives contributions from $Y_i \, Y_{i'}$ for all $i<i'$ with $i\geq j$ and $i'\geq j'$. We compute the simple term like in \eqref{eq:Y1_contribution}:
\begin{equation*}
	\begin{aligned}
	Y_1 \, Y_2 = Y_2 \, Y_1 & = x_{23} \cdots x_{2N} \, \hat{q}_2 \, x_{21} \ x_{12} \, x_{13} \cdots x_{1N} \, \hat{q}_1 \\
	& = x_{23} \cdots x_{2N} \, \hat{q}_2 \, x_{13} \cdots x_{1N} \, \hat{q}_1 \\
	& = x_{23} \cdots x_{2N} \, x_{13} \cdots x_{1N} \, \hat{q}_1 \, \hat{q}_2 \\
	& = 
	(a_{23} + \dots) \cdots (a_{2N} + \dots ) \, (a_{13} + \dots ) \cdots (a_{1N} + \dots) \, \hat{q}_1 \, \hat{q}_2 \\
	& = A_{12}(\vect{z}) \, \hat{q}_1 \, \hat{q}_2 \ + \ \text{contributions to all other } A_{j j'}(\vect{z}) \, \hat{q}_j \, \hat{q}_{j'} \, ,
	\end{aligned}
\checkedMma
\end{equation*}
where in the middle equality we commuted $\hat{q}_2$ through operators independent of $z_2$.
The computation is similar for higher~$r$:
\begin{equation} \label{eq:Y_1...r_contribution}
	\begin{aligned}
	Y_1 \cdots Y_r = Y_r \cdots Y_1 & = x_{r,r+1} \cdots x_{rN} \ \cdots \ x_{1,r+1} \cdots x_{1N} \, \hat{q}_1 \cdots \hat{q}_r \\
	& = A_{1\cdots r}(\vect{z}) \, \hat{q}_1 \cdots \hat{q}_r \ + \ \text{contrib.\ to all other }  A_{j_1 \cdots j_r}(\vect{z}) \, \hat{q}_{j_1} \!\cdots \hat{q}_{j_r} \, ,
	\end{aligned}
\checkedMma
\end{equation}
from which we read off $A_{1\cdots r}(\vect{z}) = \prod_{j (\leq r)} \prod_{\bar{\jmath}\mspace{1mu}(>r)} a_{j \mspace{1mu} \bar{\jmath}}\mspace{1mu}$.
\checkedMma
If $r=N$ there are no $x_{ij}$ left and we already get \eqref{eq:D_N}. The remaining $A_{j_1 \cdots j_r}(\vect{z})$ can again readily be obtained by a suitable conjugation. Indeed, in terms of the notation \eqref{eq:perm_grassmannian} we have
\begin{equation*}
	e_r(\vect{Y}) \, = \, s_{\{j_1,\To,j_r\}}^{\vphantom{-1}} \, e_r(\vect{Y})  \, s_{\{j_1,\To,j_r\}}^{-1} \qquad \text{on} \quad \mathbb{C}[\vect{z}]^{\mathfrak{S}_N} \, .
\end{equation*}
Applying the same conjugation to \eqref{eq:Y_1...r_contribution} we conclude \eqref{eq:D_r}.
\end{proof}

The expression \eqref{eq:D_r} becomes more complicated as $r$ increases to $\lfloor N/2\rfloor$, but starts to simplify again beyond the `equator', cf.~\eqref{eq:D_N}. By \eqref{eq:e_to_D} we have for all $0 \leq r \leq N $
\begin{equation} \label{eq:D_symmetry}
	D_{N-r} = D_N \, D_{-r}\, , \qquad D_{-r} = e_r(Y_1^{-1},\To,Y_N^{-1}) \qquad \text{on} \quad \mathbb{C}[\vect{z}]^{\mathfrak{S}_N} \ . 
\end{equation}
Note that the affine generators are indeed invertible: \eqref{eq:Yi} and \eqref{eq:Yi_via_x} imply
\begin{equation*}
	\begin{aligned}
	Y_i^{-1} & = T_{i-1}^\text{pol} \cdots T_1^\text{pol} \, \pi^{-1} \, T_{N-1}^{\text{pol}\,-1} \cdots T_i^{\text{pol}\,-1} \\
	& = x_{i-1,i} \cdots x_{1i} \, \hat{q}_i^{-1} x_{Ni} \cdots x_{i+1,i} \, .
	\end{aligned}
\checkedMma
\end{equation*}
\begin{proposition} \label{prop:D_-r} 
The operators defined in \eqref{eq:D_symmetry} are given by
\begin{equation} \label{eq:D_-r}
	D_{-r} = \!\!\! \sum_{I \colon \# I = r} \!\!\!\!\! A_{-I}(\vect{z}) \, \hat{q}^{-1}_I \, , \qquad
	A_{-I}(\vect{z}) \coloneqq \prod_{i \in I \niton \bar\imath} \!\!\! a_{\bar\imath\mspace{1mu} i} \, .
\checkedMma
\end{equation}
\end{proposition}
\noindent Note that the arguments in the definition of $A_{-I}(\vect{z})$ are the inverse of those in \eqref{eq:D_1}.

\begin{proof}
We proceed as in the proof of Proposition~\ref{prop:D_r}. Consider $r=1$. On $\mathbb{C}[\vect{z}]^{\mathfrak{S}_N}$ the result will have normal form $\sum_i \widetilde{Y}_{-i}$ where $\widetilde{Y}_{-i} = A_{-i}(\vect{z})\, \hat{q}_i^{-1}$. This time the rational function $A_{-i}(\vect{z})$ receives contributions from all $Y_j^{-1}$ with $j\geq i$. The simple term thus is $\widetilde{Y}_{-N}$, with sole contribution coming from
\begin{equation*}
	\begin{aligned}
	Y_N^{-1} & = x_{N-1,N} \cdots x_{1N} \, \hat{q}_N^{-1} \\
	& = (a_{N-1,N} + b_{N-1,N} \, s_{N-1,N}) \cdots (a_{1N} + b_{1N} \, s_{1N}) \, \hat{q}_N^{-1} \\
	& = a_{1N} \cdots a_{N-1,N} \, \hat{q}_N^{-1} \ + \ \text{contributions to all other $\widetilde{Y}_{-i}$} \, .
	\end{aligned}	
\end{equation*}
The remaining terms are found from this like before, now conjugating by $s_i \cdots s_{N-1}$. Higher $r$ are treated analogously.
\end{proof}

Since the $Y_i$ commute they can be simultaneously diagonalised. Their joint eigenfunctions are labelled by (weak) \emph{compositions}~$\alpha$ with at most $N$ parts. Let $\alpha^+$ be the corresponding partition of length $\ell(\alpha^+)\leq N$, viewed as a weak partition by appending zeros if necessary. Write $w^\alpha$ for the shortest permutation such that $\alpha_i = \alpha^+_{w^\alpha(i)}$. The monomial basis $\vect{z}^\alpha \coloneqq z_1^{\alpha_1} \cdots z_N^{\alpha_N}$ of $\mathbb{C}[\vect{z}]$ has a (partial) ordering induced by the \emph{dominance} (partial) order on compositions. Define
\begin{subequations} \label{eq:dominance}
	\begin{gather}
	\lambda \geq \nu \qquad\quad \text{if{f}} \qquad\quad \sum_{i=1}^n \lambda_i \geq \sum_{i=1}^n \nu_i \quad \text{for all $1\leq n\leq N$} \, ,
\shortintertext{and write $\lambda > \nu$ if $\lambda \geq \nu$ but $\lambda \neq \nu$. This is refined to compositions as~\cite{BG+_93,Opd_95}}
	\alpha \succ \beta \qquad\quad \text{if{f}} \qquad\quad \text{either} \quad \alpha^+ > \beta^+ \quad \text{or} \quad \alpha^+ = \beta^+ , \, \alpha > \beta \, .
	\end{gather}
\end{subequations}
We will say that $\vect{z}^\beta$ is \emph{lower} than $\vect{z}^\alpha$ if $\alpha \succ \beta$.

\begin{definition}
The \emph{nonsymmetric Macdonald polynomial} $E_\alpha(\vect{z}) \coloneqq E_\alpha(\vect{z};q,t)$ \cite{Mac_03} is the unique polynomial such that
\begin{equation} \label{eq:nonsymm_Macdonalds}
	\begin{aligned}
	E_\alpha(\vect{z}) & = \vect{z}^\alpha + \text{lower monomials} \, , \\
	Y_i \, E_\alpha(\vect{z}) & = t^{(N - 2\,w^\alpha(i)\, +1)/2} \, q^{\alpha_i} \, E_\alpha(\vect{z}) \, , \qquad 1\leq i\leq N \, .
	\end{aligned}
\end{equation}
\end{definition}
\begin{proof}[Proof of uniqueness (sketch)]
To show that the $Y$\/s are triangular it suffices by \eqref{eq:Yi_via_x} to verify that the $x_{ij}$ are already triangular. The eigenvalue in~\eqref{eq:nonsymm_Macdonalds} can then be read off as the coefficient in $Y_i \, \vect{z}^\alpha = \text{coeff} \times \vect{z}^\alpha + \text{lower}$. The uniqueness follows since the joint spectrum of the $Y_i$ is simple (multiplicity free).
\end{proof}

By (\textsf{q}-)symmetrising one obtains \emph{(symmetric) Macdonald polynomials}:
\begin{equation} \label{eq:Macd_polyn} 
	P_\lambda(\vect{z}) = \!\! \sum_{\alpha \, :\, \alpha^+ = \lambda} \!\!\!\!\! E_\alpha(\vect{z}) = \text{cst} \times\, \Pi_+^\text{pol} \, E_\alpha(\vect{z}) \, ,
\end{equation}
where $\Pi_+^\text{pol}$ is the projector from~\eqref{eq:Hecke_projectors_pol} and the constant is such that $P_\lambda$ is monic.
These are joint eigenfunctions of the Macdonald operators~\eqref{eq:D_r},
\begin{equation} \label{eq:Macd_eigenvalues} \checkedMma
	D_r \, P_\lambda(\vect{z}) = \Lambda_r(\lambda) \, P_\lambda(\vect{z}) \, , \qquad \Lambda_r(\lambda) \coloneqq \sum_{I \colon \# I = r} \ \prod_{i\in I} t^{(N-2\,i+1)/2} \, q^{\lambda_i} \, .
\end{equation}
Macdonald polynomials are orthogonal, cf.\ e.g.\ \cite[\textsection3.4]{Mac_98}, \cite[\textsection{VI.9}]{Mac_95}. Set
\begin{equation} \label{eq:Macd_measure}
	\mu_{q,t}(\vect{z}) \coloneqq \prod_{i\neq j}^N \frac{(z_i/z_j;q)_\infty}{(t\,z_i/z_j;q)_\infty} \, , \qquad (z;q)_\infty \coloneqq \prod_{k=0}^\infty (1-z\,q^k) \, ,
\end{equation}
where the infinite products truncate if $t = q^k$ for $k \in \mathbb{N}$. Define the scalar product~\cite{Mac_00}
\begin{equation*}
	(F,G) \coloneqq \text{constant term}\bigl(\mu_{q,t} \, F^* \, G\bigr) \, , \qquad F^*(\vect{z}) \coloneqq F(z_1^{-1},\To,z_N^{-1}) \, .
\end{equation*}
Then the Macdonald polynomials~\eqref{eq:Macd_polyn} can also be uniquely characterised as
\begin{equation} \label{eq:Macd_orthog}
	\begin{aligned}
	P_\lambda(\vect{z}) & = m_\lambda(\vect{z}) + \text{lower terms} \, , \\ 
	(P_\lambda,P_\nu) & = 0 \quad \text{if} \quad \lambda \neq \nu \, ,
	\end{aligned}
\end{equation}
where $m_\lambda$ is a monomial symmetric polynomial, see \eqref{eq:monomial}.

\begin{figure}[h]
	\centering
	\begin{tikzpicture}[scale=3]
		\draw[->] (-.1,0) -- node[below] {$W_\lambda(q)$} (1.15,0) node[right] {$q$};
		\draw[->] (0,-.1) -- node[left] {$P_\lambda(t)$} (0,1.15) node[above]{$t$};
		\draw [fill=black] (0,0) circle (.02cm) node [below left] {$0$} node [above left] {$s_\lambda$};
		\draw[smooth,samples=100,domain=0:1] plot(\x,\x*\x);
		\node at (.5,.25) [shift={(.5,0)}] {$P_\lambda^\star$};
		\draw[dashed] (0,0) -- node[shift={(-.5,0)}] {$s_\lambda$} (1.1,1.1);
		\draw[dotted] (1,0) node[below] {$1$} -- node[below right] {$e_{\lambda'}$} (1,1.1);
		\draw[dotted] (0,1) node[left] {$1$} -- node[above left] {$m_\lambda$} (1.1,1);
		\draw (1-.04,1-2*.04) -- (1+.04,1+2*.04) node [above] {$P^{(1/2)}_\lambda$};
		\draw (1-2*.04,1-.04) -- (1+2*.04,1+.04) node [right] {$Z_\lambda$};;
	\end{tikzpicture}
	\caption{(Adapted from \cite{Mac_98}.) The parameter space of the Macdonald polynomial $P_\lambda(q,t)$, where the dependence on $\vect{z}$ is suppressed. Special cases include elementary symmetric (for the conjugate partition) when $q=1$, monomial symmetric for $t=1$, Schur on the diagonal $t=q$ as well as the origin (and infinity) approached from any direction, Hall--Littlewood at $q=0$, and \textit{q}-$\!$Whittaker at $t=0$. The quantum spherical zonal polynomial, which we denote by $P_\lambda^\star$ and features in our wave function~\eqref{eq:intro_polynomial}, lives on the parabola $t = q^2$. Jack polynomials~$P^{(\alpha)}_\lambda$, in the (monic) `P-normalisation', are associated to tangent lines at $t=q=1$ with slope $\alpha^{-1} = k$. This includes $P^{(1)}_\lambda = s_\lambda$, the spherical zonal polynomial $P^{(1/2)}_\lambda$ from \eqref{eq:HS_polynomial}, and the zonal polynomial $P^{(2)}_\lambda = Z_\lambda$.}
	\label{fg:Macdonalds}
\end{figure}
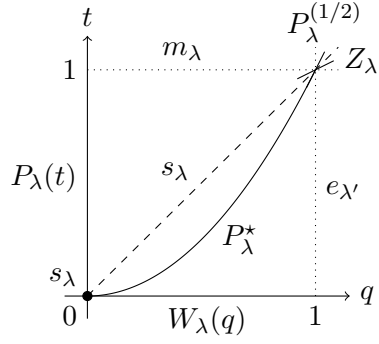

\subsubsection{Ruijsenaars model} \label{s:Ruijsenaars}
From a physical point of view \eqref{eq:phys_space_scalar} describes $N$~indistinguishable (\textsf{q}-)bosons moving on a circle with coordinates $z_j$. Macdonald operators can be understood as `effective' (gauge-transformed) Hamiltonians. Namely, Koornwinder observed~\cite{vD_95} that conjugation by the square root of the measure~\eqref{eq:Macd_measure} turns \eqref{eq:D_r} and \eqref{eq:D_-r} into Ruijsenaars's hermitian operators~\cite{Rui_87}. Indeed,
\begin{equation*}
	\frac{(t\,z_i/z_j;q)_\infty}{(z_i/z_j;q)_\infty} \ \hat{q}_i \ \frac{(z_i/z_j;q)_\infty}{(t\,z_i/z_j;q)_\infty} = \frac{t\,z_i-z_j}{z_i-z_j} \ \hat{q}_i \ \frac{z_i-z_j}{t\,z_i-z_j} \, ,
\end{equation*}
whence 
\begin{equation} \label{eq:Macd_to_Ruij}
	\mu_{q,t}(\vect{z})^{-1/2} \, D_{\pm r} \ \mu_{q,t}(\vect{z})^{1/2} = D^\text{Rui}_{\pm r} \, , \qquad D^\text{Rui}_{\pm r} \coloneqq \!\! \sum_{J \colon \# J = r} \!\!\!\!\!
	A_{\pm J}^{1/2} \ \hat{q}^{\pm 1}_J \, A_{\mp J}^{1/2} \, .
\end{equation}
The physical Hamiltonian and momentum operator~\cite{Rui_87} are $H^\text{rel} \coloneqq (D^\text{Rui}_1 + D^\text{Rui}_{-1})/2$ and $P^\text{rel} \coloneqq (D^\text{Rui}_1 - D^\text{Rui}_{-1})/2$, cf.\ \textsection\ref{s:app_CalSut_limit}. The model is relativistic in the sense that it enjoys two-dimensional Poincaré invariance with `boost' $B^\text{rel} \coloneqq \log(z_1 \cdots z_N)/\log q$,
\begin{equation*}
	[P^\text{rel},H^\text{rel}] = 0 \, , \qquad [H^\text{rel},B^\text{rel}] = P^\text{rel} \, , \qquad [P^\text{rel},B^\text{rel}] = H^\text{rel} \, .
\end{equation*}
 
The square root of \eqref{eq:Macd_measure} can be interpreted as the ground-state wave function, with energy $\varepsilon_0 = [N]$, as follows from \eqref{eq:Macd_to_Ruij} and the identity
\begin{equation} \label{eq:D_r_q=1}
	\sum_{J : \# J = r} \!\!\!\! A_J = D_r\,\big|_{q=1} = \begin{bmatrix} N \\ r \end{bmatrix} \, ,
\end{equation}
which is a consequence of \eqref{eq:Macd_eigenvalues} as the eigenvalues become independent of $\lambda$ at $q=1$, whence $D_r$ diagonal.
The other eigenfunctions now follow from \textsection\ref{s:Macdonald}.

In the \textsf{q}-fermionic case \eqref{eq:phys_space_scalar} is replaced by
\begin{equation} \label{eq:phys_space_scalar_fermionic}
	\Delta_t(\vect{z})\,\mathbb{C}[\vect{z}]^{\mathfrak{S}_N} = \bigcap_{i=1}^{N-1} \! \ker\bigl(T_i^\text{pol} + t^{-1/2} \bigr) \, .
\end{equation}
Though we will not explicitly use it, the spin-generalisation of this space plays an important role in the background in \textsection\ref{s:explicit_evrs}. We plan to return to this in the future. The corresponding hierarchy is obtained from \eqref{eq:e_to_D} using
\begin{lemma}[cf.~(5.8.12) in~\cite{Mac_03}] \label{lem:conj_by_qVandermonde}
Conjugation with the \textsf{q}-$\mspace{-2mu}$Vandermonde factor gives
\begin{equation} \label{eq:conj_by_qVandermonde}
	e_r(\vect{Y}) \, \Delta_t(\vect{z}) \, = \, q^{r\mspace{1mu}(N-1)/2} \, \Delta_t(\vect{z}) \, e_r(\vect{Y}') \qquad \text{on} \quad \mathbb{C}[\vect{z}]^{\mathfrak{S}_N} \, ,
\checkedMma
\end{equation}	
where $Y'_i$ denotes the affine generator~\eqref{eq:Yi} with parameters $q'=q$ and $t'=q\,t$, shifting $k'=k+1$ for $t = q^k$ and likewise with primes.
\end{lemma}
\begin{proof}
Since $e_r(\vect{Y})$ preserves \eqref{eq:phys_space_scalar_fermionic}, $\Delta_t(\vect{z})^{-1} \, e_r(\vect{Y}) \, \Delta_t(\vect{z})$ does so for $\mathbb{C}[\vect{z}]^{\mathfrak{S}_N}$. We can thus proceed like in \textsection\ref{s:Macdonald}. Write $\Delta_t(\vect{z})^{-1} \, e_r(\vect{Y}) \, \Delta_t(\vect{z}) = \sum_I A'_I(\vect{z}) \, \hat{q}_I$ in normal form. By symmetry it suffices to find $A'_{1\cdots r}(\vect{z})$, with only contribution due to
\begin{equation*}
	\begin{aligned}
	\Delta_t(\vect{z})^{-1} \, e_r(\vect{Y}) \, \Delta_t(\vect{z}) & = \Delta_t(\vect{z})^{-1} \, \bigl( A_{1\cdots r}(\vect{z}) \, \hat{q}_{1\cdots r} + \cdots \bigr) \, \Delta_t(\vect{z}) \\
	& = A_{1\cdots r}(\vect{z}) \, \Biggl( q^{r\,(r-1)/2} \, \prod_{j=1}^r \prod_{\bar\jmath(>r)}^N \!\! \frac{q\,t\,z_j - z_{\bar\jmath}}{t\,z_j - z_{\bar\jmath}} \Biggr) \, \hat{q}_{1\cdots r} + \cdots \, .
	\end{aligned}
\end{equation*}
Hence $A'_{1\cdots r}(\vect{z})$ equals $q^{r\mspace{1mu}(N-1)/2}$ times $A_{1\cdots r}(\vect{z})$ with $t$ replaced by $q\,t$.
\end{proof}

In the nonrelativistic limit $q=t^\alpha$ (so $\alpha = k^{-1}$), $t\to 1$, the \textit{Y}$\mspace{-2mu}$-operators reduce to Dunkl(--Cherednik) operators, Macdonald polynomials to Jack polynomials, and the Ruijsenaars model to the trigonometric Calogero--Sutherland model. This is summarised in \textsection\ref{s:app_nonrlt_limit}.

\subsection{Spin side} \label{s:spin} The spin-chain Hilbert space is $\mathcal{H} \coloneqq V^{\otimes N}$, where $V = \mathbb{C} \, \ket{\uparrow} \oplus \mathbb{C} \, \ket{\downarrow}$ is the spin-1/2 (defining, or vector) one-particle Hilbert space. A good introduction to quantum $\mathfrak{sl}_2$ and the quantum-loop algebra of $\mathfrak{sl}_2$ can be found in \cite{Jim_92}.

\subsubsection{Hecke, Temperley--Lieb and quantum $\mathfrak{sl}_2$} \label{s:Hecke_TL_Uqsl}
The second well-known representation of the Hecke algebra $\mathfrak{H}_N$ (\textsection\ref{s:Hecke_AHA}) \textsf{q}-deforms the natural action of $\mathfrak{S}_N$ on $\mathcal{H}$: 
\begin{equation} \label{eq:Hecke_spin}
	T_i^\text{sp} \coloneqq \id^{\otimes(i-1)} \otimes \,T^\text{sp} \otimes \id^{\otimes(N-i-1)} \, , \qquad T^\text{sp} \coloneqq \begin{pmatrix} 
	\, t^{1/2} & \!\! \color{gray!80}{0} & \color{gray!80}{0} \! & \color{gray!80}{0} \, \\
	\, \color{gray!80}{0} & \!\! t^{1/2} - t^{-1/2} & 1 \! & \color{gray!80}{0} \, \\
	\, \color{gray!80}{0} & \!\! 1 & 0 \! & \color{gray!80}{0} \, \\
	\, \color{gray!80}{0} & \!\! \color{gray!80}{0} & \color{gray!80}{0} \! & t^{1/2} \, \\
	\end{pmatrix} \, .
\end{equation}
Here the matrix is given with respect to the standard basis $\ket{\uparrow\uparrow},\ket{\uparrow\downarrow},\ket{\downarrow\uparrow},\ket{\downarrow\downarrow}$ of $V^{\otimes 2}$. 

The spin analogue of the eigenspace decomposition~\eqref{eq:T_pol_eigenspaces} is (see also \textsection\ref{s:app_stochastic_twist})
\begin{subequations} \label{eq:T_sp_eigenspaces}
	\begin{gather}
	V \otimes V \cong \, \mathrm{Sym}_{t}^2(V) \oplus \Lambda_t^{\!2}(V) \, ,
\shortintertext{where we write}
	\begin{aligned}
	\mathrm{Sym}_{t}^2(V) & \coloneqq \mathbb{C} \, \ket{\uparrow\uparrow} \, \oplus \, \mathbb{C} \, \bigl( t^{1/4} \, \ket{\uparrow\downarrow} + t^{-1/4} \, \, \ket{\downarrow\uparrow}\bigr) \, \oplus \, \mathbb{C} \, \ket{\downarrow\downarrow} \, , \\
	\Lambda_t^{\!2}(V) & \coloneqq \mathbb{C} \, \bigl(t^{-1/4} \, \ket{\uparrow\downarrow}- t^{1/4} \,\ket{\downarrow\uparrow}\bigr)
	\end{aligned}
	\end{gather}
\end{subequations}	
for the (\textsf{q}-symmetric) $t^{1/2}$- and (\textsf{q}-antisymmetric) $-t^{-1/2}$-eigenspaces, respectively.

A close cousin is the \emph{Temperley--Lieb algebra}~$\mathfrak{T}_N(\beta)$ of type $A_{N-1}$, with `loop fugacity'~$\beta\in\mathbb{C}^\times$. This is the unital associative algebra with generators~$e_1,\To,e_{N-1}$ (not to be confused with elementary symmetric polynomials) and defining relations
\begin{equation} \label{eq:TL}
	\begin{gathered}
	e_i \, e_{i\pm 1} \, e_i = e_i \ , \qquad\qquad e_i \, e_j = e_j \, e_i \quad \text{if } |i-j|>1 \, , \\
	e_i^2 = \beta\,e_i \ .
	\end{gathered}
\end{equation}
It can be obtained as a quotient of $\mathfrak{H}_N$ if $[2] = t^{1/2} + t^{-1/2} \neq 0$. Indeed, consider the shifted generator $e'_i \coloneqq t^{1/2} - T_i$. By the Hecke condition $e'_i$ is, up to normalisation, the projector onto the $-t^{-1/2}$-eigenspace of $T_i$. In fact, \eqref{eq:Hecke} implies that the $e'_i$ obey \eqref{eq:TL} with $\beta = [2]$, except that the first relation in \eqref{eq:TL} is replaced by the weaker condition $e'_i \, e'_{i+1} \, e'_i - e'_i = e'_{i+1} \, e'_i \, e'_{i+1} - e'_{i+1}$. Requiring both sides to vanish one arrives at \eqref{eq:TL}.
The relevance for us is that all of these relations are satisfied by
\begin{equation} \label{eq:TL_spin}
	\begin{aligned}
	e_i^\text{sp} \coloneqq {} & t^{1/2} - T_i^\text{sp} \\
	= {} & \id^{\otimes(i-1)} \otimes \, e^\text{sp} \otimes \id^{\otimes(N-i-1)} \, ,
	\end{aligned}
	\qquad e^\text{sp} = \begin{pmatrix} 
	\, 0 & \color{gray!80}{0} & \! \color{gray!80}{0} & \color{gray!80}{0} \, \\
	\, \color{gray!80}{0} & \!\!\hphantom{-}t^{-1/2}\! & \!-1\! & \color{gray!80}{0} \, \\
	\, \color{gray!80}{0} & \!\!-1\! & \!\hphantom{-}t^{1/2}\! & \color{gray!80}{0} \, \\
	\, \color{gray!80}{0} & \color{gray!80}{0} & \! \color{gray!80}{0} & 0 \, \\
	\end{pmatrix} \, .
\end{equation}
That is, the spin representation~\eqref{eq:Hecke_spin} of $\mathfrak{H}_N$ on $\mathcal{H}$ factors through $\mathfrak{T}_N\bigl([2]\bigr)$. Up to normalisation $e^\text{sp}$ is the projector onto $\Lambda_t^{\!2}(V)$ in \eqref{eq:T_sp_eigenspaces}.	
	
Next, 
\begin{definition} The quantum group $\mathfrak{U} \coloneqq U_{t^{1/2}}(\mathfrak{sl}_2)$ is the unital associative algebra generated by $E,F,K,K^{-1}$ with relations $K \, K^{-1} = K^{-1} \, K = 1$ and
\begin{equation} \label{eq:Uqsl2}
	\begin{aligned}
	K \, E\, K^{-1} & = t \, E \, , \\ 
	K \, F \, K^{-1} & = t^{-1} \, F \, , 
	\end{aligned} \qquad \qquad [E,F] = \frac{K-K^{-1}}{t^{1/2}-t^{-1/2}} \, .
\checkedMma
\end{equation}
It comes equipped with a coproduct $\Delta \colon \mathfrak{U} \longrightarrow \mathfrak{U} \otimes \mathfrak{U}$,
\begin{equation} \label{eq:Uqsl2_coprod}
	\begin{aligned} \checkedMma
	\Delta E & \coloneqq E \otimes 1 + K \otimes E \, , \\
	\Delta F & \coloneqq F \otimes K^{-1} + 1\otimes\, F \, , 
	\end{aligned} 
	\qquad\quad
	\Delta K^{\pm1} \coloneqq K^{\pm1} \otimes K^{\pm1} \, ,
\end{equation}
as well as a counit and antipode. It is the \textsf{q}-deformation of \eqref{eq:sl_2}.
\end{definition}

For generic $t$ the representation theory parallels that of $\mathfrak{sl}_2$ (except that, due to the Hopf-algebra automorphism $K\mapsto -K, E \mapsto -E$, there are two non-isomorphic irreps for each dimension). As in \textsection\ref{s:intro_HS} $\sigma^x,\sigma^y,\sigma^z$ are the Pauli matrices on $V$. The vector ($\Box$) representation of $\mathfrak{U}$ is given by $\sigma^\pm = \sigma^x \pm \I \, \sigma^y$ for $E,F$ and $k \coloneqq t^{\sigma^z\!/2} = \diag(t^{1/2},t^{-1/2})$ for $K$. Repeated application of the coproduct yields a (reducible) representation on $\mathcal{H}$,
\begin{equation} \label{eq:Uqsl2_spin}
	\begin{aligned}
	E_1^\text{sp} & = \sum_{i=1}^N k_1 \cdots k_{i-1} \, \sigma^+_i \, , \\
	F_1^\text{sp} & = \sum_{i=1}^N \sigma^-_i \, k^{-1}_{i+1} \cdots k^{-1}_N \, ,
	\end{aligned} 	
	\qquad\quad
	K_1^\text{sp} = k_1 \cdots k_N = t^{S^z} \, .
\end{equation}
The reason for the subscript~`1' on the left-hand sides, not to be confused with the `tensor-leg' subscripts on the right-hand sides, will become clear soon (\textsection\ref{s:quantum-affine}). 
Of course \eqref{eq:weight_decomp} is also a weight decomposition for $\mathfrak{U}$, with $\mathcal{H}_M = \ker(K_1^\text{sp} - t^{(N-2M)/2})$.

Jimbo noted that the \textit{R}-matrix for $\mathfrak{U}$ is essentially a Hecke generator: $\check{R} = P \, R = t^{-1/2} \, T^\text{sp}$. Since $\check{R}$ commutes with the case $N=2$ of \eqref{eq:Uqsl2_spin} it follows that the $T_i^\text{sp}$ commute with \eqref{eq:Uqsl2_spin} for general~$N$, whence so do the Temperley--Lieb generators~$e_i^\text{sp}$: this is the \textsf{q}-analogue of Schur--Weyl duality~\cite{Jim_86a}, see also \cite{CP_94}. A concrete example is \eqref{eq:T_sp_eigenspaces}, where the $T^\text{sp}$-eigenspaces are $\mathfrak{U}$-irreps under the case $N=2$ of \eqref{eq:Uqsl2_spin}.

\subsubsection{Quantum-loop algebra of $\mathfrak{sl}_2$} \label{s:quantum-affine} We will use two descriptions. The first is

\begin{definition}[Drinfeld--Jimbo presentation]
The quantum-loop algebra $\widehat{\mathfrak{U}}$
\footnote{\ \label{fn:U'gl} The quantum-loop algebra is denoted by $U_{t^{\smash{1/2}}}(L\,\mathfrak{gl}_2) = U'_{t^{\smash{1/2}}}(\widehat{\vphantom{t}\smash{\mathfrak{gl}_2}})_{c=0}$ in \cite{CP_94} and \cite{JM_95} respectively. We don't require the quantum determinant to be unity. The prime indicates the absence of the degree operator. As we only deal with finite-dimensional modules we focus on `level' $c = 0$, which is why we get a quantum loop algebra rather than quantum-affine algebra.} 
is the unital associative algebra generated by two copies of the Chevalley generators of $\mathfrak{U}$, which we denote by $E_b,F_b,K_b$ ($b=0,1$), each obeying \eqref{eq:Uqsl2} with cross relations $K_0 \, K_1 = 1$ (i.e.\ the `level' $c = 0$ condition), $[E_b,F_{b'}] = 0$ if $b \neq b'$, 
\checkedMma 
while
\begin{equation*}
	\begin{gathered}
	E_b^3 \, E_{b'}^{\vphantom{1}} - [3] \, E_b^2 \, E_{b'}^{\vphantom{1}} \, E_b^{\vphantom{1}} + [3] \, E_b^{\vphantom{1}} \, E_{b'}^{\vphantom{1}} \, E_b^2 - E_{b'}^{\vphantom{1}} \, E_b^3 = 0 \\ 
	F_b^3 \, F_{b'}^{\vphantom{1}} - [3] \, F_b^2 \, F_{b'}^{\vphantom{1}} \, F_b^{\vphantom{1}} + [3] \, F_b^{\vphantom{1}} \, F_{b'}^{\vphantom{1}} \, F_b^2 - F_{b'}^{\vphantom{1}} \, F_b^3 = 0
	\end{gathered}
	\qquad (b \neq b') \, .
\checkedMma
\end{equation*}
In these final (\textsf{q}-Serre) relations $[3] = t + 1 + t^{-1}$.
\end{definition} 

By $\mathfrak{U} \subset \widehat{\mathfrak{U}}$ we will mean the copy of $\mathfrak{U}$ generated by $E_1,F_1,K_1$. The representation~\eqref{eq:Uqsl2_spin} of $\mathfrak{U}$ on $\mathcal{H}$ can be `affinised'~\cite{Jim_86a,Cha_95} to get a module of $\widehat{\mathfrak{U}}$. (We use the homogeneous, rather than principal, gradation~\cite{JM_95}.) Namely, given $N$~parameters~$z_i$, supplement \eqref{eq:Uqsl2_spin} by\,%
\footnote{\ This is compatible with the coproduct \eqref{eq:Uqsl2_coprod}, cf.~\cite{JM_95}, and matches \cite[\textsection2.3]{Jim_92}. Notice that \cite[\textsection2.1]{CP_96} uses the opposite coproduct; thus, the expression from \cite[\textsection4.2]{CP_96} is the opposite (left-right reverse) of \eqref{eq:affinisation_z}. Besides a difference in normalisation the monodromy matrix of \cite{BG+_93} is built from $\bar{R}(u) \coloneqq \check{R}(u)\,P$ as $\bar{L}_a(u) = \bar{R}_{a1}(u) \cdots \bar{R}_{aN}(u)$, cf.~\eqref{eq:L_inh}; this yields \eqref{eq:affinisation_z} with $z_i$ inverted if one proceeds as in \textsection\ref{s:app_presentations}. The monodromy matrix of \cite[\textsection2.3--2.4]{JK+_95b} differs in several aspects from \eqref{eq:L_inh}; the resulting Chevalley generators \cite[\textsection4.1]{JK+_95a}, \cite[\textsection2.4]{JK+_95b} are again the opposite of \eqref{eq:Uqsl2_spin}--\eqref{eq:affinisation_z}, matching \cite{CP_96}.\label{fn:coprod}}
\begin{equation} \label{eq:affinisation_z}
	\begin{aligned}
	E^\text{inh}_0 & = \sum_{i=1}^N z_i \ k^{-1}_1 \cdots k^{-1}_{i-1} \,  \sigma^-_i \, , \\
	F^\text{inh}_0 & = \sum_{i=1}^N z_i^{-1} \sigma^+_i \,  k_{i+1} \cdots k_N \, , 
	\end{aligned}
	\qquad\quad
	K^\text{inh}_0 = k_1^{-1} \cdots k_N^{-1} = t^{-S^z} \, .
\checkedMma
\end{equation}
Then the relations of $\widehat{\mathfrak{U}}$ hold. The superscript is for `inhomogeneous', see~\textsection\ref{s:spin_chains}.

\begin{definition} \label{p:hw_def} 
In the affine case one can make different choices of Borel subalgebra. Note that even for $N=1$ the usual highest-weight condition $E_0^\text{inh} \, \ket{\Psi} = E_1^\text{sp} \, \ket{\Psi} = 0$ implies $\ket{\Psi} = 0$, cf.~e.g.~\cite[\textsection2.3]{Jim_92}. We will take a \emph{pseudo highest-weight vector} for $\widehat{\mathfrak{U}}$ to mean a vector that is an eigenvector of both $K_b$ and annihilated by $E_0$ and $F_1$ (rather than both $E_b$). This property is called `pseudo highest weight' in \cite{CP_94} and `\textit{l}-highest weight' in \cite{Nak_01}. It is the usual notion for quantum-integrable spin chains: both $E_0^\text{inh}$ and $F_1^\text{sp}$ flip one spin up, mapping $\mathcal{H}_M$ to $\mathcal{H}_{M-1}$, while $F_0^\text{inh}$ and $E_1^\text{sp}$ act from $\mathcal{H}_M$ to $\mathcal{H}_{M+1}$.
\end{definition}

For $N=1$ \eqref{eq:affinisation_z} is the \emph{evaluation representation} of $\widehat{\mathfrak{U}}$ on $V(z_1) = V \otimes \mathbb{C}[z_1^{\pm 1}]$. For $N=2$ we get a tensor product of two such modules; for generic values of the parameters it is irreducible, and $V(z_1) \otimes V(z_2)$ and $V(z_2) \otimes V(z_1)$ are isomorphic as $\widehat{\mathfrak{U}}$-irreps. Thus there exists an intertwiner $\check{R}(z_1,z_2) \colon V(z_1) \otimes V(z_2) \longrightarrow V(z_2) \otimes V(z_1)$ that is generically invertible. Following Jimbo~\cite{Jim_85,Jim_86a} one can write $\check{R}$ as a linear combination of $\mathfrak{U}$-invariants and determine the coefficients (up to a common normalisation) from the intertwining property. The result only depends on the ratio $u=z_1/z_2$: this is the `difference property' in multiplicative notation. This gives Jimbo's quantum-affine $\mathfrak{sl}_2$ \textit{R}-matrix (in the `homogeneous grading', cf.~\textsection5.4 in \cite{JM_95}) from~\eqref{eq:R_mat}:
\begin{equation} \label{eq:baxterisation} \checkedMma
	\check{R}(u) = t^{1/2} \, \frac{u \, T^\text{sp} - T^{\text{sp}\,-1}}{t \, u - 1} = f(u)\, T^\text{sp} + g(u) = 1 - f(u) \, e^\text{sp} \, ,
\end{equation}
where as in \textsection\ref{s:intro} we switch from the rational functions~$a,b$ defined in \eqref{eq:a,b} to
\begin{equation} \label{eq:f,g}
	f(u) = t^{1/2} \, \frac{u-1}{t\,u-1} = \frac{1}{a(u)} \, , \quad\qquad g(u) = \frac{t-1}{t\,u-1} = -\frac{b(u)}{a(u)} \, ,
\end{equation}
obeying $t^{1/2} \mspace{1mu} f(u) + g(u) = t^{-1/2} \mspace{1mu} f(u) + u \, g(u) = 1$. The usual symmetry property of the \textit{R}-matrix (from the `principal gradation') is broken, $P\,\check{R}(u) \, P \neq \check{R}(u)$ for $t\neq 1$. Note that the final expression in \eqref{eq:baxterisation}, together with \eqref{eq:TL_spin}, implies the relation \eqref{eq:R'(1)}. Since $\det \check{R}(u) = (t-u)/(t\,u-1)$ the \textit{R}-matrix is invertible unless $u = t$ ($u=t^{-1}$), where it becomes proportional to the \textsf{q}-(anti)symmetriser.

The Hecke-algebra relations~\eqref{eq:Hecke} guarantee~\cite{Jon_90} that on $V \otimes V \otimes V$ the \textit{R}-matrix obeys the Yang--Baxter equation in the braid-like form:
\begin{subequations} \label{eq:R_relations}
\begin{gather} 
	\check{R}_{12}(u/v) \, \check{R}_{23}(u) \, \check{R}_{12}(v) = \check{R}_{23}(v) \, \check{R}_{12}(u) \, \check{R}_{23}(u/v) \, , \label{eq:YBE}
\shortintertext{where $\check{R}_{12}(u) = \check{R}(u)\otimes \id$ and $\check{R}_{23}(u) = \id \otimes \, \check{R}(u)$. The normalisation in~\eqref{eq:baxterisation} is chosen such that the (braiding) unitarity and `initial' conditions read}
	\check{R}(u)\,\check{R}(1/u) = \id\otimes\id \, , \qquad \check{R}(1) = \id\otimes\id \, . \label{eq:unitarity_initial}
\end{gather}
\end{subequations}
Together, these properties imply that we can depict the \textit{R}-matrix as in \eqref{eq:R_graphical}, where by unitarity we do not have to distinguish between under- and overcrossings. Setting $w=u\,v$ we may then translate \eqref{eq:R_relations} to
\begin{equation} \label{eq:R_conditions_graphical}
	\tikz[baseline={([yshift=-.5*12pt*0.6]current bounding box.center)},	xscale=0.5,yscale=0.4,font=\footnotesize]{
		\draw[rounded corners=3.5pt,->] (3,0) node[below]{$v$} -- (3,1.5) -- (2,2.5) node[inner sep=1.5pt,fill=white]{$v$} -- (1,3.5) -- (1,4) node[above]{$v$};
		\draw[rounded corners=3.5pt,->] (2,0) node[below]{$u$} -- (2,0.5) -- (1,1.5) -- (1,2) node[inner sep=1.5pt,fill=white]{$u$} -- (1,2.5) -- (2,3.5) -- (2,4) node[above]{$u$};
		\draw[rounded corners=3.5pt,->] (1,0) node[below]{$w$} -- (1,0.5) -- (2,1.5) node[inner sep=1.5pt,fill=white]{$w$} -- (3,2.5) -- (3,4) node[above]{$w$};
	}
	=
	\tikz[baseline={([yshift=-.5*12pt*0.6]current bounding box.center)},	xscale=0.5,yscale=0.4,font=\footnotesize]{
		\draw[rounded corners=3.5pt,->] (3,0) node[below]{$v$} -- (3,0.5) -- (2,1.5) node[inner sep=1.5pt,fill=white]{$v$} -- (1,2.5) -- (1,4) node[above]{$v$};
		\draw[rounded corners=3.5pt,->] (2,0) node[below]{$u$} -- (2,0.5) -- (3,1.5) -- (3,2) node[inner sep=1.5pt,fill=white]{$u$} -- (3,2.5) -- (2,3.5) -- (2,4) node[above]{$u$};
		\draw[rounded corners=3.5pt,->] (1,0) node[below]{$w$} -- (1,1.5) -- (2,2.5) node[inner sep=1.5pt,fill=white]{$w$} -- (3,3.5) -- (3,4) node[above]{$w$};
	}
	\, , \qquad
	\tikz[baseline={([yshift=-.5*12pt*0.6]current bounding box.center)},	xscale=0.5,yscale=0.4,font=\footnotesize]{
		\draw[rounded corners=3.5pt,->] 
		(2,0) node[below]{$v$} -- (2,0.5) -- (1,1.5) -- (1,2) node[inner sep=1.5pt,fill=white]{$v$} -- (1,2.5) -- (2,3.5) -- (2,4) node[above]{$v$};
		\draw[rounded corners=3.5pt,->] (1,0) node[below]{$w$} -- (1,0.5) -- (2,1.5) -- (2,2) node[inner sep=1.5pt,fill=white]{$w$} -- (2,2.5) -- (1,3.5) -- (1,4) node[above]{$w$};
	}
	=
	\tikz[baseline={([yshift=-.5*12pt*0.6]current bounding box.center)},	xscale=0.5,yscale=0.4,font=\footnotesize]{
		\draw[->] (1,0) node[below]{$v$} -- (1,4) node[above]{$v$};
		\draw[->] (2,0) node[below]{$w$} -- (2,4) node[above]{$w$};
	}
	, \qquad 
	\tikz[baseline={([yshift=-.5*12pt*0.6]current bounding box.center)},	xscale=0.5,yscale=0.4,font=\footnotesize]{
		\draw[rounded corners=3.5pt,->] 
		(2,0) node[below]{$u$} -- (2,0.5) -- (1,1.5) -- (1,2) node[above]{$u$};
		\draw[rounded corners=3.5pt,->] (1,0) node[below]{$u$} -- (1,0.5) -- (2,1.5) -- (2,2) node[above]{$u$};
	}
	=
	\tikz[baseline={([yshift=-.5*12pt*0.6]current bounding box.center)},	xscale=0.5,yscale=0.4,font=\footnotesize]{
		\draw[->] (1,0) node[below]{$u$} -- (1,2) node[above]{$u$};
		\draw[->] (2,0) node[below]{$u$} -- (2,2) node[above]{$u$};
	}
\end{equation}

The \textit{R}-matrix is key for the second characterisation. 
\begin{definition}[Faddeev--Reshetikhin--Takhtajan presentation]
Consider an auxiliary space $V_a \cong V$ with spectral (affine) parameter~$u$ and a monodromy matrix  
\begin{equation} \label{eq:ABCD}
	L_a(u) = \begin{pmatrix}
	A(u) & B(u) \\ C(u) & D(u)
	\end{pmatrix}_{\!a}
\end{equation}
which should be understood as a matrix on $V_a$ with noncommutative entries.
Introduce another copy $V_b \cong V$ of the auxiliary space. Then $\widehat{\mathfrak{U}}$ is the unital associative algebra generated by the four operators in \eqref{eq:ABCD}, or more precisely by the operator-valued coefficients (`modes') in expansions as a formal power series in $u^{\pm1}$, with defining relations expressed on $V_a\otimes V_b$ as
\begin{equation} \label{eq:RLL}
	R_{ab}(u/v) \, L_a(u)\,L_b(v) = L_b(v)\,L_a(u)\,R_{ab}(u/v) \, , \qquad R(u) \coloneqq P \, \check{R}(u) \, .
\end{equation}
The coproduct is $L_a(u) \mapsto L_a(u) \otimes L_a(u) = L_{a1}(u) \, L_{a2}(u)$, and the antipode the inverse of \eqref{eq:ABCD} as a $2\times 2$ matrix with noncommutative entries.
\end{definition}
The centre of $\widehat{\mathfrak{U}}$ is generated by the quantum determinant of the monodromy matrix, which is obtained by fusion in the auxiliary space. Indeed, remove the denominators of $R$ in \eqref{eq:RLL} and take $u=t^{-1}\,v$ to get $(t^{1/2} - t^{-1/2}) \, e_{ab}^\text{sp}$, i.e.\ essentially the \textsf{q}-antisymmetriser on $V_a\otimes V_b$, times
\begin{equation} \label{eq:qdet_ABCD}
	\begin{aligned}
	\qdet_a L_a(u) 
	& = A(t \, u) \, D(u) - t^{1/2} \, B(t\,u) \, C(u) = D(t \, u) \, A(u) - t^{-1/2} \, C(t\,u) \, B(u) \! \\
	& = A(u) \, D(t \, u) - t^{1/2} \, C(u) \, B(t \, u) = D(u) \, A(t \, u) - t^{-1/2} \, B(u) \, C(t\,u) \, .\!
	\end{aligned}
\end{equation}

Representations of $\widehat{\mathfrak{U}}$ can be directly constructed from the \textit{R}-matrix, which itself obeys \eqref{eq:RLL} for $N=1$. Repeated application of the (opposite) coproduct yields a (`global') representation on $\mathcal{H}$:\,%
\footnote{\ Note that the order in $R(u) = P \,\check{R}(u)$, see \eqref{eq:RLL}, and the order of the \textit{R}-matrices in \eqref{eq:L_inh} are reversed compared to \cite{BG+_93,JK+_95b,Ugl_95u}; cf.~Footnote~\ref{fn:coprod} on p.\,\pageref{fn:coprod}. The normalisation of \eqref{eq:L_inh} differs from \cite{BG+_93,Ugl_95u} to avoid a pole at $u=1$, cf.~\eqref{eq:xxz}.}
\begin{equation} \label{eq:L_inh}
	\begin{aligned}
	L_a^\text{inh}(u;\vect{z}) \coloneqq \ & R_{aN}(u/z_N) \cdots R_{a2}(u/z_2) \, R_{a1}(u/z_1) \\
	= \ & P_{(a12\cdots N)} \, \check{R}_{N-1,N}(u/z_N) \cdots \check{R}_{12}(u/z_2) \, \check{R}_{a1}(u/z_1) \, .
	\end{aligned}
\end{equation}
The Drinfeld--Jimbo presentation by \eqref{eq:Uqsl2_spin}--\eqref{eq:affinisation_z} is recovered by expanding in $u^{\pm 1}$, as we show \textsection\ref{s:app_presentations}. 
In particular a pseudo highest-weight vector now is as in \eqref{eq:intro_ABCD_on_shell}. The quantum determinant is multiplicative, yielding a multiple of the identity
\begin{equation} \label{eq:qdet}
	\qdet_a L_a^\text{inh}(u;\vect{z})  = \prod_{i=1}^N \qdet_a R_{ai}(u/z_i) = t^{N/2} \prod_{i=1}^N \frac{u-z_i}{t \, u - z_i} \, .
\end{equation}

\subsubsection{Integrable spin chains} \label{s:spin_chains} The \textit{RLL}-relations~\eqref{eq:RLL} yield a one-parameter family of commuting operators via the transfer matrix
\begin{equation*}
	\tau(u) \coloneqq \tr_a L_a(u) = A(u) + D(u) \, , \qquad\quad \bigl[\tau(u),\tau(v)\bigr]=0 \, .
\end{equation*}
This is the generating function for an abelian subalgebra of $\widehat{\mathfrak{U}}$, sometimes called the Bethe subalgebra, whose elements are commuting charges of quantum-integrable models. Consider the representation \eqref{eq:L_inh} in the `homogeneous limit' $\tau^\textsc{xxz}(u) \coloneqq \mathrm{ev}_1 \, \tau^\text{inh}(u;\vect{z})$, where $\mathrm{ev}_1 \colon z_j \longmapsto 1$ for all $j$. This is the transfer matrix of the six-vertex model. It generates the symmetries of the (spin-1/2) Heisenberg \textsc{xxz} spin chain. Indeed,
\begin{equation} \label{eq:xxz_transl}
	\tau^\textsc{xxz}(1) = \tr_a P_{(a12\cdots N)} = P_{(12\cdots N)}
\checkedMma
\end{equation}
is the (right) translation operator, while the logarithmic derivative
\begin{equation} \label{eq:xxz}
	\begin{gathered}
	H^\textsc{xxz} = -(t^{1/2}-t^{-1/2}) \, \frac{\partial}{\partial \mspace{1mu}u}\bigg|_{u=1} \!\! \log \tau^\textsc{xxz}(u) = -\! \sum_{i \in \mathbb{Z}_N} \!\! h_{i,i+1}^\textsc{xxz}  = \sum_{i=1}^{N-1} e^\text{sp}_i + P_{(1\To N)}^{-1} \, e^\text{sp}_1 \, P_{(1\To N)} \, , \\
	h_{i,i+1}^\textsc{xxz} \coloneqq \frac{[2]}{2} \, \frac{\sigma^z_i \, \sigma^z_{i+1} - 1}{2} + \sigma^+_i \sigma^-_i + \sigma^-_i \sigma^+_i	\, ,
	\end{gathered}
\checkedMma
\end{equation}
is the spin-chain Hamiltonian with anisotropy parameter~$\Delta = [2]/2 = (t^{1/2}+t^{-1/2})/2$. Here we used \cite{TL_71}
\begin{equation} \label{eq:TL_vs_XXZ}
	-(t^{1/2}-t^{-1/2}) \, \check{R}'_{i,i+1}(1) = e^\text{sp}_i = -h_{i,i+1}^\textsc{xxz} -\frac{t^{1/2} - t^{-1/2}}{2} \, \frac{\sigma^z_i - \sigma^z_{i+1}}{2} \, .
\checkedMma
\end{equation} 
The periodic boundary conditions, visible in the term $P_{(1\To N)}^{-1} \, e^\text{sp}_1 \, P_{(1\To N)}$ in \eqref{eq:xxz}, break the $\mathfrak{U}$-invariance of the monodromy matrix. We stress that, although $\widehat{\mathfrak{U}}$ plays an important role in its exact solution, this Hamiltonian does \emph{not} have quantum-affine \emph{symmetries}: this is the whole point of the algebraic Bethe ansatz, where $B^\textsc{xxz}(u)$ is used to construct the model's (highest-weight) eigenvectors; its action changes the energy.

There are also Heisenberg-type spin chains for which the $\mathfrak{U}$-symmetry is preserved. One of these is the \emph{open} Temperley--Lieb spin chain~\cite{PS_90}, with Hamiltonian
\begin{equation*}
	\sum_{i=1}^{N-1} e^\text{sp}_i = -\! \sum_{i=1}^{N-1} \! h_{i,i+1}^\textsc{xxz} - \frac{t^{1/2} - t^{-1/2}}{2} \, \frac{\sigma^z_1 - \sigma^z_N}{2} \, .
\end{equation*}
The final term can be interpreted as carefully adjusted boundary magnetic fields. This Hamiltonian can be obtained from a double-row transfer matrix~\cite{Skl_88,KS_91}.

In general \eqref{eq:L_inh} yields an `inhomogeneous' version of the Heisenberg spin chain, with \emph{inhomogeneities}~$z_j$. These inhomogeneities are natural from the six-vertex model's perspective. Although they are often considered a mere computational tool for the spin chain one can view the inhomogeneous Heisenberg spin chain as a bona fide spin chain in its own right. It has $N$ commuting `inhomogeneous translation operators' $G^\text{inh}_i = \tau^\text{inh}(z_i;\vect{z})$, including $G^\text{inh}_1 = P_{(1\cdots N)} \, \widetilde{G}$ where the latter is as in \eqref{eq:q-translation}, that obey $G^\text{inh}_1 \cdots G^\text{inh}_N = 1$. Interestingly, the Hamiltonian at $u=z_1$ features \emph{long-range} interactions, with $N-1$ terms that are very similar to the (unevaluated) summands of~\eqref{eq:ham_right} along with a truly cyclic term. We will return to this connection in the future.

\section{Derivations} \label{s:proofs}

\noindent With these preliminaries in place we are all set to combine the polynomial and spin sides to get the setting in which the \textsf{q}-deformed Haldane--Shastry model is best understood.

\subsection{Spin-Ruijsenaars model} \label{s:spin_Ruijsenaars} 
Consider the tensor product\,---\,over $\mathbb{C}$ but see \eqref{eq:phys_space} and \eqref{eq:phys_space_alt}\,---\,of $\mathcal{H} = V^{\otimes N}$ and the ring of polynomials in the coordinates,
\begin{equation} \label{eq:big_space}
	\mathcal{H}[\vect{z}] \coloneqq \mathcal{H}[z_1,\To,z_N] = \mathcal{H} \otimes \mathbb{C}[z_1,\To,z_N] \cong V[z]^{\,\otimes N} \, .
\end{equation}
The physical picture is that of $N$ particles with spin~$1/2$ and coordinates $z_j$. 

More mathematically the picture seems to be that the parameter~$z$ of the evaluation module $V(z) = V \otimes \mathbb{C}[z,z^{-1}]$ is reinterpreted as a coordinate (cf.\ the definition of loop algebras). We only consider the positive modes $V[z] \subset V(z)$, which will eventually be justified by the evaluation~\eqref{eq:ev} of the $z_j$ for the spin chain; it is also in accordance with \textsection\ref{s:Macdonald}. We should also point out that in this section $\mathcal{H}$ may be replaced by the Hecke algebra itself, viewed as a $\mathfrak{H}_N$-module in the obvious way. The present setting is recovered when picking the spin representation, the $\mathfrak{gl}_n$-case arises if instead $V = \mathbb{C}^n$, and the (spinless) setting from \textsection\ref{s:Macdonald}--\ref{s:Ruijsenaars} corresponds to $V = \mathbb{C}$ the trivial representation $T_i \mapsto t^{1/2}$. We will elaborate on this elsewhere.

In \cite{BG+_93} it was shown that, analogously to the spinless case from \textsection\ref{s:Macdonald} the big vector space~\eqref{eq:big_space} has a `physical' subspace on which the action of the centre of the affine Hecke algebra gives rise to a spin-version of the Ruijsenaars model. (The connection with the latter was made more explicit in \cite{Kon_96}.) Importantly, as we will see, this model has quantum-affine symmetry.

\subsubsection{Physical (\textsf{q}-bosonic) space} \label{s:physical_space} We want to think of elements of \eqref{eq:big_space} as indistinguishable particles with spin~$1/2$ and coordinates $z_j$. Ordinary bosons (fermions) are defined by their (anti)symmetry under the exchange operators, $s_i \, P_{i,i+1} = \pm 1$. In the form $P_{i,i+1} = \pm s_i$  this is a relation between the spin and polynomial representation of the symmetric group $\mathfrak{S}_N$ acting on either factor of \eqref{eq:big_space}. This is the relevant setting for the spin-Calogero--Sutherland model and isotropic Haldane--Shastry, but breaks the structure from \textsection\ref{s:setup}. The appropriate generalisation to the \textsf{q}-case can be described in terms of $\mathfrak{H}_N$ and in terms of $\mathfrak{S}_N$. We begin with the former characterisation.

Morally the \emph{\textsf{q}-bosonic Fock space} or \emph{physical space}, which we will denote by $\widetilde{\mathcal{H}}$, is the subspace of \eqref{eq:big_space} on which the spin and polynomial representations of the Hecke algebra coincide~\cite{BG+_93}. This couples the two Hecke-representations from \textsection\ref{s:setup}. To motivate the precise definition consider a vector $\ket{\widetilde{\Psi}} = \ket{\widetilde{\Psi}(\vect{z})} \in \mathcal{H}[\vect{z}]$ on which $T_i^\text{pol} \, \ket{\widetilde{\Psi}} = T_i^\text{sp} \, \ket{\widetilde{\Psi}}$ for all~$i$. (We retain the notation $T_i^\text{sp},T_i^\text{pol}$ for the actions from \textsection\ref{s:Hecke_AHA} and \textsection\ref{s:Hecke_TL_Uqsl} extended to $\mathcal{H}[\vect{z}]$ by acting by the identity on the other factor.) To ensure that this extends from the generators to all of $\mathfrak{H}_N$ we more precisely have to ask for the two actions to \emph{anti}coincide. Indeed, in
\begin{equation} \label{eq:Hecke_anticoincide}
	T_i^\text{pol} \, T_j^\text{pol} \, \ket{\widetilde{\Psi}} = T_i^\text{pol} \, T_j^\text{sp} \, \ket{\widetilde{\Psi}} = T_j^\text{sp} \, T_i^\text{pol} \, \ket{\widetilde{\Psi}} = T_j^\text{sp} \, T_i^\text{sp} \, \ket{\widetilde{\Psi}}
\end{equation}
the order of the generators is reversed, so we should treat one representation as a \emph{left} and the other as a \emph{right} action. (See \textsection\ref{s:app_stochastic_twist} for another incarnation of this.) In more mathematical language:

\begin{definition}[\cite{BG+_93}]
The \emph{physical space} is a tensor product over the Hecke algebra
\begin{equation} \label{eq:phys_space}
	\begin{aligned}
	\widetilde{\mathcal{H}} \coloneqq \, {} & \mathcal{H} \! \underset{\mathfrak{H}_N}{\otimes} \! \mathbb{C}[\vect{z}] \, = \, \raisebox{.3em}{$\mathcal{H}[\vect{z}]$} \left/ \raisebox{-.3em}{$\mathcal{N}$} \right. \qquad\qquad \mathcal{N} \coloneqq \sum_{i=1}^{N-1} \mathrm{im}\bigl(T_i^\text{sp} - T_i^\text{pol}\bigr) \subset \mathcal{H}[\vect{z}] \\
	\cong \, {} & \mathcal{B} \, \coloneqq \bigcap_{i=1}^{\smash{N-1}} \ker\bigl(T_i^\text{sp} - T_i^\text{pol}\bigr) 
	= \mathcal{H}[\vect{z}]^{\,\mathfrak{H}_N} \, .
	\end{aligned}
\end{equation}
The first line is the quotient of \eqref{eq:big_space} by the (vector, not direct) sum of the images $\mathrm{im}(T_i^\text{sp} - T_i^\text{pol})$. In the second line we instead view $\widetilde{\mathcal{H}}$ as a subspace of \eqref{eq:big_space}. 
\end{definition}

These two descriptions are isomorphic (as vector spaces; neither is an $\mathfrak{H}_N$-module). Before we explain the final equality in \eqref{eq:phys_space}, describing $\mathcal{B}$ as the $\mathfrak{H}_N$-invariants in $\mathcal{H}[\vect{z}]$, let us demonstrate that $\widetilde{\mathcal{H}} \cong \mathcal{B}$.

\begin{proof}[Proof of isomorphism in \eqref{eq:phys_space}] 
First consider the case $N=2$. From Table~\ref{tb:q-bosons/fermions} we read off that
\begin{subequations}
	\begin{align*}
	\mathrm{im}(T^\text{sp}_1 - T^\text{pol}_1) & = \mathrm{Sym}^2_{t}(V)\otimes (t\,z_1 - z_2) \, \mathbb{C}[z_1,z_2]^{\mathfrak{S}_2} \, \oplus \, \Lambda^{\!2}_{t}(V) \otimes \mathbb{C}[z_1,z_2]^{\mathfrak{S}_2} \, ,
\intertext{and that killing this subspace yields the \textsf{q}-bosonic space}
	\widetilde{\mathcal{H}}^{(N=2)} & = \, \raisebox{.3em}{$V^{\otimes 2} \otimes \mathbb{C}[z_1,z_2]$} \left/ \raisebox{-.3em}{$\mathrm{im}(T^\text{sp}_1 - T^\text{pol}_1)$} \right. \\
	& \cong \ker(T^\text{sp}_1 - T^\text{pol}_1) \\
	& = \mathrm{Sym}^2_{t}(V) \otimes \mathbb{C}[z_1,z_2]^{\mathfrak{S}_2} \, \oplus \, \Lambda^{\!2}_{t}(V) \otimes (t\,z_1 - z_2) \, \mathbb{C}[z_1,z_2]^{\mathfrak{S}_2} \, .
	\end{align*}
\end{subequations}

For general $N$ consider a vector in the complement of $\mathcal{B}$ in $\mathcal{H}[\vect{z}]$. This means there is some $1\leq i \leq N-1$ for which that vector does not lie in the kernel of $T^\text{sp}_i - T^\text{pol}_i$. By the preceding argument it therefore lies in the image of $T^\text{sp}_i - T^\text{pol}_i$, i.e.\ it belongs to $\mathcal{N}$. Hence $\mathcal{H}[\vect{z}] = \mathcal{B} \oplus \mathcal{N}$, which implies the isomorphism.
% not true that $\mathcal{H}[\vect{z}]$ is the set-theoretic union of $\cap \ker$ and $\sum \im$. rather, $\cap \ker$ can be complemented as a vector space by $\sum \im$ [needs argument] to get $\mathcal{H}[\vect{z}]$
\begin{table}[h]
	\begin{tabular}{l ccc} \toprule
		subspace & $T_i^\text{sp} - T_i^\text{pol}$ & & $T_i^\text{sp} + T_i^{\text{pol}\,-1} = T_i^{\text{sp}\,-1} + T_i^\text{pol}$ \\ \midrule
		$\mathrm{Sym}^2_{t}(V)\otimes \mathbb{C}[z_i,z_{i+1}]^{\mathfrak{S}_2}$ & $0$ & & $[2]$ \\
		%%%
		$\mathrm{Sym}^2_{t}(V)\otimes (t\,z_i - z_{i+1}) \, \mathbb{C}[z_i,z_{i+1}]^{\mathfrak{S}_2}$ & $[2]$ & & $0$ \\
		%%%
		$\mspace{23mu}\Lambda^{\!2}_{t}(V) \otimes \mathbb{C}[z_i,z_{i+1}]^{\mathfrak{S}_2}$ & $\mathllap{-}[2]$ & & $0$ \\
		%%%
		$\mspace{23mu}\Lambda^{\!2}_{t}(V) \otimes (t\,z_i - z_{i+1}) \, \mathbb{C}[z_i,z_{i+1}]^{\mathfrak{S}_2}$ & $0$ & & $\mathllap{-}[2]$ \\
		%%%
		\bottomrule
		\hline
	\end{tabular}
	\caption{The eigenvalues of $T_i^\text{sp} \mp T_i^{\text{pol}\,\pm 1}$ on the four direct summands of $V^{\otimes 2} \otimes \mathbb{C}[z_i,z_{i+1}] \cong V[z]^{\,\otimes 2}$ decomposed according to \eqref{eq:T_pol_eigenspaces} and \eqref{eq:T_sp_eigenspaces}.}
	\label{tb:q-bosons/fermions}
\end{table}
\end{proof}

The final equality in \eqref{eq:phys_space} gives a more intrinsic characterisation of $\mathcal{B} \subset \mathcal{H}[\vect{z}]$. Rather than coupling two commuting Hecke actions define \cite{GRV_94} (cf.~\cite{TU_98} and the references therein)
\begin{equation} \label{eq:Hecke_tot}
	T_i^\text{tot} \coloneqq  s_i \, (T_i^\text{sp} - T_i^\text{pol}) + t^{1/2} \, .
\end{equation}
This generates a diagonal action of $\mathfrak{H}_N$ on $\mathcal{H}[\vect{z}]$ that \textsf{q}-deforms $s_i \, P_{i,i+1}$ in a nontrivial way. (The obvious guess $T_i^{\text{sp}\,\pm1} \, T_i^{\text{pol}\,\mp1}$ fails the Hecke condition.) It clearly commutes with the action of $\mathfrak{U}$ on the spin factor. The presence of $s_i$ makes direct verification of the Hecke-algebra relations~\eqref{eq:Hecke} rather tedious. For the Hecke condition one can use $s_i \, T_i^\text{pol} \, s_i = [2] \, s_i - T_i^{\text{pol}\mspace{1mu}-1}$. The explicit decomposition of $V^{\otimes 2} \otimes \mathbb{C}[z_i,z_{i+1}] \cong V[z]^{\,\otimes 2}$ into $T_i^\text{tot}$-eigenspaces is given in Table~\ref{tb:total_Hecke}. Comparing this with Table~\ref{tb:q-bosons/fermions} shows that $\ker(T_i^\text{sp} - T_i^\text{pol}) = \ker(T_i^\text{tot} - t^{1/2})$. Since $T_i^\text{tot} = t^{1/2}$ is as close as it gets to invariance under the Hecke algebra given our normalisation of the Hecke generators, with the Hecke condition as in \eqref{eq:Hecke}, we identify 
\begin{equation} \label{eq:phys_space_via_Ttot}
	\mathcal{B} = \bigcap_{i=1}^{N-1} \ker(T_i^\text{sp} - T_i^\text{pol}) = \bigcap_{i=1}^{N-1} \ker(T_i^\text{tot} - t^{1/2}) \eqqcolon \mathcal{H}[\vect{z}]^{\,\mathfrak{H}_N} \, .
\end{equation}
This allows us to describe $\mathcal{B}$ as the totally \textsf{q}-symmetric subspace, obtained by projecting with the total \textsf{q}-symmetriser $\Pi_+^\text{tot}$, cf.~\eqref{eq:Hecke_projectors}.

\begin{table}[h]
	\begin{tabular}{l cc} \toprule
		subspace & & $T_i^\text{tot}$ \\ \midrule
		$\mathrm{Sym}^2_{t}(V)\otimes \mathbb{C}[z_i,z_{i+1}]^{\mathfrak{S}_2}$ & & $t^{1/2}$ \\
		%%%
		$\mathrm{Sym}^2_{t}(V)\otimes (t\,z_{i+1} - z_i) \, \mathbb{C}[z_i,z_{i+1}]^{\mathfrak{S}_2}$ & & $-t^{-1/2}$ \\
		%%%
		$\mspace{23mu}\Lambda^{\!2}_{t}(V) \otimes \mathbb{C}[z_i,z_{i+1}]^{\mathfrak{S}_2}$ & & $-t^{-1/2}$ \\
		%%%
		$\mspace{23mu}\Lambda^{\!2}_{t}(V) \otimes  (t\,z_i - z_{i+1}) \, \mathbb{C}[z_i,z_{i+1}]^{\mathfrak{S}_2}$ & & $t^{1/2}$ \\
		%%%
		\bottomrule
		\hline
	\end{tabular}
	\caption{Eigenvalues of \eqref{eq:Hecke_tot} on four direct summands of $V^{\otimes 2} \otimes \mathbb{C}[z_i,z_{i+1}] \cong V[z]^{\,\otimes 2}$. The first, third and fourth rows are immediate from Table~\ref{tb:q-bosons/fermions}, while the second row requires a computation. This time both decompositions $\mathbb{C}[z_i,z_{i+1}] \cong \mathbb{C}[z_i,z_{i+1}]^{\mathfrak{S}_2} \oplus (t\,z_i - z_{i+1}) \, \mathbb{C}[z_i,z_{i+1}]^{\mathfrak{S}_2} \cong \mathbb{C}[z_i,z_{i+1}]^{\mathfrak{S}_2} \oplus (t\,z_{i+1} - z_i) \, \mathbb{C}[z_i,z_{i+1}]^{\mathfrak{S}_2}$ appear.}
	\label{tb:total_Hecke}
\end{table}

For generic $t \in \mathbb{C}^\times$ we have $\mathfrak{H}_N \cong \mathbb{C}[\mathfrak{S}_N]$ (as algebras), where the latter is the group algebra of the symmetric group. Accordingly the physical space also admits a characterisation in terms of a ($t$-dependent) representation of $\mathfrak{S}_N$. Using the functions $f,g$ from \eqref{eq:f,g} and the `Baxterisation' formula~\eqref{eq:baxterisation} we can recast
\begin{equation} \label{eq:phys_space_equivalence}
	\begin{aligned}
	T_i^\text{sp} - T_i^\text{pol} & = T_i^\text{sp} - (a_{i,i+1} \, s_i + b_{i,i+1}) \\
	& = a_{i,i+1} \, \bigl(f_{i,i+1} \, T_i^\text{sp} + g_{i,i+1} - s_i \bigr) \\
	& = a_{i,i+1} \, \bigl(\check{R}_{i,i+1}(z_i/z_{i+1}) - s_i \bigr) \\
	& = a_{i,i+1} \, s_i \, (s^\text{tot}_i - 1) \, ,
	\end{aligned}
\end{equation}
where in the last line we defined (cf.~\textsection10.2 in \cite{Gau_83}), Prop.~6.2 in \cite{FR_92}) 
\begin{equation} \label{eq:s_tot}
	s^\text{tot}_i \coloneqq s_i \, \check{R}_{i,i+1}(z_i/z_{i+1}) \, .
\end{equation}
Thanks to \eqref{eq:R_relations} the latter obeys the braid relations and is an involution, $(s_i^\text{tot})^2 = 1$, yielding a representation of $\mathfrak{S}_N$ on $\mathcal{H}[\vect{z}]$ that depends on~$t$ and deforms $s_i \, P_{i,i+1}$ too. We'll write $s^\text{tot}_w$ for the operator representing $w \in \mathfrak{S}_N$ in this way. This gives 
\begin{proposition}
The physical space may also be characterised as
\begin{equation} \label{eq:phys_space_alt}
	\normalfont
	\begin{aligned}
	\widetilde{\mathcal{H}} & = \, \mathcal{H} \! \underset{\mathfrak{S}_N}{\otimes} \! \mathbb{C}[\vect{z}] = \, \raisebox{.3em}{$\mathcal{H}[\vect{z}]$} \left/ \raisebox{-.3em}{$\mathcal{N}$} \right. \qquad\qquad\qquad \mathcal{N} = \sum_{i=1}^{N-1} \mathrm{im}\bigl(s^\text{tot}_i - 1\bigr) \\
	& \cong \, \mathcal{B} \, = \, \mathcal{H}[\vect{z}]^{\,\mathfrak{S}_N} = \bigcap_{i=1}^{N-1} \ker(s^\text{tot}_i - 1) \, .
	\end{aligned}
\end{equation}
\end{proposition}
\noindent Here $s_i^\text{tot} = 1$ is the `local condition' from the quantum Knizhnik--Zamolodchikov (\textsf{q}KZ) system (reduced \textsf{q}KZ equation) \cite{Smi_86,FR_92,Che_92b}.

Next we turn to the elements of $\widetilde{\mathcal{H}}$. Physical vectors have a rather nice form with respect to the coordinate basis. Since the weight decomposition~\eqref{eq:weight_decomp} is preserved by the $T^\text{sp}_i$ the physical space decomposes into \textit{M}-particle sectors
\begin{equation} \label{eq:phys_space_weight}
	\widetilde{\mathcal{H}} = \bigoplus_{M=0}^N \widetilde{\mathcal{H}}_M \, , \qquad  \widetilde{\mathcal{H}}_M \coloneqq \mathcal{H}_M \! \underset{\mathfrak{H}_N}{\otimes} \! \mathbb{C}[\vect{z}] \cong \mathcal{B} \cap (\mathcal{H}_M \otimes \mathbb{C}[\vect{z}]) \, .
\end{equation}
Vectors in this \textit{M}-particle sector have the explicit form given in Proposition~\ref{prop:intro_physical_vectors} from \textsection\ref{s:intro_phys_space}:

\begin{proposition}[cf.~\cite{RSZ_07}] \label{prop:physical_vectors}
A vector in $\mathcal{H}_M \otimes \mathbb{C}[\vect{z}]$ is physical, i.e.\ lies in $\widetilde{\mathcal{H}}_M \subset \widetilde{\mathcal{H}}$, if{f} with respect to the coordinate basis~\eqref{eq:coord_basis} it has the form\,%
\footnote{\ Note that this `Hecke form' of physical vectors is closely related to the characterisation of the physical space in terms of the Hecke algebra~$\mathfrak{H}_N$. It has an analogue corresponding to the characterisation of $\widetilde{\mathcal{H}}$ via $\mathfrak{S}_N$, which gives an `$\mspace{-1mu}$\textit{R}-matrix form' for physical vectors. Surprisingly, we find that we then have to replace $\Delta_t \, \Delta_{1/t} \rightsquigarrow \Delta_1^2$ in \eqref{eq:intro_polynomial}. We will get back to this in the future. \label{fn:R-mat_form}}
\begin{equation}\label{eq:phys_vector}
	\normalfont
	\begin{aligned} 
		\ket{\widetilde{\Psi}(\vect{z})} \coloneqq {} & \!\! \sum_{i_1 < \cdots < i_M}^N \!\!\!\!\! T_{\{i_1,\To,i_M\}}^\text{pol} \widetilde{\Psi}(\vect{z}) \, \cket{i_1,\To,i_M} \\
		= {} & \!\! \sum_{w \in \mathfrak{S}_N/(\mathfrak{S}_M \times \mathfrak{S}_{N-M})} \!\!\!\!\!\!\!\!\!\!\!\!\!\!\!\!\! T_w^\text{pol} \, \widetilde{\Psi}(\vect{z}) \, \cket{w\,1,\To,w\,M} \, , \\
	\end{aligned} 
	\qquad\quad \widetilde{\Psi}(\vect{z}) \in \mathbb{C}[\vect{z}]^{\mathfrak{S}_M \times \mathfrak{S}_{N-M}} \, ,
\end{equation}
where in the equality we recognise the (Grassmannian) permutations $w=\{i_1,\To,i_M\}$ as representatives for the coset $\mathfrak{S}_N/(\mathfrak{S}_M \times \mathfrak{S}_{N-M})$. In other words, each \textit{M}-particle sector in \eqref{eq:phys_space_weight} has a `polynomial avatar' consisting of polynomials with definite symmetry: we have a bijection
\begin{equation} \label{eq:phys_space_M_particle}
	\begin{aligned}
		\widetilde{\mathcal{H}}_M \ \ & \, \xrightarrow{\quad\sim\quad} \, \mathbb{C}[\vect{z}]_M \coloneqq \, \mathbb{C}[\vect{z}]^{\mathfrak{S}_M \times \mathfrak{S}_{N-M}} \, , \\
		\rotatebox{90}{$\in$} \ \ \ \, & \, \hphantom{\xmapsto{\quad\sim\quad}} \qquad\!\! \rotatebox{90}{$\in$} \\[-.2\baselineskip]
		\ket{\widetilde{\Psi}(\vect{z})} & \, \xmapsto{\hphantom{\quad\sim\quad}} \, \ \widetilde{\Psi}(\vect{z}) \ = \, \cbraket{1,\To,M}{\widetilde{\Psi}(\vect{z})} \, .
	\end{aligned}
\end{equation}
\end{proposition}

We will call $\widetilde{\Psi}(\vect{z})$ the \emph{simple component} of $\ket{\widetilde{\Psi}(\vect{z})}$.

\begin{proof}
A straightforward, if tedious, check shows that the generators of the two Hecke actions coincide on any vector of the form \eqref{eq:phys_vector}, so the latter lies in $\mathcal{B} \cong \widetilde{\mathcal{H}}$. It remains to show that any $\ket{\widetilde{\Psi}} \in \widetilde{\mathcal{H}}_M$ is of this form. 
	
By \eqref{eq:Hecke_spin} we have $\bra{\downarrow\uparrow}\,T^\text{sp} = \bra{\uparrow\downarrow}$ whence $\cbra{i-1}\,T_{i-1}^\text{sp} = \cbra{i}$. Iterating this yields $\cbra{i} = \cbra{1} \, T^\text{sp}_1 \cdots T^\text{sp}_{i-1}$. The physical condition~\eqref{eq:phys_space} thus interrelates the components of vectors in $\widetilde{\mathcal{H}}$ with respect to the coordinate basis. Let us show this in detail for $M=1$: 
\begin{equation*}
	\begin{aligned}
	\cbraket{i}{\widetilde\Psi} & = \cbra{1}\, T^\text{sp}_1 \cdots T^\text{sp}_{i-2} \, T^\text{sp}_{i-1} \, \ket{\widetilde\Psi} \\
	& = \cbra{1}\, T^\text{sp}_1 \cdots T^\text{sp}_{i-2} \, T^\text{pol}_{i-1} \, \ket{\widetilde\Psi} \\
	& = T^\text{pol}_{i-1} \, \cbra{1}\, T^\text{sp}_1 \cdots T^\text{sp}_{i-2} \, \ket{\widetilde\Psi} \\
	& = \cdots \\ 
	& = T^\text{pol}_{i-1} \, T^\text{pol}_{i-2} \cdots T^\text{pol}_1 \, \cbraket{1}{\widetilde\Psi} \\ 
	& = T^\text{pol}_{(i,\To,1)} \, \widetilde\Psi(\vect{z}) \, .
	\end{aligned}
\end{equation*}
For general $M$ we just repeat this:
\begin{equation*}
	\begin{aligned}
	\cbraket{i_1,i_2,\To,i_M}{\widetilde{\Psi}}
	& = T_{(i_1,\To,1)}^\text{pol} \cbraket{1,i_2,\To,i_M}{\widetilde{\Psi}} \\
	& = \cdots \\
	& = T_{(i_1, \To, 1)}^\text{pol} \cdots T_{(i_M, \To, M)}^\text{pol} \cbraket{1,\To,M}{\widetilde{\Psi}} \\
	& = T_{\{i_1,\To,i_M\}}^\text{pol} \, \widetilde{\Psi}(\vect{z}) \, .
	\end{aligned}
\end{equation*}
Moreover, the simple component inherits the symmetry of $\cbra{1,\To,M}$,
\begin{equation*}
	(T_i^\text{pol} - t^{1/2}) \, \cbraket{1,\To,M}{\widetilde{\Psi}} = \cbra{1,\To,M} \, (T_i^\text{sp} - t^{1/2}) \, \ket{\widetilde{\Psi}} = 0 \, , \qquad i\neq M \, .
	\qedhere
\end{equation*}
\end{proof}
\noindent This proves Proposition~\ref{prop:intro_physical_vectors}.

An operator on the big vector space $\mathcal{H}[\vect{z}]$ descends to the physical space if it preserves $\mathcal{N}$ in~\eqref{eq:phys_space}, or equivalently if it preserves $\widetilde{\mathcal{H}} \cong \mathcal{B} \subset \mathcal{H}[\vect{z}]$. That is,
\begin{definition}
An operator $\widetilde{O}$ on $\mathcal{H}[\vect{z}]$ is called \emph{physical} if for any $\widetilde{\Psi} \in \widetilde{\mathcal{H}}$ we have $T_i^\text{sp} \, \widetilde{O} \, \widetilde{\Psi} = T_i^\text{pol} \, \widetilde{O} \, \widetilde{\Psi}$ for all $1\leq i \leq N-1$.
\end{definition}
Since a physical vector is completely determined by a single component with definite symmetry we may forget about the spin part and pass to the world of polynomials. In particular any physical operator is completely determined by its action on the simple component, inducing an action on polynomials. That is, any linear operator $\widetilde{O}^\text{sp}$ on $\widetilde{\mathcal{H}}$ is equivalent to some $\widetilde{O}^\text{pol}$ acting on polynomials such that 
\begin{equation*}
	\widetilde{O}^\text{sp} \, \ket{\widetilde{\Psi}(\vect{z})} = \ket{\widetilde{O}^\text{pol}\, \widetilde{\Psi}(\vect{z})} \, .
\end{equation*}
More specifically,
\begin{definition}
Assume that $\widetilde{O}^\text{sp} \, \widetilde{\mathcal{H}}_M \subset \widetilde{\mathcal{H}}_{M'}$ for some $M'$ (typically depending on $M$); this holds for all physical operators that we will consider with $M' \in \{M-1,M,M+1\}$. Taking the simple component of the preceding equality leads us to define $\widetilde{O}^\text{pol}_{M,M'} \colon \mathbb{C}[\vect{z}]_M \longrightarrow \mathbb{C}[\vect{z}]_{M'}$ by
\begin{equation} \label{eq:induced_pol}
	\begin{aligned}
	\widetilde{O}^\text{pol}_{M,M'} \, \widetilde{\Psi}(\vect{z}) \, \coloneqq \, & \cbra{1,\To,M'} \, \widetilde{O}^\text{sp} \, \ket{\widetilde{\Psi}(\vect{z})} \\
	= \, & \!\! \sum_{i_1 < \cdots < i_M}^N \!\!\!\!\! \cbra{1,\To,M'} \, \widetilde{O}^\text{sp} \, \cket{i_1,\To,i_M} \, T_{\{i_1,\To,i_M\}}^\text{pol} \widetilde{\Psi}(\vect{z}) \, .
	\end{aligned}
\end{equation}
We set $\mathbb{C}[\vect{z}]_{-1} = \mathbb{C}[\vect{z}]_{N+1} = \{0\}$.
\end{definition}
This trick will be particularly useful for dealing with the nonabelian symmetries in \textsection\ref{s:nonabelian_tilde}.

\subsubsection{Abelian symmetries (spin-Macdonald operators)} \label{s:abelian_tilde} 
The Hecke action on the polynomial factor of the big vector space~\eqref{eq:big_space} readily extends to an action of the \textsc{aha} as in \textsection\ref{s:polynomial}. We retain the notation $Y_i$ for the operators $Y_i^\text{pol} = \id \otimes Y_i$ acting nontrivially on polynomials.

In \cite{BG+_93} it was realised that elements of the \emph{centre} of the polynomial action of the \textsc{aha} are physical operators, and so is generating function~\eqref{eq:Delta}. (See \textsection\ref{s:nonabelian_tilde} for another, more subtle example of a physical operator.)

We can now derive a hierarchy of commuting difference operators on the physical space proceeding like in \textsection\ref{s:Macdonald}. The spin analogues of the Macdonald operators arise as 
\begin{equation} \label{eq:spin_D_r_def}
	\widetilde{D}_r \coloneqq e_r(\vect{Y}) \qquad \text{on}  \quad \widetilde{\mathcal{H}} \, , \qquad 0\leq r\leq N \, .
\end{equation}
Note that $\widetilde{D}_N = \hat{q}_1 \cdots \hat{q}_N = D_N$ is the same as in the scalar case. As in \eqref{eq:D_symmetry} we have
\begin{equation} \label{eq:spin_D_symmetry}
	\widetilde{D}_{N-r} = \widetilde{D}_N \, \widetilde{D}_{-r} \, , \qquad\quad \widetilde{D}_{-r} \coloneqq \!\!\! \sum_{i_1<\cdots<i_r}^N \!\!\!\!\! Y_{i_1}^{-1} \cdots Y_{i_r}^{-1} \qquad \text{on}  \quad \widetilde{\mathcal{H}} \, , \qquad 0\leq r\leq N \, .
\end{equation}
Using the physical condition~\eqref{eq:phys_space} we obtain explicit expressions for these operators, which by construction commute. This yields Theorem~\ref{thm:intro_spin_D_r} from \textsection\ref{s:intro_abelian_tilde} in terms of the convenient shorthand (cf.~Proposition~6.3 in \cite{FR_92})
\begin{equation} \label{eq:Rcheck_w}
	\check{R}_w \coloneqq s_{w^{-1}} \, s^\text{tot}_w \qquad \text{so that} \qquad s^\text{tot}_w = \check{R}_{w^{-1}}^{-1} \, s_w = s_w \, \check{R}_w \, , \qquad w\in \mathfrak{S}_N \, .
\end{equation}

\begin{theorem}[cf.~\cite{Che_94a}]\label{thm:spin_D_r}
The spin analogue~\eqref{eq:spin_D_r_def} of the Macdonald operators are
\begin{equation} \label{eq:spin_D_r}
	\widetilde{D}_r = \!\!\! \sum_{J \colon \# J = r} \!\!\!\!\! A_J(\vect{z}) \, \check{R}_{\{j_1,\To,j_M\}^{-1}}^{-1} \, \hat{q}^{\phantom{x}}_J \, \check{R}_{\{j_1,\To,j_M\}^{-1}}^{\vphantom{-1}} \, ,
\end{equation}
with $A_J(\vect{z})$ from \eqref{eq:D_r}.
In particular $\widetilde{D}_{N-1} = \hat{q}_1 \cdots \hat{q}_N \, \widetilde{D}_{-1} = \widetilde{D}_{-1} \, \hat{q}_1 \cdots \hat{q}_N$, where
\begin{equation} \label{eq:spin_D_-1}
	\widetilde{D}_{-1} = \sum_{i=1}^N A_{-i}(\vect{z}) \, \check{R}_{(N,N-1 \cdots i)}^{-1} \, \hat{q}_i^{-1} \, \check{R}_{(N,N-1 \cdots i)} \, , \qquad 
	A_{-i}(\vect{z}) = \prod_{\bar\imath\mspace{1mu}(\neq i)}^N \!\! a_{\bar\imath \mspace{1mu} i} \, .
\end{equation}
The other $\widetilde{D}_{-r}$ are similarly obtained by conjugating with more layers of \textit{R}-matrices.
\end{theorem}

Before we get to the proof let us illustrate this formula. The coefficient $A_j(\vect{z})$ is just as in the scalar case, see~\eqref{eq:D_1}.
Let us illustrate the notation~\eqref{eq:Rcheck_w} with some examples. To start with, $\check{R}_i \coloneqq \check{R}_{s_i} = \check{R}_{i,i+1}(z_i/z_{i+1})$. In general $\check{R}_w$ is obtained by drawing the braid diagram for a reduced decomposition of $w$ and reinterpreting it in terms of graphical notation~\eqref{eq:R_graphical} for the \textit{R}-matrix. (The $s_w$ ensures that all $z_j$, which are moved around by the \textit{R}-matrices, end up at their original positions.) Observe that \eqref{eq:Rcheck_w} is \emph{not} a representation of $\mathfrak{S}_N$; for example $\check{R}_{(321)} = \check{R}_{23}(z_1/z_3)\,\check{R}_{12}(z_1/z_2)$ is not the inverse of $\check{R}_{(123)} = \check{R}_{12}(z_1/z_3)\,\check{R}_{23}(z_2/z_3)$. Now recall that $\{j_1,\To,j_M\} \in \mathfrak{S}_N$ was defined in \eqref{eq:perm_grassmannian}, so that
\begin{equation}
\check{R}_{\{j_1,\To,j_M\}^{-1}} = \check{R}_{(M,\To,j_M)} \cdots \check{R}_{(1,\To,j_1)} = 
\tikz[baseline={([yshift=-.5*11pt*0.4]current bounding box.center)},	xscale=0.5,yscale=-0.4,font=\scriptsize]{
	\draw[<-] (7,0) -- (7,5); 
	\draw[<-] (8.5,0) node[above]{$z_N$} -- (8.5,5) node[below]{$z_N$};
	\foreach \xx in {-1,...,1} \node at (7.75+.2*\xx,2.5) {$\cdot\mathstrut$};
	\draw[rounded corners=3.5pt,<-] (6,0) -- (6,3.5) -- (5,4.5) -- (5,5);
	\draw[rounded corners=3.5pt,<-] (5,0) -- (5,2.5) -- (3,4.5) -- (3,5);
	\draw[rounded corners=3.5pt,<-] (4,0) -- (4,1.5) -- (2,3.5) -- (2,5);
	\draw[rounded corners=3.5pt,<-] (3,0) node[above]{$z_1$} -- (3,.5) -- (0,3.5) -- (0,5) node[below]{$z_1$};
	\draw[rounded corners=3.5pt,<-] (2,0) node[above]{$z_{j_M}$} -- (2,.5) -- (6,4.5) -- (6,5) node[below]{$z_{j_M}$};
	\draw[rounded corners=3.5pt,<-] (1,0) node[above]{$\cdots$} -- (1,1.5) -- (4,4.5) -- (4,5) node[below]{$\cdots$};
	\draw[rounded corners=3.5pt,<-] (0,0) node[above]{$z_{j_1}$} -- (0,2.5) -- (1,3.5) -- (1,5) node[below]{$z_{j_1}$};
} \, .
\end{equation}
For $r=1$ Theorem~\ref{thm:spin_D_r} gives \eqref{eq:intro_spin-D_1}:
\begin{equation} \label{eq:spin-D_1}
	\widetilde{D}_1 = \sum_{j=1}^N A_j(\vect{z}) \, \check{R}_{(12 \cdots j)}^{-1} \, \hat{q}_j \, \check{R}_{(12\cdots j)}^{\vphantom{-1}} \, , \qquad A_j(\vect{z}) = \prod_{\bar\jmath(\neq j)}^N \!\! a_{j \bar\jmath} \, .
\end{equation}
The expression for $r=2$ is given in \eqref{eq:intro_spin-D_2}.

\begin{proof}[Proof of Theorem~\ref{thm:spin_D_r}]
Our proof of Proposition~\ref{prop:D_r} in \textsection\ref{s:Macdonald} readily adapts to the spin case. As in the scalar case the Hamiltonian can be written in normal form $\widetilde{D}_1 = \sum_j \widetilde{A}_j(\vect{z})$, where $\widetilde{A}_j(\vect{z})$ acts on polynomials only by a rational factor times $\hat{q}_j$. The computation of $\widetilde{A}_1(\vect{z}) = A_1(\vect{z}) \, \hat{q}_1$ is as before. The only new feature in the spin case is that the trick for getting the other $\widetilde{A}_j(\vect{z})$ now involves conjugation by the $t$-dependent $\mathfrak{S}_N$-action~\eqref{eq:s_tot}:
\begin{equation*}
	\widetilde{D}_1 = s^\text{tot}_{(j \cdots 21)} \, \widetilde{D}_1 \, s^\text{tot}_{(12\cdots j)} \qquad \text{on} \quad \widetilde{\mathcal{H}} \, .
\end{equation*}
On the right-hand side the contribution to the term with $\hat{q}_j$ is easy to compute, whence $\widetilde{A}_j(\vect{z}) = s^\text{tot}_{(j \cdots 21)} \, \widetilde{A}_1(\vect{z}) \, s^\text{tot}_{(12\cdots j)} = \check{R}_{(12\cdots j)}^{-1} \, s_{(j \cdots 21)} \, A_1(\vect{z}) \, \hat{q}_1 \, s_{(12\cdots j)} \, \check{R}_{(12\cdots j)} = A_j(\vect{z}) \,\check{R}_{(12\cdots j)}^{-1} \, \hat{q}_j \, \check{R}_{(12\cdots j)}$.

The higher spin-Macdonald operators \eqref{eq:spin_D_r} are found analogously. Finally, for $r=-1$ the proof of Proposition~\ref{prop:D_-r} in \textsection\ref{s:Macdonald} readily adapts to the spin-case as well. We obtain \eqref{eq:spin_D_-1} from
\begin{equation*}
	\widetilde{D}_{-1} = s^\text{tot}_{(i\cdots N-1,N)} \, \widetilde{D}_{-1} \, s^\text{tot}_{(N,N-1 \cdots i)} \qquad \text{on} \quad \widetilde{\mathcal{H}} \, .
	% = \sum_{i=1}^N s^\text{tot}_{(i\cdots N-1,N)} \, A_{-N} \, \hat{q}_N^{-1} \, s^\text{tot}_{(N,N-1 \cdots i)} \, .
	\qedhere
\end{equation*}
\end{proof}

We will denote the generating function of the spin-Macdonald operators by
\begin{equation} \label{eq:Delta_tilde}
	\widetilde{\Delta}(u) \coloneqq \prod_{i=1}^N (1 + u \, Y_i) = \sum_{r=0}^N u^r \, \widetilde{D}_r  \qquad \text{on}  \quad \widetilde{\mathcal{H}} \, .
\end{equation}

\subsubsection{Nonabelian symmetries} \label{s:nonabelian_tilde} 
The interesting new feature of the spin version of the Ruijsenaars model is the presence of $\widehat{\mathfrak{U}}$-symmetry~\cite{BG+_93}, cf.~\cite{CP_96}. We begin with the \textsc{frt} presentation (Theorem~\ref{thm:monodromy_tilde} from \textsection\ref{s:intro_nonab_dyn}):

\begin{theorem}[\cite{BG+_93}] \label{thm:nonabelian_tilde}
Introduce an auxiliary space~$V_a \cong V$ with spectral parameter~$u$ and define on $V_a \otimes \mathcal{H}[\vect{z}]$ the monodromy matrix
\begin{equation} \label{eq:L_tilde}
	\begin{aligned}
	\widetilde{L}_a(u) \coloneqq {} & R_{aN}(u\,Y_N) \cdots R_{a2}(u\,Y_2) \, R_{a1}(u\,Y_1) \\
	= {} & \normalfont \frac{(-t^{1/2})^N}{\Delta(-t\,u)} \, P_{(a1\cdots N)} \, \bigl( u \, T_{N-1}^\text{sp} \, Y_N - T_{N-1}^{\text{sp}\,-1} \bigr) \cdots	\bigl(u \, T_0^\text{sp} \, Y_1 - T_0^{\text{sp}\,-1} \bigr) \, .
	\end{aligned}
\end{equation}
This endows $\widetilde{\mathcal{H}}$ with an action of $\widehat{\mathfrak{U}}$ that commutes with the spin-Macdonald operators. 
\end{theorem}
\noindent The appearance of some sort of inhomogeneities in~\eqref{eq:L_tilde} is not surprising from the Heisenberg point of view, whose nearest-neighbour interactions~\eqref{eq:xxz} are deformed to long-ranged ones away from the homogeneous point. In the more algebraic language of the \textsc{daha} \eqref{eq:L_tilde} is the dual, $Y_i \leftrightarrow Z_i^{-1}$ (if we allow for Laurent polynomials in the~$z_i$), of the inhomogeneous \textsc{xxz} monodromy matrix~\eqref{eq:L_inh}. Since the following proof only relies on the \textsc{aha} relations it follows that the monodromy matrix~\eqref{eq:L_inh}, and therefore the inhomogeneous \textsc{xxz} spin chain, act on the physical space too. This played an important role in e.g.~\cite{DZ_05a,Pas_06,DZ_05b}.

\begin{proof}[Proof of Theorem~\ref{thm:nonabelian_tilde}]
As $R_{ab}(u/v)$ commutes with the $Y_i$ the (level $c=0$) \textit{RLL}-relations~\eqref{eq:RLL} are verified as usual. As $\widetilde{\Delta}(u)$ is central it is furthermore clear that
\begin{equation} \label{eq:commutator_L_Delta}
	\bigl[ \widetilde{L}_a(u) , \widetilde{\Delta}(v) \bigr] = 0 \, ,
\end{equation}
which entails commutativity with the spin-Macdonald operators.

The crucial step is to prove that \eqref{eq:L_tilde} descends to the physical space. Viewing $\widetilde{\mathcal{H}} \cong \mathcal{B}$ as a subspace of $\mathcal{H}[\vect{z}]$ this amounts to showing that $\widetilde{L}_a(u)$ preserves $\mathcal{B}$. We will demonstrate that for any $1\leq i \leq N-1$ we have $\widetilde{L}_a(u) \, \ker(T_i^\text{sp} - T_i^\text{pol}) \subset \ker(T_i^\text{sp} - T_i^\text{pol})$, i.e.\ that $T_i^\text{sp}$ and $T_i^\text{pol}$ coincide on $\widetilde{L}_a(u) \, \ker(T_i^\text{sp} - T_i^\text{pol})$. 

The equality in \eqref{eq:L_tilde} uses the `Baxterisation' formula \eqref{eq:baxterisation}. Together, $T_0^\text{sp} \coloneqq T^\text{sp}_{a1}$ and the other $T_i^\text{sp}$ form a representation of $\mathfrak{H}_{N+1}$ on $V_a \otimes \mathcal{H} \cong V^{\otimes (N+1)}$. Let us remove the (central) denominator of \eqref{eq:L_tilde}. Since
\begin{equation*}
	T_i^\text{sp} \, P_{(a1\cdots N)} = P_{(a1\cdots N)} \, T_{i-1}^\text{sp} 
\end{equation*}
we have to show that $T_{i-1}^\text{sp}$ and $T_i^\text{pol}$ act in the same way on the factors
\begin{equation*}
	\begin{aligned}
	& \bigl( u \, T_i^\text{sp}\, Y_{i+1} - T_i^{\text{sp}\, -1} \bigr) \bigl( u \, T_{i-1}^\text{sp} \, Y_i - T_{i-1}^{\text{sp}\, -1} \bigr) \\
	& = u^2 \, T_i^\text{sp} \, T_{i-1}^\text{sp} \, Y_i \, Y_{i+1} - u\, (T_i^\text{sp} \, T_{i-1}^{\text{sp}\mspace{1mu}-1} \, Y_{i+1} + T_i^{\text{sp}\mspace{1mu}-1} \, T_{i-1}^\text{sp} \, Y_i) + T_i^{\text{sp}\mspace{1mu}-1}\,T_{i-1}^{\text{sp}\mspace{1mu}-1}
	\end{aligned}
\end{equation*}
where on the right we can use the physical condition. We proceed order by order in $u$.

Quadratic order in $u$ just uses the braid relation and the fact that $T_i^\text{pol}$ commutes with $Y_i \, Y_{i+1}$. Order $u^0$ is straightforward too as the Hecke condition and the braid relation for the inverse Hecke generators imply
\begin{equation*}
	T_{i-1}^\text{sp} \, T_i^{\text{sp}\mspace{1mu}-1} \, T_{i-1}^{\text{sp}\mspace{1mu}-1} = T_i^{\text{sp}\mspace{1mu}-1} \, T_{i-1}^{\text{sp}\mspace{1mu}-1} \, T_i^\text{sp}	\, .
\end{equation*}
Finally, for the part linear in $u$ we rewrite
\begin{equation*}
	T_i^\text{sp} \, T_{i-1}^{\text{sp}\mspace{1mu}-1} \, Y_{i+1} + T_i^{\text{sp}\mspace{1mu}-1} \, T_{i-1}^\text{sp} \, Y_i = T_i^\text{sp} \, T_{i-1}^\text{sp} \, (Y_i + Y_{i+1}) - (t^{1/2}-t^{-1/2}) (T_i^\text{sp} \, Y_{i+1} + T_{i-1}^\text{sp} \, Y_i) \, .
\end{equation*}
Commutation with the part featuring $Y_i + Y_{i+1}$ is again simple. For the remainder use $T_i^\text{sp} \, Y_{i+1} = Y_{i+1} \, T_i^\text{sp} = Y_{i+1} \, T_i^\text{pol} = T_i^{\text{pol}\mspace{1mu}-1} \, Y_i$ to see that on $T_i^{\text{pol}\mspace{1mu}-1} + T_{i-1}^\text{sp} = T_i^\text{pol} + T_{i-1}^{\text{sp}\mspace{1mu}-1}$ the actions of $T_i^\text{pol}$ and $T_{i-1}^\text{sp}$ coincide too.
\end{proof}

Replacing $z_j \rightsquigarrow Y_j^{-1}$ in \eqref{eq:qdet} yields the quantum determinant \cite{BG+_93}
\begin{equation} \label{eq:qdet_Y}
	\qdet_a \widetilde{L}_a(u) = t^{N/2} \, \frac{\Delta(-u)}{\Delta(-t \,u)} \, .
\end{equation}
This is a scalar as far as the spins are concerned, but still acts nontrivially on polynomials. 

\begin{proposition}[\cite{CP_96,JK+_95b}, cf.~Footnote~\ref{fn:coprod} on p.\,\pageref{fn:coprod}] \label{prop:chevalley_tilde}
The Chevalley generators obtained from $\widetilde{L}_a(u)$ are \eqref{eq:Uqsl2_spin}, together with \eqref{eq:affinisation_z} where $z_i \rightsquigarrow Y_i^{-1}$, i.e.
\begin{equation} \label{eq:affinisation_Y}
	\normalfont
	\begin{aligned}
	\widetilde{E}_0^\text{sp} & = \sum_{i=1}^N Y_i^{-1} \, k^{-1}_1 \cdots k^{-1}_{i-1} \,  \sigma^-_i \, , \\
	\widetilde{F}_0^\text{sp} & = \sum_{i=1}^N Y_i \ \sigma^+_i \,  k_{i+1} \cdots k_N \, , 
	\end{aligned}
	\qquad\quad
	\widetilde{\!K}_0^\text{sp} = K_1^{\text{sp}\mspace{1mu}-1} \, .
\checkedMma
\end{equation}
Since $\widetilde{L}_a(u)$ preserves the physical space all Chevalley generators do so too.
\end{proposition}

For practical purposes the preceding description of the $\widehat{\mathfrak{U}}$-action on the physical space is rather cumbersome. It will be much more convenient to work with the polynomial action it induces using the trick from the end of \textsection\ref{s:physical_space}. First we consider the presentation by Chevalley generators from Proposition~\ref{prop:chevalley_tilde}. To avoid a proliferation of subscripts let us write
\begin{equation} \label{eq:Uqaff_induced_polynomial}
	\begin{aligned}
	E_M^\text{pol} & \coloneqq (E_1)_{M,M-1}^\text{pol} \, , \qquad & \widehat{E}_M^\text{pol} & \coloneqq (\widetilde{E}_0)_{M,M+1}^{\text{pol}} \, , \\
	F_M^\text{pol} & \coloneqq (F_1)_{M,M+1}^\text{pol} \, , & \widehat{F}_M^\text{pol} & \coloneqq (\widetilde{F}_0)_{M,M-1}^\text{pol} \, , \\
	K_M^\text{pol} & \coloneqq (K_1)_{M,M}^\text{pol} \, , & \,\widehat{\!K}_{\!M}^\text{pol} & \coloneqq (\,\widetilde{\!K}_0)_{M,M}^\text{pol} \, .
	\end{aligned}
\end{equation}

\begin{proposition} \label{prop:chevalley_pol}
These Chevalley generators induce the following action on polynomials:
\begin{equation} \label{eq:Uqsl2_pol}
	\normalfont
	\begin{aligned}
	K_M^\text{pol} & = t^{(N-2M)/2} \, , \\
	E_0^\text{pol} = 0 \, , \qquad 
	E_M^{\text{pol}} & = t^{(1-M)/2} \sum_{j=M}^N t^{(j-M)/2} \ T^{\text{pol}}_{(j,j-1 \cdots M)} \, , \\
	F_M^{\text{pol}} & = t^{(M+1-N)/2} \sum_{i=1}^{M+1} \! t^{(M+1-i)/2} \ T^{\text{pol}}_{(i,i+1 \cdots M+1)} \, , \qquad
	F_N^\text{pol} = 0 \, ,
	\end{aligned}
\checkedMma
\end{equation}
and
\begin{equation} \label{eq:affinisation_pol} 
	\normalfont
	\begin{aligned} 
	\widehat{\!K}_{\!M}^\text{pol} & = t^{-(N-2M)/2} \, , \\
	\widehat{E}_M^{\text{pol}} & = t^{M/2} \, \Biggl( \sum_{i=1}^{M+1} \! t^{-(M+1-i)/2} \, T^{\text{pol}\,-1}_i \cdots T^{\text{pol}\,-1}_M \Biggr) Y_{M+1}^{-1} \, , 
	\qquad \widehat{E}_N^{\text{pol}}  = 0 \, , \\
	\widehat{F}_0^{\text{pol}} = 0 \, , \qquad
	\widehat{F}_M^{\text{pol}} & = t^{(N-M)/2} \, \Biggl( \sum_{j=M}^N t^{-(j-M)/2} \, T^{\text{pol}\,-1}_{j-1} \cdots T^{\text{pol}\,-1}_M \Biggr) Y_M \, .
	\end{aligned}
\checkedMma
\end{equation}
\end{proposition}
\noindent Up to normalisation the generators of $\mathfrak{U} \subset \widehat{\mathfrak{U}}$, are just partial Hecke symmetrisers, cf.~\eqref{eq:Hecke_projectors}, ensuring that the resulting polynomial has the correct symmetry type. (Note that $j-M = \ell(j,j-1,\To,M)$ and $M+1-i = \ell(i,i+1,\To,M+1)$.) 

The affine generators also involve projectors onto the right symmetry type, now by $t^{-1}$-symmetrising. Indeed, the sums in $\widehat{E}_M^{\text{pol}}$ and $F_M^{\text{pol}}$, and those in $\widehat{F}_M^{\text{pol}}$ and $E_M^{\text{pol}}$, are related by the bar involution $t\mapsto t^{-1}$ and $T_w \mapsto T^{-1}_{w^{-1}}$ of the Hecke algebra~\cite{KL_79}.

\begin{proof}[Proof of Proposition~\ref{prop:chevalley_pol}]
The diagonal operators immediately follow from the definition \eqref{eq:weight_decomp} of the $M$-particle sector and the `level-zero' condition $K_{\!M}^\text{pol} \, \widehat{\!K}_{\!M}^\text{pol}=1$.

The matrix element $\cbra{1,\To,M-1} \, E_1^\text{sp} \, \cket{i_1,\To,i_M}$ can only be nonzero if $i_m = m$ for all $m<M$. Denoting $j = i_M$ we find from \eqref{eq:Uqsl2_spin}
\begin{equation*}
	\cbra{1,\To,M-1} \, E_1^\text{sp} \, \cket{1,\To,M-1,j} = t^{(j-2M+1)/2} \, , \qquad j\geq M \, .
\end{equation*}
This gives the expression for $E_M^\text{pol}$ since $T^\text{pol}_{\{1,\To,M-1,j\}} = T^\text{pol}_{(j,j-1 \cdots M)}$.

Next, $\cbra{1,\To,M+1} \, F_1^\text{sp} \, \cket{i_1,\To,i_M}$ survives precisely if
$\{i_1,\To,i_M \} \subset \{ 1,\To,M+1\}$. Writing $i$ for the element in $\{ 1,\To,M+1\} \setminus \{i_1,\To,i_M \}$ we obtain
\begin{equation*}
	\cbra{1,\To,M+1} \, F_1^\text{sp} \, \cket{1,\To \widehat\imath \To,M+1} = t^{-(N-2(M+1)+i))/2} \, , \qquad i \leq M+1 \, ,
\end{equation*}
where the caret denotes omission. But $T^\text{pol}_{\{1,\To \, \hat\imath\, \To,M+1\}} = T^\text{pol}_i \cdots T^\text{pol}_M = T^\text{pol}_{(i,i+1,\To,M+1)}$.

With the help of the \textsc{aha} relations~\eqref{eq:AHA} one likewise obtains~\eqref{eq:affinisation_pol}.
\end{proof}

By construction these polynomial operators obey the relations of $\widehat{\mathfrak{U}}$, though direct verification of most relations is tedious.
The expressions \eqref{eq:Uqsl2_pol}--\eqref{eq:affinisation_pol} can be simplified drastically using the symmetry of $\mathbb{C}[\vect{z}]_M$.

\begin{proposition} \label{prop:chevalley_pol_alt}
The nontrivial Chevalley generators from Proposition~\ref{prop:chevalley_pol} reduce to
\begin{equation} \label{eq:Uqsl2_pol_alt}
	\normalfont
	\begin{aligned} 
	E_M^\text{pol} & = t^{(N-2M+1)/2} \, \Biggl(\, \sum_{j=M}^N \! A_{-j}(\vect{z}) \, s_{Mj} \Biggr) \prod_{m=1}^{M-1} \! f_{mM} \, , \\
	F_M^\text{pol} & = t^{-(N-2(M+1)+1)/2} \, \Biggl(\, \sum_{i=1}^M \! s_{i,M+1} \Biggr) \prod_{m=1}^M \! f_{M+1,m} \; ,
	\end{aligned}
\checkedMma
\end{equation}
with $A_{-j}$ as in \eqref{eq:spin_D_-1} and $f_{ij} = 1/a(z_i/z_j)$ from \eqref{eq:f,g}, along with
\begin{equation} \label{eq:affinisation_pol_alt}
	\normalfont
	\begin{aligned}
	\widehat{E}_M^\text{pol} & = t^{-(N-2(M+1)+1)/2} \, \Biggl(\, \sum_{i=1}^M \! s_{i,M+1} \Biggr) \Biggl(\,\prod_{m=1}^M \! f_{m,M+1}\Biggr) \hat{q}_{M+1}^{-1} \, , \\
	\widehat{F}_M^\text{pol} & = t^{(N-2M+1)/2} \, \Biggl(\, \sum_{j=M}^N \! A_j(\vect{z}) \, s_{Mj} \Biggr) \Biggl(\, \prod_{m=1}^{M-1} \! f_{Mm} \Biggr) \hat{q}_M \, .
	\end{aligned}
\checkedMma
\end{equation}
where $A_j(\vect{z})$ is as in \eqref{eq:spin-D_1}. 
\end{proposition}

\begin{proof}
We use a symmetry argument like in \textsection\ref{s:Macdonald}. By \eqref{eq:Hecke_pol_2b} we can write
\begin{equation*}
	E_M^\text{pol} = \sum_{j=M}^N \! c_j(\vect{z}) \, s_{M j} 
\end{equation*}
for some coefficients $c_j(\vect{z})$ that we wish to determine. It is easy to find the coefficient with $j=N$, for which the only contribution comes from the term $j=N$ in \eqref{eq:Uqsl2_pol_alt}. Since
\begin{equation*}
	\begin{aligned}
	T^\text{pol}_{(N,N-1,\To,M)} & = (a_{N-1,N} \, s_{N-1} + b_{N-1,N}) \cdots (a_{M,M+1} \, s_M + b_{M,M+1}) \\ 
	& = a_{N-1,N} \, s_{N-1} \cdots a_{M,M+1} \, s_M + \cdots \\
	& = a_{N-1,N} \cdots a_{M,N} \, s_{M,N} + \text{contributions to other terms}
	\end{aligned}
\end{equation*}
we find $c_N(\vect{z}) = t^{(N-2M+1)/2} \, A_{-N}(\vect{z}) \, f_{1N} \cdots f_{M-1,N}$. By symmetry in the range $\mathbb{C}[\vect{z}]_{M-1}$ of $E_M^\text{pol}$ this determines the other coefficients via $c_k(\vect{z}) = s_{kN} \, c_N(\vect{z}) \, s_{kN}$.

One likewise determines the coefficients in
\begin{equation*}
	F_M^\text{pol} = \sum_{i=1}^{M+1} \! c'_i(\vect{z}) \, s_{i,M+1}
\end{equation*}
from the case $i=1$. This establishes \eqref{eq:Uqsl2_pol_alt}.

For the affine generators we first use \eqref{eq:Yi} to compute
\begin{equation*}
	\begin{aligned}
	\widehat{E}_M^{\text{pol}} & = \Biggl( \sum_{i=1}^{M+1} \! t^{(i-1)/2} \, T^\text{pol}_{i-1} \cdots T^\text{pol}_1 \Biggr) \pi^{-1} \, T^{\text{pol}\,-1}_{N-1} \cdots T^{\text{pol}\,-1}_{M+1} \\
	& = t^{-(N-(M+1))/2} \Biggl( \sum_{i=1}^{M+1} \! t^{(i-1)/2} \, T^\text{pol}_{i-1} \cdots T^\text{pol}_1 \Biggr) \, s_1 \cdots s_M \, \hat{q}_{M+1}^{-1} \, , \\
	\end{aligned}
\end{equation*}
where we use that $T_i^{\text{pol}\,-1} = t^{-1/2}$ and $s_i = 1$ for $i>M$ on $\widetilde{\Psi}(\vect{z}) \in \mathbb{C}[\vect{z}]_M$. At this point we can proceed as before, where the coefficient of $s_{i,M+1}$ with $i=M+1$ is easily found.

We similarly rewrite
\begin{equation*}
	\begin{aligned}
	\widehat{F}_M^{\text{pol}} & = \Biggl( \sum_{j=M}^N \! t^{(N-j)/2} \, T^\text{pol}_j \cdots T^\text{pol}_{N-1} \Biggr) \pi \, T^{\text{pol}\,-1}_1 \cdots T^{\text{pol}\,-1}_{M-1} \\
	& = t^{(1-M)/2} \Biggl( \sum_{j=M}^N \! t^{(N-j)/2} \, T^\text{pol}_j \cdots T^\text{pol}_{N-1} \Biggr) \, s_{N-1} \cdots s_M \, \hat{q}_M \, .
	\end{aligned}
\end{equation*}
Here the simple coefficient is that of $s_{Mj}$ with $j=M$.
\end{proof}

Next we turn to the generators of $\widehat{\mathfrak{U}}$ by the `quantum operators' obtained from $\widetilde{L}_a(u)$ as in \eqref{eq:ABCD}. To highlight the origin of the following expressions let us denote the entries of the \text{R}-matrix as the weights of the asymmetric six-vertex model,
\begin{equation}
	R(u) = P \, \check{R}(u) = 
	\begin{pmatrix} 
	\, 1 & \color{gray!80}{0} & \color{gray!80}{0} & \color{gray!80}{0} \, \\
	\, \color{gray!80}{0} & b_+(u) & c_-(u) & \color{gray!80}{0} \, \\
	\, \color{gray!80}{0} & c_+(u) & b_-(u) & \color{gray!80}{0} \, \\
	\, \color{gray!80}{0} & \color{gray!80}{0} & \color{gray!80}{0} & 1 \, \\
	\end{pmatrix}
	\, , \qquad
	\begin{aligned}
	b_\pm(u) & = f(u) \, , \\
	c_+(u) & = u \, g(u) \, , \\
	c_-(u) & = g(u) \, . 
	\end{aligned}
\end{equation}
\begin{proposition} \label{prop:ABCD_pol}
The polynomial action induced by the quantum operators is
\begin{equation} \label{eq:ABCD_Y_pol}
	\normalfont
	\begin{aligned}
	\widetilde{A}^\text{pol}_{M,M}(u) & = \prod_{m=1}^M \! b_+(u\,Y_m) \, , \\
	\widetilde{B}^\text{pol}_{M,M+1}(u) & = \sum_{i=1}^{M+1} c_-(u\,Y_i) \, \Biggl( \, \prod_{j(>i)}^{M+1} \! b_+(u\,Y_j) \Biggr) \, T^\text{pol}_{(i,i+1 \To M+1)} \, , \\
	\widetilde{C}^\text{pol}_{M,M-1}(u) & = \sum_{j=M}^N
	\Biggl( \prod_{m=1}^{M-1} \! b_+(u\,Y_m) \Biggr) \, c_+(u\,Y_j) \, \Biggl( \, \prod_{k(>j)}^N \! b_-(u\,Y_k) \Biggr) \, T^\text{pol}_{(j,j-1 \To M)} \, , \\ 
	\widetilde{D}^\text{pol}_{M,M}(u) & = \! \prod_{i(>M)}^N \!\!\! b_-(u \, Y_i) + \sum_{m=1}^M \sum_{j(>M)}^N \!\! c_-(u\,Y_m) \, \Biggl( \, \prod_{m'(>m)}^M \!\!\!\! b_+(u\,Y_{m'}) \Biggr) \\ 
	& \hphantom{ = \! \prod_{i(>M)}^N \!\! b_-(u \, Y_i) + \sum_{m=1}^M \sum_{j(>M)}^N \!\! } \times c_+(u\,Y_j) \, \Biggl( \, \prod_{k(>j)}^N \!\! b_-(u\,Y_k) \Biggr) \, T^\text{pol}_{\{1, \To \, \widehat{m} \, \To,M-1,j\}} \, .
	\end{aligned}
\end{equation}
\end{proposition}
\begin{proof}[Proof (sketch)] 
The proof of the expression for $\widetilde{A}$ is easy: since $\bra{\uparrow\uparrow} \, R(u) = \bra{\uparrow\uparrow}$ we have 
\begin{equation*}
	\begin{aligned}
	\cbra{1,\To,M} \, \widetilde{A}(u) & = \bra{\underset{a\vphantom{1}}{\uparrow}\underset{1}{\downarrow}\To \underset{M}{\downarrow} \uparrow \To \underset{N}{\uparrow}} \, R_{aN}(u\,Y_N) \cdots R_{aM}(u\,Y_M) \cdots R_{a1}(u\,Y_1) \, \ket{\underset{a\vphantom{1}}{\uparrow}} \\
	& = \bra{\underset{a\vphantom{1}}{\uparrow}\underset{1}{\downarrow}\To \underset{M}{\downarrow} \uparrow \To \underset{N}{\uparrow}} \, R_{aM}(u\,Y_M) \cdots R_{a1}(u\,Y_1) \, \ket{\underset{a\vphantom{1}}{\uparrow}} \\
	& = \mspace{7mu} \bra{\underset{1}{\downarrow}\To \underset{M}{\downarrow} \uparrow \To \underset{N}{\uparrow}} \prod_{m=1}^M \! b_+(u\,Y_m) \, ,
	\end{aligned}
\end{equation*}
where the last equality again follows from the ice rule (weight conservation) for the \textit{R}-matrix and from the presence of $\ket{\uparrow}$. The computation for $\widetilde{B}, \widetilde{C}$ parallels the computation yielding \eqref{eq:Uqsl2_pol}. As the result attests there are various contributions to take into account for $\widetilde{D}$. Since we will only use $\widetilde{A}$ later we omit the details; all of these expressions are readily obtained using standard graphical notation for $R(u)$, cf.~e.g.~\cite{Lam_14}. 
\end{proof}

\subsubsection{Explicit eigenvectors} \label{s:explicit_evrs_tilde}
The spin-Ruijsenaars model can be diagonalised following \cite{TU_97}. Since this is not our main topic we suffice with an example of some simple eigenspaces. For $M=0$ the space $\mathbb{C}[\vect{z}]_0 = \mathbb{C}[\vect{z}]^{\mathfrak{S}_N}$ consists of (completely) symmetric polynomials. Consider the basis of Macdonald polynomials $P_\lambda$ indexed by partitions $\lambda$ of length $\ell(\lambda) \leq N$. Each $\ket{P_\lambda(\vect{z})} = P_\lambda(\vect{z}) \, \cket{\varnothing}$ is an eigenvector of the spin-Macdonald operators (abelian symmetries) with eigenvalues as in \eqref{eq:Macd_eigenvalues}. It furthermore has pseudo highest weight (nonabelian symmetries). Indeed, $\widetilde{C}(u)$ acts by zero, while e.g.\ from~\eqref{eq:ABCD_Y_pol} $\widetilde{A}(u)$ acts by $1$ and $\widetilde{D}(u)$ by $t^{N/2} \, \Delta(-u)/\Delta(-t\,u)$, whose value on $P_\lambda$ follows from \eqref{eq:nonsymm_Macdonalds} or \eqref{eq:Macd_eigenvalues}. As a consistency check we note that \eqref{eq:qdet_ABCD} yields \eqref{eq:qdet_Y}. The Drinfeld polynomial is $\Delta(-u) = \prod_{i=1}^N (1 - t^{(N-2i+1)/2} \, q^{\lambda_i} \, u)$, cf.~\eqref{eq:intro_ABCD_on_shell}; for generic $q$ the number of $t$-strings\,---\,whence in particular the dimension, cf.~the paragraph following \eqref{eq:Drinfeld_polyn}\,---\,depends on the number of repetitions in $\lambda$.

\subsection{Freezing} \label{s:freezing} The spin-Macdonald operators~\eqref{eq:spin_D_r} still act nontrivially on polynomials through the difference operators~$\hat{q}_i$. To extract a spin chain we proceed along the lines of Uglov~\cite{Ugl_95u}, who in turn followed Talstra and Haldane~\cite{TH_95}. In the `static limit' $q\to 1$ the kinetic ($q$-difference) part of the Hamiltonian is suppressed with respect to the potential energy. By Appendix~\ref{s:app_CalSut_limit} we have $q = t^{\hbar/k}$ in terms of the physical interpretation as a quantum many body system. The limit $q\to1$ can thus be viewed either as letting $k \to \infty$ or as the \emph{classical} limit~$\hbar\to 0$. The dynamical (polynomial) and spin parts are treated differently, however; the latter will remain fully quantum mechanical. The physical picture is that the particles moving on the circle come to a halt and `freeze', so that only the spin interactions remain. The idea of \emph{freezing} can already be found in \cite{Sha_88}, and was worked out more concretely in \cite{Pol_93}.

At the point $q=1$ the spin-Macdonald operators become trivial: by \eqref{eq:D_r_q=1}, the `classical' spin-Macdonald operators just are (Gaussian binomial) multiples of the identity,
\begin{equation} \label{eq:Macdonald_gen_fn_class_pt}
	\widetilde{\Delta}(u)\big|_{q=1} = \Delta(u)\big|_{q=1} = \prod_{i=1}^N \bigl(1 + t^{(N-2\mspace{1mu}i+1)/2} \, u \bigr) = \sum_{r=0}^N \, \begin{bmatrix} N \\ r \end{bmatrix} \, u^r \, .
\end{equation}
We therefore have to consider a neighbourhood of the classical point and linearise in $q$, i.e.\ take the \emph{semiclassical} limit. 
\begin{definition} 
We will write
\begin{equation} \label{eq:semiclass_expansion}
	\widetilde{O}^\circ \coloneqq \bigl. \widetilde{O} \bigr|_{q=1} \, , \quad \
	\delta \widetilde{O} \coloneqq \frac{\partial \widetilde{O}}{\partial \mspace{1mu} q}\bigg|_{q=1} \qquad 
	\text{so} \qquad 
	\widetilde{O} = \widetilde{O}^\circ + (q-1) \, \delta \widetilde{O} + \mathcal{O}(q-1)^2 \, .
\end{equation}
The physical condition is independent of $q$ so both $\widetilde{O}^\circ$ and $\delta \widetilde{O}$ are physical operators.
\end{definition}

\subsubsection{Abelian spin-chain symmetries} \label{s:abelian_freezing}
To start let us focus on the spin-Macdonald operator~$\widetilde{D}_1$ from~\eqref{eq:spin-D_1}. 
Taylor expanding at $q= 1$ gives
\begin{equation*}
	\widetilde{D}_1 = [N] + (q-1)\, \delta\widetilde{D}_1 + \mathcal{O}(q-1)^2 \, .
\end{equation*}
\begin{lemma} \label{lem:decoupling}
The semiclassical limit of the first spin-Macdonald operator takes the form 
\begin{equation} \label{eq:H_tilde_expansion}
	\normalfont
	\delta \widetilde{D}_1 = \sum_{j=1}^N A_j(\vect{z}) \, z_j \, \partial_{z_j} + (t^{1/2}-t^{-1/2}) \, \sum_{j=1}^N A_j(\vect{z}) \sum_{i=1}^{j-1} V\mspace{-1mu}(z_i,z_j) \, S_{[i,j]}^{\mspace{1mu}\textsc{l}} \, ,
\end{equation}
where the spin part features the potential~\eqref{eq:pot} and long-range interactions~\eqref{eq:Sij_left},
\begin{equation*}
	\normalfont
	S_{[i,j]}^{\mspace{1mu}\textsc{l}} = \check{R}_{(i+1,\To,j-1,j)}^{-1} \, e_i^\text{sp} \, \check{R}_{(i+1,\To,j-1,j)} \, .
\end{equation*}
\end{lemma}
\noindent The decoupling between kinetic and spin terms in \eqref{eq:H_tilde_expansion} was observed in~\cite{TH_95,Ugl_95u}.

\begin{proof}
The decoupling is a simple consequence of our expression~\eqref{eq:spin-D_1} for $\widetilde{D}_1$; cf.~the sketch of the proof in \textsection\ref{s:intro_freezing} in terms of the graphical notation. Write the summand of \eqref{eq:spin-D_1} as
\begin{equation*} 
	A_j(\vect{z}) \, \check{R}_{j-1,j}(z_j/z_{j-1}) \cdots \check{R}_{12}(z_j/z_1) \, \check{R}_{12}(z_1/q\,z_j) \cdots \check{R}_{j-1,j}(z_{j-1}/q\,z_j) \, \hat{q}_j \, .
\end{equation*}
We'll show that its linearisation in~$q$ is the summand of \eqref{eq:H_tilde_expansion}.

When the derivative hits $\hat{q}_j = 1 + (q-1) \, z_j \,\partial_{z_j} + \mathcal{O}(q-1)^2$ all \textit{R}-matrices, now at $q=1$, cancel pairwise by unitarity~\eqref{eq:R_relations}. This yields the first term in \eqref{eq:H_tilde_expansion}. 

By the Leibniz rule the derivative of the spin part produces a sum over $i(<j)$. Consider the term where $\delta$ hits the $i$th \textit{R}-matrix that was affected by $\hat{q}_j$,
\begin{equation*}
	\delta \check{R}_{i,i+1}(z_i/q \, z_j) = -\frac{z_i}{z_j} \, \check{R}'_{i,i+1}(z_i/z_j) \, . 
\end{equation*} 
The \textit{R}-matrices to its left again cancel in pairs by unitarity. The derivate of the \textit{R}-matrix
can be easily evaluated using \eqref{eq:baxterisation} (cf.~the `change of variables' in \cite{Lam_18}):
\begin{equation*}
	\check{R}'_{i,i+1}(z_i/z_j) = -f'(z_i/z_j) \, e^\text{sp}_i \, ,
\end{equation*}
which allows us to recognise $e^\text{sp}_i = -f'(1)^{-1} \check{R}'_{i,i+1}(1) = -(t^{1/2}-t^{-1/2})\,\check{R}'_{i,i+1}(1)$ from \eqref{eq:R'(1)}. By \eqref{eq:baxterisation} and the Temperley--Lieb relation~\eqref{eq:TL} we moreover have
\begin{equation*}
	\check{R}_{i,i+1}(z_j/z_i) \, e^\text{sp}_i = (1 - [2]\,f_{ji}) \, e^\text{sp}_i = -\frac{f_{ji}}{f_{ij}} \, e^\text{sp}_i .
\end{equation*}
Hence
\begin{equation} \label{eq:pot_origin}
	\begin{aligned}
	\check{R}_{i,i+1}(z_j/z_i) \, \delta \check{R}_{i,i+1}(z_i/q\,z_j) & = (-z_i/z_j)\,(-f'_{ij})\,(-f_{ij}/f_{ji}) \, e^\text{sp}_i \\
	& = (t^{1/2}-t^{-1/2})\,V\mspace{-1mu}(z_i,z_j) \, e^\text{sp}_i \, .
	\end{aligned}
\end{equation}
The remaining \textit{R}-matrices combine to give the long-range spin interaction~\eqref{eq:Sij_left}. 
\end{proof}

Note that the physical space~\eqref{eq:phys_space} is not affected by the limit $q\to 1$; in particular Proposition~\ref{prop:physical_vectors} remains valid.

It remains to get rid of the kinetic term in~\eqref{eq:H_tilde_expansion} to get an operator that can be viewed as acting on the spin-chain Hilbert space~$\mathcal{H}$. Uglov argued as follows. Consider the abelian symmetries, i.e.\ the tower of higher Hamiltonians generated by $\widetilde{\Delta}(u)$. Following~\cite{TH_95,Ugl_95u} we expand the commutation relation~\eqref{eq:Delta} around $q=1$. Dropping all commutators with the constants~\eqref{eq:Macdonald_gen_fn_class_pt} the first nontrivial relations appear at quadratic order~\cite{Ugl_95u}:
\begin{equation} \label{eq:semiclassical_symmetries}
	0 = [\widetilde{\Delta}(u),\widetilde{\Delta}(v)] = (q-1)^2 \, [\delta\widetilde{\Delta}(u) , \delta\widetilde{\Delta}(v)] + \mathcal{O}(q-1)^3 \, .
\end{equation}
That is, the abelian symmetries~\eqref{eq:spin_D_r} survive at the semiclassical level. In \textsection\ref{s:abelian_tilde} we already noted that one of these is particularly simple: the total degree operator
\begin{equation} \label{eq:total_degree_semiclassical}
	\widetilde{D}_N = \hat{q}_1 \cdots \hat{q}_N = 1 + (q-1) \sum_{j=1}^N z_j \, \partial_{z_j} \, + \mathcal{O}(q-1)^2 \, .
\end{equation}
Happily, the commutation \eqref{eq:semiclassical_symmetries} implies that we may modify $\delta \widetilde{D}_1$ from~\eqref{eq:H_tilde_expansion} by adding any multiple of the total degree operator~$\delta \widetilde{D}_N$. The result still acts on the physical space $\widetilde{\mathcal{H}}$ and commutes with all other operators in the expansion of $\widetilde{\Delta}(u)$. In this way we can get rid of the derivates in $\delta \widetilde{D}_1$ provided we can make all their coefficients $A_j(\vect{z})$ equal. 

This is where the evaluation comes in: we need to find a value for $\vect{z}$ where the $A_j(\vect{z})$ become independent of the value of~$j$~\cite{Ugl_95u}. Solving $A_1(\vect{z}) = \cdots = A_N(\vect{z})$ for the coordinates yields $\vect{z} = z_1\,(1,\omega,\cdots,\omega^{N-1})$, or any permutation thereof, for $\omega\coloneqq \E^{2\pi\I/N}$. These are precisely the stationary (equilibrium) positions, cf.~e.g.~\textsection5.2 in \cite{Rui_95}, of the trigonometric classical Ruijsenaars--Schneider model~\cite{RS_86},
with constant centre-of-mass (angular) momentum. Omitting the latter we come to the following
\begin{definition} 
Define the \emph{evaluation} (specialisation) map $\ev \colon \widetilde{\mathcal{H}} \longrightarrow \mathcal{H}$ as in \eqref{eq:ev}, and for a physical operator $\widetilde{O}$ by 
\begin{equation} \label{eq:ev_operator}
	\ev(\widetilde{O}) \, \ev = \ev \circ \, \widetilde{O} \, .
\end{equation}
Let us denote equality upon evaluation, or \emph{on-shell} equality, by
\begin{equation*}
	\widetilde{O}_1 \overset{\text{ev}}{=} \widetilde{O}_2 \, \qquad \text{as shorthand for } \qquad \ev \widetilde{O}_1 = \ev \widetilde{O}_2 \, .
\end{equation*}
\end{definition}

Since $\sum_j A_j(\vect{z}) = \widetilde{D}_1^\circ$ the common value is
\begin{equation} \label{eq:ev_Aj}
	A_j(\vect{z}) \, \overset{\mathrm{ev}}{=} \, \frac{\widetilde{D}_1^\circ}{N} = \frac{[N]}{N} \, .
\end{equation}
Thus we are finally led to the Hamiltonian~\eqref{eq:ham_left}: by construction,
\begin{equation} \label{eq:ham_left_from_freezing}
	\begin{aligned}
	\frac{1}{t^{1/2}-t^{-1/2}} \, \biggl(\delta \widetilde{D}_1 - \frac{[N]}{N} \, \delta \widetilde{D}_N \biggr) \ & \overset{\text{ev}}{=} \ \sum_{j=1}^N A_j(\vect{z}) \sum_{i=1}^{j-1} V\mspace{-1mu}(z_i,z_j) \, S_{[i,j]}^{\mspace{1mu}\textsc{l}} \\
	& \overset{\text{ev}}{=} \ \frac{[N]}{N} \sum_{i<j}^N V\mspace{-1mu}(z_i,z_j) \, S_{[i,j]}^{\mspace{1mu}\textsc{l}} = \widetilde{H}^\textsc{l}
	\end{aligned}
\end{equation}
acts nontrivially on spins only while preserving the physical space. Here we removed a factor of $t^{1/2}-t^{-1/2}$ to ensure the isotropic limit $t \to 1$ is nontrivial. We have arrived at the \textsf{q}-deformed Haldane--Shastry spin chain.

As a corollary we readily obtain the possible eigenvalues of \eqref{eq:ham_left_from_freezing} from those of the \mbox{(spin-)}Macdonald operators. This gives Theorem~\ref{thm:intro_energy_left}~(i):
\begin{proposition}[\cite{Ugl_95u}] \label{prop:energy_left}
Any eigenvalue of \eqref{eq:ham_left_from_freezing} can be written as $\normalfont E^\textsc{l}(\mu)$ from \eqref{eq:energy_left}.
\end{proposition}
\begin{proof}
We use the expression \eqref{eq:ham_left_from_freezing} in terms of the Hamiltonian in terms of (symmetric combinations of) the \textit{Y}$\!$-operators, whose eigenvalues we know (\textsection\ref{s:Macdonald}). Let $\lambda$ be any partition with $\ell(\lambda)\leq N$. By adding a string of zeros at the end if necessary we can view $\lambda$ as a weak partition with $N$~parts. By \eqref{eq:Macd_eigenvalues} the eigenvalues of $\widetilde{D}_1$ and $\widetilde{D}_N$ are given by
\begin{equation*}
	\Lambda_1(\lambda) = \sum_{i=1}^N t^{(N-2 \, i+1)/2} \, q^{\lambda_i} \, , \qquad\qquad \Lambda_N(\lambda) = q^{|\lambda|} \, , \quad |\lambda| \coloneqq \sum_{i=1}^N \lambda_i \, .
\end{equation*} 
The eigenvalues of the frozen Hamiltonian~\eqref{eq:ham_left_from_freezing} follow by linearisation. The crucial step is to recognise contributions of $M \coloneqq \lambda_1$ separate magnons. This goes as follows~\cite{Ugl_95u}. The linear part in $q$ of the eigenvalue of $\widetilde{D}_1$ is
\begin{equation} \label{eq:energy_identity}
	\begin{aligned}
	\delta \Lambda_1(\lambda) = \sum_{i=1}^N \lambda_i \, t^{(N-2\,i+1)/2} & = \sum_{i=1}^N \sum_{m=1}^{\lambda_i} t^{(N-2\,i+1)/2} \\
	& = \sum_{m=1}^M \sum_{i=1}^{\lambda'_m} t^{(N-2\,i+1)/2} = \sum_{m=1}^M t^{(N-\lambda'_m)/2}\, [\lambda'_m] \, .
	\end{aligned}
\end{equation}
In the second equality we reinterpret the sum on the first line as a double sum with one term for each box in the Young diagram of~$\lambda$, contributing $t^{(N-2\,i+1)/2}$ for each of the $\lambda_i$ boxes in row~$i$. In the second line we perform the sum per column instead to pass to the conjugate partition~$\lambda'$, with $\ell\bigl(\lambda'\bigr) = M$. In the final equality we summed a geometric progression. 
Combining this with $\delta \Lambda_N(\lambda) = |\lambda| = |\lambda'|$ we obtain
\begin{equation*}
	\frac{1}{t^{1/2}-t^{-1/2}} \, \biggl(\delta \Lambda_1(\lambda) - \frac{[N]}{N} \, \delta \Lambda_N(\lambda) \biggr) = \frac{1}{t^{1/2}-t^{-1/2}} \, \sum_{m=1}^M \biggl( t^{(N-\lambda'_m)/2}\, [\lambda'_m] \, - \frac{[N]}{N} \,\lambda'_m \biggr) \, .
\end{equation*}
Upon renaming $\mu_m \coloneqq \lambda'_{M-m+1}$ we arrive at Uglov's expression for $E^\textsc{l}(\mu)$.
\end{proof}

It remains to show which of the above possible eigenvalues actually occur. In \textsection\ref{s:explicit_evrs} we will prove that the eigenspace $\mathcal{H}^\mu$ is nontrivial if $\mu \in \mathcal{M}_N$ is a motif by explicitly constructing the corresponding pseudo highest-weight vector. The result will be the wave functions from \textsection\ref{s:intro_explicit_evrs}. We have not yet found a satisfactory way to verify that its energy is given by $E^\textsc{l}(\mu)$ by direct computation, except in special cases; cf.~the remarks on p.\,\pageref{rmk:direct_verification}.

Next we turn to the higher spin-chain Hamiltonians from Theorem~\ref{thm:freezing_abelian}. First of all we observe that continuing the expansion~\eqref{eq:semiclassical_symmetries} gives two nontrivial commutators at cubic order in $q-1$, and so on, so we are \emph{not} guaranteed to get any further symmetries of $\widetilde{D}_1$ at higher order in the expansion. The abelian symmetries of the spin chain instead just arise by freezing the higher spin-Macdonald operators.

\begin{proof}[Proof of Theorem~\ref{thm:ham_energy_right}]
The simplest higher spin-Macdonald operator is $\widetilde{D}_{-1} = \widetilde{D}_N^{-1} \, \widetilde{D}_{N-1}$ from \eqref{eq:intro_spin-D_-1}. Note that when we push $\hat{q}_i^{-1}$ to the right the $q$ again appears in the \emph{denominator} of the arguments of the \textit{R}-matrices it has passed. The semiclassical limit is computed just as for \eqref{eq:H_tilde_expansion}. The result is
\begin{equation} \label{eq:ham_right_semiclassical}
	\delta \widetilde{D}_{-1} = - \sum_{i=1}^N A_{-i}(\vect{z}) \, z_i \, \partial_{z_i} + (t^{1/2}-t^{-1/2}) \sum_{i=1}^N A_{-i}(\vect{z}) \! \sum_{j=i+1}^N \!\! V\mspace{-1mu}(z_i,z_j) \, S_{[i,j]}^{\mspace{1mu}\textsc{r}} \, ,
\end{equation}
where
\begin{equation*}
	S_{[i,j]}^{\mspace{1mu}\textsc{r}} = \check{R}_{(j-1,\To,i+1,i)}^{-1} \, e_{j-1}^\text{sp} \, \check{R}_{(j-1,\To,i+1,i)} \, .
\end{equation*}
As $A_{-i}(\vect{z}) \overset{\mathrm{ev}}{=} [N]_{t^{-1/2}}/N = [N]_{t^{1/2}}/N$ we find the spin-chain Hamiltonian~\eqref{eq:ham_right},
\begin{equation} \label{eq:ham_right_from_freezing}
	\begin{aligned}
	\frac{1}{t^{1/2}-t^{-1/2}} \, \biggl(\delta \widetilde{D}_{-1} + \frac{[N]}{N} \, \delta \widetilde{D}_N \biggr) \ & \overset{\text{ev}}{=} \ \sum_{i=1}^N A_{-i}(\vect{z}) \! \sum_{j=i+1}^N \!\! V\mspace{-1mu}(z_i,z_j) \, S_{[i,j]}^{\mspace{1mu}\textsc{r}} \\
	& \overset{\text{ev}}{=} \ \frac{[N]}{N} \sum_{i<j}^N V\mspace{-1mu}(z_i,z_j) \, S_{[i,j]}^{\mspace{1mu}\textsc{r}} = \widetilde{H}^\textsc{r} \, .
	\end{aligned}
\end{equation}
Here we note that this is consistent with \eqref{eq:H_r} as $[N] - A_{-i}(\vect{z}) \overset{\mathrm{ev}}{=} [N]\,(N-1)/N$:
\begin{align*}
	\delta \widetilde{D}_{N-1} & = \widetilde{D}^\circ_{-1} \, \delta D_N + \widetilde{D}^\circ_N \,\delta \widetilde{D}_{-1} = [N] \, \delta D_N + \delta \widetilde{D}_{-1} \\
	& = \sum_{i=1}^N \bigl([N] - A_{-i}(\vect{z})\bigr) \, z_i \, \partial_{z_i} + (t^{1/2}-t^{-1/2}) \sum_{i=1}^N A_{-i}(\vect{z}) \! \sum_{j=i+1}^N \!\! V\mspace{-1mu}(z_i,z_j) \, S_{[i,j]}^{\mspace{1mu}\textsc{r}} \, . 
\end{align*}

The eigenvalues of $H^\textsc{r}$ follow from \eqref{eq:Macd_eigenvalues} using \eqref{eq:spin_D_symmetry}:
\begin{equation*}
	\Lambda_{-1}(\lambda) = \frac{\Lambda_{N-1}(\lambda)}{\Lambda_N(\lambda)} = \sum_{i=1}^N t^{-(N-2\mspace{1mu}i+1)/2} \, q^{-\lambda_i} \qquad \text{so} \qquad \delta \Lambda_{-1}(\lambda) = -\delta \Lambda_1(\lambda)\big|_{t \, \mapsto t^{-1}} \, .
\end{equation*}
Since $\widetilde{D}_{-1}^\circ = [N] = \widetilde{D}_1^\circ \, {} \big|_{t \, \mapsto t^{-1}}$ while $t^{1/2}-t^{-1/2}$ changes sign under inverting $t$ it follows that $E^\textsc{r}(\mu) = E^\textsc{l}(\mu)\big|_{t \, \mapsto t^{-1}}$.
\end{proof}

The other spin-chain Hamiltonians are similarly obtained from \eqref{eq:spin_D_r}:
\begin{proof}[Proof of Theorem~\ref{thm:freezing_abelian}]
It is clear that the kinetic and spin part decouple for any~$r$. To find the required multiple of $\delta \widetilde{D}_N$ needed to remove the kinetic part we compute 
\begin{equation*}
	\sum_{J \colon \# J = r} \!\!\!\!\! A_J(\vect{z}) \, \delta \hat{q}_J \, = \!\!\! \sum_{J \colon \# J = r} \Biggl( A_J(\vect{z}) \, \sum_{j \in J} z_j \, \partial_{z_j} \Biggr) = \, \sum_{j=1}^N \Biggl( \, \sum_{\substack{J \colon \# J = r \\ J \ni j}} \!\!\!\!\! A_J(\vect{z}) \Biggr) z_j \, \partial_{z_j} \, .
\checkedMma
\end{equation*}
The prudent generalisation of \eqref{eq:ev_Aj} is the identity, valid for any $j$,
\begin{equation}
	\sum_{\substack{J : \# J = r \\ J \ni j}} \!\!\!\!\! A_J(\vect{z}) \, \overset{\mathrm{ev}}{=} \, \frac{r}{N} \, \widetilde{D}_r^\circ = \frac{r}{N} \, \begin{bmatrix} N \\ r \end{bmatrix} \, , \qquad 1 \leq r \leq N \, .
\checkedMma
\end{equation}
In this way we obtain \eqref{eq:H_r}.

For the eigenvalues of these higher spin-chain Hamiltonians we use the following generalisation of \eqref{eq:energy_identity}, valid for any partition $\lambda$ with $\ell(\lambda)\leq N$,
\begin{align*}
	\delta \Lambda_r(\lambda) & = \sum_{j_1 < \cdots < j_r}^N \!\!\!\! (\lambda_{j_1}+\cdots + \lambda_{j_r}) \prod_{s=1}^r t^{(N-2\,j_s +1)/2} \\
	& = \sum_{m=1}^{\lambda_1} \sum_{s=1}^r (-1)^{s-1} \, \begin{bmatrix} N \\ r - s \end{bmatrix} \, t^{s\,(N-\lambda'_m)/2} \, \frac{[s\,\lambda'_m]}{[s]} \, .
\checkedMma
\end{align*} 
This yields the additive form~\eqref{eq:energy_r}.
\end{proof}

The only spin-Macdonald operator that does \emph{not} give rise to a spin-chain symmetry in this way is the (multiplicative) translation operator $\widetilde{D}_N = \hat{q}_1 \cdots \hat{q}_N$ from~\eqref{eq:total_degree_semiclassical}. Let us show that it nevertheless gives \textsf{q}-homogeneity on $\mathcal{H}$ in return, establishing Proposition~\ref{prop:intro_q-translation}~(i) from \textsection\ref{s:intro_abelian}:

\begin{proposition} \label{prop:q-translation}
If $\widetilde{O}$ is an operator on $\widetilde{\mathcal{H}}$ that commutes with $\widetilde{D}_N$ then the evaluation $O = \ev \widetilde{O}$ is \textsf{q}-homogeneous: $O = G\,O\,G^{-1}$ with $G$ the \textsf{q}-translation operator~\eqref{eq:q-translation}.
\end{proposition}
\begin{proof}
For any $\widetilde{O}$ acting on $\widetilde{\mathcal{H}}$ we have, cf.~the proof of \eqref{eq:spin_D_r},
\begin{equation*}
	\widetilde{O} = s^\text{tot}_{(N \cdots 21)} \, \widetilde{O} \, s^\text{tot}_{(12 \cdots N)} = s_{(N \cdots 21)} \, \widetilde{G} \, \widetilde{O} \, \widetilde{G}^{-1} \, s_{(12\cdots N)} \qquad \text{on} \quad \widetilde{\mathcal{H}} \, ,
\end{equation*}
where we used $\check{R}_{(N \cdots 21)} = \widetilde{G}$. Note that the conjugation by $s_{(12\cdots N)}$ just cyclically permutes the $z_j$ in $\widetilde{G}\,\widetilde{O}\,\widetilde{G}^{-1}$. However, commutation with the total degree operator~$\hat{q}_1 \cdots \hat{q}_N$ means that $\widetilde{O}$ is homogeneous of total degree zero in~$\vect{z}$, i.e.\ depends only on ratios of coordinates. The same holds for $\widetilde{G}$. Thus the cyclic permutation is invisible upon evaluation, and we conclude that $\widetilde{O} \overset{\mathrm{ev}}{=} \widetilde{G} \, \widetilde{O} \, \widetilde{G}^{-1}$.
\end{proof}
\noindent We recall that the second part of  Proposition~\ref{prop:intro_q-translation} was already demonstrated in \textsection\ref{s:intro_abelian}.

Besides all Hamiltonians obtained from $\widetilde{\Delta}(u)$ it follows that $\widetilde{L}_a(u)$ is \textsf{q}-homogeneous. The abelian symmetries are summarised in Table~\ref{tb:abelian_symmetries} on p.\,\pageref{tb:abelian_symmetries}.

\subsubsection{Nonabelian spin-chain symmetries} \label{s:nonabelian_freezing} 
Finally we turn to the nonabelian symmetries, which are generated by the monodromy matrix~\eqref{eq:L_tilde}. 

\begin{proof}[Proof of Theorem~\ref{thm:monodromy_frozen}]
This time we don't have to go far in the expansion~\cite{TH_95} as the zeroth order already gives a nontrivial operator:
\begin{subequations} \label{eq:L_spin_chain}
	\begin{gather}
	\widetilde{L}_a(u) = \widetilde{L}_a^\circ(u) + \mathcal{O}(q-1)^1 \, , \qquad \widetilde{L}_a^\circ(u) = R_{aN}(u\,Y_N^\circ) \cdots R_{a1}(u\,Y_1^\circ) \, .
\shortintertext{It is clear that this operator still obeys the \textit{RLL}-relations~\eqref{eq:RLL}. To check that it also remains a symmetry of the spin chain we expand \eqref{eq:commutator_L_Delta} like in \eqref{eq:semiclassical_symmetries}:}
	0 = \bigl[ \widetilde{L}_a(u) , \widetilde{\Delta}(v) \bigr] =  (q-1) \, \bigl[ \widetilde{L}_a^\circ(u) , \delta \widetilde{\Delta}(v) \bigr] + \mathcal{O}(q-1)^2 \, .
	\end{gather}
\end{subequations}
where use \eqref{eq:Macdonald_gen_fn_class_pt}. So the abelian symmetries remain $\widehat{\mathfrak{U}}$-invariant semiclassically. 
\end{proof}

The Chevalley generators are \eqref{eq:Uqsl2_spin} and \eqref{eq:affinisation_Y} with $Y_i \mapsto Y_i^\circ$. The induced polynomial action acquires a neat symmetric form: at $q=1$ \eqref{eq:affinisation_pol_alt} is related to \eqref{eq:Uqsl2_pol_alt} by
\begin{equation} \label{eq:Uhat_pol_t_inv}
	\begin{aligned}
	\widehat{E}_M^{\text{pol},\,\circ} & \propto F_M^\text{pol}\big|_{t\mapsto t^{-1}} \, , \\
	\widehat{F}_M^{\text{pol},\,\circ} & \propto E_M^\text{pol}\big|_{t\mapsto t^{-1}}\, ,
	\end{aligned}
\end{equation}
where the proportionality signs just mean that we ignore the prefactors in \eqref{eq:Uqsl2_pol_alt}--\eqref{eq:affinisation_pol_alt}.

By \eqref{eq:Macdonald_gen_fn_class_pt} the quantum determinant~\eqref{eq:qdet_Y} now becomes a true, $\vect{z}$-independent scalar:
\begin{equation*}
	\qdet_a \widetilde{L}^\circ_a(u) = t^{N/2} \, \frac{\Delta^{\!\circ}(-u)}{\Delta^{\!\circ}(-t\,u)} = t^{N/2} \, \frac{t^{(1-N)/2} \, u - 1}{t^{(1+N)/2} \, u - 1} \, .
\end{equation*}
This is one way to justify our detour through the dynamical model~\cite{TH_95}: the quantum determinant of $\widetilde{L}_a(u)$ was nontrivial from the polynomial perspective, making it a suitable candidate for generating nontrivial abelian symmetries.

\subsubsection{Explicit spin-chain eigenvectors} \label{s:explicit_evrs} 
Our final task is to construct eigenvectors of the spin-chain Hamiltonian. As usual we proceed per $M$-particle sector~$\mathcal{H}_M$, cf.~\eqref{eq:weight_decomp}. As in \cite{BG+_93} we will exploit the rich algebraic structure available off shell, i.e.\ prior to evaluation. 
\bigskip

\noindent\textbf{General considerations.} By Proposition~\ref{prop:physical_vectors} from \textsection\ref{s:physical_space} we may pass to the polynomial world and work with $\widetilde{\Psi}(\vect{z}) \in \mathbb{C}[\vect{z}]_M = \mathbb{C}[\vect{z}]^{\mathfrak{S}_M \times \mathfrak{S}_{N-M}}$ to diagonalise the Hamiltonian~\eqref{eq:ham_left_from_freezing}, viewed prior to evaluation as acting on polynomials. Then we embed the eigenfunctions in the physical space via \eqref{eq:phys_vector}, and finally evaluate to land in the $M$-particle sector of the spin chain.

The origin of Theorem~\ref{thm:phys_vector_spinchain} is the following.
\begin{proposition}\label{prop:phys_vector_freezing}
Any $M$-particle spin-chain eigenvector obtained by freezing is determined by a symmetric polynomial in just $M$ variables, $\widetilde{\Psi}(z_1,\To,z_M) \in \mathbb{C}[z_1,\To,z_M]^{\mathfrak{S}_M} \subset \mathbb{C}[\vect{z}]_M$, as
\begin{equation}
	\normalfont \sum_{i_1 < \cdots < i_M}^N \!\!\!\!\! \ev\Bigl( T_{\{i_1,\To,i_M\}}^\text{pol} \widetilde{\Psi}(z_1,\To,z_M) \Bigr) \, \cket{i_1,\To,i_M}  \, .
\end{equation}
\end{proposition}
\begin{proof}
Our starting point is Proposition~\ref{prop:physical_vectors}; we have to show that the polynomial may be taken to be independent of $z_{M+1},\To,z_N$. Consider the power-sum basis for $\mathbb{C}[\vect{z}]_M$, which is given by $p_{\lambda^{(1)}}(z_1,\To,z_M) \, p_{\lambda^{(2)}}(z_{M+1},\To,z_N)$ for two partitions with $\ell(\lambda^{(1)}) \leq M$ and $\ell(\lambda^{(2)}) \leq N-M$. Here $p_\lambda = \prod_{r \in \lambda} p_r$ and $p_r(\vect{z}) = \sum_i z_i^r$ as usual. Notice that 
\begin{equation} \label{eq:ev_powersum}
	\ev \, p_r(z_1,\To,z_N) = \sum_{i=1}^N \omega^{i\,r} = N \, \delta_{r, 0 \,\mathrm{mod}\,N} \, ,
\end{equation}
so $p_r(z_{M+1},\To,z_N) = p_r(z_1,\To,z_N) - p_r(z_1,\To,z_M) \overset{\mathrm{ev}}{=} -p_r(z_1,\To,z_M)$ for $0<r<N$. Hence on shell $p_{\lambda^{(2)}}(z_{M+1},\To,z_N) \overset{\mathrm{ev}}{=} (-1)^{\ell(\lambda^{(2)})} \, p_{\lambda^{(2)}}(z_1,\To,z_M)$ since $\lambda^{(2)}_1 < N$. In this way we land in $\mathbb{C}[z_1,\To,z_M]^{\mathfrak{S}_M}$. 
\end{proof}

Our ansatz will be that sufficiently many spin-chain eigenvectors are obtained in this way. The goal of this section will be to show that this is indeed the case and prove Theorem~\ref{thm:nice_polynomial}. Theorem~\ref{thm:phys_vector_spinchain} then follows from the nonabelian symmetries and the proof in \textsection\ref{s:hw}.

By \eqref{eq:ham_left_from_freezing} and \eqref{eq:ham_right_from_freezing} we will seek joint $\delta\widetilde{D}_r$-eigenvectors in $\widetilde{\mathcal{H}}_M$. To this end we may in fact work at the \emph{classical} level: the operators $Y_i^\circ = Y_i|_{q=1}$, which didn't play a role in \textsection\ref{s:abelian_freezing} due to \eqref{eq:Macdonald_gen_fn_class_pt}, will be pivotal for our diagonalisation. Indeed, although the sum $\sum_i Y_i^\circ = [N]$ is trivial, the individual terms certainly are not. The $Y_i^\circ$ do not preserve $\widetilde{\mathcal{H}}$, but as in \textsection\ref{s:Macdonald} we can first view the magnons as distinguishable particles to develop the nonsymmetric theory, and then \mbox{(\textsf{q}-)}symmetrise at the end. Crucially, at the intermediate step the $Y_i^\circ$ commute with the Hamiltonians: just as in \eqref{eq:L_spin_chain} we have
\begin{equation*}
	0 = \bigl[ Y_i , \widetilde{\Delta}(u) \bigr] =  (q-1) \, \bigl[ Y_i^\circ , \delta \widetilde{\Delta}(u) \bigr] + \mathcal{O}(q-1)^2 \, .
\end{equation*}
The $Y_i^\circ$ still form a commuting family of operators, and can be jointly diagonalised. At $q=1$ a part of the dependence on the partition drops out of \eqref{eq:nonsymm_Macdonalds}, but the joint spectrum remains multiplicity free when taking into account $\delta \widetilde{\Delta}(u)$. This passage to a classical spinless model is quite a simplification!

The evaluation further facilitates our task. Firstly, it gave us Proposition~\ref{prop:phys_vector_freezing}. Secondly, it allows us to restrict ourselves to polynomials with degree $<N$ in each variable. (Of course $\omega^N = 1$ will already play a role for lower powers of $z_j$ as $j$ increases, but since the generators of the \textsc{aha} preserve the total degree we should allow the maximal degree in any variable to equal to that for $z_1$, which by evaluation is $N-1$.)

The corresponding nonsymmetric theory ought to take place in $\mathbb{C}[z_1,\To,z_M] \subset \mathbb{C}[\vect{z}]$ (with degree $<N$ in each variable) and therefore involve the $Y_m^\circ$ with $1\leq m\leq M$. However, the latter are associated to $\widehat{\mathfrak{H}}_N$ and depend on all $N$ variables, so do not preserve the subspace $\mathbb{C}[z_1,\To,z_M] \subset \mathbb{C}[\vect{z}]$. They do, however, preserve the slightly larger space $\mathbb{C}[z_1,\To,z_M] \otimes \mathbb{C}[z_{M+1},\To,z_N]^{\mathfrak{S}_{M-N}} \subset \mathbb{C}[\vect{z}]$. The key point of our derivation will be that, moreover, \emph{an appropriate subspace} of $\mathbb{C}[z_1,\To,z_M]$ is \emph{on shell} preserved by these $Y_m^\circ$, which reduce to the affine generators $Y'_m$ of $\widehat{\mathfrak{H}}'_M$ with parameters $q'=t'=t^{-1}$.
\bigskip

\noindent\textbf{Kernel for the \textit{q}-shift.} As for any spin chain with some form of translational invariance the Hamiltonian~\eqref{eq:ham_left} is readily diagonalised for $M=1$ by \textsf{q}-homogeneity, see \textsection\ref{s:intro_examples}. As a warm-up for general $M$ let us redo the derivation for $M=1$ following the strategy outlined above. Fist we need to develop a piece of technology.

By evaluation we may restrict ourselves to the subspace of polynomials of degree at most $N-1$ in any variable, which we denote by $\mathbb{C}[\vect{z}]^{<N} \subset \mathbb{C}[\vect{z}]$. This subspace is clearly preserved by the \textit{q}-shift operators~$\hat{q}_i$, $1\leq i\leq N$. Define the replacement map
\begin{equation} \label{eq:rho_mj}
	r_{mj} \coloneqq \,\cdot\,\, |_{z_m \, \mapsto \, z_j} \, .
\end{equation}
\begin{lemma}[off-shell kernel for the $q$-shift] \label{lem:q_hat_offshell}
On polynomials of degree at most $N-1$ in each variable the $q$-shift operator acts by a linear combination of replacements:
\begin{equation} \label{eq:q_hat_kernel_off_shell}
	\hat{q}_i \mspace{4mu} = \ \sum_{j=1}^N  \, \Biggl( \, \prod_{k(\neq j)}^N \!\! \frac{q\,z_i - z_k}{z_j - z_k} \Biggr) \, r_{ij} \qquad \text{on} \quad \mathbb{C}[\vect{z}]^{<N} \, .
	\checkedMma
\end{equation}
\end{lemma}
\noindent In particular the right-hand side preserves $\mathbb{C}[\vect{z}]^{<N}$ despite the denominators.

\begin{proof}
Write $z$ instead of $z_i$. Let $w_1,\To,w_N \in \mathbb{C}^\times$ be pairwise distinct. Then the $N$~polynomials
\begin{equation*}
	\varphi_j(z) \coloneqq \prod_{k(\neq j)}^N \! (z - w_k) \, , \qquad \varphi_j(w_k) = \delta_{jk} \prod_{l(\neq j)}^N \! (w_j - w_l) \, ,
\end{equation*}
form a basis for $\mathbb{C}[z]^{<N}$, with linear independence because only $\varphi_r$ is nonzero at $z=w_k$. By Lagrange interpolation we can thus write any $P(z) \in \mathbb{C}[z]^{<N}$ as
\begin{equation*}
	P(z) = \sum_{j=1}^N \, \frac{P(w_j)}{\varphi_j(w_j)} \, \varphi_j(z) = \sum_{j=1}^N \, \Biggl(\, \prod_{k(\neq j)}^N \! \frac{z - w_k}{w_j - w_k} \Biggr) P(w_j) \qquad \text{on} \quad \mathbb{C}[z]^{<N} \, .
\end{equation*}
Thus the $q$-shift operator acts by
\begin{equation*}
	\hat{q} \, P(z) = \sum_{j=1}^N \, \Biggl(\, \prod_{k(\neq j)}^N \! \frac{q\,z - w_k}{w_j - w_k} \Biggr) P(w_j) \qquad \text{on} \quad \mathbb{C}[z]^{<N} \, .
\end{equation*}
Taking $\vect{w} = \vect{z}$ and $z=z_i$ now gives a nontrivial expression, and we arrive at \eqref{eq:q_hat_kernel_off_shell}.
\end{proof}

On shell the kernel \eqref{eq:q_hat_kernel_off_shell} simplifies significantly.
\begin{lemma}[on-shell kernel for the $q$-shift] \label{lem:q_hat_onshell}
On shell the $q$-shift operator can be expressed as
\begin{equation} \label{eq:q_hat_kernel_on_shell}
	\normalfont
	\hat{q}_i \mspace{4mu} \overset{\text{ev}}{=} \ \frac{q^N - 1}{N} \, \sum_{j=1}^N \frac{1}{q\,\omega^{i-j} - 1} \, r_{ij} \qquad \text{on} \quad \mathbb{C}[\vect{z}]^{<N} \, .
	\checkedMma
\end{equation}
\end{lemma}
\begin{proof}
Note that the elementary symmetric polynomials evaluate to
\begin{equation} \label{eq:ev_elementary}
	\ev e_r = \delta_{r,0} + (-1)^{N-1} \, \delta_{r,N} \, .
\end{equation}
Indeed, by Newton's identity and \eqref{eq:ev_powersum} we have
\begin{equation*}
	\ev e_r = \frac{1}{r} \sum_{s=1}^r (-1)^{s-1} \ev p_s \, \ev e_{r-s} = (-1)^{N-1} \, \delta_{r,N} \, , \qquad r>0 \, .
\end{equation*}
Hence for any $1\leq i \leq N$
\begin{equation} \label{eq:ev_products} \checkedMma
	\prod_{k=1}^N (q\,z_i - z_k) =  \sum_{r=0}^N (-1)^r \, (q\,z_i)^{N-r} \, e_r(\vect{z}) \mspace{4mu} \overset{\text{ev}}{=} \mspace{4mu} q^N - 1 \, , \qquad z_i \prod_{k(\neq i)}^N \! (z_i - z_k) \mspace{4mu} \overset{\text{ev}}{=} \mspace{4mu} N \, .
\end{equation}
(Note that former implies the latter\,---\,take the semiclassical limit $\delta$\,---\,and together they yield~\eqref{eq:ev_Aj}.) By virtue of these relations the kernel~\eqref{eq:q_hat_kernel_off_shell} reduces to \eqref{eq:q_hat_kernel_on_shell} on shell.
\end{proof}

Equipped with this tool we return to the diagonalisation for $M=1$. We wish to find eigenfunctions in $\mathbb{C}[z_1] \subset \mathbb{C}[z_1] \otimes \mathbb{C}[z_2,\To,z_N]^{\mathfrak{S}_{N-1}}$, which is preserved by first affine generator since $Y_1 \, T_i = T_i \, Y_1$ for all $i>1$. For the spin chain we focus on its classical version~$Y_1^\circ$. This operator simplifies in a way similar to what happened for the polynomial action of $\widehat{\mathfrak{U}}$ in \eqref{eq:Uqsl2_pol_alt}--\eqref{eq:affinisation_pol_alt} at the end of \textsection\ref{s:nonabelian_tilde}:
\begin{lemma}[\cite{NS_17}] \label{lem:Y1_circ}
We have
\begin{equation} \label{eq:Y1_circ}
	Y_1^\circ = \sum_{j=1}^N A_j(\vect{z}) \, \frac{b_{1j}}{a_{j1}} \, s_{1j} = A_1(\vect{z}) + \sum_{j=2}^N A_j(\vect{z}) \, \frac{b_{1j}}{a_{j1}} \, s_{1j} \qquad \text{on} \quad \mathbb{C}[z_1] \otimes \mathbb{C}[z_2,\To,z_N]^{\mathfrak{S}_{N-1}} \ .
\checkedMma
\end{equation}
Here $A_j(\vect{z})$ was defined in \eqref{eq:D_1} and $b_{mj}/a_{jm} = (t-1)\,z_j/(t\,z_j - z_m)$
\checkedMma 
by \eqref{eq:a,b}. 
\end{lemma}

\noindent This simplification was also found in \cite{NS_17} for general~$q$; we will comment on this after Proposition~\ref{prop:Ym_circ}.

\begin{proof} Although we obtained this argument independently our presentation follows the proof of Lemma~3.4 in \cite{Cha_19}. From \eqref{eq:Y1_contribution} it is clear that $Y_1^\circ$ is of the form
\begin{equation*}
	Y_1^\circ = \sum_{j=1}^N c_j(\vect{z}) \, s_{1j} \qquad \text{on} \quad \mathbb{C}[z_1] \otimes \mathbb{C}[z_2,\To,z_N]^{\mathfrak{S}_{N-1}} \, .
\end{equation*}
Indeed, as soon as we pick up a permutation in \eqref{eq:Y1_contribution} the remaining permutations act by the identity due to symmetry in $z_2,\To,z_N$. The coefficients are found as in the proof of \eqref{eq:D_1} in \textsection\ref{s:Macdonald}, now using the partial symmetry $\mathbb{C}[z_1] \otimes \mathbb{C}[z_2,\To,z_N]^{\mathfrak{S}_{N-1}}$.
Two coefficients are easy to get. In \eqref{eq:Y1_contribution} we already read off $c_1(\vect{z}) = a_{12} \cdots a_{1N} = A_1(\vect{z})$. For $j=2$ the only contribution is $b_{12} \, s_{12} \, a_{13} \cdots a_{1N}$, so $c_2(\vect{z}) = b_{12} \, a_{23} \cdots a_{2N} = A_2(\vect{z}) \, b_{21}/a_{12}$. 
The remaining coefficients follow by symmetry: on $\mathbb{C}[z_1] \otimes \mathbb{C}[z_2,\To,z_N]^{\mathfrak{S}_{N-1}}$ we have $Y_1^\circ = s_{2j} \, Y_1^\circ s_{2j}$, whence $c_j(\vect{z}) = s_{2j} \, c_2(\vect{z}) \, s_{2j}$ for all $j\geq 2$.
\end{proof}

Because of the evaluation we restrict ourselves to degree at most $N-1$ in $z_1$.
\begin{proposition} \label{prop:Y1_freezing}
On the subspace of polynomials of degree at most $N-1$ in $z_1$ we have
\begin{equation} \label{eq:Y1_p=1_via_Y'}
	Y_1^\circ \mspace{4mu} \overset{\mathrm{ev}}{=} \mspace{4mu} t^{(N-1)/2} \, \hat{q}^{\mspace{1mu}\prime}_1 \, , \qquad q' \coloneqq t^{-1} \, , \qquad \text{on} \quad \mathbb{C}[z_1]^{<N} \subset \mathbb{C}[z_1] \otimes \mathbb{C}[z_2,\To,z_N]^{\mathfrak{S}_{N-1}} \, .
\checkedMma
\end{equation}
\end{proposition}
\begin{proof}
On polynomials independent of $z_2,\To,z_N$ we may replace $s_{1j}=r_{1j}$. Further using \eqref{eq:ev_Aj} we thus find that \eqref{eq:Y1_circ} implies
\begin{equation*}
	Y_1^\circ \mspace{4mu} \overset{\mathrm{ev}}{=} \mspace{4mu} t^{(N-1)/2} \, \frac{t^{-N} - 1}{N} \sum_{j=1}^N \frac{1}{t^{-1} \, \omega^{1-j}-1} \, r_{1j} \qquad \text{on} \quad \mathbb{C}[z_1] \, . 
	\checkedMma
\end{equation*}
The result now follows from Lemma~\ref{lem:q_hat_onshell}.
\end{proof}

The action of the remaining affine generators~$Y_j^\circ$, $j>1$, on $\mathbb{C}[z_1] \otimes \mathbb{C}[z_2,\To,z_N]^{\mathfrak{S}_{N-1}}$ is more complicated, but we can do without them: in the one-particle sector the $\widetilde{D}_r$ can already be diagonalised together with just $Y_1^\circ$. Indeed, $Y_1^\circ \propto Y'_1$ has eigenfunctions $P_{(n)}(z_1) = z_1^n \in \mathbb{C}[z_1]$. The (orthogonal) plane waves $\ev z_1^n = \omega^n$ give all $N = \dim \mathcal{H}_1$ (orthogonal) eigenvectors in the one-particle sector. Like for \eqref{eq:q-magnon} the case $n=0$ is a multiple of the $\mathfrak{U}$-descendant $F_1^\text{sp} \, \cket{\varnothing}$ coming from $\mathcal{H}_0$. For $1\leq n<N$ we get $N-1$ vectors that have highest weight, at least for $\mathfrak{U}$. In \textsection\ref{s:hw} we will show that these have pseudo highest weight for $\widehat{\mathfrak{U}}$ too. This establishes \eqref{eq:intro_polynomial} for $M=1$ with $\nu = (n)$. (In this case the dependence of $P_\nu$ on the parameters drops out.)
\bigskip 

\noindent \textbf{General \textit{M}.} Now we turn to the proof of Theorem~\ref{thm:nice_polynomial}. We seek joint $M$-particle eigenvectors of the $\delta \widetilde{D}_r$ by simultaneously diagonalising $Y_m^\circ$ for $1\leq m\leq M$ on the subspace $\mathbb{C}[z_1,\To,z_M]\otimes \mathbb{C}[z_{M+1},\To,z_N]^{\mathfrak{S}_{N-M}} \subset \mathbb{C}[\vect{z}]$ in accordance with \textsection\ref{s:physical_space}. In physical terms we think of $z_1, \To,z_M$ as the coordinates of the magnons, which we treat as distinguishable particles for the moment. Let us try to proceed as for $M=1$. 
The start is easy: the analogue of Lemma~\ref{lem:Y1_circ} for general~$M$, $1\leq m\leq M$, is
\begin{proposition} \label{prop:Ym_circ}
We have
\begin{equation} \label{eq:Ym_circ}
	\begin{gathered}
	Y_m^\circ = x_{m,m+1} \cdots x_{mM} \Biggl( A_m(\vect{z}) + \! \sum_{j(>M)}^N \!\!\! A_j(\vect{z}) \, \frac{b_{mj}}{a_{jm}} \, s_{mj} \Biggr) \Biggl(\, \prod_{\bar{m}(\neq m)}^{\smash{M}} \!\!\!\! f_{m \bar{m}} \Biggr) x_{m1} \cdots x_{m,m-1} \\
	\hfill \text{on} \quad \mathbb{C}[z_1,\To,z_M]\otimes \mathbb{C}[z_{M+1},\To,z_N]^{\mathfrak{S}_{N-M}} \, ,
	\end{gathered}
\checkedMma
\end{equation}
where we recall that $f_{m\bar{m}} = 1/a_{m\bar{m}}$.
\end{proposition}
\noindent If $m=1$ arbitrary~$q$ is included by postmultiplication with $\hat{q}_1$. The resulting operator was used in \cite{NS_17} to construct `covariant' \textit{Y}$\mspace{-2mu}$-operators (\textsf{q}-deformed Heckman operators).
For $m\geq 2$, however, the \textit{q}-shift acts after $x_{m1} \cdots x_{m,m-1}$, cf.~\eqref{eq:Yi_via_x}, affecting those $x_{m\bar{m}}$\,---\,unless $q=1$, as for us.

\begin{proof}
As the $x_{mm'}$ on the right preserve $\mathbb{C}[z_1,\To,z_M]\otimes \mathbb{C}[z_{M+1},\To,z_N]^{\mathfrak{S}_{N-M}}$ it suffices to show that on this space
\begin{equation} \label{eq:x_prod_simplification}
	x_{m,M+1} \cdots x_{mN} = \Biggl(\! A_m(\vect{z}) + \! \sum_{j(>M)}^N \!\!\! A_j(\vect{z}) \, \frac{b_{jm}}{a_{mj}} \, s_{mj} \Biggr) \! \prod_{\bar{m}(\neq m)}^M \!\!\!\! f_{m\bar{m}} \, .
\checkedMma
\end{equation}
This can be shown by a symmetry argument as for \eqref{eq:Y1_circ}. The result will be of the form
\begin{equation*}
	x_{m,M+1} \cdots x_{mN} = c_m(\vect{z}) + \! \sum_{j(>M)}^N \!\!\! c_j(\vect{z}) \, s_{mj} \, .
\end{equation*}
As before we read off $c_m(\vect{z}) = a_{m,M+1}\cdots a_{mN} = A_m(\vect{z}) \, \prod_{\bar{m}(\neq m)}^M f_{m\bar{m}}$, where the $f_{m\bar{m}}$ compensate for the superfluous factors of $a_{m\bar{m}}$ in the definition of $A_m(\vect{z})$. Next, $c_{M+1}(\vect{z}) = b_{m,M+1} \, a_{M+1,M+2} \cdots a_{M+1,N} = A_{M+1}(\vect{z}) \, (b_{m,M+1}/a_{M+1,m}) \prod_{\bar{m}(\neq m)}^M f_{M+1,\bar{m}}$. The remaining coefficients are obtained from this via conjugation by $s_{M+1,j}$.
\end{proof}

Motivated by our findings for $M=1$ we would like to recognise the kernel for the \textit{q}-shift in \eqref{eq:x_prod_simplification}. However, as the proof of \eqref{eq:q_hat_kernel_off_shell} shows the latter is only valid when acting on polynomials (of sufficiently low degree). We therefore have to get rid of the denominator of the product of $f$s in \eqref{eq:x_prod_simplification}, which is the \textsf{q}-$\mspace{-1mu}$Vandermonde-type product $\prod_{\bar{m} (\neq m)}^M (t\,z_m - z_{\bar{m}})$. For $M=2$ it is not hard to see that it suffices for the polynomial to be divisible by $t \, z_1 - z_2$, where for $Y_2^\circ$ one needs identity \eqref{eq:x_to_x'} below. Let us show that in general we will need the polynomials that we act on to be divisible by $\Delta_t \coloneqq \Delta_t(z_1,\To,z_M)$. 

Theorem~\ref{thm:intro_Ym_to_Y'm} from \textsection\ref{s:intro_explicit_evrs_freezing} can be stated more precisely as
\begin{theorem} \label{thm:Ym_to_Y'm}
For $1\leq m\leq M$
\begin{equation} \label{eq:Ym_to_Y'm}
	Y_m^\circ \, \overset{\mathrm{ev}}{=} \,\, t^{(N-M)/2} \, \Delta_t \, Y'_m \, \Delta_t^{-1} \qquad \text{on} \quad \Delta_t \, \mathbb{C}[z_1,\To,z_M]^{< N-M+1} \subset \mathbb{C}[z_1,\To,z_M] \, .
\checkedMma
\end{equation}
\end{theorem}
\begin{proof}
The idea is to knead \eqref{eq:x_prod_simplification} into a form that allows us to use \eqref{eq:q_hat_kernel_off_shell}. We divide the proof into four steps. It is instructive to keep the case $M=2$ in mind; the extension to arbitrary~$M$ is mostly a matter of bookkeeping.

\textit{Step~i.~Rewriting the coefficients.} Let us first show that $A_m(\vect{z})$ and $A_j(\vect{z})$ in \eqref{eq:Ym_circ} may on shell be replaced by ($t^{(N-1)/2}$ times) the coefficients of \eqref{eq:q_hat_kernel_off_shell} with $q\rightsquigarrow q' = t^{-1}$. Indeed, for $j=m$ as well as $j \neq m$
\begin{equation}
	\begin{aligned}
	A_j(\vect{z}) \, \frac{b_{mj}}{a_{jm}} \, = {} & \, \frac{t^{-(N-1)/2}}{t\,z_j - z_m} \, \frac{\prod_{k=1}^N (t \, z_j - z_k)}{\prod_{k(\neq j)}^N (z_j - z_k)} \\
	\overset{\mathrm{ev}}{=} {} & \, \frac{t^{-(N-1)/2}}{t\,(z_j - t^{-1} z_m)} \, \frac{t^N - 1}{t^{-N}-1} \, \frac{\prod_{k=1}^N (t^{-1} \, z_m - z_k)}{\prod_{k(\neq j)}^N (z_j - z_k)} \\
	= {} & \, t^{(N-1)/2} \prod_{k(\neq j)}^N \!\!\! \frac{t^{-1} \, z_m - z_k}{z_j - z_k} \, .
	\end{aligned}
\checkedMma
\end{equation}
Here the on-shell equality uses the first evaluation in \eqref{eq:ev_products}. Importantly, the value of the latter is independent of $1\leq j\leq N$. After all, by definition~\eqref{eq:ev_operator} evaluation takes place \emph{after} any permutation has acted. But permutations at most change the value of~$j$ in \eqref{eq:ev_products}, which doesn't matter upon evaluation.

On shell \eqref{eq:Ym_circ} can therefore be rewritten as
\begin{equation} \label{eq:Ym_p=1_rewrite}
	\begin{aligned}
	Y_m^\circ \mspace{4mu} \overset{\mathrm{ev}}{=} {} \mspace{4mu} & t^{(N-1)/2} \, x_{m,m+1} \cdots x_{mM} \, \Biggl( \prod_{k(\neq m)}^N \!\!\!\! \frac{t^{-1} \, z_m - z_k}{z_m - z_k} + \! \sum_{j(>M)}^N \prod_{k(\neq j)}^N \!\!\! \frac{t^{-1} \, z_m - z_k}{z_j - z_k} \, s_{mj} \Biggr) \\
	& \times \Biggl(\prod_{\bar{m}(\neq m)}^{\smash{M}} \!\!\!\! f_{m\bar{m}}\Biggr) x_{m1} \cdots x_{m,m-1} \qquad \text{on} \quad \mathbb{C}[z_1,\To,z_M]\otimes \mathbb{C}[z_{M+1},\To,z_N]^{\mathfrak{S}_{N-M}} \, .
	\end{aligned}
\checkedMma
\end{equation}

\textit{Step~ii.~Pulling the \textsf{q}-$\!$Vandermonde through.} Next we show that the \textsf{q}-$\mspace{-1mu}$Vandermonde factor ensures the denominators of the $f_{m \bar{m}}$ are cancelled, so that we stay in the world of polynomials. For $m>1$ we first need to move $x_{m,1} \cdots x_{m,m-1}$ through $\Delta_t$. Note that $x_{m,m-1}$ commutes with $\Delta_t /(t\,z_{m-1}-z_m)$, which is symmetric in $z_{m-1} \leftrightarrow z_m$, while 
\begin{equation} \label{eq:x_to_x'}
	x_{m,m-1} \, (t\,z_{m-1}-z_m) = (z_{m-1} - t\, z_m) \, x'_{m,m-1} \, , \qquad t' \coloneqq t^{-1} \, ,
\checkedMma
\end{equation}
where $x'_{ij}$ denotes \eqref{eq:x_ij} with $t\rightsquigarrow t'$. To verify \eqref{eq:x_to_x'} recall that $x_{m,m-1} = x_{m-1,m}^{-1} = s_{m-1} \, T^{\text{pol}\,-1}_{m-1}$ and check $(t\,z_{m-1}-z_m)^{-1} \, T^\text{pol}_{m-1} \, (t\,z_{m-1}-z_m) = -T^{\prime\,\text{pol}}_{m-1}$, which can be conveniently done on $\mathbb{C}[z_{m-1},z_m]^{\mathfrak{S}_2} \oplus (t'\,z_{m-1}-z_m)\,\mathbb{C}[z_{m-1},z_m]^{\mathfrak{S}_2}$.

For $m>2$ we next move $x_{m,m-2}$ through $\Delta_t \, (t'\,z_{m-1}-z_m)/(t\,z_{m-1}-z_m)$. But besides a factor of $t\,z_{m-2}-z_m$ the latter is symmetric in $z_{m-2} \leftrightarrow z_m$ so we can argue like before. \checkedMma
Continuing in this way we see that
\begin{equation*}
	x_{m,1} \cdots x_{m,m-1} \, \Delta_t \, = \, \cdots \, = \, \Delta_t \prod_{k =1}^{m-1} \! \frac{z_k - t\, z_m}{t \, z_k - z_m} \ x'_{m,1} \cdots x'_{m,m-1} \, .
\checkedMma
\end{equation*}
The denominator of $\prod_{\bar{m}(\neq m)}^M f_{m \bar{m}}$ is now precisely cancelled by $\prod_{k=1}^{m-1} (t' \, z_k - z_m)$ along with the factors $t \, z_m - z_l$ still contained in $\Delta_t$:
\begin{equation} \label{eq:vandermonde_factor}
	\Biggl(\, \prod_{\bar{m}(\neq m)}^{\smash{M}} \!\!\!\! f_{m \bar{m}}\Biggr) \Delta_t \! \prod_{k=1}^{m-1} \! \frac{z_k - t\, z_m}{t \, z_k - z_m} = \prod_{k=1}^{m-1} \! (z_k - z_m) \!\! \prod_{l=m+1}^M \!\!\! (z_m - z_l) \ \Delta_{t}(z_1,\To \widehat{z_m} \To,z_M) \, ,
\checkedMma
\end{equation}
where the caret indicates that $z_m$ is to be omitted from the \textsf{q}-$\mspace{-1mu}$Vandermonde.

\textit{Step~iii.~Recognising the $q$-shift.} On polynomials independent of $z_{M+1},\To,z_N$ we may, like for $M=1$, replace the $s_{mj}$ by $r_{mj}$. Comparing with \eqref{eq:q_hat_kernel_off_shell} we just miss the terms with $j\in \{1,\To,m-1,m+1,\To,M \}$. We observe, however, that \eqref{eq:vandermonde_factor} vanishes when $z_m = z_k$ for any $k\neq m$: the factor $\Delta_t$ does not only ensure that the sum in \eqref{eq:Ym_p=1_rewrite} acts on polynomials, but moreover allows us to complete the sum to all values of $1\leq j\leq N$, just as for the ordinary Haldane--Shastry model, see \textsection3.3 in~\cite{BG+_93}. This allows us to use \eqref{eq:q_hat_kernel_off_shell} provided we act on polynomials of degree $<N$ (including~$\Delta_t$). Therefore on $\mathbb{C}[z_1,\To,z_M]^{< N-M+1} \subset \mathbb{C}[z_1,\To,z_M] \otimes \mathbb{C}[z_{M+1},\To,z_N]^{\mathfrak{S}_{N-M}}$ we have
\begin{equation*}
	\begin{aligned}
	Y_m^\circ \, \Delta_t \, \overset{\mathrm{ev}}{=} \, {} & t^{(N-1)/2} \, x_{m,m+1} \cdots x_{mM} \, \hat{q}'_m  \prod_{k=1}^{m-1} \!(z_k-z_m) \!\! \prod_{l=m+1}^M \!\!\! (z_m-z_l) \\
	& \times \Delta_t(z_1,\To\widehat{z_m}\To,z_M) \, x'_{m1} \cdots x'_{m,m-1} \, . 
	\end{aligned}
\checkedMma
\end{equation*}	

\textit{Step~iv.~Pulling the \textsf{q}-$\!$Vandermonde through further.} Since
\begin{equation*}
	 \hat{q}'_m  \prod_{k=1}^{m-1} \!(z_k -z_m) \!\! \prod_{l=m+1}^M \!\!\! (z_m-z_l) \, = \, t^{1-M} \prod_{k=1}^{m-1} \! (t\,z_k - z_m) \!\! \prod_{l=m+1}^M \!\!\! (z_m - t \, z_l) \ \hat{q}'_m 
	 % \, .
\checkedMma
\end{equation*}
it remains to move $x_{m,m+1} \cdots x_{mM}$ through
\begin{equation*}
	\Delta_t(z_1,\To,\widehat{z_m},\To,z_M) \prod_{k=1}^{m-1} \! (t\,z_k- z_m) \!\! \prod_{l=m+1}^M \!\!\! (z_m - t\,z_l) \, =\, t^{(M-1)/2} \, \Delta_t \!\! \prod_{l=m+1}^M \!\! \frac{z_m - t \, z_l}{t\,z_m- z_l} \, .
\checkedMma
\end{equation*}
This is done like in step~\textit{ii}: except for the factor $t\,z_m - z_M$, the latter is symmetric in $z_m \leftrightarrow z_M$, while $x_{mM}\,(z_m-t\,z_M)= (t\,z_m-z_M)\,x'_{mM}$. 
\checkedMma 
Hence
\begin{equation*}
	x_{m,m+1} \cdots x_{mM} \ \Delta_t \! \prod_{l=m+1}^M \!\! \frac{ t'\, z_m - z_l}{t\,z_m-z_l} \, = \, \cdots \, = \, \Delta_t \, x'_{m,m+1} \cdots x'_{mM} \, .
\checkedMma
\end{equation*}
Putting everything together we arrive at \eqref{eq:Ym_to_Y'm}.
\end{proof}

The $Y'_m$ are simultaneously diagonalised by the nonsymmetric Macdonald polynomials $E'_\alpha$, $q'=t'=t^{-1}$. However, to make contact with the $M$-particle Macdonald operators we need to act on symmetric polynomials, requiring conjugation by $\Delta_{1/t}$. By Lemma~\ref{lem:conj_by_qVandermonde} this changes the parameters once more, see \eqref{eq:conj_by_qVandermonde}, where now $N \rightsquigarrow M$, $q \rightsquigarrow q'=t^{-1}$, $t \rightsquigarrow t' = t^{-1}$. The new parameters $q'' = q' = t^{-1}$ and $t'' = q'\, t' = t^{-2}$ are related as $q'' = t''\mspace{1mu}{}^\alpha$ for $\alpha = 1/2$.
\bigskip

\noindent \textbf{Upshot.} On polynomials divisible by the symmetric square of the \textsf{q}-$\!$Vandermonde and of degree less than $N$ in each variable the operator $e_r(Y_1^\circ,\To,Y_M^\circ)$ --- in $N$ variables, so \eqref{eq:D_r_q=1} does not apply for $M<N$ --- is transformed, on shell, to a quantum spherical zonal Macdonald operator in $M$ variables:
\begin{equation} \label{eq:parameter_shift_upshot}
	\begin{aligned}
	& e_r(Y_1^\circ,\To,Y_M^\circ) \, \Delta_t \,\Delta_{1/t} && &&&& q^\circ=1, \ t^\circ = t \\
	& \overset{\mathrm{ev}}{=} t^{r \,(N-M)/2} \,\Delta_t \, e_r(Y'_1,\To,Y'_M) \, \Delta_{1/t} && \text{on} \quad \mathbb{C}[z_1,\To,z_M]^{<N-2\mspace{1mu}M+2} &&&& \, q' = t' = t^{-1} \\
	& = t^{r\,(N-2\mspace{1mu}M+1)/2} \, \Delta_t \,\Delta_{1/t} \, e_r(Y''_1,\To,Y''_M) && \text{on} \quad \mathbb{C}[z_1,\To,z_M]^{\mathfrak{S}_M}  &&&& q'' = t''{}^{1/2} = t^{-1}
	\end{aligned}
\end{equation}	
The joint eigenfunctions of the $e_r(Y''_1,\To,Y''_M) = D_r''$ (on valid on $\mathbb{C}[z_1,\To,z_M]^{\mathfrak{S}_M}$) are Macdonald polynomials (\textsection\ref{s:Macdonald}). Using the latter's invariance under simultaneous inversion of both parameters (\textsection{VI.4}~(iv) in \cite{Mac_95}) we conclude that the polynomials we set out to find are
\begin{equation} \label{eq:wavefn_polynomial}
	\widetilde{\Psi}_\nu = \Delta_{t} \, \Delta_{1/t} \, P_\nu'' = \Delta_{t} \, \Delta_{1/t}	\, P_\nu^\star \, \in \, \mathbb{C}[z_1,\To,z_M]^{\mathfrak{S}_M} \, , \qquad \ell(\nu) \leq M \, ,
\end{equation}
where finally $q^\star = t^\star\mspace{1mu}{}^\alpha = t$, still for the quantum spherical zonal case $\alpha = 1/2$. This proves Theorem~\ref{thm:nice_polynomial} from \textsection\ref{s:intro_explicit_evrs}. Our derivation is valid provided $P_\nu^\star$ has degree $\nu_1 \leq N-2\, M+1$ in each variable. This reproduces the motif condition, cf.~the line below \eqref{eq:motifs_vs_partitions}.
% \end{upshot}

\subsubsection{Pseudo highest-weight property} \label{s:hw} Let us show that our eigenvectors are pseudo highest weight for $\widehat{\mathfrak{U}}$, i.e.\ that they are annihilated by the Chevalley generators $E_1^\text{sp}$ and $\widetilde{F}_0^{\text{sp},\mspace{1mu}\circ}$ that come from the \textit{C}-operator (see \textsection\ref{s:app_presentations}). In the polynomial setting we have to show that their simple component is annihilated by $E_M^\text{pol}$ and $\widehat{F}_M^{\text{pol},\mspace{1mu}\circ}$ from \eqref{eq:Uqsl2_pol_alt}--\eqref{eq:affinisation_pol_alt}. This is Theorem~\ref{eq:thm_hw_condition} from \textsection\ref{s:intro_nonabelian_2}:
\begin{theorem} \label{thm:hw_property}
The simple component of our eigenvectors have the pseudo highest-weight property
\begin{equation}
	\normalfont
	E_M^{\text{pol}} \, \widetilde{\Psi}_\nu \, \overset{\mathrm{ev}}{=} \, \widehat{F}_M^{\text{pol},\mspace{1mu}\circ} \, \widetilde{\Psi}_\nu \, \overset{\mathrm{ev}}{=} \, 0 \qquad \text{if{f}} \qquad \ell(\nu) = M \, .
\end{equation}
\end{theorem}

\begin{proof}
The Chevalley operators \eqref{eq:Uqsl2_pol_alt} and \eqref{eq:affinisation_pol_alt} simplify further on shell: by \eqref{eq:ev_Aj} we have
\begin{equation*}
	\begin{aligned} 
	E_M^{\text{pol}} & \, \overset{\mathrm{ev}}{=} \, t^{(N-2M+1)/2} \, \frac{[N]}{N} \, \Biggl(\, \sum_{j=M}^N \! s_{Mj} \Biggr) \prod_{m=1}^{M-1} \! f_{mM} \, , \\
	\widehat{F}_M^{\text{pol},\mspace{1mu}\circ} & \, \overset{\mathrm{ev}}{=} \, t^{(N-2M+1)/2} \, \frac{[N]}{N} \, \Biggl(\, \sum_{j=M}^N \! s_{Mj} \Biggr) \prod_{m=1}^{M-1} \! f_{Mm} \, .
	\end{aligned}
\checkedMma
\end{equation*}
Note the symmetry from \eqref{eq:Uhat_pol_t_inv}. As \eqref{eq:wavefn_polynomial} is invariant under inverting~$t$ the proof for the two operators is parallel; we will give it for $E_M^{\text{pol}}$. We use various ingredients from the proof in \textsection\ref{s:explicit_evrs}.

The denominator of $\prod_m f_{mM}$ cancels with some factors of $\Delta_t$ in \eqref{eq:wavefn_polynomial}. On functions independent of $z_{M+1},\To,z_N$ we can replace $s_{Mj}$ by $r_{Mj}$ for $j>M$. Moreover, the numerator of $f_{mM}$ vanishes if $z_M$ is replaced by $z_m$ with $m\leq M$, so we can again complete the sum:
\begin{equation*}
	E_M^{\text{pol}} \, \overset{\mathrm{ev}}{=} \, t^{(N-2M+1)/2} \,  \frac{[N]}{N} \, \Biggl(\, \sum_{j=1}^{N} r_{Mj} \Biggr) \prod_{m=1}^{M-1} \! f_{mM} \qquad \text{on} \quad \mathbb{C}[z_1,\To,z_M]^{\mathfrak{S}_M} \, .
\checkedMma
\end{equation*}
When we act on polynomials that contain a factor of $\prod_{m=1}^{M-1} (t \, z_m -z_M)$ to cancel the denominators of the $f_{mM}$ the result will again be a polynomial.
Now take $q=0$ in the expressions \eqref{eq:q_hat_kernel_off_shell} and \eqref{eq:q_hat_kernel_on_shell} for the \textit{q}-shift to get the `annihilator'
\begin{equation*}
	\begin{aligned}
	\hat{0}_i \mspace{4mu} & = \ \sum_{j=1}^N  \Biggl(\, \prod_{k(\neq j)}^N \!\! \frac{- z_k}{z_j - z_k} \Biggr) \, r_{ij} \qquad \text{on} \quad \mathbb{C}[z_1,\To,z_N]^{<N} \\
	& \overset{\text{ev}}{=} \ \frac{1}{N} \, \sum_{j=1}^N r_{ij} \, .
	\end{aligned}
\checkedMma
\end{equation*}
Comparing this with our on-shell expression for $E_M^{\text{pol}}$, and its analogue for $\widehat{F}_M^{\text{pol},\mspace{1mu}\circ}$, we see that on the intersection
\begin{equation*}
	\Delta_t(z_1,\To,z_M) \, \Delta_{1/t}(z_1,\To,z_M) \, \mathbb{C}[z_1,\To,z_M]^{S_M} \cap \mathbb{C}[z_1,\To,z_N]^{<N}
\end{equation*} 
the action of both of these operators is proportional to the annihilator $\hat{0}_M$. The latter can be moved through $(\,\prod_{m(<M)} f_{mM}) \, \Delta_t \, \Delta_{1/t}$ and $(\,\prod_{m(<M)} f_{Mm}) \, \Delta_t
\, \Delta_{1/t}$, which both contain terms that do survive setting $z_M = 0$. We conclude that on shell our eigenvectors have pseudo highest weight for $\widehat{\mathfrak{U}}$ if{f} $P_\nu^\star$ has degree at least one in $z_M$, i.e.\ if{f} $\ell(\nu) = M$.
\end{proof}

To conclude we prove that the Drinfeld polynomial of the $\widehat{\mathfrak{U}}$-irrep determined by $\widetilde{\Psi}_\nu$ from \eqref{eq:wavefn_polynomial} is given by \eqref{eq:Drinfeld_polyn}.

\begin{proof}[Proof of Proposition~\ref{prop:Drinfeld}]
We will compute the ratio of eigenvalues of $\widetilde{A}^\circ(u)$ and $\widetilde{D}^\circ(u)$ as in \eqref{eq:intro_ABCD_on_shell}. Use the polynomial action \eqref{eq:ABCD_Y_pol} on our simple component~\eqref{eq:wavefn_polynomial} to calculate
\begin{equation*}
	\begin{aligned}
	\widetilde{A}^{\text{pol}\,\circ}_M(u) \, \widetilde{\Psi}_\nu & = t^{M/2} \prod_{m=1}^M \frac{1 - u\,Y_m^\circ}{1 - t\, u\,Y_m^\circ} \ \, \Delta_t \, \Delta_{1/t} \, P^\star_\nu \\
	& = t^{M/2} \, \Delta_t \, \Delta_{1/t} \prod_{m=1}^M \frac{1 - t^{(N-2M+1)/2} \, u\,Y''_m}{1 - t^{(N-2M+3)/2}\, u\,Y''_m} \ \, P''_\nu \\
	& = t^{M/2} \prod_{m=1}^M \frac{1 - t^{(N-4(M-m)-2\nu_m-1)/2} \, u}{1 - t^{(N-4(M-m)-2\nu_m +1)/2}\, u} \ \, \widetilde{\Psi}_\nu \, .
	\end{aligned}
\end{equation*}
In the second equality we used \eqref{eq:parameter_shift_upshot} to move the symmetric combination of $Y_m^\circ$ through the symmetric square of the \textsf{q}-$\!$Vandermonde factor, and in the second equality we used \eqref{eq:nonsymm_Macdonalds} or \eqref{eq:Macd_eigenvalues} with $N'' = M$ and $q'' = t^{-1}$, $t'' = t^{-2}$. With the help of the identification~\eqref{eq:motifs_vs_partitions} between the partition~$\nu$ and motif~$\mu$ we thus obtain the eigenvalue
\begin{equation*}
	\alpha^\mu(u) = t^{M/2} \prod_{m=1}^M \frac{1 - t^{(N-2\mu_m-1)/2} \, u}{1 - t^{(N-2\mu_m+1)/2}\, u} \, .
\end{equation*}

Concerning $\widetilde{D}^\circ(u)$ we can avoid the complicated formula \eqref{eq:ABCD_Y_pol} by exploiting the quantum determinant of $\widetilde{L}^\circ_a(u)$. By $\widehat{\mathfrak{U}}$-invariance we may consider the pseudo highest-weight vector $\ket{\mu}$. On this vector \eqref{eq:qdet_ABCD} and \eqref{eq:qdet_Y} together imply
\begin{equation*}
	\mathrm{qdet}_a \widetilde{L}^\circ_a(u) \, \ket{\mu} = \Bigl( \widetilde{A}^\circ(t\,u) \, \widetilde{D}^\circ(u) - 0 \Bigr) \ket{\mu} = t^{N/2} \, \frac{\Delta^{\!\circ}(-u)}{\Delta^{\!\circ}(-t\,u)} \, \ket{\mu} \, ,
\end{equation*}
so the eigenvalue of $\widetilde{D}^\circ(u)$ is given by
\begin{equation*}
	\delta^\mu(u) = t^{N/2} \, \frac{\Delta^{\!\circ}(-u)}{\Delta^{\!\circ}(-t\,u)} \, \alpha^\mu(t\,u)^{-1} \, .
\end{equation*}

Now compute the ratio in \eqref{eq:intro_ABCD_on_shell}. First consider the empty motif, $\mu = 0$. In this case $\alpha^0(u) = 1$ and the only contribution comes from the quantum determinant,
\begin{equation*}
	\frac{\alpha^0(u)}{\delta^0(u)} = t^{-N/2} \, \frac{\Delta^{\!\circ}(-t\,u)}{\Delta^{\!\circ}(-u)} = t^{-N/2} \prod_{i=1}^N \frac{1 - t^{(N-2i+3)/2} \, u}{1 - t^{(N-2i+1)/2} \, u} \, .
\end{equation*}
Hence $P^0(u) = \Delta^{\!\circ}(-u) = \prod_{i=1}^N (1-t^{(N-2i+1)/2} \, u)$. For any other motif we have to correct the preceding by a factor of
\begin{equation*}
	\alpha^\mu(u)\, \alpha^\mu(t\,u) = t^M \prod_{m=1}^M \frac{1 - t^{(N-2\mu_m-1)/2} \, u}{1 - t^{(N-2\mu_m+1)/2}\, u} \, \frac{1 - t^{(N-2\mu_m+1)/2} \, u}{1 - t^{(N-2\mu_m+3)/2}\, u} \, .
\end{equation*}
The numerator tells us which factors to delete from $\Delta^{\!\circ}(-u)$, yielding \eqref{eq:Drinfeld_polyn}.
\end{proof}

% \newpage

\appendix

\section{Nonrelativistic/isotropic limit} \label{s:app_nonrlt_limit}

\subsection{Dunkl and Calogero--Sutherland limit} \label{s:app_CalSut_limit} To facilitate comparison with the literature on the Haldane--Shastry model let us review the nonrelativistic limit in some detail. Setting $q = t^\alpha$ and letting $t\to1$ the \textsc{aha} generators \eqref{eq:Hecke_pol_2a} and \eqref{eq:Yi_via_x} behave as
\begin{equation*}
	\begin{aligned}
	T_i^\text{pol} & = s_i + \frac{1}{2} \, (t-1) \biggl(1 - \frac{z_i + z_{i+1}}{z_i - z_{i+1}} \, (1-s_i) \biggr) + \mathcal{O}(t-1)^2 \, , \\
	Y_i & = 1 + (t-1) \, d_i + \mathcal{O}(t-1)^2 \, , 
	\end{aligned}
\end{equation*}
where the trigonometric Dunkl(--Cherednik) operators \cite{Dun_89,Che_91} are
\begin{equation} \label{eq:Dunkl}
	\begin{aligned} 
	d_i \coloneqq {} & \alpha \, z_i \,\partial_{z_i} + \frac{1}{2} \sum_{j(\neq i)}^N \frac{z_i + z_j}{z_i - z_j} \, (1-s_{ij}) - \frac{1}{2} \sum_{j=1}^{i-1} s_{ij} + \frac{1}{2} \! \sum_{j=i+1}^N \!\! s_{ij} \\
	= {} & \alpha \, z_i \,\partial_{z_i} + \frac12 \, (N-2\,i+1) + \sum_{j=1}^{i-1} \frac{z_i}{z_i - z_j} \, (1-s_{ij}) + \!\! \sum_{j=i+1}^N \frac{z_j}{z_i - z_j} \, (1-s_{ij}) \, .
	\end{aligned}
\end{equation}
This is the \emph{basic} representation of the \emph{degenerate \textsc{aha}}~\cite{Dri_86,Lus_89} whose relations can be obtained by expanding \eqref{eq:Hecke} and \eqref{eq:AHA} in $t-1$. The $s_i$ obey the relations of $\mathfrak{S}_N$, the $d_i$ form an abelian subalgebra, while the cross relations read
\begin{equation*}
	s_i \, d_i - d_{i+1} \, s_i = 1 \, , \qquad\quad s_i \, d_j = d_j \, s_i \quad \text{if } j\notin \{i,i+1\} \, .
\end{equation*}
Note that shifts of the $d_i$ by a common constant, which do not change the relations, occur in the literature. In \cite{BG+_93} the \eqref{eq:Dunkl}, which act on the space of polynomials, were called `gauge transformed' Dunkl operators.

The nonrelativistic limit of the Macdonald operators $D_{\pm 1}$ is obtained using
\begin{equation*}
	\begin{aligned}
	A_{\pm i}(\vect{z}) = 1 & \pm \frac{1}{2} \, (t-1) \, \sum_{j(\neq i)}^N \frac{z_i+z_j}{z_i-z_j} \\
	& + \frac{1}{4} \, (t-1)^2 \, \Biggl(\, \sum_{j(\neq i)}^N \! \biggl(\mp \frac{z_i+z_j}{z_i-z_j} + \frac{1}{2} \biggr) + \sum_{\substack{j\neq k \\ (\neq i)}}^N \frac{z_i+z_j}{z_i-z_j} \, \frac{z_i+z_k}{z_i-z_k} \Biggr) +\mathcal{O}(t-1)^3
	\end{aligned}
\end{equation*}
along with \eqref{eq:q_hat} for $q=t^\alpha$, noting that $z_i^2 \, \partial_{z_i}^2 = z_i \, \partial_{z_i} \, (z_i \, \partial_{z_i} - 1)$. The result is
\begin{equation} \label{eq:D_pm1_limit}
	D_{\pm 1} = N \pm (t-1) \, \alpha \, P^{\text{nr}} + \frac{1}{2} \, (t-1)^2 \, (\alpha^2 \, H^{\text{eff},\mspace{1mu}\text{nr}} \mp \alpha \, P^{\text{nr}}) + \mathcal{O}(t-1)^3 \, .
\end{equation}
Here $P^{\text{nr}} = \sum_j z_j \, \partial_{z_j}$ is the nonrelativistic limit of $(D_1 - D_{-1})/2$, giving the usual total momentum operator in multiplicative notation $z_j = \E^{2\pi \I\, x_j/L}$. (This operator also arises in the semiclassical limit of $D_N$, cf.~\textsection\ref{s:abelian_freezing} and especially \eqref{eq:total_degree_semiclassical}.) The combination $(D_1 + D_{-1})/2$ contains, besides the rest mass $N$ (times $m\,c^2$),
\begin{equation} \label{eq:CS_eff}
	\begin{aligned}
	H^{\text{eff},\mspace{1mu}\text{nr}} = {} & \sum_{i=1}^N \bigl(z_i\,\partial_{z_i}\bigr)^{\!2} + k \sum_{i<j}^N \frac{z_i+z_j}{z_i-z_j} \, (z_i \, \partial_{z_i} - z_j \, \partial_{z_j}) + \varepsilon^{\text{nr}}_0 \\
	= {} & \sum_{i=1}^N \bigl(z_i\,\partial_{z_i}\bigr)^{\!2} - k \sum_{i=1}^N (N-2\,i+1)\,z_i\,\partial_{z_i} \\
	& + 2\,k \sum_{i<j}^{\smash{N}} \frac{z_i}{z_i-z_j} \, (z_i \, \partial_{z_i} - z_j \, \partial_{z_j}) + \varepsilon^{\text{nr}}_0 \, ,
	\end{aligned} 
	\qquad\quad k = \frac{1}{\alpha}\, ,
\end{equation}
where $2\,\varepsilon^{\text{nr}}_0/k^2 = \binom{N}{2} + \binom{N}{3} = N\,(N^2-1)/6$ are tetrahedral numbers. We recognise \eqref{eq:CS_eff} as the effective (gauge-transformed) Hamiltonian of the trigonometric Calogero--Sutherland model~\cite{Sut_71,Sut_72} in multiplicative notation. In terms of~\eqref{eq:Dunkl} we have
\begin{equation*}
	\sum_{i=1}^N d_i = \alpha \, P^{\text{nr}} \, , \qquad\qquad \sum_{i=1}^N d_i^2 = \alpha^2 \, H^{\text{eff},\mspace{1mu}\text{nr}} \qquad \text{on} \quad \mathbb{C}[\vect{z}]^{\mathfrak{S}_N} \, .
\end{equation*}
Joint eigenfunctions~\cite{Sut_72} arise from Macdonald polynomials, cf.~\eqref{eq:Macd_orthog}:
\begin{equation} \label{eq:Jack}
	P^{(\alpha)}_\lambda(\vect{z}) = \lim_{t\to 1} \, P_\lambda(\vect{z})|_{q=t^\alpha} \, , \qquad P^{(\alpha)}_\lambda(\vect{z}) = m_\lambda(\vect{z}) + \text{lower} \, .
\end{equation}
Here $P^{(\alpha)}_\lambda$ are Jack polynomials in the (monic) `P-normalisation'~\cite{Jac_70,Mac_95}. The ground state is $P^{(\alpha)}_0(\vect{z}) = 1$ with energy $\varepsilon^{\text{nr}}_0$. Schur polynomials arise in the special case $\alpha=k=1$, where the particles are free, as is clear from the following.

The physical nonrelativistic Hamiltonian~$H^{\text{nr}}$ is \eqref{eq:spin-CS} with $P_{ij} \rightsquigarrow 1$. It is obtained by either of expanding the Ruijsenaars operators \eqref{eq:Macd_to_Ruij}~at $t=1$ or conjugating \eqref{eq:CS_eff} by the square root of the nonrelativistic limit\,---\,use $(z;q)_\infty/(q^k\mspace{1mu} z;q)_\infty \to (1-z)^k$ as $q\to 1$\,---\,of \eqref{eq:Macd_measure},
\begin{equation} \label{eq:psi_0_nonrlt}
\mu_k^\text{nr}(\vect{z}) = \prod_{i\neq j}^N (1 - z_i/z_j)^k = \prod_{i< j}^N  2 \, \Bigl|\sin\frac{x_i - x_j}{2}\Bigr|^{2k}  \, , \qquad z_j = \E^{\I\, x_j} \, .
\end{equation}
The square root of this measure is the ground-state wave function of $H^{\text{nr}}$, with energy $\varepsilon^{\text{nr}}_0 = k^2 \, N\,(N^2-1)/12$. (Cf.~\textsection\ref{s:Ruijsenaars}: $\varepsilon_0 = [N] = N + \frac{1}{2}\,(t-1)^2 \, \alpha^2 \, \varepsilon^{\text{nr}}_0 + \mathcal{O}(t-1)^3$.)

\subsection{Spin-Calogero--Sutherland limit} \label{s:app_spin_CS_limit} For the `nonrelativistic limit' $q=t^\alpha$, $t\to1$, we can use the results of \textsection\ref{s:app_CalSut_limit}. We only need to determine the limit for the long-range spin interactions from \eqref{eq:spin-D_1}. This gives
\begin{equation} \label{eq:dyn_spin_interactions_nonrlt}
	\check{R}_{(12\cdots j)}^{-1} \, \hat{t}^{\,\alpha}_j \, \check{R}_{(12\cdots j)}^{\vphantom{-1}} = 1 + \frac{1}{2} \, (t-1)^2 \; 2\,\alpha \sum_{i=1}^{j-1} \frac{-z_i\,z_j}{(z_i-z_j)^2} \, (1-P_{ij}) + \mathcal{O}(t-1)^3 \, . 
\end{equation}
The first nontrivial terms conveniently appear at order $(t-1)^2$, so we just have to add these to \eqref{eq:CS_eff} in order to get the effective Hamiltonian of the trigonometric spin-Calogero--Sutherland model. Conjugation by the ground-state wave function~\eqref{eq:psi_0_nonrlt} yields the physical Hamiltonian \eqref{eq:spin-CS} derived in \cite{BG+_93}. 

In this limit the physical space \eqref{eq:phys_space} describes bosons with spin-1/2 and coordinates~$z_j$:
\begin{equation} \label{eq:phys_space_isotropic}
	\widetilde{\mathcal{H}}^{\text{nr}} \coloneqq \, \bigcap_{i=1}^{N-1} \ker\bigl(P_{i,i+1} - s_i\bigr) \, \cong \, \mathcal{H} \underset{\mathfrak{S}_N}{\otimes} \mathbb{C}[\vect{z}] = \mathcal{H}/\mathcal{N}_{t=1} \, \qquad \mathcal{N}_{t=1} = \sum_{i=1}^{N-1} \mathrm{im}(P_{i,i+1} - s_i) \, .
\end{equation}
Physical vectors in the $M$-particle sector acquire the simple form
\begin{equation*}
	\sum_{i_1 < \cdots < i_M}^N \!\!\!\!\! s_{\{i_1,\To,i_M\}} \widetilde{\Psi}(\vect{z}) \, \cket{i_1,\To,i_M} \, , \qquad \widetilde{\Psi}(\vect{z}) \in \mathbb{C}[\vect{z}]^{\mathfrak{S}_M \times \mathfrak{S}_{N-M}} \, ,
\end{equation*}
where $s_{\{i_1,\To,i_M\}} \widetilde{\Psi}(\vect{z}) = \widetilde{\Psi}(z_{i_1},\To,z_{i_M};z_{j_1},\To,z_{j_{N-M}})$ with $j_n$ labelling the components of $\vect{z}$ not among the $z_{i_1},\To,z_{i_M}$.
The eigenvectors are thus obtained by embedding the eigenfunctions of the spinless (scalar) Calogero--Sutherland model, cf.~\textsection\ref{s:app_CalSut_limit}.

The spin-Calogero--Sutherland model inherits the quantum-affine symmetry of spin-Ruijsenaars in the form of a (double) Yangian symmetry, which can be described more explicitly than in the \textsf{q}-case, cf.~\cite{HH+_92}. The isotropic limit of the \textit{R}-matrix, cf.~\eqref{eq:baxterisation}, is
\begin{equation}
	R\bigl(u = t^{\lambda}\bigr) \, \to \, \frac{\lambda + P}{\lambda + 1} \, , \qquad t\to1 \, .
\end{equation}

\section{More about the spin side}

\subsection{Stochastic twist} \label{s:app_stochastic_twist} In \textsection\ref{s:spin} we chose a particular representation of the Hecke algebra on $\mathcal{H}=V^{\otimes N}$. Another, slightly different choice is often used too. Generalise \eqref{eq:Hecke_spin} to
\begin{equation} \label{eq:app_Hecke_spin}
	T^\text{sp} \coloneqq \begin{pmatrix} 
	\, t^{1/2} & \!\! \color{gray!80}{0} & \color{gray!80}{0} \! & \color{gray!80}{0} \, \\
	\, \color{gray!80}{0} & \!\! t^{1/2} - t^{-1/2} & t^{\epsilon/2} \! & \color{gray!80}{0} \, \\
	\, \color{gray!80}{0} & \!\! t^{-\epsilon/2} & 0 \! & \color{gray!80}{0} \, \\
	\, \color{gray!80}{0} & \!\! \color{gray!80}{0} & \color{gray!80}{0} \! & t^{1/2} \, \\
	\end{pmatrix} \, .
\end{equation}
Though this yields an action of $\mathfrak{H}_N$ on $\mathcal{H}$ for any value of $\epsilon$ we will only consider $\epsilon \in \{0,1\}$. If $\epsilon = 0$, as in the main text, $T^\text{sp}$ is symmetric (and hermitian if further $t^{1/2}\in\mathbb{R}^\times$) while for $\epsilon = 1$ its column sums are fixed. The two conventions are related by a `gauge transformation' or `stochastic twist'. Indeed, on $V \otimes V$ we have
\begin{subequations} \label{eq:gauge_transf}
	\begin{gather}
	T^\text{sp}|_{\epsilon=1} = \theta_2^{\vphantom{-1}} \, T^\text{sp}|_{\epsilon=0} \ \theta_2^{-1} \, , \qquad \theta_2 \coloneqq k_1^{1/4} \, k_2^{-1/4} = \diag(1,t^{1/4},t^{-1/4},1) \, ,
\intertext{where we recall that $k = t^{\sigma^z\mspace{-1mu}/2} = \diag(t^{1/2},t^{-1/2})$. This extends to $\mathcal{H}$ as}
	T^\text{sp}_i|_{\epsilon=1} = \theta_N^{\vphantom{-1}} \, T^\text{sp}_i|_{\epsilon=0} \, \theta_N^{-1} \, , \qquad \theta_N \coloneqq \prod_{i=1}^N k^{(N-2i+1)/4}_i \, .
	\end{gather}
\end{subequations}

For $\epsilon=0$ expressions are a bit simpler, yet $\epsilon=1$ is nice from the following viewpoint: it can be obtained from the polynomial Hecke action~\eqref{eq:Hecke_pol_2a}, as follows.\,\footnote{\ We should point out $x_{ij}$ from \eqref{eq:x_ij} likewise yields $P\,T^\text{sp}$, the \textit{R}-matrix of $\mathfrak{U}$ up to a factor of $t^{-1/2}$, yet the `physical condition' from \textsection\ref{s:physical_space} does \emph{not} imply that the two coincide on the physical space $\widetilde{\mathcal{H}}$.} Let us identify the subspace $\mathbb{C} \oplus \mathbb{C}\, z = \mathbb{C}[z]^{<2} \subset \mathbb{C}[z]$ with $V^*$ via $1 \leftrightarrow \bra{\uparrow}$ (`empty') and $z \leftrightarrow \bra{\downarrow}$ (`occupied'); the reason for the dual will become clear momentarily. Likewise, for $N=2$ the subspace $\mathbb{C}[z_1,z_2]^{<2} \subset \mathbb{C}[z_1,z_2]$ can be thought of as $\mathcal{H}^* = V^* \otimes V^*$ under the identification $1 \leftrightarrow \bra{\uparrow\uparrow}$, $z_2 \leftrightarrow \bra{\uparrow\downarrow}$, $z_1 \leftrightarrow \bra{\downarrow\uparrow}$, $z_1 \, z_2 \leftrightarrow \bra{\downarrow\downarrow}$. The operator $T_1^\text{pol}$ from  \eqref{eq:Hecke_pol_2a} preserves the total degree and thus this subspace, on which it acts by \eqref{eq:Hecke_spin} for $\epsilon=1$.

If we had used $V$ and $V \otimes V$ rather than their duals we would have obtained the transpose of \eqref{eq:Hecke_spin}. This is related to the observation that the decomposition of $V \otimes V$ into eigenspaces $V \otimes V \cong \, \mathrm{Sym}_{t}^2(V) \oplus \Lambda_t^{\!2}(V)$, 
\begin{equation*}
	\begin{aligned}
	\mathrm{Sym}_{t}^2(V) & = \mathbb{C} \, \ket{\uparrow\uparrow} \, \oplus \, \mathbb{C} \, \bigl( t^{(\epsilon+1)/4} \, \ket{\uparrow\downarrow} + t^{-(\epsilon+1)/4} \, \, \ket{\downarrow\uparrow}\bigr) \, \oplus \, \mathbb{C} \, \ket{\downarrow\downarrow} \, , \\
	\Lambda_t^{\!2}(V) & = \mathbb{C} \, \bigl(t^{(\epsilon-1)/4} \, \ket{\uparrow\downarrow}- t^{-(\epsilon-1)/4} \,\ket{\downarrow\uparrow} \bigr) \, ,
	\end{aligned}
\end{equation*}	
is somewhat unsatisfactory at $\epsilon=1$ in that $t$s feature in the \textsf{q}-symmetric (`triplet') eigenspace, rather than the \textsf{q}-antisymmetric (`singlet') eigenspace as in the polynomial case \eqref{eq:T_pol_eigenspaces}. However, the dual eigenspace decomposition is entirely analogous to \eqref{eq:T_pol_eigenspaces} when $\epsilon=1$:
\begin{equation} \label{eq:T_sp_eigenspaces_dual}
	\begin{aligned} 
	\mathrm{Sym}_{t}^2(V^*) & = \mathbb{C} \, \bra{\uparrow\uparrow} \, \oplus \, \mathbb{C} \, \bigl(t^{-(\epsilon-1)/4} \, \bra{\uparrow\downarrow}+ t^{(\epsilon-1)/4} \, \bra{\downarrow\uparrow}\bigr) \, \oplus \, \mathbb{C} \, \bra{\downarrow\downarrow} \, , \\ 
	\Lambda_t^{\!2}(V^*) & = \mathbb{C} \, \bigl( t^{-(\epsilon+1)/4} \, \bra{\uparrow\downarrow} - t^{(\epsilon+1)/4} \,\bra{\downarrow\uparrow}\bigr) \, .
	\end{aligned}
\end{equation}

In the remainder of this appendix we give the $\epsilon$-generalisations of the spin expressions from the main text. The Temperley--Lieb generators~\eqref{eq:TL_spin} now feature
\begin{equation} \label{eq:app_TL_spin}
	e^\text{sp} = \begin{pmatrix} 
	\, 0 & \color{gray!80}{0} & \! \color{gray!80}{0} & \color{gray!80}{0} \, \\
	\, \color{gray!80}{0} & \!\!\hphantom{-}t^{-1/2}\! & \!-t^{\epsilon/2}\! & \color{gray!80}{0} \, \\
	\, \color{gray!80}{0} & \!\!-t^{-\epsilon/2}\! & \!\hphantom{-}t^{1/2}\! & \color{gray!80}{0} \, \\
	\, \color{gray!80}{0} & \color{gray!80}{0} & \! \color{gray!80}{0} & 0 \, \\
	\end{pmatrix} \, .
\end{equation}
For $\epsilon = 1$ the column sums vanish: the matrix is \emph{stochastic}. This plays an important role in the connection with models in quantum-integrable stochastic models such as the asymmetric exclusion process~(\textsc{asep})~\cite{GS_92}.

Inserting $\epsilon$ as in~\eqref{eq:gauge_transf} the action~\eqref{eq:Uqsl2_spin} of $\mathfrak{U}$ on $\mathcal{H}$ becomes
\begin{equation} \label{eq:app_Uqsl2_spin}
\begin{aligned}
	E_1^\text{sp} & = \sum_{i=1}^N t^{+\epsilon\,(N-2\mspace{1mu}i+1)/4} \mspace{4mu} k_1 \cdots k_{i-1} \, \sigma^+_i \, , \\
	F_1^\text{sp} & = \sum_{i=1}^N t^{-\epsilon\,(N-2\mspace{1mu}i+1)/4} \, \sigma^-_i \, k^{-1}_{i+1} \cdots k^{-1}_N \, , \\
	\end{aligned} 	
	\qquad\quad
	K_1^\text{sp} = k_1 \cdots k_N \, .
\end{equation}
Any operator that is $\mathfrak{U}$-invariant and annihilates $\cbra{\varnothing} = \bra{\uparrow\cdots\uparrow}$ is stochastic. Indeed,
\begin{equation*}
	(1,1,\To,1) = \bra{\uparrow\cdots\uparrow} \exp(S^+) = \bra{\uparrow\cdots\uparrow} \, \sum_{n=0}^N \frac{t^{-n\,(N-n)/4}}{[n]!} \, \bigl(E_1^\text{sp}|_{\epsilon=1}\bigr)^n
\checkedMma
\end{equation*}
then is a $t$-independent (left) eigenvector with eigenvalue zero. But acting on this vector from the right is nothing but computing the column sums. 

Now we move to the affine setting (\textsection\ref{s:quantum-affine}). The (minimal) affinisation~\eqref{eq:affinisation_z} becomes
\begin{equation} \label{eq:app_affinisation_z}
	\begin{aligned}
	E_0^\text{inh} & = \sum_{i=1}^N t^{-\epsilon\,(N-2i+1)/4} \ \, z_i \ k^{-1}_1 \cdots k^{-1}_{i-1} \,  \sigma^-_i \, , \\
	F_0^\text{inh} & = \sum_{i=1}^N t^{+\epsilon\,(N-2i+1)/4} \ z_i^{-1} \, \sigma^+_i \,  k_{i+1} \cdots k_N \, , \\
	\end{aligned}
	\qquad\quad 
	K_0^\text{inh} = k_1^{-1} \cdots k_N^{-1} \, .
\checkedMma
\end{equation}
Baxterisation~\eqref{eq:baxterisation} gives the \textit{R}-matrix
\begin{equation} \label{eq:app_baxterisation}
	\check{R}(u) = t^{1/2} \, \frac{u \, T^\text{sp} - T^{\text{sp}\,-1}}{t \, u - 1} = f(u)\, T^\text{sp} + g(u) = 
	\begin{pmatrix} 
	\, 1 & \color{gray!80}{0} & \color{gray!80}{0} & \color{gray!80}{0} \, \\
	\, \color{gray!80}{0} & u\,g(u) & \!t^{\epsilon/2} f(u)\! & \color{gray!80}{0} \, \\
	\, \color{gray!80}{0} & \!t^{-\epsilon/2} f(u)\! & g(u) & \color{gray!80}{0} \, \\
	\, \color{gray!80}{0} & \color{gray!80}{0} & \color{gray!80}{0} & 1 \, \\
	\end{pmatrix}
	\, .
\end{equation}
If $\epsilon =0$ the \textit{R}-matrix is hermitian when $t^{1/2} \in \mathbb{R}^\times$ and $u\in S^1 \subset\mathbb{C}$; for $\epsilon=1$ its column sums equal unity. The monodromy matrix~\eqref{eq:L_inh} is defined as before. Its quantum determinant
\begin{equation*}
	\begin{aligned}
	\qdet_a L_a^\text{inh}(u) 
	& = A(t \, u) \, D(u) - t^{(1-\epsilon)/2} \, B(t\,u) \, C(u) = D(t \, u) \, A(u) - t^{-(1-\epsilon)/2} \, C(t\,u) \, B(u)  \\
	& = A(u) \, D(t \, u) - t^{(1+\epsilon)/2} \, C(u) \, B(t \, u) = D(u) \, A(t \, u) - t^{-(1+\epsilon)/2} \, B(u) \, C(t\,u)
	\end{aligned}
\end{equation*}
is, perhaps somewhat surprisingly, proportional to $K_1^\text{sp}$ when $\epsilon=1$:
\begin{equation*}
	\qdet_a L_a^\text{inh}(u)  = \prod_{i=1}^N \qdet_a R_{ai}(u/z_i) = t^{N/2} \prod_{i=1}^N \frac{u-z_i}{t \, u - z_i} \ (K_1^\text{sp})^\epsilon \, .
\end{equation*}
Indeed, \eqref{eq:gauge_transf} holds for $\check{R}$ too. But $P \, \theta_2 = \theta_2^{-1} \, P$ so $R(u) = P \, \check{R}(u)$ obeys $R(u)|_{\epsilon=0} = \theta_2 \, R(u)|_{\epsilon=1} \, \theta_2$. Thus $\qdet{\!}_a R_{ai}(u)|_{\epsilon=1} = k_i \qdet{\!}_a R_{ai}(u)|_{\epsilon=0}$ since $\qdet \theta_2 = k^{1/2}$.

Now we move to the spin chain. As any operator built from the Hecke operators or \textit{R}-matrix the spin-chain Hamiltonians inherit the property \eqref{eq:gauge_transf}. For $\epsilon =0$ the Hamiltonian $H^\textsc{l}|_{\epsilon=0}^\textsc{t} = H^\textsc{l}|_{\epsilon=0}$ is hermitian~\cite{Lam_18} if $t^{1/2} \in \mathbb{R}^\times$. For $\epsilon=1$, instead, it is stochastic. The entries of our Hamiltonian depend on the coordinates~$z_j$ and are complex in general, so it is less clear how they can be interpreted as transition amplitudes; though probabilistic models with complex weights have been considered in the literature~\cite{PRV_20}.

The choice $\epsilon = 1$ allows for a simple way to understand the very mild dependence of the Hamiltonian on the sign of $t^{1/2} = \mathsf{q}$, reflected in the dependence of \eqref{eq:intro_polynomial} on $t = \mathsf{q}^2$ rather $t^{1/2}$. Indeed, the potential~\eqref{eq:pot} clearly depends on $t$. The same is true for the \textit{R}-matrix \eqref{eq:app_baxterisation} when $\epsilon=1$. The Temperley--Lieb generator \eqref{eq:app_TL_spin} for $\epsilon=1$ instead acquires a sign if $\mathsf{q} \mapsto -\mathsf{q}$, while $[N] \mapsto (-1)^{N+1} \, [N]$. This proves that
\begin{equation} \label{eq:ham_left_q_to_-q}
	H^\textsc{l}\big|_{\mathsf{q} \mapsto -\mathsf{q}} = (-1)^N \, H^\textsc{l} \qquad \text{if} \quad \epsilon = 1 \,.
\end{equation}

The physical vectors (\textsection\ref{s:physical_space}) change a little as well. By \eqref{eq:app_Hecke_spin} we now have $T^\text{sp}\,\ket{\downarrow\uparrow} = t^{\epsilon/2} \, \ket{\uparrow\downarrow}$ and $\bra{\downarrow\uparrow} \, T^\text{sp} = t^{-\epsilon/2} \bra{\uparrow\downarrow}$. Hence \eqref{eq:phys_vector} becomes
\begin{equation*}
	\sum_{w \in \mathfrak{S}_N/(\mathfrak{S}_M \times \mathfrak{S}_{N-M})} \!\!\!\!\!\!\!\!\!\!\!\!\!\!\!\! t^{\epsilon\, \ell(w)/2} \, T_w^\text{pol} \, \widetilde{\Psi}(\vect{z}) \, \cket{w\,1,\To,w\,M} 
\checkedMma
\end{equation*}
and
\begin{equation*}
	\sum_{w \in \mathfrak{S}_N/(\mathfrak{S}_M \times \mathfrak{S}_{N-M})} \!\!\!\!\!\!\!\!\!\!\!\!\!\!\!\! t^{-\epsilon\, \ell(w)/2} \,  \cbra{w\,1,\To,w\,M} \, T_w^\text{pol} \, \widetilde{\Psi}(\vect{z}) 
\checkedMma
\end{equation*}
where $w = \{i_1 ,\To, i_M\}$ from \eqref{eq:perm_grassmannian} has length $\ell(w) = \sum_{m=1}^M (i_m - m)$.

Finally consider the crystal limit (\textsection\ref{s:intro_crystal_limit}). For $\epsilon = 1$ the limits
\begin{equation*}
	\frac{e^\text{sp}}{[2]} \to 
	\begin{cases}
	\bigl( \ket{\downarrow\uparrow} - \epsilon \, \ket{\uparrow\downarrow}\bigr) \, \bra{\downarrow\uparrow} \, , \qquad & t \to \infty \, , \\
	\bigl( \ket{\uparrow\downarrow} - \epsilon \, \ket{\downarrow\uparrow}\bigr) \, \bra{\uparrow\downarrow} \, , \qquad & t \to 0 \, ,
	\end{cases} 
\end{equation*}
are no longer diagonal to ensure stochasticity. This is inherited by the long-range spin interactions, with a strictly triangular part when $\epsilon=1$ that depends on~$j$ for $S_{[i,j]}^{\mspace{1mu}\textsc{l}}$ and on $i$ for $S_{[i,j]}^{\mspace{1mu}\textsc{r}}$. The Hamiltonians thus acquire off-diagonal corrections to \eqref{eq:crystal_ham}.

\subsection{Relation between presentations} \label{s:app_presentations}
Consider the six-vertex \textit{R}-matrix
\begin{equation*}
	R(u) = P \check{R}(u) = \frac{1}{t\,u-1}
	\begin{pmatrix} 
	\, t \,u-1 & \color{gray!80}{0} & \color{gray!80}{0} & \color{gray!80}{0} \, \\
	\, \color{gray!80}{0} & t^{(1-\epsilon)/2} \, (u-1) & t-1 & \color{gray!80}{0} \, \\
	\color{gray!80}{0} & (t-1)\,u & t^{(1+\epsilon)/2} \, (u-1) & \color{gray!80}{0} \, \\
	\, \color{gray!80}{0} & \color{gray!80}{0} & \color{gray!80}{0} & t \, u-1 \, \\
	\end{pmatrix}
	\, ,
\end{equation*}
where we included $\epsilon \in \{ 0,1\}$ in accordance with \textsection\ref{s:app_stochastic_twist}.
In \textsection\ref{s:quantum-affine} we consider the monodromy matrix~\eqref{eq:L_inh} for $N$ sites with inhomogeneities $z_1,\To,z_N$,
\begin{equation} \label{eq:app_monodromy_spin}
	L_a^\text{inh}(u) = L_a^\text{inh}(u;\vect{z}) = R_{aN}(u/z_N) \cdots R_{a1}(u/z_1) \, .
\end{equation}
This operator obeys the \textit{RLL}-relations, yielding a finite-dimensional representation of $\widehat{\mathfrak{U}}$ in the \textsc{frt} presentation. 

The Chevalley generators of the Drinfeld--Jimbo presentation arise by expanding in~$u$ around $0$ and $\infty$. Viewed as a (formal) power series in $u^{\pm1}$ the \textit{R}-matrix has the form
\begin{equation} \label{eq:app_R-expansion} 
	\begin{aligned}
	R(u) & =
	\begin{pmatrix} 
	\, t^{(1-\epsilon)/4} \, k^{-(1-\epsilon)/2} + \mathcal{O}(u) & (1-t) \, \sigma^- + \mathcal{O}(u) \, \\
	\, (1-t) \, u\,\sigma^+ + \mathcal{O}(u^2) & t^{(1+\epsilon)/4}\,k^{(1+\epsilon)/2} + \mathcal{O}(u) \, \\
	\end{pmatrix} 
	\\
	& =
	\begin{pmatrix} 
	\, t^{-(1+\epsilon)/4} \, k^{(1+\epsilon)/2} + \mathcal{O}(u^{-1}) & (1-t^{-1}) \, u^{-1} \, \sigma^- + \mathcal{O}(u^{-2}) \, \\
	\, (1-t^{-1}) \, \sigma^+ + \mathcal{O}(u^{-1}) & t^{-(1-\epsilon)/4} \, k^{-(1-\epsilon)/2} + \mathcal{O}(u^{-1}) \, \\
	\end{pmatrix}
	\, ,
	\end{aligned}
\checkedMma
\end{equation}
where $k = \diag(t^{1/2},t^{-1/2})$. (If we had removed the denominator of the \textit{R}-matrix these would correspond to the highest and lowest orders in $u$.) Consider the Gauss decomposition of the monodromy matrix likewise expanded in $u^{\pm 1}$ to get two operators
\begin{equation*}
	L_a^{\pm}(u) = 
	\begin{pmatrix} 
	\, A^{\pm}(u) & B^{\pm}(u) \, \\
	\, C^{\pm}(u) & D^{\pm}(u) \, \\
	\end{pmatrix}_{\!a} 
	=
	\begin{pmatrix} 
	\, 1 & \mathsf{f}^\pm(u) \, \\
	\, 0 & 1 \,
	\end{pmatrix}_{\!a}  
	\begin{pmatrix} 
	\, \mathsf{k}_0^\pm(u) & 0 \, \\
	\, 0 & \mathsf{k}_1^\pm(u) \,
	\end{pmatrix}_{\!a} 
	\begin{pmatrix} 
	\, 1 & 0 \, \\
	\, \mathsf{e}^\pm(u) & 1 \,
	\end{pmatrix}_{\!a}	\, .
\end{equation*}
For \eqref{eq:app_monodromy_spin} use \eqref{eq:app_R-expansion} to recover the $\widehat{\mathfrak{U}}$-representation \eqref{eq:app_Uqsl2_spin}, \eqref{eq:app_affinisation_z} at the lowest order: 
\begin{equation*}
	\begin{aligned}
	\mathsf{e}^{+,\mspace{1mu}\text{sp}}(u) & = {-}t^{-\epsilon\,(N+1)/4} \, (t-1) \, u \, K_0^\text{inh} \, F_0^\text{inh} \!\!\!\!\! && + \mathcal{O}(u^2) \, , \\
	\mathsf{f}^{+,\mspace{1mu}\text{sp}}(u) & = -t^{-\epsilon\,(N+1)/4}\,(t^{1/2}-t^{-1/2}) \, F_1^\text{sp} && + \mathcal{O}(u) \, , \\
	\mathsf{k}_0^{+,\mspace{1mu}\text{sp}}(u) & = \hphantom{-}t^{(1-\epsilon)\,N/4} \, \bigl(\mspace{-1mu} K_0^\text{inh} \bigr)^{(1-\epsilon)/2} && + \mathcal{O}(u) \, , \\
	\mathsf{k}_1^{+,\mspace{1mu}\text{sp}}(u) & = \hphantom{-}t^{(1+\epsilon)\,N/4} \, \bigl(\mspace{-1mu}K_1^\text{sp}\bigr)^{(1+\epsilon)/2} && + \mathcal{O}(u) \, ,
	\end{aligned}
\checkedMma
\end{equation*}
and % cf CP p431 rmk 4
\begin{equation*}
	\begin{aligned}
	\mathsf{e}^{-,\mspace{1mu}\text{sp}}(u) & = t^{-\epsilon\,(N+1)/4} \,(t^{1/2}-t^{-1/2}) \, E_1^\text{sp} && + \mathcal{O}(u^{-1}) \, , \\
	\mathsf{f}^{-,\mspace{1mu}\text{sp}}(u) & = t^{-\epsilon\,(N+1)/4} \, (1-t^{-1}) \, u^{-1} \, E_0^\text{inh} \, \bigl(\mspace{-1mu}K_0^\text{inh}\bigr)^{\!-1} \!\!\!\!\!\! && + \mathcal{O}(u^{-2}) \, , \\
	\mathsf{k}_0^{-,\mspace{1mu}\text{sp}}(u) & = t^{-(1+\epsilon)\,N/4} \, \bigl(\mspace{-1mu} K_0^\text{inh}\bigr)^{-(1+\epsilon)/2} && + \mathcal{O}(u^{-1}) \, , \\
	\mathsf{k}_1^{-,\mspace{1mu}\text{sp}}(u) & = t^{-(1-\epsilon)\,N/4} \, \bigl(\mspace{-1mu} K_1^\text{sp}\bigr)^{-(1-\epsilon)/2} && + \mathcal{O}(u^{-1}) \, .
	\end{aligned}
\checkedMma
\end{equation*}
See also \textsection{A} of \cite{JK+_95b}, which contains the Drinfeld presentation via currents as well.
\bigskip

\section{Glossary} \label{s:app_glossary}

\setlist[itemize]{leftmargin=*,label={\quad},itemsep=2pt}

\noindent The symbols in this glossary are ordered (approximately) alphabetically according to the English spelling of its pronunciation; for example, $\omega$ can be found under `o'.
\bigskip

\begin{itemize}
	\item $A(u)$ \ quantum operator (entry of $L_a(u)$): \ \eqref{eq:intro_ABCD_on_shell}, \eqref{eq:ABCD}, \eqref{eq:ABCD_Y_pol}
	\item $a_{ij} \coloneqq a(z_i/z_j)$, $a(u) = 1/f(u)$ \ rational function: \ \eqref{eq:a,b}
	\item affine Hecke algebra (\textsc{aha}) \ $\widehat{\mathfrak{H}}_N = \widehat{\mathfrak{H}}_N(t^{1/2})$: \ \eqref{eq:AHA}, \eqref{eq:AHA'}
	\item $A_J(\vect{z}), A_j(\vect{z}), A_{-I}(\vect{z}), A_{-i}(\vect{z})$ \ coefficients of Macdonald operators: \ \textsection\ref{s:intro_abelian_tilde}, \textsection\ref{s:Macdonald}, \textsection\ref{s:abelian_tilde}
	\item $\alpha = 1/k$ \ Jack parameter ($\mathsf{p} = \mathsf{q}^{2\alpha}$, $q = t^\alpha$): \ \textsection\ref{s:intro_spin-CS}, Figure~\ref{fg:Macdonalds}, \textsection\ref{s:app_nonrlt_limit}
	\item $\alpha,\beta$ \ compositions: \ p.\,\pageref{eq:dominance}
	\begin{itemize}
		\item $\alpha^+$ \ corresponding partition
	\end{itemize}
	\item $\alpha^\mu(u)$ \ eigenvalue of $A(u)$: \ \eqref{eq:Drinfeld_polyn}, \textsection\ref{s:hw} 
\bigskip
	\item $\mathcal{B} \cong \widetilde{\mathcal{H}}$ \ physical (\textsf{q}-bosonic) space: \ \textsection\ref{s:intro_phys_space}, \textsection\ref{s:physical_space}: \eqref{eq:phys_space}, \eqref{eq:phys_space_via_Ttot}, \eqref{eq:phys_space_alt}
	\item $B(u)$ \ quantum operator (entry of $L_a(u)$): \ \eqref{eq:intro_ABCD_on_shell}, \eqref{eq:ABCD}, \eqref{eq:ABCD_Y_pol}
	\item $b_{ij} \coloneqq b(z_i/z_j)$, $b(u) = {-}g(u)/f(u)$ \ rational function: \ \eqref{eq:a,b}
\bigskip
	\item $C(u)$ \ quantum operator (entry of $L_a(u)$): \ \eqref{eq:intro_ABCD_on_shell}, \eqref{eq:ABCD}, \eqref{eq:ABCD_Y_pol}
	\item $\mathbb{C}[\vect{z}] \coloneqq \mathbb{C}[z_1,\To,z_N]$ \ ring of polynomials: \ start of \textsection\ref{s:polynomial}
	\begin{itemize}
		\item $\mathbb{C}[\vect{z}]_M \coloneqq \mathbb{C}[z_1,\To,z_N]^{\mathfrak{S}_M \times \mathfrak{S}_{N-M}}$ \ polynomial analogue of $\mathcal{H}_M$: \ \eqref{eq:phys_space_M_particle}
		\item $\mathbb{C}[\vect{z}]^{\mathfrak{S}_N} \coloneqq \mathbb{C}[z_1,\To,z_N]^{\mathfrak{S}_N}$ \ ring of symmetric polynomials: \ start of \textsection\ref{s:polynomial}, \eqref{eq:Hecke_projectors_pol}, \eqref{eq:phys_space_scalar}
	\end{itemize}
	\item $\chi \coloneqq \ket{\downarrow\uparrow} \bra{\downarrow\uparrow} = \diag(0,0,1,0)$ : \ \textsection\ref{s:intro_crystal_limit}
	\item Chevalley generators: \ \eqref{eq:Uqsl2}--\eqref{eq:Uqsl2_coprod}, start of \textsection\ref{s:quantum-affine}, \textsection\ref{s:app_presentations}
	\begin{itemize}
		\item $E_0^\text{inh}$, $F_0^\text{inh}$, $K_0^\text{inh}$ \ minimal affinisation: \ \eqref{eq:affinisation_z}, \eqref{eq:app_affinisation_z}
		\item $\widetilde{E}_0^\text{sp}$, $\widetilde{F}_0^\text{sp}$, $\widetilde{\!K}_0^\text{sp}$ \ with $Y_i$ instead of $1/z_i$: \ \eqref{eq:affinisation_Y}
		\begin{itemize}
			\item $\widehat{E}^\text{pol}_M$, $\widehat{F}^\text{pol}_M$, $\,\widehat{\!K}_{\!M}^\text{pol}$ \ its induced action on polynomials: \ \eqref{eq:Uqaff_induced_polynomial}, \eqref{eq:affinisation_pol}, \eqref{eq:affinisation_pol_alt}
		\end{itemize}
		\item $E_1^\text{sp}$, $F_1^\text{sp}$, $K_1^\text{sp}$ \ spin representation: \ \eqref{eq:Uqsl2_spin}, \eqref{eq:app_Uqsl2_spin}
		\begin{itemize}
			\item $E^\text{pol}_M$, $F^\text{pol}_M$, $K^\text{pol}_M$ \ its induced action on polynomials: \ \eqref{eq:Uqaff_induced_polynomial}, \eqref{eq:Uqsl2_pol}, \eqref{eq:Uqsl2_pol_alt}
		\end{itemize}
	\end{itemize}
	\item $\,\cdot\,^\circ \coloneqq \,\cdot\,|_{q=1}$ \ classical limit: \ \eqref{eq:semiclass_expansion}: \ \textsection\ref{s:intro_nonabelian_2}, start of \textsection\ref{s:intro_freezing}, Table~\ref{tb:abelian_symmetries}, \eqref{eq:semiclass_expansion}
	\item coordinate basis \ $\cket{i_1,\To,i_M} = \sigma^-_{i_1} \cdots \sigma^-_{i_M} \, \ket{\uparrow\cdots\uparrow}$: \ \eqref{eq:coord_basis}
\bigskip
	\item $D(u)$ \ quantum operator (entry of $L_a(u)$): \ \eqref{eq:intro_ABCD_on_shell}, \eqref{eq:ABCD}, \eqref{eq:ABCD_Y_pol}
	\item $\delta_M \coloneqq (M-1,M-2,\cdots)$ \ staircase partition of length~$M-1$: \ \eqref{eq:motifs_vs_partitions}, p.\,\pageref{eq:largest_monomial}, p.\,\pageref{eq:Macdonald_shift_property}
	\item $\Delta_t(\vect{z})$ \ \textsf{q}-$\mspace{-2mu}$Vandermonde polynomial: \ \eqref{eq:q-Vand}
	\item $\Delta(u)$ \ generating function of Macdonald operators: \ \eqref{eq:Delta}
	\begin{itemize}
		\item $\widetilde{\Delta}(u)$ \ its spin-generalisation: \ \eqref{eq:Delta_tilde}
	\end{itemize}
	\item $\delta \coloneqq \partial/\partial q|_{q=1}$ \ semiclassical limit: \ start of \textsection\ref{s:intro_freezing}, \eqref{eq:semiclass_expansion}
	\item $\delta^\mu(u)$ \ eigenvalue of $A(u)$: \ \eqref{eq:Drinfeld_polyn}, \textsection\ref{s:hw}
	% \item degenerate affine Hecke algebra: \ \textsection\ref{s:app_CalSut_limit}
	\item $d_i$ \ Dunkl operator (nonrelativistic limit of $Y_i$): \ \eqref{eq:Dunkl}
	\item dominance order: \ \eqref{eq:dominance}
	\item $D_{\pm r}$ \ Macdonald operators: \ \textsection\ref{s:Macdonald}
	\begin{itemize}
		\item $\widetilde{D}_{\pm r}$ \ their spin-generalisations: \ \textsection\ref{s:intro_abelian_tilde}, \textsection\ref{s:abelian_tilde}
		\item $D_{\pm r}^\text{Rui}$ \ Ruijsenaars's hermitian form: \ \eqref{eq:Macd_to_Ruij}, \eqref{eq:CS_eff}
		\item their eigenvalues \ $\Lambda_r(\lambda)$: \ \eqref{eq:Macd_eigenvalues}, \textsection\ref{s:abelian_freezing}
		\item their nonrelativistic limit: \ \eqref{eq:CS_eff}
	\end{itemize}
	\item Drinfeld polynomial \ $P^\mu(u)$: \ \eqref{eq:Drinfeld_polyn}, \eqref{eq:intro_ABCD_on_shell}, \textsection\ref{s:hw}
\bigskip
	\item $E_\alpha(\vect{z}) \coloneqq E_\alpha(\vect{z};q,t)$ \ nonsymmetric Macdonald polynomial: \ \eqref{eq:nonsymm_Macdonalds}
	\item $E^\bullet(\mu)$ \ spin-chain energy 
	\begin{itemize}
		\item $E^\text{full}(\mu)$ \ for $H^\text{full} = (H^\textsc{l} + H^\textsc{r})/2$: \ Figure~\ref{fg:dispersions}, \eqref{eq:dispersion_full}
		\item $E^{\mspace{1mu}\textsc{l}}(\mu)$ \ for $H^\textsc{l}=H_1$: \ \eqref{eq:energy_left}, Figure~\ref{fg:dispersions}, \textsection\ref{s:abelian_freezing}
		\item $E^\textsc{hs}(\mu)$ \ for Haldane--Shastry: \ \eqref{eq:HS_dispersion}
		\item $E_r(\mu)$ \ for $H_r$: \ \eqref{eq:energy_r}, cf.~\eqref{eq:energy_r_full}
		\item $E^{\mspace{1mu}\textsc{r}}(\mu)$ \ for $H^\textsc{r} = H_{N-1}$: \ \eqref{eq:energy_right}, Figure~\ref{fg:dispersions}, \textsection\ref{s:abelian_freezing}
	\end{itemize}
	\item $E_\bullet^\bullet$: \ \textit{see} Chevalley generators
	\item $e_i$ ($e_i^\text{sp}$) \ (spin representation of) Temperley--Lieb generator: \ \eqref{eq:R'(1)}, \eqref{eq:TL}--\eqref{eq:TL_spin}, \eqref{eq:TL_vs_XXZ}
	\item $e_\lambda(\vect{z}), e_r(\vect{z})$ \ elementary symmetric polynomials: \ \eqref{eq:elementary}
	\begin{itemize}
		\item $e_r(\vect{Y})$ \ elementary symmetric polynomials in the $Y_i$: \ \eqref{eq:Delta}
		\item its evaluation: \ \eqref{eq:ev_elementary}
	\end{itemize}
	\item $\epsilon \in \{0,1\}$ \ parameter in \textsection\ref{s:app_stochastic_twist}
	\item $\varepsilon^\bullet(\mu_m)$ \ spin-chain dispersion: \ \textit{see} $E^\bullet$
	\item $\cket{\varnothing} = \ket{\uparrow\cdots\uparrow}$ \ pseudovacuum: \ \textit{see} $\cket{i_1,\To,i_M}$
	\item $\ev: z_j \mapsto \omega^j$, $\omega = \E^{2\pi\I/N}$, \ evaluation: \ \eqref{eq:ev}, \eqref{eq:ev_operator}, \eqref{eq:ev_elementary}--\eqref{eq:ev_products}
	\begin{itemize}
		\item $\widetilde{O}_1 \overset{\mathrm{ev}}{=} \widetilde{O}_2$ \ on-shell equality: \ p.\,\pageref{eq:ev_operator}
	\end{itemize}
\bigskip
	\item $F_\bullet^\bullet$: \ \textit{see} Chevalley generators
	\item $f_{ij} \coloneqq f(z_i/z_j)$, $f(u) = 1/a(u)$ \ rational function: \ \eqref{eq:R_mat}, \eqref{eq:f,g}
\bigskip
	\item $G$ \ \textsf{q}-translation operator: \ \eqref{eq:q-translation}, Proposition~\ref{prop:q-translation}
	\begin{itemize}
		\item $\bar{G}$ \ its crystal limit: \ \eqref{eq:crystal_q-translation}
	\end{itemize}
	\item $g_{ij} \coloneqq g(z_i/z_j)$, $g(u) = {-}b(u)/a(u)$ \ rational function: \ \eqref{eq:R_mat}, \eqref{eq:f,g}
\bigskip
	\item $\mathcal{H} \coloneqq V^{\otimes N}$, $V \coloneqq \mathbb{C}\,\ket{\uparrow} \oplus \mathbb{C}\,\ket{\downarrow}$, \ spin-chain Hilbert space: \ start of \textsection\ref{s:intro_HS}, start of \textsection\ref{s:spin}
	\begin{itemize}
		\item $\mathcal{H}_M \coloneqq \ker\bigl[S^z - \bigl(\tfrac{1}{2} N-M\bigr)\bigr]$ \ its \textit{M}-particle sector (weight space): \ \eqref{eq:weight_decomp}
		\item $\mathcal{H}^\mu$ \ joint eigenspace of abelian symmetries, irrep for nonabelian symmetries: \ \textsection\ref{s:intro_nonabelian_1}
	\end{itemize}
	\item $\widetilde{\mathcal{H}} \cong \mathcal{B}$ \ physical (\textsf{q}-bosonic) space: \ \textsection\ref{s:intro_phys_space}, \textsection\ref{s:physical_space}: \eqref{eq:phys_space}, \eqref{eq:phys_space_via_Ttot}, \eqref{eq:phys_space_alt}
	\begin{itemize}
		\item $\widetilde{\mathcal{H}}_M$ \ its \textit{M}-particle sector (weight space): \ \eqref{eq:phys_space_weight}, \eqref{eq:phys_space_M_particle}
		\item $\widetilde{\mathcal{H}}^\text{nr}$ \ its nonrelativistic/isotropic limit: \ \eqref{eq:phys_space_isotropic}
		\item $\widetilde{\mathcal{H}}[\vect{z}] \coloneqq \mathcal{H} \otimes \mathbb{C}[z_1,\To,z_N]$ \ its ambient vector space: \ \eqref{eq:big_space}		
	\end{itemize}
	\item $H^\bullet$ Hamiltonian
	\begin{itemize}
		\item $H^{\text{eff},\mspace{1mu}\text{nr}}$ \ effective Calogero--Sutherland: \ \eqref{eq:CS_eff}
		\item $H^\text{full} = (H^\textsc{l} + H^\textsc{r})/2$ \ \textsf{q}-deformed Haldane--Shastry (full Hamiltonian): \ \eqref{eq:ham_full}
		\item $H^\textsc{l} = H_1$ \ \textsf{q}-deformed Haldane--Shastry: \ \eqref{eq:ham_left}, \eqref{eq:ham_left_from_freezing}
		\begin{itemize}
			\item $\bar{H}^\textsc{l}$ \ its crystal limit: \ \eqref{eq:crystal_ham}, \eqref{eq:crystal_dispersion}
			\item its braid limit: \ \eqref{eq:ham_formal_limit}
			\item its stochastic version: \ \textsection\ref{s:app_stochastic_twist}
			\item its symmetry in $\mathsf{q}\mapsto -\mathsf{q}$: \ \eqref{eq:ham_left_q_to_-q}
		\end{itemize}
		\item $\widetilde{H}^{\text{nr}}$ \ spin-Calogero--Sutherland: \ \eqref{eq:spin-CS}
		\item $H^\textsc{hs}$ \ ($\mathsf{q}=1$) Haldane--Shastry: \ \eqref{eq:HS}
		\item $H^\textsc{r} = H_{N-1}$ \ \textsf{q}-deformed Haldane--Shastry (opposite chirality): \ \eqref{eq:ham_right}, \eqref{eq:ham_right_from_freezing}
		\begin{itemize}
			\item $\bar{H}^\textsc{r}$ \ its crystal limit: \ \eqref{eq:crystal_ham}, \eqref{eq:crystal_dispersion}
		\end{itemize}
		\item $H_r$ \ higher Hamiltonians: \ \eqref{eq:H_r}, \eqref{eq:H_2}
		\item $H^\textsc{xxz}$ \ Heisenberg Hamiltonians: \ \textsection\ref{s:spin_chains}, cf.~\eqref{eq:ham_formal_limit}
	\end{itemize}
	\item $\mathfrak{H}_N \coloneqq \mathfrak{H}_N(t^{1/2})$ \ Hecke algebra: \ \eqref{eq:Hecke}
	\item $\widehat{\mathfrak{H}}_N \coloneqq \widehat{\mathfrak{H}}_N(t^{1/2})$ \ affine Hecke algebra (\textsc{aha}): \ \eqref{eq:AHA}, \eqref{eq:AHA'}
\bigskip
	\item $\cket{i_1,\To,i_M} \coloneqq \sigma^-_{i_1} \cdots \sigma^-_{i_M} \, \ket{\uparrow\cdots\uparrow}$ \ coordinate basis: \ \eqref{eq:coord_basis}
	\item $\{i_1,\To,i_M \} \in \mathfrak{S}_N/(\mathfrak{S}_M \times \mathfrak{S}_{N-M})$ \ Grassmannian permutation: \ \eqref{eq:perm_grassmannian}, \eqref{eq:phys_vector}
	\item induced action on polynomials: \ : \ Corollary to Proposition~\ref{prop:intro_physical_vectors}, \eqref{eq:induced_pol}
\bigskip
	\item Jack polynomial: \ \textit{see} $P_\lambda^{(\alpha)}(\vect{z})$
\bigskip
	\item $K_\bullet^\bullet$: \ \textit{see} Chevalley generators
	\item $k = t^{\sigma^z\mspace{-1mu}/2} = \diag(t^{1/2},t^{-1/2})$: \ \textsection\ref{s:Hecke_TL_Uqsl}
	\item $k = 1/\alpha$ \ reduced coupling ($\mathsf{p} = \mathsf{q}^{2/k}$, $q = t^{1/k}$): \ \eqref{eq:spin-CS}, \textsection\ref{s:app_nonrlt_limit}
\bigskip
	\item \textit{L}-operator: \ \textsection\ref{s:quantum-affine}, \textsection\ref{s:app_presentations}
	\begin{itemize}
		\item $L_a(u) = \ev \widetilde{L}_a^\circ(u)$ \ for \textsf{q}-deformed Haldane--Shastry: \ \eqref{eq:intro_L_spin_chain}, \eqref{eq:intro_L_spin_chain_def}, \eqref{eq:L_spin_chain}
		\item $L_a^\text{inh}(u) = R_{aN}(u/z_N) \cdots R_{a1}(u/z_1)$ \ for inhomogeneous \textsc{xxz}: \ \eqref{eq:L_inh}, \eqref{eq:app_monodromy_spin}
		\item $\widetilde{L}_a(u) = R_{aN}(u \, Y_N) \cdots R_{a1}(u\,Y_1)$ \ for spin-Ruijsenaars: \ \eqref{eq:L_tilde}, \eqref{eq:ABCD_Y_pol}
	\end{itemize}
	\item $\lambda$ \ partition: \ p.\,\pageref{eq:HS_polynomial}
	\begin{itemize}
		\item $\lambda' = \mu^+$ \ (conjugate) partition associated to  $\mu \in \mathcal{M}_N$: \ \eqref{eq:motifs_vs_partitions_lambda}
		\item $|\lambda| \coloneqq \sum \lambda_i$ \ its weight
		\item $\lambda \geq \nu$ dominance ordering: \ \eqref{eq:dominance}
	\end{itemize}
	\item $\Lambda_r(\lambda)$ \ eigenvalues of Macdonald operators: \ \eqref{eq:Macd_eigenvalues}, \textsection\ref{s:abelian_freezing}
	\item $\ell(\mu),\ell(\lambda)$ \ length (number of nonzero parts)
\bigskip
	\item \textit{M}-particle sector (weight space) 
	\begin{itemize}
		\item of spin-chain Hilbert space \ $\mathcal{H}_M \coloneqq \ker\bigl[S^z - \bigl(\tfrac{1}{2} N-M\bigr)\bigr]$: \ \eqref{eq:weight_decomp} 
		\item of physical space \ $\widetilde{\mathcal{H}}_M$: \ \eqref{eq:phys_space_weight}, \eqref{eq:phys_space_M_particle}
	\end{itemize}
	\item $m_\lambda(\vect{z})$ \ monomial symmetric polynomial: \ \eqref{eq:monomial}, cf.~\eqref{eq:Macd_orthog}
	\item $\mathcal{M}_N$ \ set of all motifs (for $N$ sites): \ \eqref{eq:motif}
	\item Macdonald operator: \ \textit{see} $D_r$, \textsection\ref{s:Macdonald}
	\item Macdonald polynomial: \ \textit{see} $E_\alpha(\vect{z})$ (nonsymmetric), $P_\lambda(\vect{z})$ (symmetric)	
	\item monodromy matrix: \ \textit{see} \textit{L}-operator
	\item $\mu \in \mathcal{M}_N$ \ motif: \ \eqref{eq:motif}
	\begin{itemize}
		\item $\mu^+ = \lambda' = \nu + 2\,\delta_{\ell(\mu)}$ \ corresponding partition: \ \eqref{eq:motifs_vs_partitions}, \eqref{eq:motifs_vs_partitions_lambda}
		\item $|\mu| \coloneqq \sum \mu_i$ \ its weight
		\item $\ket{\mu}$ \ our pseudo highest-weight eigenvectors: \ \eqref{eq:intro_wavefn}, \eqref{eq:intro_polynomial}
		\item $\cket{\mu}$ \ crystal limit of $\ket{\mu}$: \ \eqref{eq:intro_eigenvector_crystal}
	\end{itemize}
	\item $\mu_{q,t}(\vect{z})$ \ Macdonald measure: \ \eqref{eq:Macd_measure}, \eqref{eq:Macd_to_Ruij}
	\begin{itemize}
		\item $\mu_k^\text{nr}(\vect{z})$ its nonrelativistic limit: \ \eqref{eq:psi_0_nonrlt}
	\end{itemize}
\bigskip
	\item $N$ \ number of sites/particles
	\item $[n]$ \ Gaussian integer: \ \eqref{eq:N_q}, \eqref{eq:q-numbers_etc}
	\item $\nu$ \ partition associated to $\mu \in \mathcal{M}_N$: \ \eqref{eq:motifs_vs_partitions}
\bigskip
	\item $\omega = \E^{2\pi\I/N}$, \ primitive $N$th root of unity
	\item $\widetilde{O}$ \ physical operator: \ \textsection\ref{s:intro_phys_space}, end of \textsection\ref{s:physical_space}
	\begin{itemize}
		\item $\widetilde{O}^\circ \coloneqq \widetilde{O}|_{q=1}$ \ its classical limit: \ \eqref{eq:semiclass_expansion}
		\item $\delta \widetilde{O}\coloneqq \partial \widetilde{O}/\partial q|_{q=1}$ \ its semiclassical limit: \ \eqref{eq:semiclass_expansion}
		\item $\widetilde{O}^\text{pol}_{M,M'}$ \ its induced action on polynomials: \ Corollary to Proposition~\ref{prop:intro_physical_vectors}, \eqref{eq:induced_pol}
	\end{itemize}
	\item on shell \ (upon evaluation): \ \textit{see} $\ev$
\bigskip
	\item $p(\mu)$ \ \textsf{q}-momentum (setting the eigenvalue of $G$): \ p.\,\pageref{eq:q-translation}, \eqref{eq:crystal_mtm}
	\item $\mathsf{p} = q$ \ parameter: \ \textit{see} $q$
	\item $P_{i,j} = (1+\vec{\sigma}_i \cdot \vec{\sigma}_j)/2$ \ spin representation of permutation: \ start of \textsection\ref{s:intro_HS}
	\item $P_\lambda(\vect{z}) \coloneqq P_\lambda(\vect{z};q,t)$ \ (symmetric) Macdonald polynomial: \ \eqref{eq:Macd_polyn}, \eqref{eq:Macd_eigenvalues}, \eqref{eq:Macd_orthog}, Figure~\ref{fg:Macdonalds}
	\begin{itemize}
		\item $P_\lambda^{(\alpha)}(\vect{z})$ \ its Jack limit: \ \eqref{eq:HS_polynomial},  \eqref{eq:Jack}
		\item $P_\lambda^\star(\vect{z}) \coloneqq P_\lambda(\vect{z};t,t^2)$ \ its quantum zonal spherical special case: \ \eqref{eq:intro_polynomial}, \textsection\ref{s:intro_examples}, \eqref{eq:wavefn_polynomial}
	\end{itemize}
	\item $p_\lambda(\vect{z}), p_r(\vect{z})$ \ power sum: \ p.\,\pageref{eq:ev_powersum}, \eqref{eq:ev_powersum}
	\item $P^\mu(u)$ \ Drinfeld polynomial: \ \eqref{eq:Drinfeld_polyn}, \eqref{eq:intro_ABCD_on_shell}, \textsection\ref{s:hw}
	\item physical (\textsf{q}-bosonic) space \ $\widetilde{\mathcal{H}} \cong \mathcal{B}$: \ \textsection\ref{s:intro_phys_space}, \textsection\ref{s:physical_space}: \eqref{eq:phys_space}, \eqref{eq:phys_space_via_Ttot}, \eqref{eq:phys_space_alt}
	\begin{itemize}
		\item physical operator: \ \textsection\ref{s:intro_phys_space}, end of \textsection\ref{s:physical_space}
		\item physical vector: \ \eqref{eq:phys_vector_intro}, \eqref{eq:phys_vector} 
	\end{itemize}
	\item pseudo highest weight (\textit{l}-highest weight): \ p.\,\pageref{eq:HS_polynomial}, p.\,\pageref{p:intro_hw_def}, p.\,\pageref{p:hw_def}, \textsection\ref{s:hw}
	\item $\Psi_\nu(i_1,\To,i_M) \coloneqq \cbraket{i_1,\To,i_M}{\mu}$ \ wave function for $\ket{\mu}$: \ \eqref{eq:intro_wavefn}, cf.~\eqref{eq:wavefn_hecke_left}
	\item $\widetilde{\Psi}_\nu(z_1,\To,z_M)$ \ polynomial determining pseudo highest-weight vectors: \ \eqref{eq:intro_polynomial}, \eqref{eq:wavefn_polynomial}
	\item $\ket{\widetilde{\Psi}(\vect{z})}$ \ physical vector: \ \eqref{eq:phys_vector_intro}, \eqref{eq:phys_vector}
\bigskip
	\item $q = \mathsf{p}$ \ parameter: \ start of \textsection\ref{s:setup}
	\begin{itemize}
		\item its physical interpretation: \ start of \textsection\ref{s:intro_abelian_tilde}, start of \textsection\ref{s:intro_freezing}, start of \textsection\ref{s:freezing}, \textsection\ref{s:app_nonrlt_limit}
		\item $q^\circ =1$ \ its classical value: \ \textsection\ref{s:intro_nonabelian_2}, \textsection\ref{s:intro_explicit_evrs_freezing}, \eqref{eq:parameter_shift_upshot}
		\item $\hat{q}_i : z_i \mapsto q\,z_i$ \ difference operator (\textit{q}-shift): \ \eqref{eq:p_hat_graphical}, \eqref{eq:q_hat}, \eqref{eq:q_hat_kernel_off_shell}, \eqref{eq:q_hat_kernel_on_shell}
		\item $q^\star = (t^\star)^{1/2} = t$ \ its quantum zonal spherical value: \ Theorem~\ref{thm:nice_polynomial}, \textsection\ref{s:intro_explicit_evrs_freezing}, Figure~\ref{fg:Macdonalds}, \eqref{eq:wavefn_polynomial}
	\end{itemize}
	\item $\mathsf{q} = t^{1/2}$ \ deformation parameter: \ \textit{see} $t$
	\item \textsf{q}-deformed Dunkl operator: \ \textit{see} $Y_i$
	\item \textsf{q}-momentum (setting the eigenvalue of $G$) \ $p(\mu)$: \  p.\,\pageref{eq:q-translation}, \eqref{eq:crystal_mtm}
	\item \textsf{q}-translation operator \ $G$: \ \eqref{eq:q-translation}, Proposition~\ref{prop:q-translation} 
	\item \textsf{q}-$\mspace{-2mu}$Vandermonde polynomial \ $\Delta_t(\vect{z})$: \ \eqref{eq:intro_polynomial}, \eqref{eq:q-Vand}
	\item $\qdet_a L_a(u)$ \ quantum determinant: \ \eqref{eq:qdet_ABCD},  \eqref{eq:qdet}, \eqref{eq:qdet_Y}, p.\,\pageref{eq:app_baxterisation}
\bigskip
	\item $r_{mj}: z_m \mapsto z_j$ \ replacement map: \ \eqref{eq:rho_mj}
	\item $\check{R}(u)$ \ \textit{R}-matrix (braid-like form, homogeneous gradation): \ \eqref{eq:R_mat}, \eqref{eq:baxterisation}, \eqref{eq:app_baxterisation}
	\begin{itemize}
		\item $R(u) = P\,\check{R}(u)$ \ \textit{R}-matrix: \ \eqref{eq:intro_L_spin_chain}, \eqref{eq:L_inh}, \eqref{eq:app_monodromy_spin},
		\item $\check{R}_w = s_{w^{-1}} \, s^\text{tot}_w $: \ \eqref{eq:Rcheck_w}, \eqref{eq:spin_D_r}, Lemma~\ref{lem:decoupling}, Proof of Proposition~\ref{prop:q-translation}
	\end{itemize}
\bigskip
	\item semiclassical limit \ $\delta \coloneqq \partial/\partial q|_{q=1}$: \ start of \textsection\ref{s:intro_freezing}, \eqref{eq:semiclass_expansion}
	\item $s_i = s_{i,i+1} : z_i \leftrightarrow z_{i+1}$ \ generator of $\mathfrak{S}_N$ in polynomial representation: \ start of \textsection\ref{s:polynomial}
	\item $\sigma^\pm,\sigma^z$ \ Pauli matrices: \ start of \textsection\ref{s:intro_HS}, \eqref{eq:sl_2}
	\begin{itemize}
		\item $s_i^\text{tot} \coloneqq s_i \, \check{R}_{i,i+1}(z_i/z_{i+1})$ \ generator of $\mathfrak{S}_N$ in diagonal representation on $\mathcal{H}[\vect{z}]$: \ \eqref{eq:s_tot} 
	\end{itemize}
	\item $S^\textsc{l}_{[i,j]}, S^\textsc{r}_{[i,j]}$ \ \textsf{q}-deformed exchange interaction: \ \eqref{eq:Sij_left}, \eqref{eq:Sij_left_ex}, \eqref{eq:ham_right}
	\item $s_\lambda$ \ Schur polynomial: \ \eqref{eq:intro_pol_crystal} , \eqref{eq:Schur} 
	\item simple component \ $\cbraket{1,\To,M}{\Psi}$: \ \eqref{eq:phys_space_M_particle}, \textsection\ref{s:intro_phys_space}
	\item $\mathfrak{S}_N$ \ symmetric group: \ start of \textsection\ref{s:Hecke_AHA}
	\item $S^\pm, S^z$ \ (spin representation of) generators of $\mathfrak{sl}_2$: \ \eqref{eq:sl_2}
\bigskip
	\item $t = \mathsf{q}^2$ \ deformation parameter: \ start of \textsection\ref{s:setup}
	\begin{itemize}
		\item its physical interpretation: \ start of \textsection\ref{s:intro_abelian_tilde}, start of \textsection\ref{s:intro_freezing}, start of \textsection\ref{s:freezing}, \textsection\ref{s:app_nonrlt_limit}
		\item $t^\star = (q^\star)^2 = t^2$ \ its quantum zonal spherical value: \ Theorem~\ref{thm:nice_polynomial}, \textsection\ref{s:intro_explicit_evrs_freezing}, Figure~\ref{fg:Macdonalds}, \eqref{eq:wavefn_polynomial}
	\end{itemize}
	\item $T_i$ \ Hecke generator: \ \textsection\ref{s:Hecke_TL_Uqsl}
	\begin{itemize}
		\item $T_i^\text{pol}$ \ its polynomial (Demazure--Lusztig) representation: \ \eqref{eq:intro_Hecke_pol}, \eqref{eq:Hecke_pol_2a}, \eqref{eq:Hecke_pol_2b}
		\item $T_i^\text{sp}$ \ its spin representation: \ \eqref{eq:Hecke_spin}, \eqref{eq:app_Hecke_spin}, \eqref{eq:ham_formal_limit}
		\item $T_i^\text{tot}$ \ its diagonal representation (on $\mathcal{H}[\vect{z}]$): \ \eqref{eq:Hecke_tot}		
	\end{itemize}
	\item $\mathfrak{T}_N(\beta)$ \ Temperley--Lieb algebra: \ \eqref{eq:TL}
\bigskip
	\item $\mathfrak{U} \coloneqq U_{\mathsf{q}}(\mathfrak{sl}_2)$ \ quantum $\mathfrak{sl}_2$: \ \eqref{eq:Uqsl2}
	\item $\widehat{\mathfrak{U}} \coloneqq U_{\mathsf{q}}(L\,\mathfrak{gl}_2) = U'_{\mathsf{q}}(\widehat{\vphantom{t}\smash{\mathfrak{gl}_2}})_{c=0}$ \ quantum-loop algebra of $\mathfrak{gl}_2$: \  \textsection\ref{s:intro_nonabelian_1}, \textsection\ref{s:intro_nonabelian_2}, \textsection\ref{s:quantum-affine}
\bigskip
	\item $V = \mathbb{C}\,\ket{\uparrow} \oplus \mathbb{C}\,\ket{\downarrow}$ \ one-site Hilbert space: \ start of \textsection\ref{s:intro_HS}, start of \textsection\ref{s:spin}
	\item $V\mspace{-1mu}(z_i,z_j)$ \ potential: \ \eqref{eq:pot}, Figure~\ref{fg:pot}, proof on p.\,\pageref{p:V_computation}, \eqref{eq:pot_origin}
\bigskip
	\item $w \in \mathfrak{S}_N$ \ permutation
\bigskip
	\item $x_{ij}$ \ triangular building blocks of the $Y_i$: \ \eqref{eq:intro_Yi_circ}, \eqref{eq:x_ij}
\bigskip
	\item $Y_i$ \ \textsf{q}-deformed Dunkl operator ($\!$\textit{Y}-operator): \ \eqref{eq:intro_Y-ops}, \eqref{eq:Yi}, \eqref{eq:Yi_via_x}
	\begin{itemize}
		\item its eigenvalues: \ \eqref{eq:nonsymm_Macdonalds}
		\item $Y_i^\circ$ \ its classical limit (no difference part): \ \eqref{eq:intro_Yi_circ}, \textsection\ref{s:explicit_evrs} 
	\end{itemize}
\bigskip
	\item $\vect{z}^\alpha = z_1^{\alpha_1} \cdots z_N^{\alpha_N}$ \ monomial: \ p.\,\pageref{eq:largest_monomial}, p.\,\pageref{eq:Macdonald_shift_property}, p.\,\pageref{eq:nonsymm_Macdonalds}
	\begin{itemize}
		\item its (dominance) ordering: \ \eqref{eq:dominance}
	\end{itemize}
	\item $\hat{0}_i$ \ annihilator: \ \textsection\ref{s:hw}
\end{itemize}

\end{document}